%% file: aos_LowrankLLMEvaluation.tex
\documentclass[aos,preprint]{imsart}

\RequirePackage{amsthm,amsmath,amsfonts,amssymb,mathtools,bm,bbm}
\RequirePackage[authoryear]{natbib}
\RequirePackage[colorlinks,citecolor=blue,urlcolor=blue,linkcolor=blue]{hyperref}
\RequirePackage{graphicx}
\RequirePackage{booktabs}
\RequirePackage{multirow}
\RequirePackage{enumitem}
\RequirePackage{url}
\usepackage{algorithm}

\startlocaldefs

\numberwithin{equation}{section}

\theoremstyle{plain}
\newtheorem{theorem}{Theorem}[section]
\newtheorem{lemma}[theorem]{Lemma}
\newtheorem{proposition}[theorem]{Proposition}
\newtheorem{corollary}[theorem]{Corollary}

\newtheorem{remark}[theorem]{Remark}
\newtheorem{assumption}[theorem]{Assumption}

\newtheorem{example}[theorem]{Example}

\newcommand{\E}{\mathbb{E}}

\newcommand{\Pn}{\mathbb{P}_n}
\newcommand{\Pstar}{\mathbb{P}^\star}
\newcommand{\R}{\mathbb{R}}
\newcommand{\Hcal}{\mathcal{H}}

\newcommand{\ip}[2]{\left\langle #1,#2\right\rangle}

\newcommand{\norm}[1]{\left\lVert #1\right\rVert}
\newcommand{\Fnorm}[1]{\left\lVert #1\right\rVert_{F}}
\newcommand{\infnorm}[1]{\left\lVert #1\right\rVert_{\infty}}
\newcommand{\oneNorm}[1]{\left\lVert #1\right\rVert_1}
\newcommand{\twoninfnorm}[1]{\left\lVert #1\right\rVert_{2,\infty}}
\newcommand{\inftyinfnorm}[1]{\left\lVert #1\right\rVert_{\infty\to\infty}}
\newcommand{\psione}{\psi_1}

\newcommand{\one}{\mathbf 1}
\newcommand{\Id}{I}
\newcommand{\dstar}{d^\star}
\newcommand{\Gammavec}{\Gamma}

\newcommand{\mode}[2]{\times_{#1} #2}
\newcommand{\PU}[1]{P_{U_{#1}}}
\newcommand{\PUhat}[1]{P_{\widehat U_{#1}}}
\newcommand{\Var}{\mathrm{Var}}

\endlocaldefs

\begin{document}

\begin{frontmatter}


\title{LLM Evaluation as Tensor Completion: Low-Rank Structure and Semiparametric Efficiency
}

\begin{aug}
\author[A]{\fnms{Jiachun} \snm{Li}}\ead{jiach334@mit.edu}
\author[A]{\fnms{David} \snm{Simchi-Levi}}\ead{dslevi@mit.edu}
\author[B]{\fnms{Will Wei} \snm{Sun}}\ead{sun244@purdue.edu}
\runauthor{Li, Simchi-Levi and Sun}
\address[A]{Massachusetts Institute of Technology}
\address[B]{Purdue University}
\end{aug}

\begin{abstract}
Large language model (LLM) evaluation platforms increasingly rely on pairwise human judgments, where a user compares two model responses and selects the preferred one. These data are noisy, sparse, and highly non-uniform, yet the resulting leaderboards are typically reported with limited uncertainty quantification. We study this problem as semiparametric inference for a low-rank latent score tensor observed through pairwise comparisons. Specifically, we model preference outcomes using Bradley-Terry-Luce-type comparison models and represent model performance across task categories and other contexts by a latent score tensor $T^\star$. This formulation places modern LLM evaluation in a new low-rank tensor completion setting, where the observation pattern is highly structured, the sampling is non-uniform, and each sample reveals only a pairwise contrast rather than a direct tensor entry. Our inference target is a smooth functional $\psi(T^\star)$, including both linear estimands such as ability gaps and category-specific contrasts and nonlinear estimands such as average win probabilities against a reference pool. We derive the information operator on the identifiable low-rank tangent space, the associated efficient influence function, and the resulting semiparametric efficiency bound. We then construct a one-step debiased estimator with asymptotical normality. A central technical challenge is that, under heterogeneous matchups and non-uniform sampling, the information operator is anisotropic and generally does not commute with projection onto the low-rank tangent space, creating a bottleneck absent from isotropic tensor completion models. To address this, we introduce a score-whitening method that equalizes local Fisher information across comparisons and restores stable inference at the optimal sample-complexity scale. Our results provide a principled framework for uncertainty quantification in LLM evaluation and more broadly for inference on low-rank latent structures from pairwise comparison data.
\end{abstract}



\end{frontmatter}

\section{Introduction}\label{sec:intro}

Large language model (LLM) evaluation has rapidly become a statistical problem of central importance. Modern foundation models are deployed across a wide range of tasks, updated frequently, and often differ only subtly in quality. As a result, the practical question is no longer merely how to obtain a point estimate of model performance, but how to compare models reliably, quantify uncertainty in those comparisons, and do so under highly irregular observation patterns. Human-preference platforms such as \textsc{Arena} have emerged as a central solution in practice: a user submits a prompt, two anonymous models generate responses side by side, and the user votes for the better response; the resulting pairwise outcomes are then aggregated into a public leaderboard \citep{chiang2024chatbot}. Figure~\ref{fig:lmarena-battle} illustrates this pipeline. This approach has become widely popular because it evaluates LLM models with real prompts and real users, rather than through a fixed benchmark alone.

\begin{figure}[t]
\centering
\begin{minipage}[c]{0.65\textwidth}
    \centering
    \includegraphics[width=\textwidth]{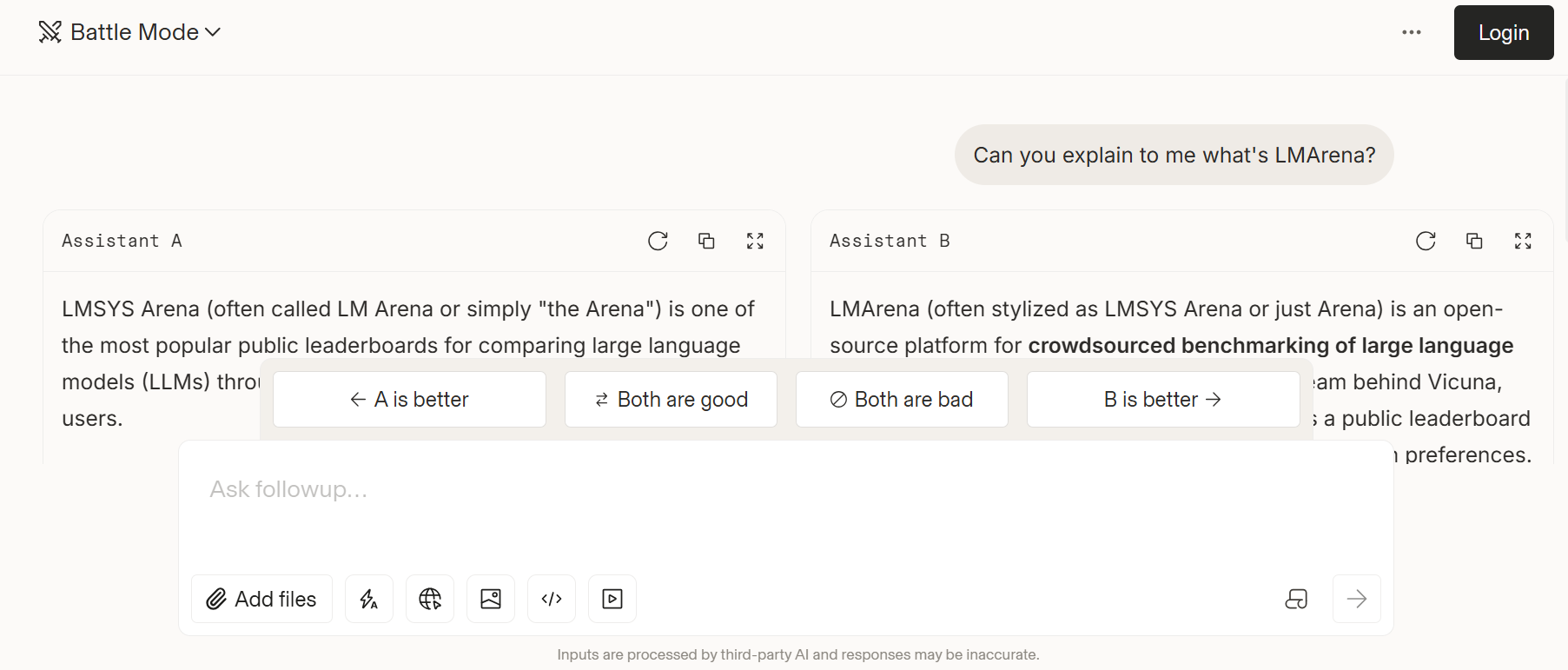}
\end{minipage}
\hspace{0.01\textwidth}
\begin{minipage}[c]{0.01\textwidth}
    \centering
    $\Rightarrow$
\end{minipage}
\hspace{0.01\textwidth}
\begin{minipage}[c]{0.26\textwidth}
    \centering
    \includegraphics[width=\textwidth]{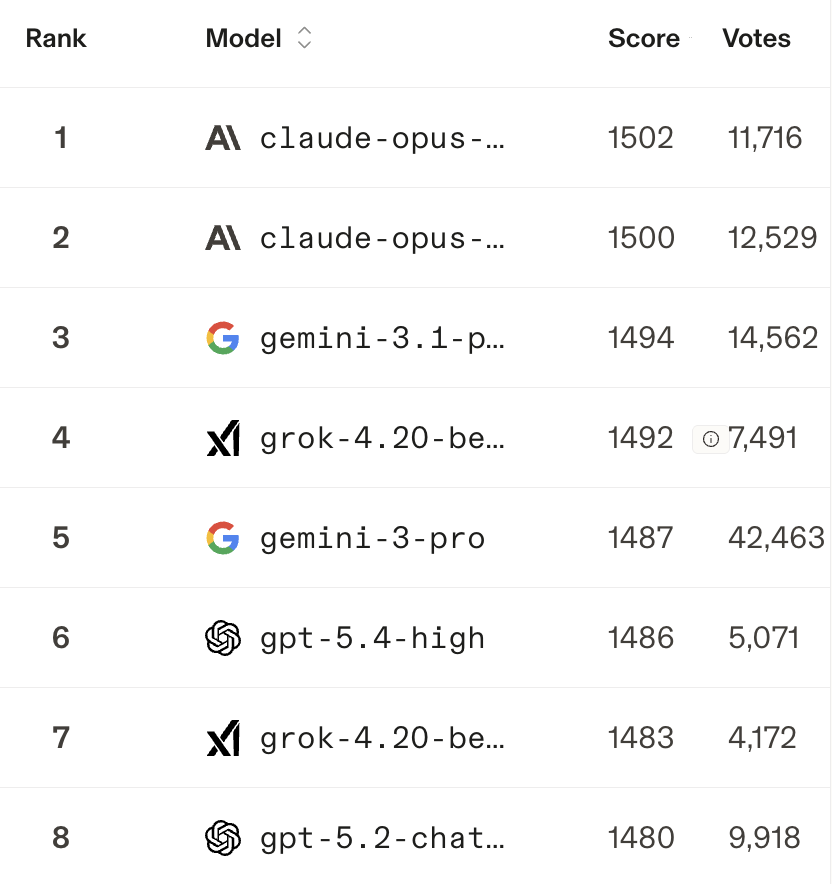}
\end{minipage}
\caption{Left plot is the \textsc{Arena} battle-mode interface. A user submits a prompt, two anonymous LLMs respond side by side, and the user votes on which response is preferred. The right plot is a category-specific \textsc{Arena} leaderboard (Text category). Scores are obtained by aggregating pairwise comparison outcomes, while the vote counts vary substantially across models. This imbalance reflects non-uniform sampling across model pairs.}
\label{fig:lmarena-battle}
\end{figure}

The statistical object underlying such leaderboards is naturally a collection of latent pairwise-comparison strengths.  
A standard starting point is the Bradley-Terry-Luce (BTL) model, under which the probability that model $a$ beats model $b$ depends on the difference of their latent scores \citep{bradley1952rank,luce1959individual}.  
When evaluation is broken down by task category, prompt type, or user segment, the latent score is no longer a single scalar per model but rather a structured array.  
For example, if $j\in[d_1]$ indexes models and $t\in[d_2]$ indexes task categories, one may posit a latent matrix $T^\star\in\mathbb{R}^{d_1\times d_2}$, where $T^\star_{j,t}$ measures model $j$'s latent ability on category $t$.  
More generally, additional modes can be added for user populations, evaluation criteria, or languages, leading to a latent tensor
$
T^\star \in \mathbb{R}^{d_1\times \cdots \times d_m}.
$
A single battle then reveals only a noisy comparison between two entries of this latent object: if models $a$ and $b$ are compared, one observes a binary outcome
\[
Y \in \{0,1\}, \quad 
\mathbb{P}(Y=1\mid X,T^\star)=\sigma\!\big(\langle T^\star,X\rangle\big),
\]
where $\sigma(z)=1/(1+e^{-z})$ and $X$ is a design tensor encoding the relevant pairwise contrast.

This viewpoint connects modern LLM evaluation to \emph{matrix and tensor completion with structured missingness}, with the important difference that the data consist of pairwise contrasts rather than direct entries. The latent score array $T^\star$ is the scientific object of interest, but only a small and highly uneven subset of pairwise comparisons is observed. Popular frontier models receive many more battles than weaker or newly released models, and prompt traffic is driven by user interests rather than experimental design. Thus the observation pattern is neither balanced nor uniform.

The inferential goals in LLM evaluation go well beyond point ranking. In many applications, one seeks to estimate a specific functional of the latent score tensor, such as an ability gap between two models, a category-specific average score, or an average win probability against a reference pool. A generic linear target takes the form
\[
\psi(T^\star)=\langle \Gamma,T^\star\rangle,
\]
where \(\Gamma\) selects or averages entries of \(T^\star\). In addition, many practical targets are nonlinear. For a fixed category \(t\) and a fixed pair of models \((a,b)\), one example is the win probability
\begin{equation}\label{eq:nonlinear-winrate-intro}
\psi(T^\star)
=
\sigma\!\big(T^\star_{a,t}-T^\star_{b,t}\big),
\end{equation}
which gives the probability that model \(a\) is preferred to model \(b\) on category \(t\).
Other examples include category-weighted leaderboard scores, robust summaries such as tail performance, and calibrated transformations of latent scores into public-facing rating scales. This leads to the main question of the paper:

\medskip
\emph{How can we construct statistically efficient estimators and valid confidence intervals for linear and nonlinear functionals of a latent low-rank score tensor when the data consist of noisy pairwise comparisons collected under highly non-uniform sampling?}
\medskip

Current leaderboard practice does not address this question. Existing platforms focus primarily on ranking models and reporting model-level uncertainty summaries. For example, \textsc{Arena} uses Bradley--Terry-type ranking models and, more recently, the open-source \texttt{Arena-Rank} package, which implements contextual ranking extensions, reweighting for underrepresented models, and confidence intervals based on asymptotic normality of $M$-estimators \citep{chiang2024chatbot,arena2025rank}. These methods are useful in practice, but they do not explicitly exploit low-rank structure across models, tasks, and related evaluation dimensions, nor do they address semiparametric efficient inference for structured functionals under identification constraints and non-uniform sampling. 

Low-rank structure provides a natural way to make this problem statistically tractable. In LLM evaluation, it is unrealistic to treat every model-by-task performance cell as unrelated. Rather, performance is typically driven by a small number of latent factors, such as reasoning ability, coding proficiency, instruction following, multilingual robustness, or style sensitivity. A low-rank matrix or tensor model captures this shared structure, reduces effective dimension, and allows information to be borrowed across models and tasks. This is the same principle that underlies low-rank methods in matrix completion, collaborative filtering, and tensor estimation \citep{candes2009exact,keshavan2010matrix,negahban2012restricted,koltchinskii2011nuclear,cai2022nonconvex,zhang2025generalized}. In the present setting, it turns a collection of sparse pairwise battles into a coherent structured inference problem.

Meanwhile, LLM evaluation is substantially harder than standard completion problems. We observe binary pairwise comparisons rather than entries of \(T^\star\), so the model is only identifiable up to additive shifts and requires normalization. The observation model is logistic rather than Gaussian, which makes the information carried by each comparison depend on the local Fisher information \(I(\eta)=\sigma(\eta)\{1-\sigma(\eta)\}\): closely matched models are informative, while lopsided comparisons are not. The sampling design is also highly uneven across prompts, categories, and model pairs, leading to anisotropic information geometry. Finally, the inferential target is typically not full recovery of \(T^\star\), but efficient estimation of a functional \(\psi(T^\star)\). These place the problem naturally in the realm of semiparametric inference for structured high-dimensional parameters and motivate the framework developed in this paper.

\subsection{Our contributions}

We study semiparametric inference for low-rank latent score tensors observed through pairwise comparisons. Our contributions are fourfold.

\noindent\textbf{A statistical formulation of low-rank LLM evaluation.}
We formulate leaderboard inference as estimation of a smooth functional $\psi(T^\star)$ of a low-rank latent score tensor $T^\star$ under generalized pairwise-comparison models. This unifies linear targets, such as ability gaps and category-specific scores, and nonlinear targets, such as average win probabilities, within a single framework. A key modeling point is that the observed battles correspond to a highly incomplete and non-uniformly sampled collection of pairwise contrasts, so the problem is naturally connected to tensor completion under general missingness mechanisms.

\noindent\textbf{Semiparametric efficiency bounds via an information equation.}
For smooth targets $\psi(T^\star)$, we derive the efficient influence function (EIF) and the corresponding semiparametric efficiency bound.  
For a linear functional $\psi(T)=\langle \Gamma,T\rangle$, the EIF direction is characterized by the operator equation
$
A H^\star = P_{\mathbb{T}}\Gamma,
$
$A=P_{\mathbb{T}}GP_{\mathbb{T}},
$
where $\mathbb{T}$ is the tangent space of the low-rank model under the identification constraint, $P_{\mathbb{T}}$ is the tangent-space projection, and $G$ is the Fisher-information operator induced by the pairwise-comparison model and the sampling distribution.  
The resulting variance bound
$
V_{\mathrm{eff}}(\psi)
=
\big\langle P_{\mathbb{T}}\Gamma,\,
A^{-1}P_{\mathbb{T}}\Gamma
\big\rangle
$
plays the role of a generalized Cram\'er--Rao bound for structured comparison-based inference.  
To our knowledge, this is the first such efficiency characterization tailored to low-rank pairwise-comparison models with identification constraints and general non-uniform sampling.

\noindent\textbf{One-step inference and a new bottleneck.}
We construct a one-step debiased estimator based on the efficient influence function. Its analysis reveals a new difficulty: under heterogeneous pairwise information, the Fisher-information operator $G$ generally does not commute with the tangent-space projector $P_{\mathbb{T}}$. As a result, the restricted inverse $(P_{\mathbb{T}}GP_{\mathbb{T}})^{-1}$ may exhibit unfavorable entrywise amplification, creating a dimension-dependent bottleneck for standard debiasing arguments. This issue is absent in isotropic completion settings where the information operator is proportional to the identity.

\noindent\textbf{Score whitening, inverse-probability weighting, and nonlinear targets.}
To overcome this, we introduce a score-whitening method that normalizes the score by the local Fisher information and restores an isotropic effective information operator on the tangent space. This can be viewed as a semiparametric preconditioning. The resulting estimator attains asymptotic normality at the optimal sample-complexity scale even when the original information geometry is highly heterogeneous. We further combine this idea with inverse-probability weighting to handle non-uniform sampling and extend
the analysis from linear to nonlinear targets.

\noindent\textbf{Broader implications for tensor completion.}
Although our main focus is LLM evaluation, the framework also applies more broadly to low-rank inference problems with heterogeneous information geometry, including tensor completion with heteroscedastic noise, non-uniform sampling, and 1-bit observations. We develop these detailed results in Appendix~\ref{sec:applications}.

\subsection{Closely related literature}

Our work is related to two fields: low-rank matrix and tensor inference, and statistical modeling of LLM evaluation through pairwise comparisons.

\noindent\textbf{Low-rank matrix and tensor completion.} Classical work studies recovery of a low-rank matrix from uniformly sampled entries \citep{candes2009exact,keshavan2010matrix,negahban2012restricted,koltchinskii2011nuclear}. A parallel tensor literature develops analogous methods for low-rank tensors under structured decompositions such as CP and Tucker models \citep{kolda2009tensor,cai2022nonconvex,ma2024statistical}. More recent work considers non-uniform observation schemes, including covariate-assisted weighting, inverse-propensity correction, and missing-not-at-random mechanisms in matrix and tensor completion \citep{ma2019covariate,yang2021matrix,chao2021weighted,zhang2025generalized,duan2025statistical}. Our setting differs in two key respects: the data consist of pairwise contrasts rather than direct entries, and the goal is not only recovery of the latent low-rank object, but semiparametric efficient inference for smooth functionals under pairwise observations and identification constraints.

\noindent\textbf{LLM evaluation and BTL models.}
A rapidly growing empirical literature studies LLM evaluation through human preferences, pairwise battles, benchmark design, and AI judges \citep{chiang2024chatbot,li2024benchbuilder,petrova2026unpacking, dong2026evaluating}.  
The \textsc{Arena} line of work established pairwise crowdsourced evaluation as a credible public ranking mechanism \citep{chiang2024chatbot, li2024benchbuilder}.  
At the same time, recent studies point to vulnerabilities and heterogeneities in leaderboard data, including demographic variation, benchmark misalignment, and the sensitivity of rankings to data-collection choices \citep{singh2025leaderboard,petrova2026unpacking}. This empirical literature makes the statistical need for principled uncertainty quantification especially clear.  

Recent theory has made substantial progress on estimation and uncertainty quantification in BTL-type models for pairwise preference data \citep{gao2023btluq,fan2024care,fan2025uncertainty,fan2026spectral}. Closely related ideas also appear in LLM alignment, where BTL-type preference models underlie modern preference-optimization methods and their robust variants \citep{ouyang2022training,bose2025lore, xu2025doubly, su2026large}. Related semiparametric ideas have also begun to appear in LLM evaluation; for example, \citet{dong2026evaluating} use auxiliary pairwise comparison signals to improve efficiency in mathematical reasoning evaluation.  
Our work is complementary but distinct. We study direct pairwise-comparison data generated from a latent low-rank score tensor and develop semiparametric efficiency theory for structured linear and nonlinear functionals of that tensor. Once task categories, user groups, or other contextual dimensions are introduced, the inferential target becomes structured and high dimensional, and low-rank regularity becomes essential for statistical efficiency. In this sense, our paper provides an efficient inferential framework tailored to the sparse, structured, and non-uniform pairwise-comparison data that arise in modern LLM evaluation.


\subsection{Organization}

The paper is organized as follows. Section~\ref{sec:model-setting} introduces the low-rank pairwise-comparison model for LLM evaluation. Section~\ref{sec:semiparametric_formulation} develops the semiparametric efficiency theory, including the information equation, efficient influence function, and efficiency bound. Section~\ref{sec:upper-bound} studies the one-step estimator and the non-commutativity bottleneck under heterogeneous information. Section~\ref{sec:score-whitening} presents score whitening, inverse-probability weighting, and the extension to nonlinear functionals. Section~\ref{sec:simulations} reports synthetic and real-data experiments. Proofs and additional technical details are deferred to the appendices.


\section{Model setting}\label{sec:model-setting}

We now formalize the statistical model underlying pairwise evaluation platforms such as \textsc{Arena}; see Figure~\ref{fig:lmarena-battle}. In a typical interaction, 
each observation is not a direct measurement of a model's latent quality, but a \emph{pairwise comparison} between two models under a particular evaluation context such as task category, evaluation criterion, language, or time period. 
We begin with the simplest and most concrete formulation, which will serve as the primary running example throughout the paper.

\begin{example}[Model-topic pairwise comparisons via BTL model]\label{ex:pairwise}
Suppose a platform evaluates $d_1$ LLM models across $d_2$ task categories, such as coding, mathematics, creative writing, and multilingual tasks. Let
\[
T^\star \in \mathbb{R}^{d_1\times d_2}
\]
be a latent score matrix, where $T^\star_{j,u}$ represents the latent ability of model $j$ on task category $u$. At evaluation round $i$, the platform observes a task category $u_i$, selects two models $p_i$ and $q_i$, and records a binary outcome
$
Y^{(i)} \in \{0,1\},
$
where $Y^{(i)}=1$ means that model $p_i$ is preferred to model $q_i$ on category $u_i$. The key point is that the outcome depends on the latent scores only through their difference,
$
\eta_i = T^\star_{p_i,u_i}-T^\star_{q_i,u_i}.
$
Under the Bradley-Terry-Luce (BTL) model,
\begin{equation}\label{eq:btl-model-setting}
\mathbb{P}\bigl(Y^{(i)}=1 \mid p_i,q_i,u_i,T^\star\bigr)
=
\sigma(\eta_i)
=
\frac{\exp(\eta_i)}{1+\exp(\eta_i)}.
\end{equation}
Equivalently, if $e_j$ denotes the $j$th standard basis vector, we may write the design matrix as
\[
X^{(i)}=(e_{p_i}-e_{q_i})e_{u_i}^\top,
\]
so that
$
\eta_i=\langle T^\star,X^{(i)}\rangle.
$
Thus each battle reveals one noisy contrast of the latent score matrix rather than one matrix entry.
\end{example}

This formulation is directly connected to the Arena example in Figure~\ref{fig:lmarena-battle}. The left panel corresponds to the data-collection mechanism and the right panel corresponds to the statistical aggregation step to estimate latent model strengths and produce a leaderboard. In the matrix formulation above, the leaderboard is driven by the latent matrix $T^\star$, while each observed battle reveals only one binary comparison involving two of its entries.

A first important feature of the BTL model is \emph{identifiability}. Because only score differences affect the feedback, the latent matrix is not identifiable without normalization. Indeed, for any vector $c\in\mathbb{R}^{d_2}$, replacing $T^\star$ by
$
T^\star + \mathbf{1}_{d_1} c^\top
$
does not change any difference $T^\star_{p,u}-T^\star_{q,u}$ and therefore leaves all pairwise probabilities unchanged. To fix a unique representative in each equivalence class, we impose a linear identification constraint. The canonical choice is the zero-sum normalization along the model dimension,
\begin{equation}\label{eq:row-sum-zero-model-setting}
\mathbf{1}_{d_1}^\top T^\star = 0,
\end{equation}
meaning that, within each task category, the model scores sum to zero. This constraint will play an essential role in the semiparametric formulation developed later.

The BTL model in \eqref{eq:btl-model-setting} is only one instance of a broader pairwise-feedback framework. More generally, let the latent signal live in the space of real matrices or tensors $T^\star\in\mathbb{R}^{d_1\times\cdots\times d_m}$, equipped with the trace inner product $\langle U,V\rangle = \sum_{j_1,\dots,j_m} U_{j_1\dots j_m} V_{j_1\dots j_m}$. We observe independent pairs
\[
(X^{(i)},Y^{(i)})\in \mathbb{R}^{d_1\times\cdots\times d_m}\times\mathcal{Y}, \qquad i=1,\dots,n,
\]
where the design $X^{(i)}$ is drawn from a distribution $\Pi^\star$ and, conditional on $X^{(i)}$, the observation $Y^{(i)}$ depends on $T^\star$ only through the scalar index
$
\eta_i=\langle T^\star,X^{(i)}\rangle.
$
That is, the conditional log-likelihood takes the single-index form
\begin{equation}\label{eq:general-glm-setting}
\log p\bigl(Y^{(i)}\mid X^{(i)},T^\star\bigr)
=
\ell\bigl(Y^{(i)},\langle T^\star,X^{(i)}\rangle\bigr)
\end{equation}
for a known function $\ell:\mathcal{Y}\times\mathbb{R}\to\mathbb{R}$. The corresponding empirical loss is
\[
\mathcal{L}_n(T)=\frac{1}{n}\sum_{i=1}^n \phi\bigl(Y^{(i)},\langle T,X^{(i)}\rangle\bigr),
\]
where $\phi$ is the negative log-likelihood loss associated with $\ell$. The BTL model is the special case in which $\mathcal{Y}=\{0,1\}$ and $\ell(y,\eta)=y\eta-\log(1+e^\eta)$.

In realistic LLM evaluation, performance may vary simultaneously across models, task categories, user groups, languages, evaluation dimensions, or time. This naturally leads to a general tensor formulation. Let
\[
T^\star \in \mathbb{R}^{d_1\times \cdots \times d_m}
\]
be an order-$m$ latent score tensor. The first mode will typically index models, while the remaining modes index contextual variables. An observed comparison again depends on $T^\star$ only through a linear contrast $\langle T^\star,X\rangle$, where $X$ is a signed design tensor encoding which two models are compared and under which context. As in the matrix case, only certain contrasts of $T^\star$ are observable, so the tensor is identifiable only after imposing suitable linear constraints. We therefore assume that $T^\star$ satisfies
\begin{equation}\label{eq:general-linear-constraint-model-setting}
\mathcal{C}(T^\star)=0,
\end{equation}
where $\mathcal{C}$ is a known linear operator. In the pairwise-comparison setting, the most important example is again a zero-sum constraint along the model mode, but the formulation in \eqref{eq:general-linear-constraint-model-setting} allows for more general identification restrictions.

To make inference feasible in high dimensions, we assume that $T^\star$ has low-rank Tucker structure \citep{kolda2009tensor}
\begin{equation}\label{eq:tucker-model-setting}
T^\star = \mathcal{C}^\star \times_1 U_1 \times_2 U_2 \cdots \times_m U_m,
\end{equation}
where $U_k\in\mathbb{R}^{d_k\times r_k}$ has orthonormal columns and $\mathcal{C}^\star\in\mathbb{R}^{r_1\times\cdots\times r_m}$ is a core tensor. In the matrix case $m=2$, this reduces to the singular value decomposition. The low-rank assumption indicates that LLM's model performance is driven by a relatively small number of latent factors such as reasoning ability, coding proficiency, instruction following, multilingual robustness, or stylistic preference. Thus low rank is not merely a computational regularizer; it is a structural assumption that allows the model to borrow strength across tasks and contexts. 

We further assume the standard $\mu$-incoherence condition on $T^\star$, which ensures that the singular vectors are not aligned with any single coordinate direction: for each mode $j\in[m]$,
\begin{equation}\label{eq:incoh}
\max_{a\in[d_j]}\|e_a^\top U_j\|_2^2 \;\le\; \frac{\mu\, r_j}{d_j},
\end{equation}
for some incoherence parameter $\mu\ge 1$.

We now collect several dimensional quantities and structural parameters that will be used throughout the paper. Define
\begin{equation}\label{eq:dim-notation}
\bar d := \max_{j\in[m]} d_j, \qquad
d^\star := \prod_{j=1}^m d_j, \qquad
D := \prod_{k=2}^m d_k,
\end{equation}
so that $\bar d$ is the largest mode dimension, $d^\star$ is the total number of tensor entries, and $D$ is the number of contexts (the product of all non-comparison mode sizes). In the matrix case $m=2$, we have $D=d_2$, $\bar d = \max(d_1,d_2)$, and $d^\star = d_1 d_2$. The rank tuple is $r=(r_1,\ldots,r_m)$ and we write $r^\star = \prod_{j=1}^m r_j$ for the core tensor size.

Let $\mathcal{C}_{(k)}^\star\in\mathbb{R}^{r_k\times \prod_{j\neq k} r_j}$ denote the mode-$k$ unfolding of the core tensor $\mathcal{C}^\star$. Following \citet{ma2024statistical}, we define the maximum and minimum singular values across all mode unfoldings as
\begin{equation}\label{eq:lambda-max-def}
\lambda_{\max} := \max_{k\in[m]} \|\mathcal{C}_{(k)}^\star\|_{\mathrm{op}},
\qquad
\lambda_{\min} := \min_{k\in[m]} \sigma_{r_k}(\mathcal{C}_{(k)}^\star),
\end{equation}
and the condition number of $T^\star$ is
\begin{equation}\label{eq:kappa-def}
\kappa := \frac{\lambda_{\max}}{\lambda_{\min}}.
\end{equation}
Since the factor matrices $U_k$ have orthonormal columns, the Frobenius norm satisfies $\|T^\star\|_F = \|\mathcal{C}^\star\|_F$, and the singular values of each mode-$k$ unfolding of $T^\star$ coincide with those of $\mathcal{C}_{(k)}^\star$. The Frobenius norm is therefore bounded by
\begin{equation}\label{eq:Fnorm-lambda}
\sqrt{r^\star}\,\lambda_{\min} \;\le\; \|T^\star\|_F \;\le\; \sqrt{r^\star}\,\lambda_{\max},
\end{equation}
so that $\|T^\star\|_F \asymp \lambda_{\max} \asymp \kappa\,\lambda_{\min}$ when $r$ and $\kappa$ are treated as bounded constants.

Under the pairwise-comparison observation model, the noise variance for each comparison is determined by the Fisher information $I(\eta)=\sigma(\eta)(1-\sigma(\eta))$, which depends on the local ability gap $\eta=T^\star_{p,u}-T^\star_{q,u}$. We assume throughout that $T^\star$ has bounded entries:
\begin{equation}\label{eq:Tstar-bounded}
\|T^\star\|_\infty \le B
\end{equation}
for a constant $B>0$. This is natural in the LLM evaluation context, where latent ability gaps are bounded. Under \eqref{eq:Tstar-bounded}, the pairwise ability gaps satisfy $|\eta|\le 2B$, and hence the Fisher information is bounded: $c_B \le I(\eta) \le C_B$ for all observed $\eta$, where $c_B,C_B>0$ depend only on $B$. That is, the noise level is heteroscedastic but uniformly of constant order. In particular, no explicit noise parameter $\sigma$ needs to appear in the error bounds. For simplicity, we also assume that the target functional $\psi$ has bounded support:

\begin{assumption}[Bounded support]\label{ass:bounded-support}
The functional gradient $\nabla\psi(T^\star)=\Gamma$ has support size $|\mathrm{supp}(\Gamma)| \le M$ for a constant $M$, and bounded $\ell_1$ norm: $\|\Gamma\|_1\le C_\psi$.
\end{assumption}

Our inferential goal is not necessarily full recovery of the entire latent matrix or tensor. In many applications, one is interested in a lower-dimensional summary of $T^\star$. As we mentioned in Section \ref{sec:intro}, the most basic class consists of linear functionals
\begin{equation}\label{eq:linear-functional-model-setting}
\psi(T^\star)=\langle \Gamma,T^\star\rangle,
\end{equation}
where $\Gamma\in\mathbb{R}^{d_1\times\cdots\times d_m}$ is a known direction. This class includes individual ability scores, pairwise contrasts between two models in a fixed category, and averages over subsets of contexts. For example, if one wishes to compare models $a$ and $b$ on task category $u$, one may choose $\Gamma=(e_a-e_b)e_u^\top$ in the matrix case, so that
\[
\psi(T^\star)=T^\star_{a,u}-T^\star_{b,u}.
\]
More generally, if $\psi$ is Gateaux differentiable, one may also consider nonlinear targets such as the win probability in \eqref{eq:nonlinear-winrate-intro}, calibrated leaderboard scores, or robust performance summaries. In this paper, linear functionals will serve as the main starting point of our theory, while nonlinear functionals will be handled later through their local linearization and influence-function representation.

In summary, we study a low-rank latent score matrix or tensor \(T^\star\), subject to linear identification constraints and observed through noisy pairwise comparisons. This formulation captures the basic structure of \textsc{Arena}-style evaluation and provides a unified framework for pairwise-comparison modeling, low-rank estimation, and semiparametric inference.

\section{Semiparametric efficiency: score, information, and the efficiency bound}\label{sec:semiparametric_formulation}

We now study the fundamental limit for estimating a target functional from pairwise comparison data. Since the target is a low-dimensional functional of a high-dimensional latent score matrix or tensor, with the sampling distribution \(\Pi^\star\) acting as a nuisance parameter, the problem is naturally semiparametric. The key objects are the score, the information operator, and the efficient influence function, which together determine the semiparametric efficiency bound.


\subsection{Low-rank tangent space and identifiability}\label{sec:lb-notation}

A first step in semiparametric analysis is to identify the local perturbation directions that are statistically meaningful. In our setting, not every perturbation of $T^\star$ is admissible. There are two reasons.

First, $T^\star$ is assumed to be low rank. Let $\mathbb{M}_r$ denote the manifold of matrices (or tensors) of rank at most $r$. Local perturbations should remain within this manifold, at least to first order, which leads to the tangent space of $\mathbb{M}_r$ at $T^\star$, denoted by $\mathrm{Tan}_{T^\star}(\mathbb{M}_r)$. Informally, this is the linear space of first-order perturbations that preserve the low-rank structure. Second, under pairwise comparisons, $T^\star$ is identifiable only after normalization. As discussed in Section~\ref{sec:model-setting}, adding a common constant along the model mode does not change any pairwise probabilities. These unidentifiable directions must therefore be removed by imposing an identification constraint, for example the linear constraint \(\mathcal{C}(T^\star)=0\) in \eqref{eq:general-linear-constraint-model-setting}.


Combining these two restrictions leads to the \emph{signal tangent space}
\begin{equation}\label{eq:signal_tangent}
\mathbb{T}
:=
\mathrm{Tan}_{T^\star}(\mathbb{M}_r)
\cap
\bigl\{\Delta\in\mathbb{R}^{d_1\times\cdots\times d_m}:\mathcal{C}(\Delta)=0\bigr\}.
\end{equation}
This is the space of local perturbations that are both compatible with the low-rank model and identifiable from the data. In the semiparametric analysis, $\mathbb{T}$ plays the role of the effective parameter space. The following matrix example makes this construction concrete.

\begin{example}[Rank-$r$, first-mode centered matrix]\label{ex:row-centered}
Consider $T^\star = U\Sigma V^\top \in \mathbb{R}^{d_1\times d_2}$ of rank $r$, with $U\in\mathbb{R}^{d_1\times r}$, $V\in\mathbb{R}^{d_2\times r}$ orthonormal. The tangent space of $\mathbb{M}_r$ at $T^\star$ consists of matrices $\Delta = UA^\top + BV^\top$ for $A\in\mathbb{R}^{d_2\times r}$, $B\in\mathbb{R}^{d_1\times r}$. The linear zero-sum constraint $\mathbf{1}_{d_1}^\top\Delta = 0$ reduces to $B^\top\mathbf{1}_{d_1}=0$ (since $U^\top\mathbf{1}_{d_1}=0$ by the constraint on $T^\star$). Let $Q\in\mathbb{R}^{d_1\times(d_1-1)}$ have orthonormal columns spanning $\mathbf{1}_{d_1}^\perp$, so that $B = QC$ for some $C\in\mathbb{R}^{(d_1-1)\times r}$. The signal tangent space is then
\begin{equation}\label{eq:tangent_row0_UCQ}
\mathbb{T}
=
\bigl\{
UA^\top + QCV^\top:
A\in\mathbb{R}^{d_2\times r},\;
C\in\mathbb{R}^{(d_1-1)\times r}
\bigr\}.
\end{equation}
\end{example}


Let $\psi:\mathbb{R}^{d_1\times\cdots\times d_m}\to\mathbb{R}$ denote the target of inference, assumed Gateaux differentiable at $T^\star$. Then along any admissible direction $H\in\mathbb{T}$,
\begin{equation}\label{eq:lb-pathwise}
D\psi_{T^\star}(H)
=
\langle \nabla \psi(T^\star), H\rangle
=
\langle P_{\mathbb{T}}\nabla\psi(T^\star), H\rangle,
\qquad \forall H\in\mathbb{T},
\end{equation}
where $P_{\mathbb{T}}$ denotes orthogonal projection onto $\mathbb{T}$. Thus only the component of the gradient lying in the admissible tangent space matters for local inference. For a linear functional $\psi(T)=\langle \Gamma,T\rangle$, the gradient is simply $\nabla\psi(T^\star)=\Gamma$.

\subsection{Score and information operator}\label{sec:info-operator}

The information operator quantifies how much statistical information different types of comparisons carry about the latent score matrix or tensor. In the LLM evaluation setting, this information is highly heterogeneous. A closely contested matchup, for example, Claude versus GPT-4 on coding tasks, produces an outcome close to a 50/50 comparison and is therefore highly informative about the underlying ability gap. By contrast, a lopsided comparison between a frontier model and a much weaker baseline yields a near-deterministic outcome and contributes relatively little information. The information operator formalizes this heterogeneity and determines how the comparison design and the local noise structure together affect statistical efficiency.

We next quantify how much information the observed comparisons carry about perturbations in the direction of $\mathbb{T}$. Recall from Section~\ref{sec:model-setting} that the conditional law of $Y$ given $X$ depends on $T^\star$ only through the scalar index $\eta^\star=\langle T^\star,X\rangle$. Define the scalar score by
\begin{equation}\label{eq:scalar-score}
s_\eta(Y,\eta^\star)
:=
\partial_\eta \ell(Y,\eta)\big|_{\eta=\eta^\star},
\qquad
\eta^\star=\langle T^\star,X\rangle.
\end{equation}
If we perturb the signal along a direction $H\in\mathbb{T}$, that is, consider the submodel $T_\varepsilon=T^\star+\varepsilon H$, then the corresponding directional score is
\begin{equation}\label{eq:directional-score}
\left.\frac{\partial}{\partial\varepsilon}
\log p_{T_\varepsilon,\Pi^\star}(X,Y)\right|_{\varepsilon=0}
=
s_\eta(Y,\eta^\star)\,\langle H,X\rangle.
\end{equation}
This quantity measures the local sensitivity of the log-likelihood to changes in $T^\star$ along the direction $H$.

A useful way to interpret \eqref{eq:directional-score} is as a generalized regression score. In ordinary parametric models, the Fisher information describes how precisely a parameter can be estimated. Here the same role is played by an \emph{information operator}, because the parameter is not a finite-dimensional vector but a matrix or tensor. Let
\[
I(\eta)
:=
\mathbb{E}\!\left[s_\eta(Y,\eta)^2 \mid \eta\right]
\]
denote the scalar Fisher information at index value $\eta$. We define the information operator $G:\mathbb{R}^{d_1\times\cdots\times d_m}\to\mathbb{R}^{d_1\times\cdots\times d_m}$ by
\begin{equation}\label{eq:G-operator}
\langle G U, V\rangle
=
\mathbb{E}^\star\!\left[
I(\eta^\star)\,
\langle U,X\rangle \langle V,X\rangle
\right],
\qquad U,V\in\mathbb{R}^{d_1\times\cdots\times d_m},
\end{equation}
where $\eta^\star=\langle T^\star,X\rangle$. This operator is the natural analogue of a Fisher information matrix in the present high-dimensional setting.

The operator $G$ captures two distinct sources of heterogeneity. First, the comparison design may be highly uneven, so some contrasts $\langle U,X\rangle$ are observed much more often than others. Second, even among observed comparisons, some are more informative than others because the Fisher information $I(\eta^\star)$ depends on the difficulty of the matchup. In the BTL model, for example, closely matched models are much more informative than obviously unequal ones. The following example makes these objects concrete and illustrates why the restricted operator is invertible under natural conditions.

\begin{example}[Uniform pairwise comparisons, matrix case]\label{ex:uniform-IG}
Consider the matrix setting of Example~\ref{ex:pairwise} with $T^\star\in\mathbb{R}^{d_1\times d_2}$, and suppose models and tasks are sampled uniformly: at each round, a task $u$ is drawn uniformly from $[d_2]$ and a model pair $(p,q)$ is drawn uniformly from all $\binom{d_1}{2}$ pairs. The design matrix is $X = (e_p - e_q)e_u^\top$, and the scalar index is $\eta = T^\star_{p,u} - T^\star_{q,u}$. Under the BTL model, the scalar Fisher information at index $\eta$ is
\[
I(\eta) = \sigma(\eta)(1-\sigma(\eta)) = \frac{e^\eta}{(1+e^\eta)^2}.
\]
For a matrix $H \in \mathbb{R}^{d_1 \times d_2}$, we have $\langle H, X \rangle = H_{p,u} - H_{q,u}$. Hence the information operator acts as
\[
(GH)_{j,u}
=
\frac{1}{d_2 \binom{d_1}{2}}
\sum_{u'=1}^{d_2}
\sum_{p<q}
I(\eta_{pq,u'})\,(H_{p,u'}-H_{q,u'})\,
\bigl(\delta_{j,p}-\delta_{j,q}\bigr)\,\delta_{u,u'}.
\]
In the simplest case where all models are equally matched ($\eta^*$ = 0 and hence $T^\star = 0$ due to the zero-sum constraint), the Fisher information is constant $I(0) = 1/4$, and $G$ simplifies to
\[
(GH)_{j,u}
=
\frac{1}{4\binom{d_1}{2}}
\sum_{q\ne j}
(H_{j,u}-H_{q,u})
=
\frac{1}{4\binom{d_1}{2}}
\Bigl(
(d_1-1)\,H_{j,u} - \sum_{q\ne j} H_{q,u}
\Bigr).
\]
If $H$ satisfies the zero-sum constraint $\sum_{j=1}^{d_1} H_{j,u} = 0$, then $\sum_{q \ne j} H_{q,u} = -H_{j,u}$, so
$
(GH)_{j,u}
=
\frac{1}{2(d_1-1)}\,H_{j,u},
$
and $G$ acts as a scalar multiple of the identity: $G = \frac{1}{2(d_1-1)}\,\mathrm{Id}$. Hence, $G$ is diagonal and invertible on this subspace. Since the signal tangent space $\mathbb{T}$ lies within the zero-sum subspace (cf.\ Example~\ref{ex:row-centered}), the restricted operator $A = P_{\mathbb{T}} G P_{\mathbb{T}}$ is also invertible on $\mathbb{T}$. When $T^\star \ne 0$, the Fisher information $I(\eta_{pq,u})$ varies across comparisons, making $G$ no longer proportional to the identity. However, $G$ remains diagonally dominant on the zero-sum subspace and the restricted operator $A$ remains invertible, provided that no comparison is completely deterministic (i.e., $I(\eta) > 0$ for all observed indices).
\end{example}

Since only directions in $\mathbb{T}$ are statistically relevant, the corresponding \emph{restricted information operator} is
\begin{equation}\label{eq:A-operator}
A
:=
P_{\mathbb{T}} G P_{\mathbb{T}}
:\mathbb{T}\to\mathbb{T}.
\end{equation}
This is the operator that governs local efficiency for the target $\psi(T^\star)$.

\begin{assumption}[Invertibility on the tangent space]\label{ass:lb-invertible}
The restricted information operator $A$ is invertible on $\mathbb{T}$: there exists $c_{\mathbb{T}}>0$ such that
\[
\langle H,AH\rangle \ge c_{\mathbb{T}}\|H\|_{F}^2,
\qquad \forall H\in\mathbb{T}.
\]
\end{assumption}

Assumption~\ref{ass:lb-invertible} ensures that the data are informative along every admissible direction in the tangent space $\mathbb{T}$. That is, any nonzero perturbation of $T^\star$ within $\mathbb{T}$ must induce a nontrivial change in the distribution of the observed comparisons. Under this condition, the optimal local direction for estimating $\psi(T^\star)$ is characterized by the \emph{information equation}
\begin{equation}\label{eq:info-equation-main}
A H^\star
=
P_{\mathbb{T}}\nabla\psi(T^\star).
\end{equation}
Its unique solution is
\begin{equation}\label{eq:Hstar}
H^\star
=
A^{-1}P_{\mathbb{T}}\nabla\psi(T^\star)
\in \mathbb{T}.
\end{equation}
The tensor $H^\star$ is the analogue of the least favorable direction in semiparametric theory: it is the direction in the parameter space that is hardest for estimating the target functional.

Equation~\eqref{eq:info-equation-main} is best understood as an infinite-dimensional version of a generalized least-squares normal equation. To see the analogy concretely, consider a heteroscedastic linear regression model $Y_i = x_i^\top \theta^\star + \varepsilon_i$ with $\mathrm{Var}(\varepsilon_i \mid x_i) = \sigma_i^2$. The efficient score for estimating a linear functional $\gamma^\top \theta^\star$ weights each observation by the inverse noise variance:
\[
\hat\theta_{\rm eff} = \textstyle\arg\min_\theta \sum_{i=1}^n \sigma_i^{-2}(Y_i - x_i^\top \theta)^2,
\quad
\text{information matrix: } \mathcal{I} = \sum_{i=1}^n \sigma_i^{-2}\, x_i x_i^\top.
\]
The efficient influence function for $\gamma^\top \theta^\star$ is $\phi_i = \sigma_i^{-2}(Y_i - x_i^\top \theta^\star)\, x_i^\top \mathcal{I}^{-1} \gamma$, and the optimal direction solves $\mathcal{I}\, h = \gamma$.
In our setting, the scalar inverse-variance weight $\sigma_i^{-2}$ is replaced by the Fisher information $I(\eta^\star)$, the design vectors $x_i$ become design tensors $X$, and the information matrix $\mathcal{I}$ becomes the operator $G$. The projection $P_{\mathbb{T}}$ then restricts attention to directions that are both low-rank and identifiable. Thus the information equation says: among all admissible perturbations, the optimal one balances the gradient of the target functional against the information geometry induced by the comparison design.

\subsection{Efficient influence function and the efficiency bound}\label{sec:eff-lb}

The solution $H^\star$ in \eqref{eq:Hstar} leads directly to the efficient influence function (EIF),
\begin{equation}\label{eq:eif-main}
\phi_{\mathrm{EIF}}(X,Y)
=
s_\eta(Y,\eta^\star)\,\langle H^\star,X\rangle.
\end{equation}
More generally, an influence function may be viewed as the first-order term that captures how an estimator responds to small perturbations in the underlying data-generating distribution. In the semiparametric framework, it provides the leading correction to a plug-in estimator and determines its asymptotic variance. The EIF is the particular influence function that has the smallest asymptotic variance, and hence yields the semiparametric efficiency bound.

To see why $\phi_{\mathrm{EIF}}$ is the efficient influence function, it is enough to note two facts. First, it gives the correct first-order representation for the target functional $\psi(T^\star)$, with the required identity following from the information equation~\eqref{eq:info-equation-main}. Second, it is constructed from a direction $H^\star$ that lies in the signal tangent space $\mathbb{T}$, so it uses only the part of the data that is informative about $T^\star$ and is orthogonal to nuisance variation coming from the sampling distribution. These two properties imply that $\phi_{\mathrm{EIF}}$ is the unique influence function with the smallest asymptotic variance. The full semiparametric derivation is given in Appendix~\ref{app:semiparametric-details}.

The variance of the EIF is
\begin{equation}\label{eq:veff-def}
V_{\mathrm{eff}}(\psi)
=
\mathbb{E}^\star[\phi_{\mathrm{EIF}}(X,Y)^2]
=
\big\langle
P_{\mathbb{T}}\nabla\psi(T^\star),
A^{-1}P_{\mathbb{T}}\nabla\psi(T^\star)
\big\rangle.
\end{equation}
This quantity is the semiparametric efficiency bound. It represents the smallest asymptotic variance achievable by any regular estimator of $\psi(T^\star)$.

To state the lower bound formally, we impose a standard unbiasedness condition.

\begin{assumption}[Global unbiasedness]\label{ass:global-unbiased}
The estimator $\widehat\psi=\widehat\psi((X_i,Y_i)_{i=1}^n)$ is unbiased for $\psi(T)$ under the model:
\[
\mathbb{E}_{T,\Pi}[\widehat\psi]
=
\psi(T),
\qquad \forall\, T \text{ in the parameter space } \mathcal{P}.
\]
\end{assumption}

\begin{theorem}[Information lower bound]\label{thm:general-lb}
Under Assumptions~\ref{ass:lb-invertible} and~\ref{ass:global-unbiased}, any estimator $\widehat\psi$ with finite variance satisfies
\begin{equation}\label{eq:lb-sharp}
\mathrm{Var}^\star(\widehat\psi)
\ge
\frac{1}{n}\,V_{\mathrm{eff}}(\psi),
\qquad
V_{\mathrm{eff}}(\psi)
=
\big\langle
P_{\mathbb{T}}\nabla\psi(T^\star),
A^{-1}P_{\mathbb{T}}\nabla\psi(T^\star)
\big\rangle.
\end{equation}
Moreover, the bound is attained by the EIF in the sense that
$
\mathbb{E}^\star[\phi_{\mathrm{EIF}}(X,Y)^2]
=
V_{\mathrm{eff}}(\psi).
$
\end{theorem}

The proof is a standard information-inequality argument and is given in Appendix~\ref{app:proof-general-lb}. The key message of Theorem~\ref{thm:general-lb} is that once the low-rank geometry, the identification constraint, and the comparison design are taken into account, the efficiency bound takes the same conceptual form as in classical semiparametric theory: the optimal variance is the squared norm of the projected gradient under the inverse restricted information operator.

This perspective is especially useful in LLM evaluation because it isolates the precise ways in which the problem differs from standard low-rank completion. In standard matrix/tensor completion with homoscedastic observations and uniform sampling, the information geometry is often isotropic. Here, by contrast, both the observation model and the sampling design may be highly heterogeneous. The restricted operator $A=P_{\mathbb{T}}GP_{\mathbb{T}}$ is therefore the correct object for capturing the true precision limit.

\section{One-step estimation and upper bound analysis}\label{sec:upper-bound}

Section~\ref{sec:semiparametric_formulation} characterized the semiparametric efficiency bound through the efficient influence function. We now turn to the constructive question: can one build a feasible estimator that attains this bound in the low-rank pairwise-comparison model? Our approach is based on a one-step correction to an initial low-rank estimator. This is a standard semiparametric strategy in spirit, but its analysis here is substantially more delicate because the data are pairwise, the signal is low rank, and the information operator is generally anisotropic.

\subsection{The one-step estimator}\label{sec:starting_point}

We use the cross-fitting procedure summarized in Algorithm~\ref{alg:onestep}. The data are randomly partitioned into $K$ folds of approximately equal size. For each fold $k$, the remaining $K-1$ folds are used to construct the initial estimator and the estimated tangent-space projector, while fold $k$ serves as the evaluation sample for the debiasing step. The final estimate is obtained by averaging over the $K$ folds. Cross-fitting avoids the sample-size loss inherent in a single data split while maintaining the independence between the estimation and evaluation stages that simplifies the theoretical analysis.

\begin{algorithm}[t]
\caption{One-step estimator for a linear functional $\psi(T^\star)=\langle \Gamma,T^\star\rangle$}
\label{alg:onestep}
\begin{enumerate}
\item Randomly partition the data $\{(X_i,Y_i)\}_{i=1}^N$ into $K$ folds $\mathcal{D}_1,\ldots,\mathcal{D}_K$ of approximately equal size.
\item For each fold $k=1,\ldots,K$:
\begin{enumerate}
\item Using the out-of-fold data $\mathcal{D}_{-k}:=\bigcup_{j\neq k}\mathcal{D}_j$, compute a low-rank estimator $\hat T^{(k)}$ of $T^\star$ subject to the identification constraint, and construct the associated tangent-space projector $\widehat{P}_{\mathbb{T}}^{(k)}$.
\item Compute the plug-in information operator $\widehat G^{(k)}$ via \eqref{eq:Ghat_def} using $\mathcal{D}_{-k}$.
\item Compute the estimated EIF direction
\[
\hat H^{(k)}=(\widehat{P}_{\mathbb{T}}^{(k)} \widehat G^{(k)} \widehat{P}_{\mathbb{T}}^{(k)})^{-1}\widehat{P}_{\mathbb{T}}^{(k)}\Gamma.
\]
\item Using fold $k$ as the evaluation sample, form the fold-level one-step estimator
\[
\hat\psi^{(k)}
=
\psi(\hat T^{(k)})
+
\frac{K}{N}\sum_{i\in\mathcal{D}_k} s_\eta(Y_i,\hat\eta_i^{(k)})\,\langle \hat H^{(k)},X_i\rangle,
\qquad
\hat\eta_i^{(k)}:=\langle \hat T^{(k)},X_i\rangle.
\]
\end{enumerate}
\item Output the cross-fitted estimator $\hat\psi=K^{-1}\sum_{k=1}^K \hat\psi^{(k)}$ and the variance estimate $\widehat V = N^{-1}\sum_{i=1}^N \hat\phi_i^2$, where $\hat\phi_i = s_\eta(Y_i,\hat\eta_i^{(k_i)})\,\langle \hat H^{(k_i)},X_i\rangle$ and $k_i$ denotes the fold containing observation $i$.
\end{enumerate}
\end{algorithm}

Let $\hat T$ be an initial estimator of $T^\star$ satisfying
Assumption~\ref{ass:initial}, and let $\widehat{P}_{\mathbb{T}}$ denote the
estimated projector onto the tangent space at $\hat T$. Since the population
information operator $G$ defined in \eqref{eq:G-operator} depends on the
unknown Fisher information $I(\eta^\star)$, it is not directly available. We
therefore replace $G$ by its plug-in estimate
\begin{equation}\label{eq:Ghat_def}
\langle \widehat G\, U, V\rangle
:=
\frac{1}{n_{\rm tr}}\sum_{i\in\mathcal{D}_{\rm tr}}
I(\hat\eta_i)\,
\langle U,X_i\rangle \langle V,X_i\rangle,
\qquad
\hat\eta_i:=\langle \hat T,X_i\rangle,
\end{equation}
where $\mathcal{D}_{\rm tr}$ denotes the training sample and
$n_{\rm tr}=|\mathcal{D}_{\rm tr}|$. Under the bounded signal condition
\eqref{eq:Tstar-bounded} and Assumption~\ref{ass:initial}, the estimation error
$\widehat G-G$ is controlled by the concentration argument in
Appendix~\ref{app:pd-hatG}. We define the estimated
efficient-influence-function direction $\hat H$ by
\begin{equation}\label{eq:Hhat_def}
(\widehat{P}_{\mathbb{T}} \widehat G \widehat{P}_{\mathbb{T}})\hat H
=
\widehat{P}_{\mathbb{T}} \Gamma,
\qquad\text{that is,}\qquad
\hat H
=
(\widehat{P}_{\mathbb{T}} \widehat G \widehat{P}_{\mathbb{T}})^{-1}\widehat{P}_{\mathbb{T}} \Gamma.
\end{equation}
For each fold $k$, the evaluation sample $\mathcal{D}_k$ is independent of the auxiliary quantities $\hat T^{(k)}$ and $\hat H^{(k)}$ constructed from $\mathcal{D}_{-k}$. The fold-level one-step estimator of the linear functional $\psi(T^\star)=\langle \Gamma,T^\star\rangle$ is
\begin{equation}\label{eq:onestep-fold}
\hat\psi^{(k)}
=
\psi(\hat T^{(k)})
+
\frac{K}{N}\sum_{i\in\mathcal{D}_k} s_\eta\!\bigl(Y_i,\hat\eta_i^{(k)}\bigr)\,\langle \hat H^{(k)},X_i\rangle,
\qquad
\hat\eta_i^{(k)}:=\langle \hat T^{(k)},X_i\rangle.
\end{equation}
Within each fold, the first term is a plug-in estimator, while the second is the influence-function correction. The cross-fitted estimator is the average over all $K$ folds:
\begin{equation}\label{eq:onestep}
\hat\psi = \frac{1}{K}\sum_{k=1}^K \hat\psi^{(k)}.
\end{equation}
Intuitively, the fold-level correction removes the leading bias of the plug-in estimator and aligns the resulting estimator with the efficient local direction identified in Section~\ref{sec:semiparametric_formulation}. Cross-fitting ensures that the independence between the evaluation sample and the estimated quantities $(\hat T^{(k)}, \hat H^{(k)})$ holds within each fold, while using the full sample for both estimation and evaluation.

For general nonlinear functionals, the same construction applies after replacing $\Gamma$ by the gradient of the target functional evaluated at the initial estimator. We defer the extension of general nonlinear functionals in Section \ref{sec:nonlinear-functional}.

\subsection{Assumptions for the upper-bound analysis}\label{sec:upper-bound-assumptions}

To establish asymptotic efficiency of the one-step estimator, we impose three conditions.

\begin{assumption}[Initial estimator]\label{ass:initial}
For each cross-fitting fold, the training sample yields an estimator $\hat T$
with Tucker rank $(r_1,\ldots,r_m)$ and orthonormal mode factors
$\hat U_1,\ldots,\hat U_m$.  For some constants $c_0,c,C_1>0$, the following
conditions hold simultaneously, uniformly over the fixed number of folds,
with probability at least $1-\bar d^{-c_0}$.  First, $\hat T$ obeys the same
identification constraint as $T^\star$,
\[
\sum_{a=1}^{d_1}\hat T_{a,u}=0,
\qquad u\in[D],
\]
and, with
\begin{equation}\label{eq:init-rate}
\varepsilon_n:=\sqrt{\frac{\bar d\log^c\bar d}{n}},
\end{equation}
it satisfies
\begin{equation}\label{eq:init-entrywise}
\|\hat T-T^\star\|_\infty\le C_1\varepsilon_n.
\end{equation}
Second, for every mode $j\in[m]$,
\begin{align}
\|P_{\hat U_j}-P_{U_j}\|_{\mathrm{op}}
&\le C_1\varepsilon_n,\label{eq:init-subspace-op}\\
\|P_{\hat U_j}-P_{U_j}\|_{2,\infty}
&\le C_1\varepsilon_n\sqrt{\frac{\mu r_j}{d_j}},
\label{eq:init-subspace-row}\\
\|P_{\hat U_j}-P_{U_j}\|_{\infty}
&\le C_1\varepsilon_n\frac{\mu r_j}{d_j},
\label{eq:init-subspace-entry}
\end{align}
and the estimated factors remain incoherent:
\begin{equation}\label{eq:init-estimated-incoherence}
\max_{a\in[d_j]}\|e_a^\top\hat U_j\|_2^2
\le \frac{2\mu r_j}{d_j}.
\end{equation}
Finally, the multilinear projector formed from the estimated factors satisfies
\begin{equation}\label{eq:init-projection-stability}
\big\|\mathcal A\times_1P_{\hat U_1}\cdots\times_mP_{\hat U_m}\big\|_\infty
\le
\left(\prod_{j=1}^m\sqrt{2\mu r_j}\right)\|\mathcal A\|_\infty
\qquad\text{for every tensor }\mathcal A.
\end{equation}
\end{assumption}

The entrywise error bound in Assumption~\ref{ass:initial} is essential for one-step inference, because the score correction is evaluated entry by entry through
$
\hat\eta_i-\eta_i^\star
=
\langle \hat T-T^\star, X_i\rangle,
$
and in the pairwise-comparison design \(X_i\) is sparse.  Consequently,
\(\max_i|\hat\eta_i-\eta_i^\star|\) is controlled by
\(\|\hat T-T^\star\|_\infty\).  The subspace and incoherence conditions in
\eqref{eq:init-subspace-op}--\eqref{eq:init-projection-stability} control the
estimated tangent projector and the plug-in information direction.  The
subsequent inference theory depends on the initial estimator only through
these stated properties.

\begin{assumption}[Signal strength]\label{ass:signal-strength}
The Frobenius norm satisfies
\begin{equation}\label{eq:signal-strength}
\|T^\star\|_F \ge c\sqrt{d^\star}
\end{equation}
for an absolute constant $c>0$.
\end{assumption}

Assumption~\ref{ass:signal-strength} and \eqref{eq:Fnorm-lambda} imply that $\lambda_{\min}\asymp\sqrt{d^\star}$ when $\kappa$ and $r$ are bounded constants. Since the noise level under the pairwise-comparison model is of constant order (the Fisher information satisfies $c_B\le I(\eta)\le 1/4$), the classical signal-to-noise ratio $\lambda_{\min}/\sigma$ grows with $\sqrt{d^\star}$, which is the rate required in existing literature when the sample size satisfies
\begin{equation}\label{eq:sample-size}
n \ge C_0\, \bar d\log^c \bar d
\end{equation}
for a sufficiently large constant $C_0$ depending on $\mu,\kappa,r,m$.
Together with the asymptotic requirement
$n/(\bar d\log^c\bar d)\to\infty$, this is the regime in which
$\varepsilon_n\to0$ and the local conditions in
Assumption~\ref{ass:initial} support the one-step expansion.

\begin{assumption}[Alignment condition]\label{ass:alignment}
There exists a constant $\alpha_\Gamma>0$ such that
\begin{equation}\label{eq:alignment}
\|P_{\mathbb{T}}\Gamma\|_F
\ge
\alpha_\Gamma\,\bar d^{1/2}(d^\star)^{-1/2}\|\Gamma\|_F.
\end{equation}
\end{assumption}

Assumption~\ref{ass:alignment} requires that the target functional have a non-negligible component in the tangent space $\mathbb{T}$. In many natural examples this condition is mild; for instance, when $\Gamma$ itself lies in $\mathbb{T}$, as in entrywise contrasts or category-specific ability gaps, the condition holds with $\alpha_\Gamma=1$.

\subsection{Main error bound for the one-step estimator}\label{sec:combined_bound}

We now state the main upper-bound result for the one-step estimator, with leading term given by the efficient influence function and remainder controlled explicitly.

\begin{theorem}[Error bound for the one-step estimator]\label{thm:combined}
Under Assumptions~\ref{ass:initial}--\ref{ass:alignment}, the bounded signal condition \eqref{eq:Tstar-bounded}, and the sample size condition \eqref{eq:sample-size}, the following holds with probability at least $1-\bar d^{-c}$ conditional on the first-stage sample:
\begin{equation}\label{eq:combined_bound}
\bigl|\hat\psi - \psi(T^\star) - (\mathbb{P}_n-\mathbb{P}^\star)\phi^\star\bigr|
\;\le\;
C(\mu,\kappa,r,m)\,C_A\,\|\Gamma\|_1\,
\frac{\bar d\log^c \bar d}{n},
\end{equation}
where $\phi^\star$ is the efficient influence function and $C_A$ is defined such that
\begin{gather}\label{eq:CA_def_main}
\|A^{-1}\|_{\infty\to\infty}
\;\vee\;
\|\hat A^{-1}\|_{\infty\to\infty}
\;\le\;
C_A\, d^\star,
\\
\notag
\text{where}\quad
A:=P_{\mathbb{T}}GP_{\mathbb{T}},
\qquad
\hat A:=\widehat{P}_{\mathbb{T}} \widehat G \widehat{P}_{\mathbb{T}}.
\end{gather}
Here $C(\mu,\kappa,r,m)$ is a constant depending only on $\mu, \kappa, r, m$ (and the bound $B$ from \eqref{eq:Tstar-bounded}), and $c>0$ is an absolute constant.
\end{theorem}

Theorem~\ref{thm:combined} gives a sharp, simplified characterization of the one-step estimator in the pairwise-comparison setting. 
Building on it, we obtain the formal CLT.

\begin{theorem}[CLT and variance estimation for the one-step estimator]\label{thm:eff-clt}
Under the conditions of Theorem~\ref{thm:combined}, if additionally
\begin{equation}\label{eq:clt-condition}
C_A\,\sqrt{\frac{\bar d\log^{2c}\bar d}{n}}\;\to\;0,
\end{equation}
then the one-step estimator is asymptotically normal:
\begin{equation}\label{eq:eff-clt}
\frac{\sqrt{n}\bigl(\hat\psi-\psi(T^\star)\bigr)}{\sqrt{V_{\rm eff}}}
\;\xrightarrow{d}\;
\mathcal{N}(0,1),
\end{equation}
where the asymptotic variance is the semiparametric efficiency bound
\begin{equation}\label{eq:Veff_formula}
V_{\rm eff}
\;=\;
\bigl\langle P_{\mathbb{T}}\Gamma,\, A^{-1}P_{\mathbb{T}}\Gamma\bigr\rangle.
\end{equation}
Moreover, the plug-in variance estimator $\widehat V_{\rm eff}:=\mathbb{P}_n[\hat\phi^2]$, where $\hat\phi(X,Y):=s_\eta(Y,\hat\eta(X))\langle\hat H,X\rangle$, satisfies
\begin{equation}\label{eq:Veff-relative}
\left|\frac{\widehat V_{\rm eff}-V_{\rm eff}}{V_{\rm eff}}\right|
\;\lesssim\;
C\sqrt{\frac{\bar d\log^c\bar d}{n}}
\end{equation}
with probability at least $1-\bar d^{-c}$.
\end{theorem}

The CLT proof proceeds in two steps: the remainder bound~\eqref{eq:combined_bound} gives $\sqrt{n}\,R_n/\sqrt{V_{\rm eff}}\to 0$ under~\eqref{eq:clt-condition}, and the oracle influence-function average satisfies a conditional Lyapunov CLT (Appendix~\ref{app:clt-upgrade}). The variance estimation bound~\eqref{eq:Veff-relative} is proved in Appendix~\ref{app:pd-variance}; the key observation is that the common $d^\star$ spectral scale of $A^{-1}$ and $\widehat A^{-1}$ cancels in the relative ratio, leaving only the geometric perturbation scale~$\varepsilon_n$ and the alignment factor~$\alpha_\Gamma$.

\medskip
\noindent\textbf{The role of $C_A$.}\;
The term $C_A$ captures the entrywise behavior of the restricted inverse information operator
$
A := P_{\mathbb{T}} G P_{\mathbb{T}},
$
where $\langle H_1, G(H_2)\rangle := \mathbb{P}^\star[I(\eta^\star)\langle H_1,X\rangle\langle H_2,X\rangle]$ is the Fisher information operator. The oracle and estimated EIF directions are
\begin{equation}\label{eq:directions}
H^\star := A^{-1}P_{\mathbb{T}}\Gamma,
\qquad
\hat H := \hat A^{-1}\widehat{P}_{\mathbb{T}} \Gamma.
\end{equation}
Although the spectral norm of $A^{-1}$ is always well-behaved ($\|A^{-1}\|_{\rm op}\lesssim d^\star$), the one-step analysis requires the stronger entrywise control
\begin{equation}\label{eq:CA_def}
\|A^{-1}\|_{\infty\to\infty}
\;\vee\;
\|\hat A^{-1}\|_{\infty\to\infty}
\;\le\;
C_A\, d^\star.
\end{equation}
The next proposition characterizes the range of $C_A$ and identifies when it is dimension-free.

\begin{proposition}[Range of $C_A$]\label{prop:CA-range-main}
Under the bounded signal condition \eqref{eq:Tstar-bounded} and $\mu$-incoherence:
\begin{enumerate}
\item[\emph{(i)}] \textbf{(General bounds.)} The term $C_A$ satisfies
\begin{equation}\label{eq:CA_range}
C(\mu,r,m)
\;\le\;
C_A
\;\le\;
C(\mu,r,m)\,\sqrt{\bar d}.
\end{equation}
In the constant-weight benchmark $I(\eta^\star)\equiv 1/4$, $A_0^{-1}=2d^\star P_{\mathbb{T}}$ and $C_A=\mathrm{poly}(\mu,r,m)$.
\item[\emph{(ii)}] \textbf{(Near-constant regime.)} If $\|T^\star\|_\infty\le B_0$ for a constant $B_0=B_0(\mu,m,r)$ depending only on the structural parameters, then $C_A\le C(\mu,r,m,B_0)$: a dimension-free constant. In this regime, the CLT condition \eqref{eq:clt-condition} reduces to $\bar d\log^{2c}\bar d/n\to 0$, and the one-step estimator achieves the semiparametric efficiency bound at the optimal sample-complexity scale $n\gg\bar d\,\mathrm{polylog}(\bar d)$.
\end{enumerate}
\end{proposition}

\noindent
The proof is given in Appendix~\ref{app:CA-range}. 

\noindent\textbf{Key difficulty: non-commutativity of $G$ and $P_{\mathbb{T}}$.}\;
The gap between the $O(1)$ lower bound and the $O(\sqrt{\bar d})$ upper bound in \eqref{eq:CA_range} reflects the fundamental technical challenge of this problem: the non-commutativity of the information operator $G$ and the tangent-space projector $P_{\mathbb{T}}$. In isotropic settings (e.g., additive Gaussian noise), $G$ is proportional to the identity on $\mathbb{T}$, so $A^{-1}$ behaves like $P_{\mathbb{T}}$ and $C_A$ is automatically dimension-free. In pairwise comparisons, however, the Fisher information $I(\eta^\star)$ varies across matchups, making $G$ anisotropic. Because $G$ and $P_{\mathbb{T}}$ do not commute in general, the inverse $A^{-1}=(P_{\mathbb{T}}GP_{\mathbb{T}})^{-1}$ can have substantially worse entrywise behavior than its spectral norm suggests.

This difficulty is unique to our setting. The existing debiasing literature for low-rank matrix and tensor estimation, including work on matrix completion, tensor completion, and nuclear-norm penalized inference, assumes homoscedastic or isotropic designs under which $G\propto I$ and the operator $A$ does not arise. The quantity $C_A$ and its interplay with non-commutative operator inversion are, to our knowledge, entirely new to this line of work. 

\subsection{Proof outline and technical challenges}\label{sec:decomposition_roadmap}

We briefly explain the structure of the proof of Theorem~\ref{thm:combined} and highlight where the main technical difficulty enters. Let
\[
\Delta := \hat T - T^\star,
\qquad
\hat\eta_i := \langle \hat T, X_i\rangle,
\qquad
\eta_i^\star := \langle T^\star, X_i\rangle.
\]
Recall that the oracle efficient influence function is
\[
\phi^\star(X,Y)=s_\eta(Y,\eta^\star)\langle H^\star,X\rangle,
\qquad
H^\star=(P_{\mathbb{T}}GP_{\mathbb{T}})^{-1}P_{\mathbb{T}}\Gamma.
\]
Starting from the one-step estimator \eqref{eq:onestep}, we add and subtract the oracle correction term and obtain the decomposition
\begin{equation}\label{eq:decomp_three}
\hat\psi-\psi(T^\star)
=
\underbrace{(\mathbb{P}_n-\mathbb{P}^\star)\phi^\star}_{\text{leading empirical term}}
\;+\;
R_{\mathrm{emp}}
\;+\;
R_{\mathrm{bias}},
\end{equation}
where the empirical remainder $R_{\mathrm{emp}}$ and the bias remainder $R_{\mathrm{bias}}$ are
\[
R_{\mathrm{emp}}
:=
(\mathbb{P}_n-\mathbb{P}^\star)\bigl[S_{\hat T}(\hat H)-S_{T^\star}(H^\star)\bigr],
\quad
R_{\mathrm{bias}}
:=
\langle \Gamma,\Delta\rangle
+
\mathbb{P}^\star\bigl[S_{\hat T}(\hat H)-S_{T^\star}(H^\star)\bigr],
\]
and
$
S_T(H)(X,Y):=s_\eta\!\bigl(Y,\langle T,X\rangle\bigr)\langle H,X\rangle.
$

The first term in \eqref{eq:decomp_three} is the leading stochastic term. It is an average of i.i.d.\ centered random variables and gives the asymptotically normal limit. The task is therefore to show that the empirical remainder $R_{\mathrm{emp}}$ and the bias remainder $R_{\mathrm{bias}}$ are of smaller order.

To do so, we further split the empirical remainder by adding and subtracting $S_{\hat T}(H^\star)$:
\begin{align}\label{eq:Remp_split}
R_{\mathrm{emp}}
&=
\underbrace{(\mathbb{P}_n-\mathbb{P}^\star)\bigl[s_\eta(Y,\hat\eta)\,\langle \hat H-H^\star,X\rangle\bigr]}_{R_{\mathrm{emp}}^{H}\ (\text{direction error})}
\notag\\
&\quad+\;
\underbrace{(\mathbb{P}_n-\mathbb{P}^\star)\bigl[(s_\eta(Y,\hat\eta)-s_\eta(Y,\eta^\star))\,\langle H^\star,X\rangle\bigr]}_{R_{\mathrm{emp}}^{\eta}\ (\text{score perturbation})},
\end{align}
where $\hat\eta:=\langle \hat T,X\rangle$ and $\eta^\star:=\langle T^\star,X\rangle$. Here $R_{\mathrm{emp}}^{H}$ measures the fluctuation caused by replacing the oracle EIF direction $H^\star$ with its estimate $\hat H$, while $R_{\mathrm{emp}}^{\eta}$ captures the additional error from evaluating the score at $\hat\eta$ instead of $\eta^\star$.

Similarly, we split the bias term by decomposing $\Gamma=P_{\mathbb{T}}\Gamma+(I-P_{\mathbb{T}})\Gamma$ and again adding and subtracting $S_{\hat T}(H^\star)$:
\begin{align}\label{eq:Rbias_split}
R_{\mathrm{bias}}
&=
\underbrace{\langle (I-P_{\mathbb{T}})\Gamma,\Delta\rangle}_{R_{\mathrm{proj}}}
\;+\;
\underbrace{\mathbb{P}^\star\bigl[S_{\hat T}(\hat H)-S_{\hat T}(H^\star)\bigr]}_{R_{\mathrm{bias}}^{H}}
\notag\\
&\quad+\;
\underbrace{\langle P_{\mathbb{T}}\Gamma,\Delta\rangle+\mathbb{P}^\star\bigl[S_{\hat T}(H^\star)-S_{T^\star}(H^\star)\bigr]}_{R_{\mathrm{cancel}}}.
\end{align}
The first term $R_{\mathrm{proj}}$ is a projection leakage term. The second term $R_{\mathrm{bias}}^{H}$ is the population-level error due to estimating the EIF direction. The final term $R_{\mathrm{cancel}}$ contains the key first-order cancellation of the one-step estimator. The dominant difficulty is the control of $R_{\mathrm{emp}}^{H}$ and $R_{\mathrm{bias}}^{H}$, both of which depend on $\hat H - H^\star$ and are thus governed by $C_A$ (see the discussion in Section~\ref{sec:combined_bound}). The remaining terms are controlled by standard arguments: $R_{\mathrm{emp}}^{\eta}$ by smoothness of the score and entrywise accuracy of $\hat T$; $R_{\mathrm{proj}}$ by subspace estimation error; and in $R_{\mathrm{cancel}}$, the tangent-space plug-in bias cancels with the population correction, leaving only higher-order contributions. See Table~\ref{tab:proof_outline} for a summary.

\begin{table}[h!]
\begin{center}
\small
\begin{tabular}{llll}
\hline
\textbf{Term} & \textbf{Description} & \textbf{Simplified bound} & \textbf{Proof}\\
\hline
$R_{\mathrm{emp}}^{H}$ & direction-error emp.\ process & $C_A\,\|\Gamma\|_1\,\bar d\log^c\bar d/n$ & App.~\ref{app:pd-RempH}\\
$R_{\mathrm{emp}}^{\eta}$ & score-perturbation emp.\ process & $C_A\,\|\Gamma\|_1\,\bar d\log^c\bar d/n$ & App.~\ref{app:pd-Remp-eta}\\
$R_{\mathrm{proj}}$ & projection leakage & $\|\Gamma\|_1\,\bar d\log^c\bar d/n$ & App.~\ref{app:pd-Rproj}\\
$R_{\mathrm{bias}}^{H}$ & bias from $\hat H\neq H^\star$ & $C_A\,\|\Gamma\|_1\,\bar d/n$ & App.~\ref{app:pd-H-bias}\\
$R_{\mathrm{cancel}}$ (1st order) & $\ip{P_{\mathbb{T}}\Gammavec}{\Delta}$ cancellation & $C_A\,\|\Gamma\|_1\,\bar d\log^c\bar d/n$ & App.~\ref{app:pd-1st-cancel}\\
$R_{\mathrm{cancel}}$ (2nd order) & $O(\infnorm{\Delta}^2)$ remainder & $C_A\,\|\Gamma\|_1\,\bar d/n$ & App.~\ref{app:pd-2nd-order}\\
\textit{CLT remainder} & \textit{$\sqrt{n}\,R_n$ vs.\ $\sqrt{V_{\rm eff}}$} & \textit{$C_A\sqrt{\bar d\log^{2c}\bar d/n}$} & App.~\ref{app:clt-upgrade}\\
\hline
\end{tabular}
\caption{Summary of remainder terms and their simplified bounds after substituting the pairwise-comparison scalings $\sigma=O(1)$, $\lambda_{\min}\asymp\sqrt{d^\star}$, and $\delta=\bar d^{-c}$. All bounds hold with probability $\ge 1-\bar d^{-c}$.}
\label{tab:proof_outline}
\end{center}
\end{table}


\section{Score whitening and extensions}\label{sec:score-whitening}

Section~\ref{sec:upper-bound} showed that the main obstacle to achieving semiparametric efficiency at the optimal sample-complexity scale is the non-commutativity of the information operator $G$ and the tangent-space projector $P_{\mathbb{T}}$. This obstruction is reflected in the quantity $C_A$, which can become dimension dependent when the local Fisher information varies substantially across matchups. In Section \ref{sec:ws-estimator}, we introduce a simple remedy: \emph{score whitening}. The basic idea is to rescale each comparison by its local Fisher information so that all matchups contribute on a common information scale. This transformation turns the effective information operator into an isotropic one, removes the \(C_A\)-driven bottleneck, and yields a simpler one-step estimator with stable asymptotic behavior. We then show how the same idea extends to non-uniform sampling through inverse-probability weighting in Section \ref{sec:unknown-sampling} and to nonlinear targets through local linearization in Section \ref{sec:nonlinear-functional}.


\subsection{The score-whitened one-step estimator}\label{sec:ws-estimator}

In the BTL model, the information carried by a single comparison depends strongly on the underlying ability gap. If two models are closely matched, the comparison is highly informative; if one model dominates the other, the outcome is nearly deterministic and carries little information. As mentioned in Example \ref{ex:uniform-IG}, for a comparison with linear index $\eta$, the Fisher information is
\[
I(\eta)=\sigma(\eta)\{1-\sigma(\eta)\},
\qquad
\sigma(\eta)=(1+e^{-\eta})^{-1}.
\]
Thus the information operator is intrinsically heterogeneous. This heterogeneity is what makes \(G\) anisotropic and creates the non-commutativity problem identified in Section~\ref{sec:upper-bound}.

The key observation is that this heterogeneity can be removed by re-scaling the score. Given the score function $s_\eta(y,\eta)=\partial_\eta \ell(y,\eta)$, define the \emph{whitened score}
\begin{equation}\label{eq:whitened-score-def}
\tilde s_\eta(y,\eta)
:=
\frac{s_\eta(y,\eta)}{I(\eta)}.
\end{equation}
This whitened score remains centered:
$
\mathbb{E}^\star[\tilde s_\eta(Y,\eta^\star)\mid X]=0.
$
Second, its conditional derivative becomes constant:
$
\mathbb{E}^\star\!\left[\partial_\eta \tilde s_\eta(Y,\eta^\star)\mid X\right]=-1.
$
This second identity is crucial and means that, after whitening, the Fisher-information weight no longer varies across comparisons.

As a consequence, under a \emph{uniform pairwise-comparison design}, the effective information operator becomes isotropic. By uniform design, we mean that the comparison tensor $X$ is sampled uniformly over all admissible context--pair combinations; for example, in the model--topic setting, this means sampling the task category uniformly and the unordered model pair uniformly from all $\binom{d}{2}$ possible pairs. Under this design, every valid comparison receives the same sampling probability. More precisely, on the column-sum-zero subspace,
\begin{equation}\label{eq:G0-whitened}
G_0(H)
=
\mathbb{E}^\star[\langle H,X\rangle X]
=
\frac{1}{d^\star}H,
\end{equation}
where \(d^\star\) is the effective comparison dimension. Therefore
\[
A_0=P_{\mathbb{T}}G_0P_{\mathbb{T}}=\frac{1}{d^\star}P_{\mathbb{T}},
\qquad
A_0^{-1}=d^\star P_{\mathbb{T}}.
\]
This is exactly the isotropic structure that was missing in the general upper-bound analysis shown in Theorem \ref{thm:combined}. In particular, the inverse restricted information operator is now explicit, and the troublesome quantity \(C_A\) disappears.

Under the whitened information structure, the efficient direction simplifies dramatically. The oracle and plug-in directions become
\begin{equation}\label{eq:ws-directions}
H^\star_{\rm ws}=d^\star P_{\mathbb{T}}\Gamma,
\qquad
\hat H_{\rm ws}=d^\star \widehat{P}_{\mathbb{T}}\Gamma.
\end{equation}
No operator inversion is required; the direction is obtained simply by projecting the target gradient onto the estimated tangent space and multiplying by \(d^\star\). This is the main computational and theoretical advantage of whitening.

The corresponding whitened one-step estimator is
\begin{equation}\label{eq:ws-estimator}
\hat\psi_{\rm ws}
=
\langle \Gamma,\hat T\rangle
+
\mathbb{P}_n\!\left[\tilde s_\eta(Y,\hat\eta)\,\langle \hat H_{\rm ws},X\rangle\right],
\qquad
\hat\eta=\langle \hat T,X\rangle.
\end{equation}
Like the general one-step estimator, this is a plug-in estimator plus a first-order correction. The difference is that the correction now uses the whitened score and the simplified direction \eqref{eq:ws-directions}. See Algorithm \ref{alg:ws-onestep} for details.

\begin{algorithm}[t]
\caption{Score-whitened one-step estimator for a linear functional $\psi(T^\star)=\langle \Gamma,T^\star\rangle$}
\label{alg:ws-onestep}
\begin{enumerate}
\item Randomly partition the data $\{(X_i,Y_i)\}_{i=1}^N$ into $K$ folds $\mathcal{D}_1,\ldots,\mathcal{D}_K$ of approximately equal size.
\item For each fold $k=1,\ldots,K$:
\begin{enumerate}
\item Using the out-of-fold data $\mathcal{D}_{-k}:=\bigcup_{j\neq k}\mathcal{D}_j$, compute a low-rank estimator $\hat T^{(k)}$ of $T^\star$ subject to the identification constraint, and construct the associated tangent-space projector $\widehat{P}_{\mathbb{T}}^{(k)}$.
\item Compute the whitened EIF direction
\[
\hat H_{\rm ws}^{(k)}=d^\star\, \widehat{P}_{\mathbb{T}}^{(k)}\Gamma.
\]
\item Using fold $k$ as the evaluation sample, form the fold-level whitened one-step estimator
\[
\hat\psi_{\rm ws}^{(k)}
=
\psi(\hat T^{(k)})
+
\frac{K}{N}\sum_{i\in\mathcal{D}_k}
\frac{s_\eta(Y_i,\hat\eta_i^{(k)})}{I(\hat\eta_i^{(k)})}\,
\langle \hat H_{\rm ws}^{(k)},X_i\rangle,
\]
where $\hat\eta_i^{(k)}:=\langle \hat T^{(k)},X_i\rangle$.
\end{enumerate}
\item Output the cross-fitted estimator $\hat\psi_{\rm ws}=K^{-1}\sum_{k=1}^K \hat\psi_{\rm ws}^{(k)}$ and the variance estimate $\widehat V_{\rm ws} = N^{-1}\sum_{i=1}^N \hat{\tilde\phi}_i^2$, where $\hat{\tilde\phi}_i = \{s_\eta(Y_i,\hat\eta_i^{(k_i)})/I(\hat\eta_i^{(k_i)})\}\,\langle \hat H_{\rm ws}^{(k_i)},X_i\rangle$ and $k_i$ denotes the fold containing observation $i$.
\end{enumerate}
\end{algorithm}

The next two theorems give the error bound and CLT for the score-whitened estimator, paralleling Theorems~\ref{thm:combined} and~\ref{thm:eff-clt} for the general case.

\begin{theorem}[Error bound for the score-whitened estimator]\label{thm:ws-bound}
Under the same conditions as Theorem~\ref{thm:combined} and a uniform sampling design, with probability at least $1-\bar d^{-c}$,
\begin{equation}\label{eq:ws-expansion}
\hat\psi_{\rm ws}-\psi(T^\star)
=
(\mathbb{P}_n-\mathbb{P}^\star)\tilde\phi_0^\star
+
R_n,
\end{equation}
where $\tilde\phi_0^\star(X,Y) := \tilde s_\eta(Y,\eta^\star)\,\langle H^\star_{\rm ws},X\rangle$ is the whitened influence function and the remainder satisfies
\begin{equation}\label{eq:ws-remainder}
|R_n|
\;\le\;
C(\mu,\kappa,r,m)\,\|\Gamma\|_1\,
\frac{\bar d\log^c \bar d}{n}.
\end{equation}
\end{theorem}

The bound~\eqref{eq:ws-remainder} is the direct analogue of~\eqref{eq:combined_bound}, with the crucial difference that \emph{no $C_A$ factor appears}. The proof follows the same five-term decomposition (Appendix~\ref{app:score-whitening}), but each term simplifies: the direction error is $\hat H_{\rm ws}-H^\star_{\rm ws} = d^\star(\widehat{P}_{\mathbb{T}}-P_{\mathbb{T}})\Gamma$ (no operator inversion), and the first-order cancellation is \emph{complete} (no off-tangent residual), since $G_0(H^\star_{\rm ws}) = P_{\mathbb{T}}\Gamma$.

\begin{theorem}[CLT and variance estimation for the score-whitened estimator]\label{thm:ws-clt}
Under the conditions of Theorem~\ref{thm:ws-bound}, since $n\gg\bar d\log^c\bar d$ implies $\sqrt{n}\,R_n\to 0$, the conditional Lyapunov CLT gives
\begin{equation}\label{eq:ws-clt}
\frac{\sqrt{n}\bigl(\hat\psi_{\rm ws}-\psi(T^\star)\bigr)}{\sqrt{V_{\rm ws}}}
\;\xrightarrow{d}\;
\mathcal{N}(0,1),
\end{equation}
with asymptotic variance
\begin{equation}\label{eq:Vws}
V_{\rm ws}
=
\mathbb{E}^\star\!\left[\frac{\langle d^\star P_{\mathbb{T}}\Gamma,X\rangle^2}{I(\eta^\star)}\right].
\end{equation}
Also, the plug-in estimator $\widehat V_{\rm ws}=\mathbb{P}_n[\hat{\tilde\phi}^{\,2}]$with $\hat{\tilde\phi}(X,Y)=\tilde s_\eta(Y,\hat\eta(X))\langle d^\star\widehat{P}_{\mathbb{T}}\Gamma,X\rangle$ satisfies
\begin{equation}\label{eq:Vws-relative}
\left|\frac{\widehat V_{\rm ws}-V_{\rm ws}}{V_{\rm ws}}\right|
\;\lesssim\;
C\sqrt{\frac{\bar d\log^c \bar d}{n}}
\end{equation}
with probability at least $1-\bar d^{-c}$. In particular, no $C_A$ factor appears in~\eqref{eq:Vws-relative}, because the whitened direction $\widehat H_{\rm ws}=d^\star\widehat{P}_{\mathbb{T}}\Gamma$ involves no operator inversion.
\end{theorem}

Comparing Theorems~\ref{thm:eff-clt} and~\ref{thm:ws-clt} reveals a clear efficiency--robustness trade-off. The general one-step estimator targets the semiparametric efficiency bound $V_{\rm eff} = \langle P_{\mathbb{T}}\Gamma, A^{-1}P_{\mathbb{T}}\Gamma\rangle$, but requires $C_A\sqrt{\bar d/n}\to 0$ for the CLT. The whitened estimator achieves CLT under the weaker condition $n\gg\bar d\log^c\bar d$ (no $C_A$), but its variance $V_{\rm ws}\ge V_{\rm eff}$ may be larger. The gap vanishes when the Fisher information is nearly constant across comparisons: if $\|T^\star\|_\infty$ is small, all matchups are competitive, $I(\eta^\star)\approx 1/4$, and $V_{\rm ws}\approx V_{\rm eff}$.

\subsection{Extension to non-uniform sampling via inverse-probability weighting}\label{sec:unknown-sampling}

In practice, pairwise comparison data are far from uniformly collected. As illustrated in Figure~\ref{fig:lmarena-battle}, popular frontier models receive disproportionately many battles, while prompts concentrate on engaging topics. This non-uniformity breaks the isotropic Gram structure on which the whitened estimator relies.

\medskip
\noindent\textbf{IPW framework.}
Let $p(x):=\Pi^\star(X=x)$ be the true sampling probability and $q(x)=1/d^\star$ the uniform reference distribution. Define the importance weight $w(x):=q(x)/p(x)$ and the \emph{IPW-whitened score}
\begin{equation}\label{eq:ipw-score}
\tilde s_\eta^{\,w}(y,\eta;x):=w(x)\,\tilde s_\eta(y,\eta).
\end{equation}
The weighted Gram operator
\begin{equation}\label{eq:Gq-def}
\langle H_1,G_q(H_2)\rangle
:=
\mathbb{E}_{X\sim p}\!\big[w(X)\,\langle H_1,X\rangle\langle H_2,X\rangle\big]
=
\mathbb{E}_{X\sim q}\!\big[\langle H_1,X\rangle\langle H_2,X\rangle\big]
\end{equation}
is exactly the uniform-design operator. On the column-sum-zero subspace,
\begin{equation}\label{eq:Gq-isotropic}
G_q(H)
=
\frac{1}{d^\star}H,
\qquad
A_q^{-1}=d^\star P_{\mathbb{T}}.
\end{equation}
Hence the oracle and plug-in directions are
\begin{equation}\label{eq:ipw-directions}
H_{q,0}=d^\star P_{\mathbb{T}}\Gamma,
\qquad
\hat H_{q,0}=d^\star \widehat{P}_{\mathbb{T}}\Gamma,
\end{equation}
which are identical to the whitened directions~\eqref{eq:ws-directions}. No operator inversion is needed.

\medskip
\noindent\textbf{IPW one-step estimator and CLT.}
When the sampling law $p(x)$ is known, the IPW-whitened one-step estimator is
\begin{equation}\label{eq:ipw-estimator}
\hat\psi_{\rm ipw}
=
\langle \Gamma,\hat T\rangle
+
\mathbb{P}_n\!\Big[w(X)\,\tilde s_\eta\!\big(Y,\hat\eta\big)\,\langle \hat H_{q,0},X\rangle\Big].
\end{equation}
Since the weighted Gram $G_q$ is isotropic, the analysis of Theorems~\ref{thm:ws-bound}--\ref{thm:ws-clt} carries through with each remainder term acquiring at most a multiplicative factor of $\|w\|_\infty \le C_p/c_p$.

\begin{theorem}[CLT for the IPW-whitened estimator]\label{thm:ipw-clt}
Under the conditions of Theorem~\ref{thm:ws-bound} and the overlap condition
\begin{equation}\label{eq:overlap_bounds}
c_p/d^\star \le p(x) \le C_p/d^\star,
\end{equation}
\[
\frac{\sqrt{n}\bigl(\hat\psi_{\rm ipw}-\psi(T^\star)\bigr)}{\sqrt{V_{\rm ws}}}
\;\xrightarrow{d}\;
\mathcal{N}(0,1),
\]
where $V_{\rm ws}$ is defined in~\eqref{eq:Vws}. The remainder bound is~\eqref{eq:ws-remainder} with $C(\mu,\kappa,r,m)$ replaced by $C(\mu,\kappa,r,m,C_p/c_p)$.
\end{theorem}

\noindent\textbf{The efficiency perspective.}
Alternatively, in the large-sample regime $n\gg C_A^2\bar d\log^c\bar d$, one could use the general one-step estimator of Section~\ref{sec:upper-bound} with the original non-uniform operator $G$ and achieve the semiparametric efficiency bound $V_{\rm eff}$ (which now depends on the sampling probabilities $p(x)$). The IPW approach is useful because it restores isotropy directly, eliminating the $C_A$ bottleneck at the cost of replacing $V_{\rm eff}$ by $V_{\rm ws}\ge V_{\rm eff}$. The proof for the efficiency case follows from the same operator inversion argument as in Section~\ref{sec:upper-bound}, with $G$ modified to incorporate the non-uniform design weights; see Appendix~\ref{app:norm-score}.

\medskip
\noindent\textbf{Unknown sampling distribution.}
When $p(x)$ is unknown, it must be estimated from data. Write $\hat w(x):=q(x)/\hat p(x)$ for the feasible weight and
\begin{equation}\label{eq:ipw-feasible}
\hat\psi_{\rm ipw}
=
\langle \Gamma,\hat T\rangle
+
\mathbb{P}_n\!\Big[\hat w(X)\,\tilde s_\eta\!\big(Y,\hat\eta\big)\,\langle \hat H_{q,0},X\rangle\Big].
\end{equation}
The key requirement is an entrywise relative error guarantee for $\hat p$:
\begin{equation}\label{eq:entrywise_relative_error}
\sup_{x}\left|\frac{\hat p(x)}{p^\star(x)}-1\right|
\;\le\;
\epsilon_p
\;\ll\; 1
\end{equation}
with high probability. This guarantee is achievable when the sampling mechanism has low effective dimension. For instance, in the matrix case ($m=2$), if tasks are sampled independently of model pairs, each task category $t$ with probability $\pi(t)$ and each unordered pair $\{a,b\}$ with probability $q(a)q(b)$, then the effective parameter dimension is $d_{\rm par} = O(d_1+d_2)$, and $\epsilon_p = O(\sqrt{d_{\rm par}/n})$ is achievable by standard spectral methods. Our numerical experiments in Section~\ref{sec:simulations} confirm that such structured sampling distributions are well-estimated in practice.

The weight-estimation error does not affect the first-order limit. Since $\mathbb{E}^\star[\tilde s_\eta(Y,\eta^\star)\mid X]=0$, the leading weight error is centered out:
$
\mathbb{E}^\star[(\hat w-w)\tilde s_\eta(Y,\eta^\star)\langle H_{q,0},X\rangle]=0.
$
The surviving contribution is a second-order Gram mismatch of order $\epsilon_p\cdot\|\Delta\|_\infty$, which is negligible. The complete proof, including the overlap condition~\eqref{eq:overlap_bounds}, the entrywise relative error guarantee~\eqref{eq:entrywise_relative_error}, and explicit remainder bounds, is given in Appendix~\ref{app:norm-score}.

\subsection{Extension to nonlinear functionals}\label{sec:nonlinear-functional}

Many inferential targets in LLM evaluation are nonlinear: the win probability~\eqref{eq:nonlinear-winrate-intro}, category-weighted leaderboard scores, or calibrated transformations of latent scores. The score-whitening framework extends naturally to such targets by replacing the fixed gradient $\Gamma$ with the plug-in gradient $\hat\Gamma_\psi := \nabla\psi(\hat T)$.

\begin{assumption}[Finite-entry nonlinear functional]\label{ass:nonlinear}
There exists a fixed index set $S_\psi\subset[d_1]\times\cdots\times[d_m]$ with $|S_\psi|=s_\psi=O(1)$ and a twice continuously differentiable function $g:\mathbb{R}^{s_\psi}\to\mathbb{R}$ such that
$
\psi(T)=g(T_{S_\psi})
$
in a neighborhood of $T^\star$. Moreover, $\|\nabla g\|_2\le C_g$ and $\|\nabla^2 g\|_{\rm op}\le L_g$ throughout this neighborhood.
\end{assumption}

This assumption covers all standard LLM evaluation targets: a head-to-head win probability $\psi(T)=\sigma(T_{a,u}-T_{b,u})$ has $s_\psi=2$; a smoothed win rate averaged over a finite reference pool has $s_\psi$ equal to the pool size times the number of categories. The finite-support condition ensures $\|\Gamma_\psi\|_1=O(1)$ and $\|\hat\Gamma_\psi-\Gamma_\psi\|_1\lesssim\|\Delta\|_\infty$.

The score-whitened one-step estimator is
\begin{equation}\label{eq:onestep-nonlinear}
\hat\psi_{\rm nl}
=
\psi(\hat T)
+
\mathbb{P}_n\!\big[\tilde s_\eta(Y,\hat\eta)\,\langle \hat H_\psi,X\rangle\big],
\qquad
\hat H_\psi:=d^\star \widehat{P}_{\mathbb{T}} \hat\Gamma_\psi.
\end{equation}

\begin{theorem}[CLT for nonlinear functionals]\label{thm:nonlinear}
Under the conditions of Theorem~\ref{thm:ws-clt} and Assumption~\ref{ass:nonlinear},
\[
\sqrt{n}\bigl(\hat\psi_{\rm nl}-\psi(T^\star)\bigr)
\;\xrightarrow{d}\;
\mathcal{N}\!\left(0,\;V_{\rm ws}(\psi)\right),
\quad
V_{\rm ws}(\psi)
\;=\;
\E^\star\!\left[\frac{\ip{\dstar P_{\mathbb{T}}\Gammavec_\psi}{X}^2}{I(\eta^\star)}\right].
\]
\end{theorem}

The proof is deferred in Appendix~\ref{app:nonlinear}. The key idea is to reduce the nonlinear problem to the linear one via a three-term decomposition: (i)~the linearized estimator with the true gradient $\Gamma_\psi$, which is exactly the linear problem of Theorem~\ref{thm:ws-bound}; (ii)~a second-order plug-in Taylor remainder $R_{\mathrm{plug},\psi} = O(\|\Delta\|_\infty^2)$ from the nonlinear functional itself; and (iii)~a gradient-perturbation correction from using $\hat\Gamma_\psi$ in place of $\Gamma_\psi$, which carries an extra factor of $\|\Delta\|_\infty$ and is therefore lower order.


\input{simulation_section}

\input{aos_LowrankLLMEvaluation.bbl}
\newpage
\appendix

\renewcommand{\theHsection}{appendix.\Alph{section}}
\renewcommand{\theHsubsection}{appendix.\Alph{section}.\arabic{subsection}}
\renewcommand{\theHsubsubsection}{appendix.\Alph{section}.\arabic{subsection}.\arabic{subsubsection}}
\renewcommand{\theHequation}{appendix.\Alph{section}.\arabic{equation}}
\renewcommand{\theHtheorem}{appendix.\Alph{section}.\arabic{theorem}}

\begin{center}    
{\large\bf SUPPLEMENTARY MATERIAL}
\end{center}
\bigskip

This supplementary material contains additional experiments and technical details. Appendix~\ref{app:additional-experiments} provides extra simulation and real-data evidence omitted from the main text. Appendix~\ref{app:sec4} develops the semiparametric efficiency arguments and proofs for Section~\ref{sec:semiparametric_formulation}.
Appendix~\ref{app:proof-details} provides the detailed proofs for Section~\ref{sec:upper-bound}. Appendix~\ref{app:score-whitening} and Appendix~\ref{app:norm-score} provide proofs for Section~\ref{sec:score-whitening} that covers score whitening, inverse-probability weighting, and the nonlinear-functional extension.
Appendix~\ref{app:clt-upgrade} provides the key Berry--Esseen bound for the pairwise-comparison estimators, used in the main proof.
Appendix~\ref{sec:applications} includes broader implications of the inference framework for tensor completion.

\section{Additional numerical experiments}\label{app:additional-experiments}

This appendix collects additional numerical results for Section~\ref{sec:simulations}. We provide CLT diagnostics for both linear and nonlinear targets, broader robustness checks across sample size and rank, further diagnostics under non-uniform sampling, and additional real-data analyses on \textsc{Arena}, including preprocessing details and sensitivity to the subsample-fractions.

\subsection{Additional CLT diagnostics}\label{app:exp-clt-extra}

Figure~\ref{fig:clt-entry} displays the $z$-score histograms for the linear entry target with $n = 60{,}000$ ($n/d^2 = 1.5$). Both the one-step efficient and score-whitened estimators exhibit close agreement with the $\mathcal{N}(0,1)$ reference density, confirming the CLT.

\begin{figure}[h!]
\centering
\includegraphics[width=\textwidth]{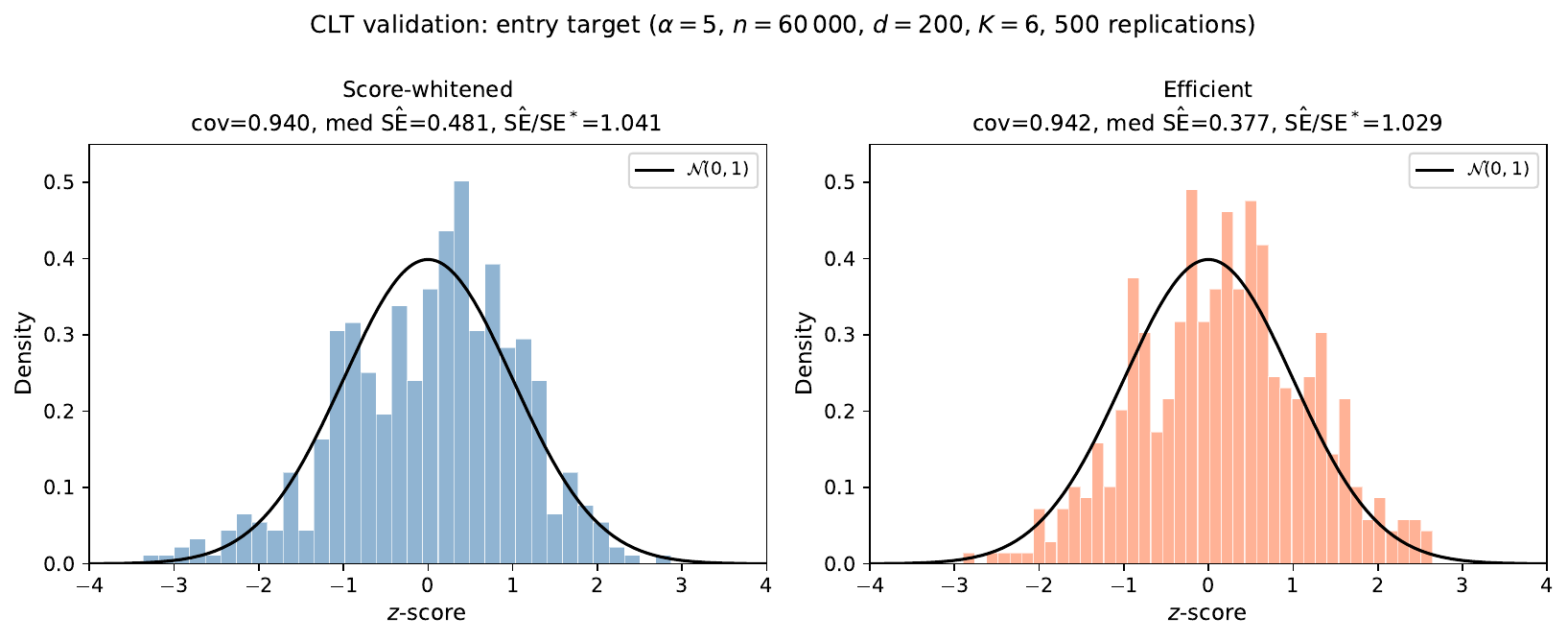}
\caption{CLT validation for the entry target ($\alpha = 5$, $n = 60{,}000$, $d = 200$, $K = 6$): $z$-score histograms for the score-whitened (left) and efficient (right) estimators, compared to the $\mathcal{N}(0,1)$ reference density. Both methods exhibit close agreement with the standard normal.}
\label{fig:clt-entry}
\end{figure}

Moreover, Figure~\ref{fig:clt-winprob} reports the histogram diagnostics for the nonlinear win-probability target in the synthetic experiment. 
The nonlinear target requires larger sample sizes for the Gaussian approximation to stabilize, which is consistent with the additional higher-order remainder terms in Section~\ref{sec:nonlinear-functional}.

\begin{figure}[h!]
\centering
\includegraphics[width=\textwidth]{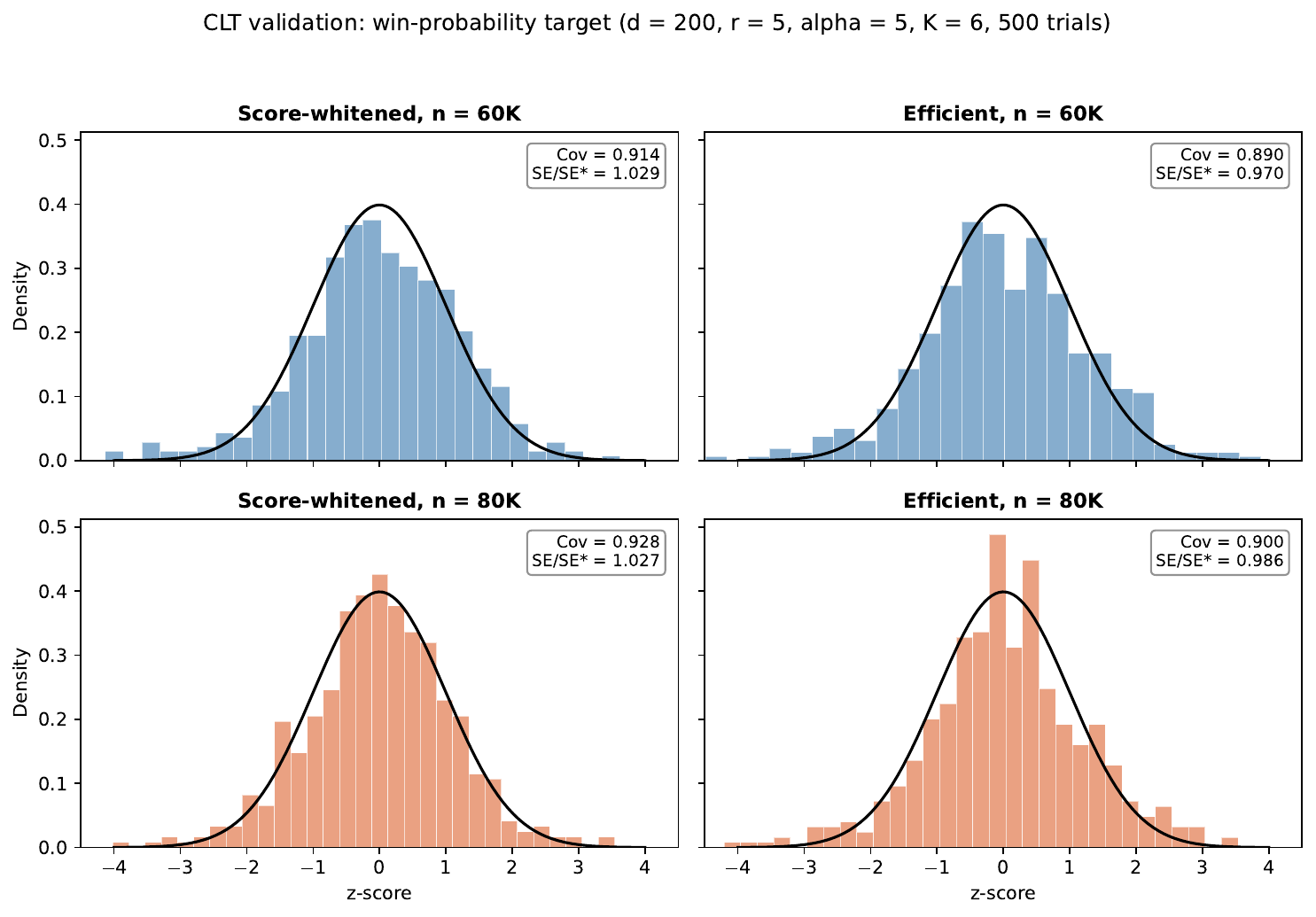}
\caption{CLT validation for the nonlinear win-probability target. 
Top row: $n=60{,}000$. Bottom row: $n=80{,}000$. 
Left: score-whitened estimator. Right: efficient estimator.}
\label{fig:clt-winprob}
\end{figure}

\subsection{Additional robustness checks}\label{sec:sim-coverage}

We further examine robustness to the sample size and the rank, while fixing $d=200$ and $\alpha=5$.

Figure~\ref{fig:coverage-n} plots the $95\%$ CI coverage and variance calibration for the entry target as $n/d^2$ increases from $0.75$ to $2.5$. Both estimators maintain coverage close to the nominal 95\% level, with the efficient estimator showing consistently stable variance calibration. The score-whitened estimator is more variable at smaller sample sizes, but improves as $n$ grows.

Figure~\ref{fig:coverage-r} reports the same quantities at fixed $n=60{,}000$ as the rank varies over $r\in\{2,3,5,8,10\}$. Coverage remains satisfactory across all ranks for both methods. As the rank increases, the score-whitened estimator becomes more variable, whereas the efficient estimator remains well calibrated. 

\begin{figure}[h!]
\centering
\includegraphics[width=0.8\textwidth]{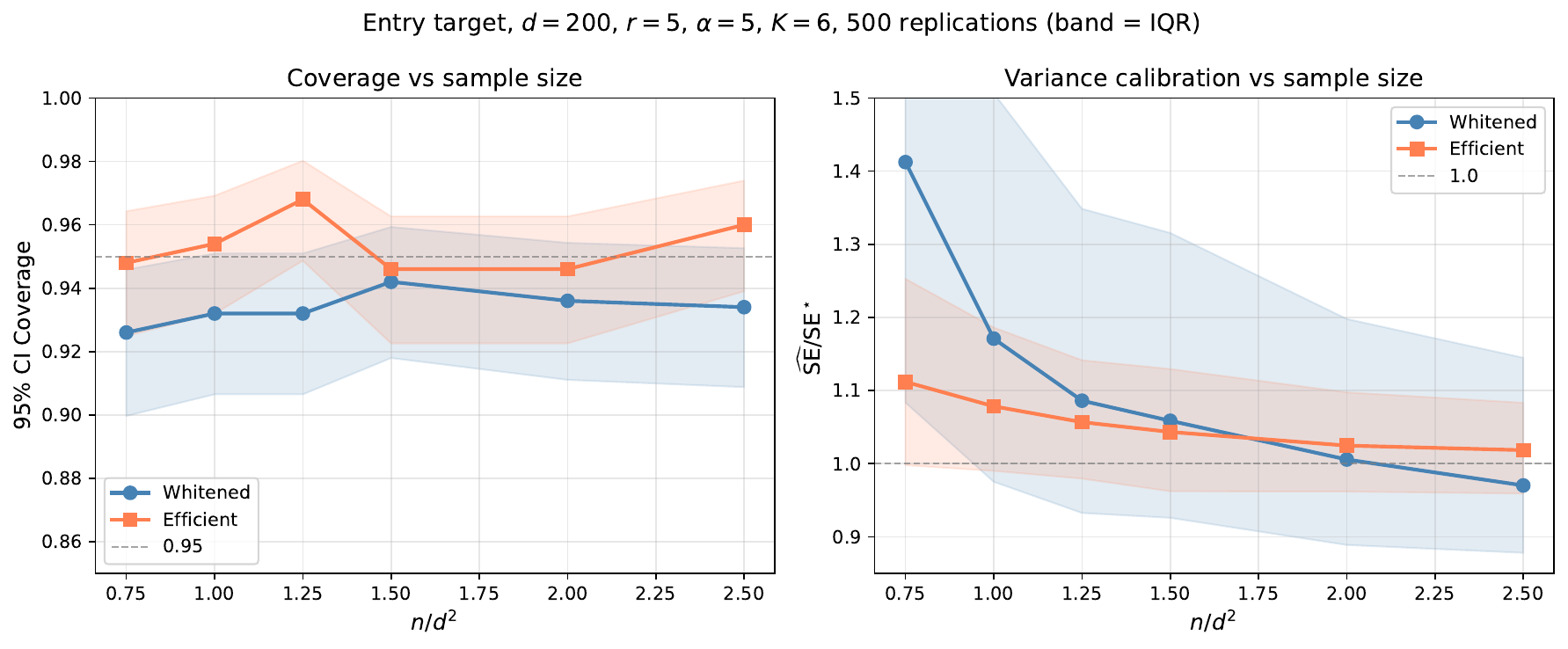}
\caption{Coverage (left) and variance calibration (right) for the entry target as a function of $n/d^2$ ($d=200$, $r=5$, $\alpha=5$, $K=6$, 500 replications).}
\label{fig:coverage-n}
\end{figure}

\begin{figure}[h!]
\centering
\includegraphics[width=0.8\textwidth]{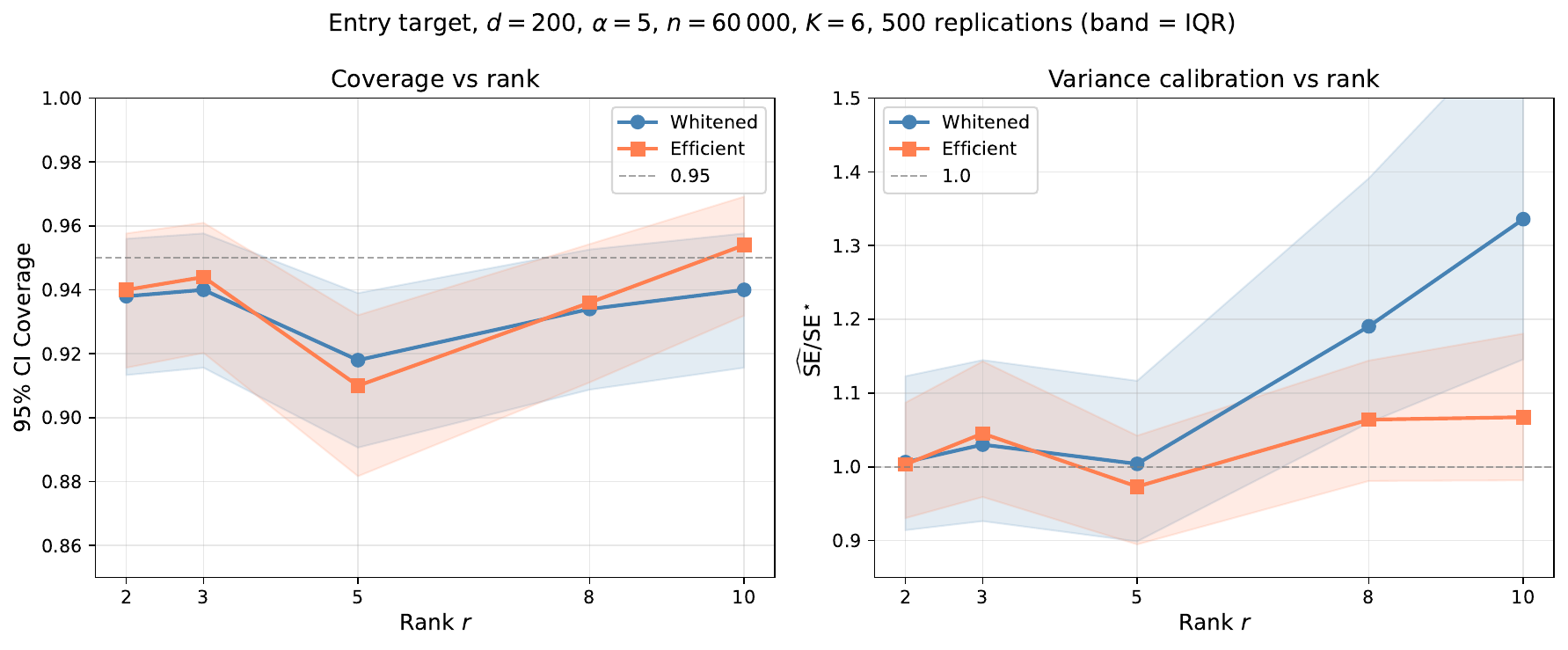}
\caption{Coverage (left) and variance calibration (right) for the entry target as a function of rank $r$ ($d=200$, $\alpha=5$, $n=60{,}000$, $K=6$, 500 replications).}
\label{fig:coverage-r}
\end{figure}

\subsection{Additional diagnostics under non-uniform sampling}\label{app:exp-nonuniform-extra}

Figure~\ref{fig:clt-nonuniform} reports the $z$-score histograms for the four methods in the non-uniform-sampling experiment. 
All four estimators exhibit approximately Gaussian behavior, and the known-$\pi$ and estimated-$\hat\pi$ versions are nearly indistinguishable.

\begin{figure}[h!]
\centering
\includegraphics[width=\textwidth]{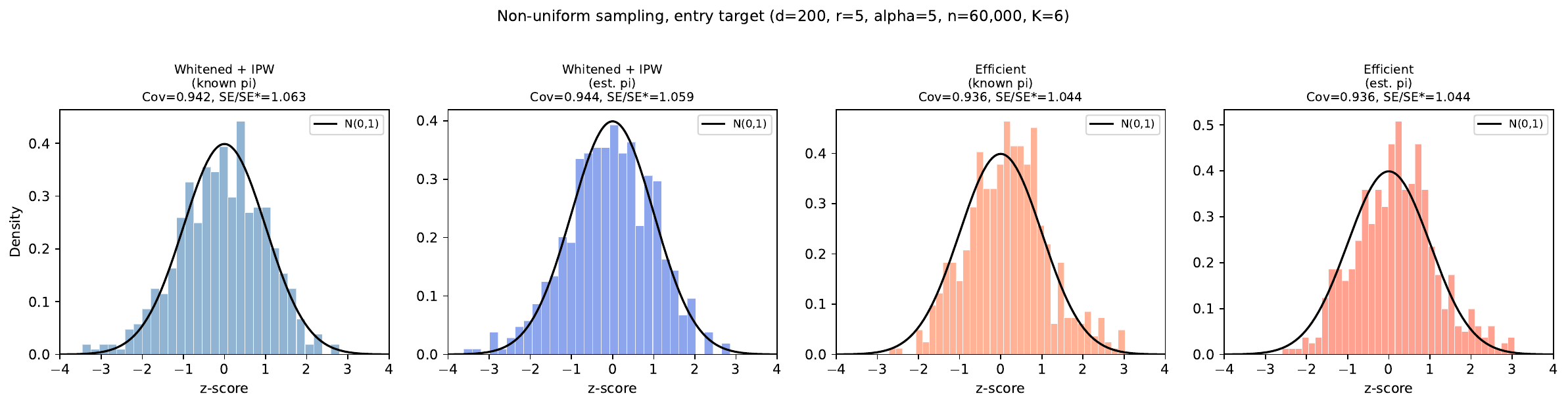}
\caption{CLT diagnostics for the non-uniform-sampling experiment.}
\label{fig:clt-nonuniform}
\end{figure}

\subsection{Real-data preprocessing and robustness evaluation}\label{app:exp-real-preprocess}

We use a public \textsc{Arena} snapshot containing about 140,000 pairwise comparisons. 
After filtering to the 30 most frequently appearing models and aggregating prompts into 10 task categories based on metadata tags, we obtain 81,150 observations. 
Table~\ref{tab:categories} lists the category counts. The distribution is substantially non-uniform: the largest category (\texttt{code\_technical}, $n = 19{,}091$, 23.5\%) contains $7.1\times$ more observations than the smallest (\texttt{code\_general}, $n = 2{,}699$, 3.3\%).

\begin{table}[h!]
\centering
\caption{Task categories and counts after filtering to the top 30 models in the \textsc{Arena} dataset.}
\label{tab:categories}
\smallskip
\begin{tabular}{lrr}
\hline
Category & $n$ & \% \\
\hline
\texttt{code\_technical}      & 19{,}091 & 23.5 \\
\texttt{general}              & 14{,}682 & 18.1 \\
\texttt{creative\_practical}  & 12{,}395 & 15.3 \\
\texttt{math}                 &  6{,}583 &  8.1 \\
\texttt{creative\_writing}    &  6{,}432 &  7.9 \\
\texttt{instruction\_following} & 5{,}885 & 7.3 \\
\texttt{analytical}           &  5{,}686 &  7.0 \\
\texttt{creative\_abstract}   &  4{,}744 &  5.8 \\
\texttt{domain\_knowledge}    &  2{,}953 &  3.6 \\
\texttt{code\_general}        &  2{,}699 &  3.3 \\
\hline
Total & 81{,}150 & 100.0 \\
\hline
\end{tabular}
\end{table}



Table~\ref{tab:real-data-full} reports the full real-data results for all four inference targets. The efficient one-step estimator achieves coverage between 0.930 and 0.976, the closest to the nominal 95\% level among the three methods, and produces the \emph{narrowest confidence intervals}.

\begin{table}[h!]
\centering
\caption{Full real-data results on \textsc{Arena}.}
\label{tab:real-data-full}
\smallskip
\begin{tabular}{llcccc}
\hline
Target & $\psi^\star$ & Method & Emp.\ std & med $\widehat{\rm SE}$ & Coverage \\
\hline
\multirow{3}{*}{$T^\star[\text{Gemini},\,\text{math}]$}
& \multirow{3}{*}{$+0.687$}
& Naive     & 0.199 & 0.274 & 0.998 \\
& & Whitened  & 0.355 & 0.523 & 0.986 \\
& & Efficient & 0.151 & 0.163 & 0.930 \\
\hline
\multirow{3}{*}{Gemini--Claude (math)}
& \multirow{3}{*}{$+0.630$}
& Naive     & 0.248 & 0.271 & 0.980 \\
& & Whitened  & 0.503 & 0.736 & 0.984 \\
& & Efficient & 0.203 & 0.226 & 0.954 \\
\hline
\multirow{3}{*}{$P(\text{Gemini}>\text{Claude}\mid \text{math})$}
& \multirow{3}{*}{$+0.652$}
& Naive     & 0.055 & 0.061 & 0.954 \\
& & Whitened  & 0.115 & 0.170 & 0.986 \\
& & Efficient & 0.047 & 0.052 & 0.954 \\
\hline
\multirow{3}{*}{Gemini--o3 (code\_tech)}
& \multirow{3}{*}{$+0.273$}
& Naive     & 0.140 & 0.155 & 0.976 \\
& & Whitened  & 0.378 & 0.476 & 0.976 \\
& & Efficient & 0.116 & 0.141 & 0.976 \\
\hline
\end{tabular}
\end{table}



Finally, to examine how the benefit of the low-rank efficient estimator changes with data availability, we repeat the naive-versus-efficient comparison at subsample fractions of \(10\%\), \(5\%\), and \(2\%\); see Table~\ref{tab:subsample-fractions}. The advantage of the efficient estimator becomes more pronounced as the sample size decreases. At \(20\%\) subsampling, both methods perform reasonably well, with the efficient estimator achieving about \(1.24\times\) smaller empirical standard deviation on average. As the fraction decreases to \(10\%\) and \(5\%\), this gain increases to about \(1.47\times\) and \(1.50\times\), respectively, while coverage remains adequate. At \(2\%\) subsampling, the naive BTL estimator becomes highly unstable, whereas the efficient estimator remains usable, though with some mild undercoverage. Overall, these results show that the low-rank structure becomes increasingly valuable in data-scarce regimes, where borrowing strength across task categories leads to substantially tighter inference.

\begin{table}[h!]
\centering
\caption{Naive versus efficient estimator at different subsample fractions of the \textsc{Arena} data.}
\label{tab:subsample-fractions}
\smallskip
\begin{tabular}{ccccccc}
\hline
& & \multicolumn{2}{c}{Naive BTL} & \multicolumn{2}{c}{Efficient} & \\
\cmidrule(lr){3-4}\cmidrule(lr){5-6}
Frac & $n_{\rm sub}$ & Emp.\ std & Cov & Emp.\ std & Cov & Std ratio \\
\hline
20\% & 16{,}230 & 0.160 & 0.977 & 0.129 & 0.954 & 1.24$\times$ \\
10\% & 8{,}115  & 0.298 & 0.956 & 0.203 & 0.952 & 1.47$\times$ \\
5\%  & 4{,}058  & 0.379 & 0.945 & 0.252 & 0.932 & 1.50$\times$ \\
2\%  & 1{,}623  & 2.245 & 0.947 & 0.408 & 0.915 & 5.50$\times$ \\
\hline
\end{tabular}
\end{table}

\section{Proofs for Section~\ref{sec:semiparametric_formulation}}\label{app:sec4}

\subsection{Proof of Theorem~\ref{thm:general-lb}}
\label{app:proof-general-lb}

Fix any $H\in\mathbb{T}$ and consider the one-dimensional submodel $T_{\varepsilon}=T^\star+\varepsilon H$
(with the nuisance held fixed at $\Pi^\star$).
By Assumption~\ref{ass:global-unbiased},
\[
\mathbb{E}_{T_{\varepsilon},\Pi}[\widehat\psi]=\psi(T_{\varepsilon})\qquad \forall \varepsilon.
\]
Differentiating both sides at $\varepsilon=0$ and using \eqref{eq:lb-pathwise} gives
\[
\left.\frac{d}{d\varepsilon}\mathbb{E}_{T_{\varepsilon},\Pi}[\widehat\psi]\right|_{\varepsilon=0}
=
D\psi_{T^\star}(H).
\]
On the other hand, writing $\mathcal{S}(H)(X,Y):=s_\eta(Y,\eta^\star)\langle H,X\rangle$ for the directional score, the standard score identity yields
\[
\left.\frac{d}{d\varepsilon}\mathbb{E}_{T_{\varepsilon},\Pi^\star}[\widehat\psi]\right|_{\varepsilon=0}
=
\mathbb{E}^\star\!\left[\widehat\psi\cdot\sum_{i=1}^n s_\eta(Y_i,\eta_i^\star)\langle H,X_i\rangle\right].
\]
Using $\mathbb{E}^\star[\widehat\psi]=\psi(T^\star)$ and $\mathbb{E}^\star[s_\eta(Y,\eta^\star)\mid X]=0$, the right-hand side equals
\[
\mathbb{E}^\star\!\left[(\widehat\psi-\psi(T^\star))\cdot\sum_{i=1}^n s_\eta(Y_i,\eta_i^\star)\langle H,X_i\rangle\right].
\]
Applying Cauchy--Schwarz and i.i.d.\ sampling gives
\[
D\psi_{T^\star}(H)^2
\le
\mathrm{Var}^\star(\widehat\psi)\cdot
\mathbb{E}^\star\!\left[\Big(\sum_{i=1}^n s_\eta(Y_i,\eta_i^\star)\langle H,X_i\rangle\Big)^2\right]
=
\mathrm{Var}^\star(\widehat\psi)\cdot n\,\mathbb{E}^\star[s_\eta(Y,\eta^\star)^2\langle H,X\rangle^2],
\]
which implies the directional bound.
The sharp form \eqref{eq:lb-sharp} follows by maximizing the Rayleigh quotient
$\langle P_{\mathbb{T}}\nabla\psi(T^\star),H\rangle^2/\langle H,A\,H\rangle$ over $H\in\mathbb{T}$, whose maximizer is $H\propto A^{-1}P_{\mathbb{T}}\nabla\psi(T^\star)$ under
Assumption~\ref{ass:lb-invertible}. That is, the submodel in the direction $H^\star$ derived in Section~\ref{sec:info-operator} is the hardest to learn among all directions $H \in \mathbb{T}$.

Under Assumption~\ref{ass:lb-invertible}, the information equation
\[
A\,H^\star=P_{\mathbb{T}}\nabla\psi(T^\star)
\]
has the unique solution $H^\star=A^{-1}P_{\mathbb{T}}\nabla\psi(T^\star)$.
The efficient influence function is
\[
\phi_{\mathrm{EIF}}(X,Y)
=s_\eta(Y,\eta^\star)\Big\langle A^{-1}P_{\mathbb{T}}\nabla\psi(T^\star),\,X\Big\rangle,
\]
and satisfies $\mathbb{E}^\star[\phi_{\mathrm{EIF}}^2]=V_{\mathrm{eff}}(\psi)$, showing that the lower bound
\eqref{eq:lb-sharp} is attainable. \qed

\subsection{Semiparametric derivation of the EIF}\label{app:semiparametric-details}

We provide the full semiparametric argument establishing that $\phi_{\mathrm{EIF}}$ defined in \eqref{eq:eif-main} is the efficient influence function.

Recall the joint density factorization $p_{T,\Pi}(X,Y)=g_{\Pi}(X)\,p(Y\mid \langle T,X\rangle)$.
Define the signal tangent space in $L^2(P^\star)$ by
$\mathcal{T}_T:=\overline{\{s_\eta(Y,\eta^\star)\langle H,X\rangle: H\in\mathbb{T}\}}^{L^2(P^\star)}$,
and the nuisance tangent space $\mathcal{T}_X:=\{h(X): \mathbb{E}^\star[h(X)]=0,\ \mathbb{E}^\star[h(X)^2]<\infty\}$.
Since $\mathbb{E}^\star[s_\eta(Y,\eta^\star)\mid X]=0$, it follows that $\mathcal{T}_T\perp\mathcal{T}_X$ in $L^2(P^\star)$, and the full tangent space is the orthogonal direct sum $\mathcal{T}=\mathcal{T}_T\oplus\mathcal{T}_X$.

An \emph{influence function} (IF) for a regular estimator of $\psi(T^\star)$ is any $\phi\in L^2_0(P^\star)$ satisfying
\[
D\psi_{T^\star}(H)
=
\mathbb{E}^\star\!\left[\phi(X,Y)\,\dot\ell_\varepsilon(X,Y)\right]
\]
for every regular submodel $\varepsilon\mapsto(T_\varepsilon,\Pi_\varepsilon)$, where $\dot\ell_\varepsilon$ is the full model score.
Since $\psi$ depends only on $T$, the derivative vanishes for pure $\Pi$-perturbations, and the identity constrains $\phi$ only through its projection onto $\mathcal{T}_T$.
It follows that IFs are not unique ($\phi+u$ is an IF whenever $u\in\mathcal{T}^\perp$), and the minimum-variance IF is
\[
\phi_{\mathrm{EIF}}:=\mathrm{Proj}_{\mathcal{T}}(\phi_0)\in\mathcal{T}_T
\]
for any IF $\phi_0$.  Since $\phi_{\mathrm{EIF}}\in\mathcal{T}_T$, there exists $H^\star\in\mathbb{T}$ such that $\phi_{\mathrm{EIF}}(X,Y)=s_\eta(Y,\eta^\star)\langle H^\star,X\rangle$.
The score identity for all signal submodels $T_\varepsilon=T^\star+\varepsilon H$ yields
\[
\langle P_{\mathbb{T}}\nabla\psi(T^\star),H\rangle
=
\mathbb{E}^\star\!\bigl[s_\eta(Y,\eta^\star)^2\langle H^\star,X\rangle\langle H,X\rangle\bigr]
=
\langle H^\star, A\, H\rangle,
\qquad\forall H\in\mathbb{T},
\]
where $A=P_{\mathbb{T}}GP_{\mathbb{T}}$, giving the information equation $AH^\star=P_{\mathbb{T}}\nabla\psi(T^\star)$.

\subsection{Detailed computations for the pairwise comparison model}\label{app:sec4-details}

We present the detailed computation of the information operator and the efficiency bound for the pairwise comparison model discussed in Section~\ref{sec:semiparametric_formulation}.

\begin{example}
We consider the pairwise comparison model where the true signal matrix belongs to
\[
\mathcal{P}_{r,0}
:=\Bigl\{T\in\mathbb{R}^{d_1\times d_2}:\ \mathrm{rank}(T)=r,\ \ \mathbf{1}_{d_1}^\top T=0\Bigr\}.
\]
Here $d_1$ indexes the items (first mode) and $d_2$ indexes the arenas/users (second mode).
The first-mode sum-to-zero constraint selects a unique representative from each shift-equivalence class and
removes the unidentifiable directions (shifts along the item mode) in preference-based observations; see also
Example~\ref{ex:row-centered}.
At each round we draw an arena $R\sim \mathrm{Unif}([d_2])$ and an unordered pair of items $(J,K)$ uniformly from
$\{(j,k):1\le j<k\le d_1\}$, and observe
\[
Y \mid (R,J,K) \sim \mathrm{Bernoulli}\bigl(g(\eta^\star_{JK,R})\bigr),
\qquad
\eta^\star_{jk,r}:=T^\star_{j,r}-T^\star_{k,r},
\]
where $g(z)=1/(1+e^{-z})$ and $I(z):=g(z)(1-g(z))$.
Define the design matrix
\[
X:= (e_J-e_K) e_R^\top\in\mathbb{R}^{d_1\times d_2},
\qquad
\langle H,X\rangle = H_{J,R}-H_{K,R}.
\]
The score along any parametric submodel $T_\varepsilon=T^\star+\varepsilon H$ is
\[
S(H)(Z)
= (Y-g(\eta^\star))\,\langle H,X\rangle
= (Y-g(\eta^\star))\,(H_{J,R}-H_{K,R}),
\quad Z=(R,J,K,Y).
\]
Recall that the tangent space is (cf.\ Example~\ref{ex:row-centered})
\[
\mathbb{T}_{r,0}
=\{\,U A^\top + Q C V^\top:\ A\in\mathbb{R}^{d_2\times r},\ C\in\mathbb{R}^{(d_1-1)\times r}\,\},
\]
where $Q\in\mathbb{R}^{d_1\times(d_1-1)}$ has orthonormal columns spanning $\mathbf{1}_{d_1}^\perp$.
For any direction $H\in\mathbb{T}_{r,0}$, the Fisher quadratic form is
\[
\|S(H)\|_{L^2(P^\star)}^2
=\mathbb{E}\bigl[I(\eta^\star)\langle H,X\rangle^2\bigr]
=\mathbb{E}\bigl[I(\eta^\star)(H_{J,R}-H_{K,R})^2\bigr].
\]
Conditioning on $R=r$ and writing $h_r\in\mathbb{R}^{d_1}$ for the $r$-th column of $H$, define
the (weighted) Laplacian
\[
L_r
:=\frac{1}{\binom{d_1}{2}}\sum_{1\le j<k\le d_1}
w_{r,jk}\,(e_j-e_k)(e_j-e_k)^\top,
\qquad
w_{r,jk}:=I(T^\star_{j,r}-T^\star_{k,r})\in(0,1/4].
\]
Then
\[
\|S(H)\|_{L^2(P^\star)}^2
=\frac1{d_2}\sum_{r=1}^{d_2} h_r^\top L_r h_r.
\]
Moreover, for each $r$ and any $u\in\mathbb{R}^{d_1}$,
\[
u^\top L_r u
=\frac{1}{\binom{d_1}{2}}\sum_{j<k} w_{r,jk}(u_j-u_k)^2,
\]
so $\ker(L_r)=\operatorname{span}\{\mathbf{1}_{d_1}\}$ and $L_r\succ 0$ on $\mathbf{1}_{d_1}^\perp$.
Since $H\in\mathbb{T}_{r,0}$ implies $\mathbf{1}_{d_1}^\top h_r=0$ for all $r$, we get
\[
\|S(H)\|_{L^2(P^\star)}^2>0
\quad\text{for all }H\in\mathbb{T}_{r,0}\setminus\{0\}.
\]
Hence the restricted Fisher information operator
\[
A
:=P_{\mathbb{T}}\circ G\circ P_{\mathbb{T}}:\ \mathbb{T}_{r,0}\to \mathbb{T}_{r,0}
\]
is injective (and, since $\mathbb{T}_{r,0}$ is finite-dimensional, invertible). This is precisely the
formal sense in which the first-mode sum-to-zero normalization removes the singular kernel present in
pairwise models without normalization (shifts along the item mode).

Let $\psi(T)=\langle \Gamma,T\rangle$ be any linear functional. Its pathwise derivative along $H$ is
$D\psi_{T^\star}(H)=\langle \Gamma,H\rangle$.
Let $b:=P_{\mathbb{T}}(\Gamma)\in\mathbb{T}_{r,0}$ be the orthogonal projection of $\Gamma$ onto $\mathbb{T}_{r,0}$.
Then the information inequality yields
\[
\mathrm{Var}_{T^\star}(\hat\psi_n)
\ \ge\
\frac1n\sup_{H\in\mathbb{T}_{r,0}\setminus\{0\}}
\frac{\langle b,H\rangle^2}{\langle H,A\,H\rangle}
\ =\
\frac1n\,\langle b,A^{-1}b\rangle.
\]
The supremum is attained at the unique \emph{hardest direction}
\[
H^\star=A^{-1}b\in\mathbb{T}_{r,0},
\qquad\text{i.e.}\qquad
A\,H^\star=b.
\]
The corresponding efficient influence function is
\[
\phi_{\mathrm{EIF}}(Z)=S(H^\star)(Z)
=(Y-g(\eta^\star))\,(H^\star_{J,R}-H^\star_{K,R}),
\]
and it satisfies
\[
\mathbb{E}[\phi_{\mathrm{EIF}}]=0,
\qquad
\mathbb{E}[\phi_{\mathrm{EIF}}^2]=\langle b,A^{-1}b\rangle,
\]
so the lower bound is sharp and is achieved by the (oracle) one-step estimator
\[
\tilde\psi_n
:=\psi(T^\star)+\frac1n\sum_{t=1}^n \phi_{\mathrm{EIF}}(Z_t).
\]
We can also provide an explicit matrix form characterization.
Recall that every direction $H$ in tangent space can be indexed by $(A,C)$ via the linear map
\[
H(A,C) \;:=\; U A^\top \;+\; Q C V^\top,\qquad
A\in\mathbb{R}^{d_2\times r},\; C\in\mathbb{R}^{(d_1-1)\times r}.
\]
We first give the Fisher information bilinear form.
In the pairwise-logistic model, for each arena $r\in[d_2]$ there exists a symmetric PSD matrix
$L_r\in\mathbb{R}^{d_1\times d_1}$ (a weighted Laplacian) such that for any two tangent directions
$H_1,H_2$,
\[
\mathbb{E}_{T^\star}[S(H_1)S(H_2)]
=\langle H_1,G(H_2)\rangle
=\sum_{r=1}^{d_2} h_{1,r}^\top L_r\, h_{2,r},
\]
where $h_{m,r}\in\mathbb{R}^{d_1}$ is the $r$-th column of $H_m$.
In the \emph{uniform} pair sampling case,
\[
L_r \;=\; \frac{1}{\binom{d_1}{2}}\sum_{1\le j<k\le d_1}
w_{r,jk}\,(e_j-e_k)(e_j-e_k)^\top,
\qquad
w_{r,jk}:=I(T^\star_{j,r}-T^\star_{k,r})=g(\eta^\star_{jk,r})(1-g(\eta^\star_{jk,r})).
\]
Define the block diagonal matrix
\[
L := \mathrm{blkdiag}(L_1,\dots,L_{d_2})\in\mathbb{R}^{(d_1d_2)\times(d_1d_2)}.
\]
Now define the induced bilinear form on the coordinate space $(A,C)$ by
\[
\mathcal{B}\bigl((A_1,C_1),(A_2,C_2)\bigr)
:= \sum_{r=1}^{d_2} h_r(A_1,C_1)^\top L_r\, h_r(A_2,C_2),
\quad h_r(A,C) := \text{$r$-th column of }H(A,C).
\]
Equip the coordinate space with the Euclidean inner product
$\langle(A_1,C_1),(A_2,C_2)\rangle_\theta:=\langle A_1,A_2\rangle+\langle C_1,C_2\rangle$.
We use the standard identities:
\[
\mathrm{vec}(AXB)=(B^\top\otimes A)\mathrm{vec}(X),
\qquad
\mathrm{vec}(X^\top)=\mathcal{K}_{m,n}\mathrm{vec}(X)\ \ (X\in\mathbb{R}^{m\times n}),
\]
where $\otimes$ is the Kronecker product and $\mathcal{K}_{m,n}$ is the commutation matrix.
Then we derive the Jacobian $J$ of the linear map $(A,C)\mapsto H(A,C)$ in vectorized coordinates.
Let
\[
\theta :=
\begin{bmatrix}
\mathrm{vec}(A)\\
\mathrm{vec}(C)
\end{bmatrix}
\in\mathbb{R}^{d_2 r + (d_1-1) r},
\qquad
\mathrm{vec}(H(A,C))\in\mathbb{R}^{d_1 d_2}.
\]
Because $H(A,C)$ is linear in $(A,C)$, there exists a (constant) matrix $J$ such that
\[
\mathrm{vec}(H(A,C)) = J\,\theta.
\]
We compute $J$ explicitly.
\[
\mathrm{vec}(U A^\top)
=(I_{d_2}\otimes U)\mathrm{vec}(A^\top)
=(I_{d_2}\otimes U)\,\mathcal{K}_{d_2,r}\,\mathrm{vec}(A).
\]
Hence
\[
J_A := (I_{d_2}\otimes U)\,\mathcal{K}_{d_2,r}
\in\mathbb{R}^{(d_1d_2)\times(d_2 r)}.
\]
Similarly,
\[
\mathrm{vec}(Q C V^\top)=(V\otimes Q)\mathrm{vec}(C),
\]
so
\[
J_C := (V\otimes Q)
\in\mathbb{R}^{(d_1d_2)\times((d_1-1) r)}.
\]
Combining the two Jacobians, we have
\[
J=\begin{bmatrix}J_A & J_C\end{bmatrix}.
\]
Using the above representation on $(A,C)$, and $\mathcal{B}((A_1,C_1),(A_2,C_2))
=\mathrm{vec}(H_1)^\top L\,\mathrm{vec}(H_2)$, we obtain
\[
\mathcal{B}\bigl((A_1,C_1),(A_2,C_2)\bigr)
= \theta_1^\top \underbrace{(J^\top L J)}_{=:K}\,\theta_2.
\]
Therefore, the unique linear operator $K$ satisfying
$\mathcal{B}(\theta_1,\theta_2)=\langle \theta_1, K\theta_2\rangle$ (Riesz representation
under the Euclidean inner product on $\theta$) has the \emph{matrix form}
\begin{equation}
    K = J^\top\,\mathrm{blkdiag}(L_1,\dots,L_{d_2})\,J.
\end{equation}
We also present the matrix $K$ in block forms where
\[
\begin{aligned}
K_{AA} &= J_A^\top L J_A
= \mathcal{K}_{d_2,r}^\top (I_{d_2}\otimes U^\top)\,L\,(I_{d_2}\otimes U)\,\mathcal{K}_{d_2,r},\\
K_{AC} &= J_A^\top L J_C
= \mathcal{K}_{d_2,r}^\top (I_{d_2}\otimes U^\top)\,L\,(V\otimes Q),\\
K_{CC} &= J_C^\top L J_C
= (V^\top\otimes Q^\top)\,L\,(V\otimes Q),\\
K_{CA} &= K_{AC}^\top.
\end{aligned}
\]
We can also calculate the gradient of a linear functional in $(A,C)$-coordinates.
For $\psi(T)=\langle\Gamma,T\rangle$, the directional derivative along $H(A,C)$ is
\[
D\psi_{T^\star}[H(A,C)]
=\langle\Gamma, U A^\top\rangle + \langle\Gamma, Q C V^\top\rangle
=\langle \Gamma^\top U,\ A\rangle + \langle Q^\top\Gamma V,\ C\rangle.
\]
Hence the coordinate gradient of $D\psi_{T^\star}$ is
\begin{equation}
\nabla_A\psi = \Gamma^\top U\in\mathbb{R}^{d_2\times r},
\qquad
\nabla_C\psi = Q^\top\Gamma V\in\mathbb{R}^{(d_1-1)\times r}.
\end{equation}
Equivalently, in vectorized coordinates
\[
g_\theta :=
\begin{bmatrix}
\mathrm{vec}(\nabla_A\psi)\\
\mathrm{vec}(\nabla_C\psi)
\end{bmatrix}
=
\begin{bmatrix}
\mathrm{vec}(\Gamma^\top U)\\
\mathrm{vec}(Q^\top\Gamma V)
\end{bmatrix}.
\]
Finally we can solve the information equation, EIF, and efficiency bound in a matrix form.
Let $\theta^\star=(A^\star,C^\star)$ denote a solution of the coordinate-space information equation
\begin{equation}
K\,\theta^\star = g_\theta,
\end{equation}
where one can take $\theta^\star=K^\dagger g_\theta$ using the Moore--Penrose pseudoinverse, as we have proved that in the tangent space, the solution to information equation is unique.
Define $H^\star := H(A^\star,C^\star)=U A^{\star\top} + Q C^\star V^\top$.
Then the efficient influence function is
\begin{equation}
\phi_{\mathrm{EIF}}(R,J,K,Y)
=(Y-g(\eta^\star_{JK,R}))\,(H^\star_{J,R}-H^\star_{K,R}),
\end{equation}
and the corresponding semiparametric efficiency bound takes the matrix form
\begin{equation}
\mathrm{Var}_{T^\star}(\widehat\psi_n)
\ \ge\ \frac1n\, g_\theta^\top K^\dagger g_\theta.
\end{equation}
\end{example}

\input{appendix_content}

\end{document}

%% file: simulation_section.tex
\section{Numerical experiments}\label{sec:simulations}

This section evaluates the main theoretical claims using both synthetic and real-data experiments. We first show that low-rank initialization is essential in sparse pairwise-comparison settings and substantially improves over a naive per-task BTL fit. We then validate asymptotic normality and variance calibration for both linear and nonlinear targets, and study the non-uniform-sampling extension. Finally, we evaluate the proposed methods on real \textsc{Arena} dataset.  Additional implementation details and supplementary experiments on robustness evaluation are deferred to Appendix~\ref{app:additional-experiments}.
The construction below is a numerical implementation of the generic initial
estimator in Assumption~\ref{ass:initial}; the theoretical results require
only the properties stated in that assumption.

All synthetic experiments use $d_1=d_2=d=200$, rank $r=5$, and the standard logistic link $\sigma(\eta)=(1+e^{-\eta})^{-1}$. 
The ground-truth matrix $T^\star\in\mathbb{R}^{d\times d}$ is generated as $T^\star=\Theta A^\top$, where $\Theta,A\in\mathbb{R}^{d\times r}$ have i.i.d.\ standard normal entries. We then center and rescale $T^\star$ so that its signal strength satisfies $\|T^\star\|_\infty=\alpha$ with $\alpha=5$. Except in the non-uniform-sampling experiment of Section~\ref{sec:sim-nonuniform}, each observation $(u_i,p_i,q_i,Y_i)$ is generated as follows: a task category $u_i$ is sampled uniformly, two distinct models $p_i\neq q_i$ are sampled uniformly, and the comparison outcome is then drawn according to
$
Y_i\sim \mathrm{Bernoulli}\!\left(\sigma\!\bigl(T^\star_{p_i,u_i}-T^\star_{q_i,u_i}\bigr)\right).
$

Throughout, we compare the semiparametrically efficient one-step estimator $\hat\psi_{\rm eff}$ from~\eqref{eq:onestep} with the score-whitened estimator $\hat\psi_{\rm ws}$ from~\eqref{eq:ws-estimator}. 
All point estimates use $K$-fold cross-fitting with $K=6$. 
The plug-in standard error $\widehat{\rm SE}$ is computed from the full sample by fitting $\hat T$ on all observations, constructing the corresponding influence-function direction, and estimating the asymptotic variance by the empirical second moment of the estimated influence function.

\subsection{Initialization and comparison with the baseline}\label{sec:sim-altmin}

We first examine the quality of the initial low-rank estimator. We parameterize the latent score matrix in low-rank form as 
$
T = UV^\top,
U,V\in\mathbb{R}^{d\times r},
$
and optimize the BTL log-likelihood over \((U,V)\) under this factorization. Since the problem is nonconvex, we start the procedure by a spectral estimator constructed from the observed pairwise-comparison matrix. Then we run an alternating minimization (AltMin) procedure to update \(V\) and \(U\). After each round, we project the fitted matrix back to rank \(r\) by truncated SVD and clip its entries to the interval \([-\alpha_0,\alpha_0]\), with \(\alpha_0=\alpha+2\), to stabilize the iterates and prevent extreme fitted logits. In our implementation, three alternating-minimization rounds are sufficient to obtain an accurate estimator \(\hat T\).

We then apply an entrywise refinement step to improve the accuracy of individual entries, which is the scale most relevant for the one-step inference procedures. Specifically, let \(\hat T=\hat U\hat \Sigma \hat V^\top\) be the rank-\(r\) SVD of the alternating-minimization estimator, and write \(\hat A=\hat V\hat \Sigma^{1/2}\) for the estimated column factor. Treating \(\hat A\) as fixed, we re-estimate each row loading of \(T^\star\) by fitting a row-wise logistic regression using only the comparisons involving that row. This yields a refined estimate of the row factor. We then reverse the roles of rows and columns and similarly refine the column factor. The resulting estimator, denoted by \(\hat T_{\rm ref}\), preserves the global low-rank structure learned by alternating minimization while correcting local entrywise bias, and in our experiments it substantially improves estimation accuracy.




Figure~\ref{fig:convergence} reports the estimation error as a function of sample size, comparing alternating minimization alone (AltMin) with alternating minimization followed by refinement (AltMin + Refinement). 
Refinement yields clear improvements in both relative Frobenius error and entrywise $\ell_\infty$ error once the sample size is moderately large. 

Moreover, to illustrate the value of the low-rank structure, we also compare with a naive approach that fits a separate BTL model independently within each task category, ignoring any shared structure across tasks. 
At $n=60{,}000$, each task receives only about 300 observations spread across $\binom{d}{2}=19{,}900$ possible model pairs, so most pairs are never observed within a task. 
As shown in Table~\ref{tab:naive-altmin}, the naive estimator performs dramatically worse: both its Frobenius and entrywise errors are an order of magnitude larger than those of the low-rank estimator (AltMin). 
This confirms that the low-rank assumption is not merely a modeling convenience, but is essential for stable estimation and inference in the sparse regime.

\begin{figure}
\centering
\includegraphics[width=0.9\textwidth]{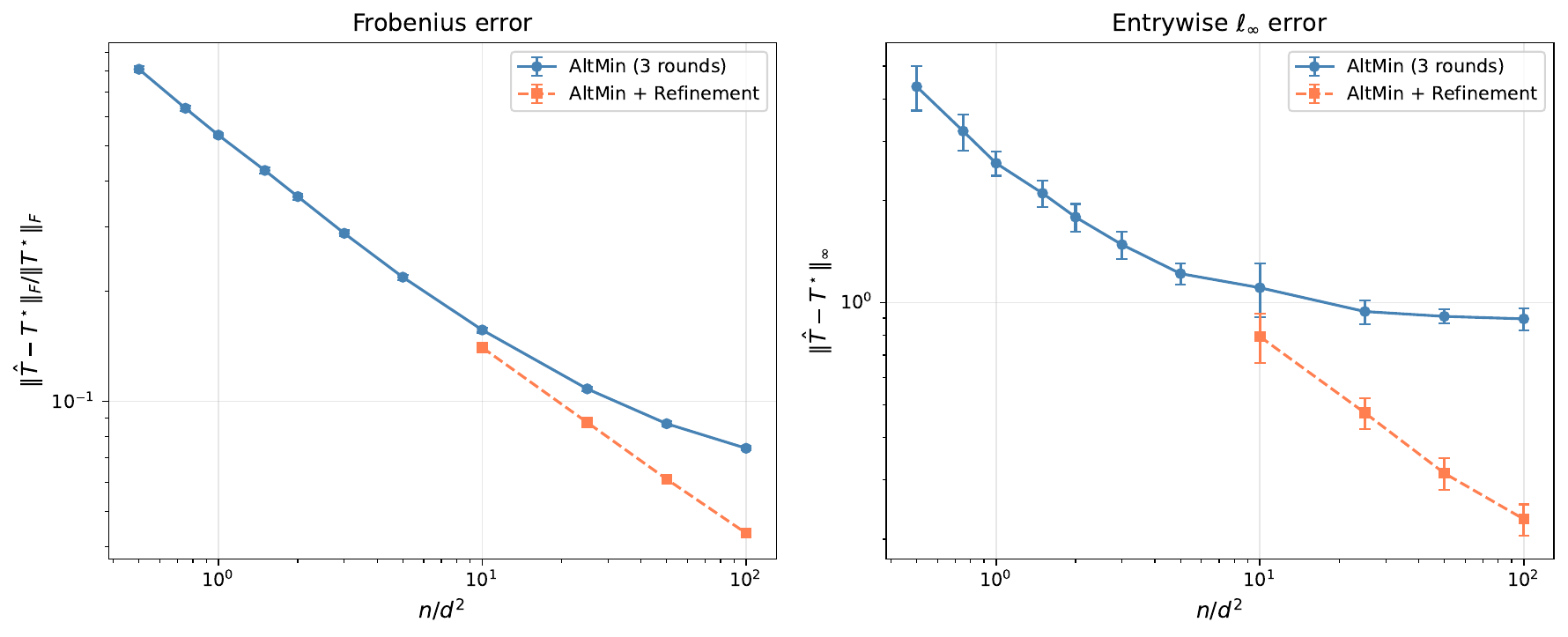}
\caption{Estimation error versus sample size: relative Frobenius error (left) and entrywise $\ell_\infty$ error (right) for alternating minimization alone and alternating minimization followed by entrywise refinement. Refinement substantially improves the accuracy of the initial estimator.}
\label{fig:convergence}
\end{figure}

\begin{table}[t]
\centering
\caption{Comparison of the naive per-task BTL estimator and the low-rank alternating-minimization estimator. }
\label{tab:naive-altmin}
\smallskip
\begin{tabular}{lccc}
\hline
Method & $\|\hat T-T^\star\|_F/\|T^\star\|_F$ & $\|\hat T-T^\star\|_\infty$ & Mean $|\hat T_{ij}-T^\star_{ij}|$ \\
\hline
Naive per-task BTL & 4.01 & 12.3 & 2.69 \\
AltMin    & 0.41 & 1.8  & 0.28 \\
\hline
\end{tabular}
\end{table}

\subsection{Asymptotic normality and variance calibration}\label{sec:sim-clt}

We evaluate the proposed inference procedures on two representative targets. The first is a \emph{linear} target,
$
\psi_1(T^\star)=T^\star_{a_0,u_0},
$
which corresponds to the latent ability of model \(a_0\) on task \(u_0\). The second is a \emph{nonlinear} target,
$
\psi_2(T^\star)=\sigma\!\bigl(T^\star_{a_0,u_0}-T^\star_{b_0,u_0}\bigr),
$
which represents the probability that model \(a_0\) beats model \(b_0\) on task \(u_0\). Throughout the experiments, we fix \((a_0,u_0,b_0)=(1,1,2)\).

For each configuration, we run 500 independent replications. 
In each replication, we generate fresh data, construct the cross-fitted initial estimator (AltMin), compute both the efficient and score-whitened one-step estimators and their standard errors. 

Table~\ref{tab:main-results} summarizes coverage and variance calibration. 
For the linear entry target, both estimators achieve near-nominal coverage and good variance calibration, confirming the asymptotic normality shown in Theorems~\ref{thm:combined} and~\ref{thm:ws-clt}. 
For the nonlinear win-probability target, the same qualitative pattern holds, although the finite-sample approximation is more demanding: coverage improves as $n$ increases, consistent with the additional higher-order remainder in the nonlinear expansion of Section~\ref{sec:nonlinear-functional}. 

Across all settings, the efficient estimator consistently yields smaller standard errors than the score-whitened estimator. This is consistent with the theory: the efficient estimator targets the semiparametric variance bound, while score whitening trades some efficiency for robustness by removing the $C_A$ bottleneck. Our experiments favor the efficient estimator because they are conducted in regimes where the efficient direction can still be estimated stably.


\begin{table}[t]
\centering
\caption{Coverage and variance calibration for the entry target and the nonlinear win-probability target. 
$\widehat{\rm SE}/{\rm SE}^\star$ is the ratio of the median plug-in standard error to the oracle asymptotic standard error computed from $T^\star$.}
\label{tab:main-results}
\smallskip
\begin{tabular}{llcccc}
\hline
Target & Method & $n$ & Coverage & med $\widehat{\rm SE}$ & $\widehat{\rm SE}/{\rm SE}^\star$ \\
\hline
\multirow{2}{*}{Entry}
& Efficient  & 60K & 0.942 & 0.377 & 1.029 \\
& Whitened   & 60K & 0.940 & 0.481 & 1.041 \\
\hline
\multirow{4}{*}{Win-prob}
& Efficient  & 60K & 0.890 & 0.121 & 0.970 \\
& Whitened   & 60K & 0.914 & 0.154 & 1.029 \\
& Efficient  & 80K & 0.900 & 0.106 & 0.987 \\
& Whitened   & 80K & 0.928 & 0.133 & 1.027 \\
\hline
\end{tabular}
\end{table}

\subsection{Non-uniform sampling}\label{sec:sim-nonuniform}

We now consider non-uniform sampling, where the task category and model pair are drawn from heterogeneous distributions. 
Specifically, the task is sampled as $u\sim\pi_J$ and the two models are sampled independently as $p,q\sim\pi_M$ (resampling if $p=q$), where $\pi_J$ and $\pi_M$ are drawn from \(\mathrm{Dirichlet}(5\cdot \mathbf{1}_d)\). This yields mildly non-uniform sampling distributions, with task and model probabilities varying by roughly a factor of 10 between the most and least likely categories, which mimics the real Arena evaluation case.

We compare four methods: the score-whitened estimator with IPW using known $\pi$, the same estimator using estimated $\hat{\pi}$, the efficient one-step estimator using the non-uniform information operator with known $\pi$, and its corresponding version using estimated $\hat{\pi}$. 
Table~\ref{tab:nonuniform} reports the results. Three conclusions are clear. First, all four methods attain coverage close to the nominal 95\% level, providing empirical support for the asymptotic normal approximation under non-uniform sampling. Second, using estimated $\hat{\pi}$ instead of the true ones makes essentially no difference, which is consistent with the theory in Section~\ref{sec:unknown-sampling}: the first-order effect of weight estimation vanishes by conditional centering. 
Third, the efficient estimator retains a substantial variance advantage over the IPW-whitened estimator. 
This gap is even larger than under uniform sampling, because IPW interacts with score whitening and amplifies heavy-tailed contributions from under-sampled comparison types.

\begin{table}[t]
\centering
\caption{Results under non-uniform sampling for the linear entry target. }
\label{tab:nonuniform}
\smallskip
\begin{tabular}{lccc}
\hline
Method & Coverage & med $\widehat{\rm SE}$ & $\widehat{\rm SE}/{\rm SE}^\star$ \\
\hline
Whitened + IPW (known $\pi$) & 0.942 & 0.574 & 1.063 \\
Whitened + IPW (est.\ $\hat\pi$) & 0.944 & 0.571 & 1.059 \\
Efficient (known $\pi$) & 0.936 & 0.372 & 1.044 \\
Efficient (est.\ $\hat\pi$) & 0.936 & 0.372 & 1.044 \\
\hline
\end{tabular}
\end{table}

\subsection{Real-data experiment on \textsc{Arena}}\label{sec:real-data}

We finally validate the proposed procedures on real pairwise-comparison data from \textsc{Arena} \citep{chiang2024chatbot}, a platform that collects anonymous side-by-side evaluations of large language models from human users. We use the publicly available \textit{arena-human-preference-140k} dataset \citep{lmarena_human_preference_140k}. After filtering to the 30 most frequently appearing models and assigning each conversation to one of 10 non-overlapping task categories, such as math, coding, and creative writing, based on metadata tags, we obtain \(N=81{,}150\) observations organized as a \(30\times 10\) model-by-category comparison matrix. Here \(T^\star_{ij}\) denotes the latent BTL ability parameter of model \(i\) in task category \(j\). The category distribution is highly non-uniform: the largest category contains about \(7.1\) times as many observations as the smallest category. This makes the dataset a natural test bed for the non-uniform-sampling setting studied in Section~\ref{sec:unknown-sampling}. Further details are in Appendix~\ref{app:exp-real-preprocess}.

To assess the low-rank assumption in the real-data setting, we first construct a full-data benchmark matrix by fitting separate BTL models within each task category via Newton's method. The singular values of this benchmark matrix decay rapidly, and the first three singular components explain about 95\% of its total squared Frobenius norm. We therefore use rank \(r=3\) in the real-data experiments. This suggests that the dominant cross-category variation is well captured by a low-rank structure, while still allowing some model misspecification.

We subsample 20\% of the data for training and compare three methods: the naive per-category BTL baseline, the score-whitened estimator with IPW correction, and the efficient one-step estimator. In the main paper, we report only two representative targets: (1) the linear entry target $T^\star[\text{Gemini-2.5-Pro},\,\text{math}]$ that measures the latent ability of Gemini-2.5-Pro on math task; and (2) the nonlinear win-probability target $P(\text{Gemini-2.5-Pro}>\text{Claude-Opus-4}\mid\text{math})$. Experiments on additional target functionals as well as the efficiency gain over varying subsample sizes are deferred to Appendix~\ref{app:exp-real-preprocess}.

\begin{table}
\centering
\caption{Real-data inference on \textsc{Arena} for one linear target and one nonlinear target. 
Naive denotes separate per-category BTL modeling; Whitened denotes score-whitening with IPW; Efficient denotes the one-step estimator based on the efficient influence function.}
\label{tab:real-data-main}
\smallskip
\begin{tabular}{llcccc}
\hline
Target & $\psi^\star$ & Method & Emp.\ std & med $\widehat{\rm SE}$ & Coverage \\
\hline
\multirow{3}{*}{$T^\star[\text{Gemini},\,\text{math}]$} 
& \multirow{3}{*}{$+0.687$}
& Naive     & 0.199 & 0.274 & 0.998 \\
& & Whitened  & 0.355 & 0.523 & 0.986 \\
& & Efficient & 0.151 & 0.163 & 0.930 \\
\hline
\multirow{3}{*}{$P(\text{Gemini}>\text{Claude}\mid \text{math})$}
& \multirow{3}{*}{$+0.652$}
& Naive     & 0.055 & 0.061 & 0.954 \\
& & Whitened  & 0.115 & 0.170 & 0.986 \\
& & Efficient & 0.047 & 0.052 & 0.954 \\
\hline
\end{tabular}
\end{table}

Table~\ref{tab:real-data-main} summarizes the real-data results over 500 replications. Our efficient one-step estimator produces the narrowest intervals and coverage closest to the nominal 95\% level, while the naive and score-whitened estimators are more conservative. Relative to the full-data per-category BTL benchmark, the naive estimator is correctly specified, whereas the low-rank estimators trade a small approximation error for substantial variance reduction by borrowing strength across categories. The score-whitened estimator further sacrifices efficiency for robustness, which leads to wider intervals. This phenomenon is consistent with the theory.

%% file: appendix_content.tex
\section{Detailed Proofs for Section~\ref{sec:upper-bound}}\label{app:proof-details}

\subsection{Setup and Notation}\label{app:pd-setup}

We consider i.i.d.\ samples $\{(Y_i,X_i)\}_{i=1}^n$ from a \emph{pairwise comparison} design with a \emph{non-constant} Fisher information weight.

\subsubsection{Comparison sampling model}
Fix one ``comparison mode'' of size $d_1$ and denote by $D:=\prod_{k=2}^m d_k$ the number of contexts (the product of the remaining mode sizes).
A single design draw $X$ is generated as:
\begin{enumerate}[label=(\roman*)]
\item sample a context $u\in[D]$ uniformly;
\item sample an unordered pair $\{a,b\}\subset[d_1]$ uniformly among $\binom{d_1}{2}$ pairs;
\item set $X = e_{u,a}-e_{u,b}$ (a signed two-sparse tensor after flattening the non-comparison modes into $u$).
\end{enumerate}
Thus $\norm{X}_F^2 = 2$ and
\begin{equation}\label{app:pd-inner-prod-diff}
\ip{H}{X} = H_{u,a}-H_{u,b}.
\end{equation}

\subsubsection{Column-sum-zero restriction}
We assume \emph{column-sum-zero along the comparison mode}: for any $H$ in the target subspace (in particular $H^\star,\widehat H$ and their difference),
\begin{equation}\label{app:pd-row-sum-zero}
\sum_{a=1}^{d_1} H_{a,u} = 0,\qquad \forall u\in[D].
\end{equation}

\subsubsection{Effective dimension $\dstar$}
We define
\begin{equation}\label{app:pd-dstar-def}
\dstar := \prod_{j=1}^m d_j,\qquad \bar d := \max_{j\in[m]} d_j.
\end{equation}
Under the pairwise comparison model on mode 1, the effective sample-space size is $D(d_1-1)/2 = \dstar\cdot(d_1-1)/(2d_1) \asymp \dstar/2$. Since they differ by at most a constant factor, all bounds in this document hold up to absolute constants under either convention.

\subsubsection{Score and Fisher information}
Let $\eta^\star=\ip{T^\star}{X}$ be the true linear predictor. We assume:
\begin{assumption}[Score normalization and Fisher comparability]\label{app:pd-score}
There exist constants $\sigma>0$ and $C_0\ge 1$ such that for all $x$,
\begin{equation}\label{app:pd-score-fisher}
\E^\star\!\left[s_\eta(Y,\eta^\star(x))\mid X=x\right]=0,
\qquad
\E^\star\!\left[\frac{s_\eta(Y,\eta^\star(x))^2}{I(\eta^\star(x))}\ \Big|\ X=x\right]\le C_0^2,
\end{equation}
and the Fisher information satisfies
\begin{equation}\label{app:pd-fisher-bounds}
\frac{c_-}{\sigma^2}\ \le\ I(\eta^\star(x))\ \le\ \frac{c_+}{\sigma^2}
\qquad\text{for some constants }0<c_-\le c_+<\infty.
\end{equation}
We further assume:
\begin{enumerate}[label=(\roman*)]
\item (Score derivatives) There exist $c_1,c_2>0$ such that for all $(y,\eta)$,
\[
|\dot s_\eta(y,\eta)|\le \frac{c_1}{\sigma^2}
\qquad\text{and}\qquad
|\ddot s_\eta(y,\eta)|\le \frac{c_2}{\sigma^2}.
\]
\item (Sub-exponential tail) $\|s_\eta(Y,\eta)\|_{\psi_1}\le C_{\psi_1}$, where $C_{\psi_1}\asymp 1/\sigma$.
\end{enumerate}
\end{assumption}
\begin{remark}[BTL verification]
For the BTL model, $\dot s = -p(1-p)/\sigma^2$ gives $c_1 = 1/4$;
$\ddot s = p(1-p)(1-2p)/\sigma^3 \le C/\sigma^2$ (since $\sigma$ is bounded below);
$s = (Y-p)/\sigma$ is bounded in $[-1/\sigma,1/\sigma]$, giving $C_{\psi_1}\asymp 1/\sigma$.
\end{remark}

\subsubsection{Operators and directions}
Define the (non-constant weight) \emph{information operator}
\begin{equation}\label{app:pd-G-def}
G := \E^\star\!\big[\, I(\eta^\star(X))\ (X\otimes X)\,\big],
\quad\text{so that}\quad
\ip{H_1}{G H_2}=\E^\star\!\big[I(\eta^\star(X))\,\ip{H_1}{X}\,\ip{H_2}{X}\big].
\end{equation}
Equivalently, $\ip{GU}{V} = -\Pstar[\dot s_\eta(Y,\eta^\star(X))\ip{U}{X}\ip{V}{X}]$, since $\E[\dot s_\eta(Y,\eta^\star)|X] = -I(\eta^\star)$.

Let $P_{\mathbb{T}}$ denote the orthogonal projector onto the tangent space $\mathcal{T}$ at $T^\star$, and $\widehat{P}_{\mathbb{T}}$ the estimated tangent projector built from estimated subspaces $\{\hat U_j\}$.
Define the \emph{restricted information operators}
\[
A := P_{\mathbb{T}}\, G\, P_{\mathbb{T}},\qquad \hat A := \widehat{P}_{\mathbb{T}}\, G\, \widehat{P}_{\mathbb{T}},
\]
both invertible on their respective ranges.
The \emph{oracle} and \emph{estimated EIF directions} are
\begin{equation}\label{app:pd-directions}
H^\star := A^{-1}P_{\mathbb{T}}\Gammavec,\qquad \hat H := \hat A^{-1}\widehat{P}_{\mathbb{T}}\,\Gammavec.
\end{equation}

\subsubsection{$\Gammavec$ coefficient vector}
Expand the functional gradient $\Gammavec\in\Hcal$ in the canonical basis:
\begin{equation}\label{app:pd-Gamma-coeff}
\Gammavec = \sum_{\omega}\gamma_\omega E_\omega,\qquad
\oneNorm{\Gammavec} = \sum_\omega |\gamma_\omega|.
\end{equation}

\subsubsection{Column-sum-zero constraint and block projections}\label{app:pd-rowsum-remark}
In the comparison model (mode 1), the orthogonal complement on mode 1 is $Q_1 - \PU{1}$
where $Q_1 := \Id_{d_1} - \frac{1}{d_1}\one_{d_1}\one_{d_1}^\top$ (projection onto column-sum-zero subspace).
However, since all tensors in our analysis already satisfy column-sum-zero (i.e., $Q_1 H = H$ for all relevant $H$), we have $Q_1 = \Id$ on the relevant subspace, so $Q_1 - \PU{1} = \Id - \PU{1}$ and all block-projection formulas remain unchanged.
See Section~\ref{app:pd-tools-blocks} for details on the block decomposition.

\subsubsection{Operator norm assumptions}\label{app:pd-opnorm-assumptions}
We collect the key assumption on operator norms used throughout.
\begin{assumption}[Inverse scaling]\label{app:pd-opnorms}
There exists an absolute constant $C_A>0$ such that
\[
\norm{\hat A^{-1}}_{\infty\to\infty}\ \vee\ \norm{A^{-1}}_{\infty\to\infty}\ \le\ C_A\,\sigma^2\dstar.
\]
\end{assumption}
Under the constant-weight BTL baseline, $A_0 = P_{\mathbb{T}}/(2\sigma^2\dstar)$ so $\norm{A_0^{-1}}_{\infty\to\infty} = 2\sigma^2\dstar\,\inftyinfnorm{P_{\mathbb{T}}}$, giving $C_A = 2\,\inftyinfnorm{P_{\mathbb{T}}} = \mathrm{poly}(\mu,r,m)$.

\begin{remark}\label{app:pd-CG}
Since $\inftyinfnorm{P_{\mathbb{T}}}\le\mathrm{poly}(\mu,r,m)$ (Lemma~\ref{app:pd-Pinfty}) and $\inftyinfnorm{G}\le c_+/(\sigma^2\dstar)$ (from $\norm{G}_{\infty\to\infty}\asymp 1/(\sigma^2\dstar)$ as it's a near-identity mapping on the full space, 
we have
\[
\begin{aligned}
\norm{A^{-1}P_{\mathbb{T}}\,G}_{\infty\to\infty}
&\;\le\;
\norm{A^{-1}}_{\infty\to\infty}\cdot\inftyinfnorm{P_{\mathbb{T}}}\cdot\inftyinfnorm{G}\\
&\;\le\;
C_A\,\sigma^2\dstar\cdot\mathrm{poly}(\mu,r,m)\cdot\frac{c_+}{\sigma^2\dstar}
\;=\;
C_A\cdot\mathrm{poly}(\mu,r,m).
\end{aligned}
\]
\end{remark}

\subsection{Proof strategy and norm scales}\label{app:pd-strategy}

This appendix proves the one-step estimator remainder bound in
Theorem~\ref{thm:combined} under Assumption~\ref{ass:initial}.  We begin with
a high-level overview and a collection of norm estimates used throughout the
proof.

\subsubsection{Proof strategy: true $G$ first, then $\hat G$ perturbation}

In practice, the estimator uses the \emph{estimated} information operator $\hat G$ (built from Fisher weights $I(\hat\eta)$ evaluated at the plug-in estimate).
However, for clarity of exposition, we carry out the entire five-term remainder analysis in Sections~\ref{app:pd-RempH}--\ref{app:pd-2nd-order} \emph{as if} the true operator $G$ were used, i.e., we set $\hat A := \widehat{P}_{\mathbb{T}} G \widehat{P}_{\mathbb{T}}$ with the \emph{true} $G$.
Section~\ref{app:pd-hatG} then shows that replacing $G$ by the plug-in estimate $\hat G$ (with Fisher weights $I(\hat\eta)$ in place of $I(\eta^\star)$) introduces only a multiplicative $(1+O(\infnorm{\Delta}))$ correction, which is negligible.
This two-stage structure---\emph{oracle analysis} followed by \emph{perturbation accounting}---simplifies the exposition considerably and isolates the source of each error term.

\subsubsection{Norm scales for tangent-space projections}\label{app:pd-norm-scales}

The following norm estimates for the tangent-space projector applied to a basis tensor $E_\omega$ appear throughout the proof.
They quantify the ``spread'' of a single coordinate direction after projection onto the tangent space, and are the key to understanding dimension cancellations in the error bounds.

\begin{lemma}[Norm scales for $P_{\mathbb{T}}E_\omega$]\label{lem:norm-scales}
Under $\mu$-incoherence, for every canonical basis tensor $E_\omega$:
\begin{equation}\label{eq:norm-scales}
\begin{aligned}
\infnorm{P_{\mathbb{T}}E_\omega} &\le C(\mu,r,m)\,\frac{\bar d}{\dstar},
\\[4pt]
\Fnorm{P_{\mathbb{T}}E_\omega} &\le C(\mu,r,m)\,\sqrt{\frac{\bar d}{\dstar}},
\\[4pt]
\oneNorm{P_{\mathbb{T}}E_\omega} &\le C(\mu,r,m)
\qquad\text{(dimension-free)}.
\end{aligned}
\end{equation}
The same bounds hold for $\widehat{P}_{\mathbb{T}}E_\omega$ under
\eqref{eq:init-estimated-incoherence}.
\end{lemma}

\noindent\textbf{Intuition.}
The tangent-space projector $P_{\mathbb{T}}$ is a sum of $m+1$ Kronecker block projectors $\Pi_S$ with $|S|\le 1$.
The ``core'' block $\Pi_\emptyset$ compresses all modes, giving an $\ell_\infty$ contribution of order $\prod_j(\mu r_j/d_j)\asymp\mu^m\prod_j r_j/\dstar$; each ``arm'' block $\Pi_{\{k\}}$ leaves mode $k$ uncompressed (contributing a factor of $\sim 1$ rather than $\sqrt{\mu r_k/d_k}$ in that mode), which costs an extra $\sqrt{d_k/(\mu r_k)}\asymp\sqrt{\bar d}$ relative to the core. Summing over blocks:
\begin{itemize}
\item $\ell_\infty$: dominated by the arm blocks, giving $\bar d/\dstar$.
\item Frobenius: since $\Fnorm{P_{\mathbb{T}}E_\omega}^2 = (P_{\mathbb{T}})_{\omega\omega}\le\infnorm{P_{\mathbb{T}}E_\omega}\le C\bar d/\dstar$, the Frobenius norm is $\sqrt{\bar d/\dstar}$.
\item $\ell_1$: each single-mode factor has $\ell_1$ norm $\sqrt{\mu r_j}$ (bounded by Cauchy--Schwarz and incoherence), so the Kronecker product gives $\prod_j\sqrt{\mu r_j}=\mathrm{poly}(\mu,r,m)$---dimension-free.
\end{itemize}

\begin{remark}[Effect of subspace estimation error]\label{rem:norm-scales-hat}
When $P_{\mathbb{T}}$ is replaced by $\widehat{P}_{\mathbb{T}}-P_{\mathbb{T}}$, each norm picks up an additional factor of $\rho$:
\[
\begin{aligned}
\infnorm{(\widehat{P}_{\mathbb{T}}-P_{\mathbb{T}})E_\omega} &\le C(\mu,r,m)\,\frac{\bar d}{\dstar}\cdot\rho,
\\
\Fnorm{(\widehat{P}_{\mathbb{T}}-P_{\mathbb{T}})E_\omega} &\le C(\mu,r,m)\,\frac{\sqrt{\bar d}}{\sqrt{\dstar}}\cdot\rho,
\\
\oneNorm{(\widehat{P}_{\mathbb{T}}-P_{\mathbb{T}})E_\omega} &\le C(\mu,r,m)\cdot\rho.
\end{aligned}
\]
Here $\rho := \frac{\sigma}{\lambda_{\min}}\sqrt{\frac{\dstar\bar d\log^c\bar d}{n}}$ is the subspace estimation parameter.
Under the pairwise-comparison scaling in Assumption~\ref{ass:signal-strength},
$\rho\asymp\varepsilon_n$, so the factor $\rho$ follows directly from the
projector bounds in Assumption~\ref{ass:initial}; see
Eq.~\eqref{app:pd-subspace-op}.
The Frobenius bound for the tangent projector difference also acquires an extra $\sqrt{\bar d}$ beyond $\rho/\sqrt{\dstar}$ from the uncompressed mode in the arm blocks of $P_{\mathbb{T}}$; see~\eqref{app:pd-hatP-minus-P-Eomega-F}.
\end{remark}

\noindent\textbf{The $\ell_1$ dimension cancellation.}
A recurring theme in the bias analysis is the interplay between $\|A^{-1}\|_{1\to 1}\asymp\sigma^2\dstar$ and the comparison-design averaging $\E^\star[|\ip{v}{X}|]\le\oneNorm{v}/\dstar$.
Because $\oneNorm{P_{\mathbb{T}}E_\omega}=O(1)$ (dimension-free), the product
\[
\frac{\|A^{-1}\|_{1\to 1}\cdot\oneNorm{P_{\mathbb{T}}E_\omega}}{\dstar}
\;=\;
\frac{C_A\,\sigma^2\dstar\cdot O(1)}{\dstar}
\;=\;
C_A\,\sigma^2\cdot O(1)
\]
is dimension-free (up to $C_A$ and the noise level). This cancellation is the key mechanism behind the dimension-free bias bounds in Propositions~\ref{app:pd-H-bias-prop}--\ref{app:pd-2nd-order-prop}.

\medskip

We now proceed to bound each remainder term individually.

\subsection{Bounding $R_{\mathrm{emp}}^H$ (Direction-Error Empirical Process)}\label{app:pd-RempH}

\noindent\textbf{Intuition.} This term captures the stochastic fluctuation from using the estimated direction $\hat H$ in place of the oracle direction $H^\star$. The key challenge is that $\hat H - H^\star$ depends on the inverse restricted information operator $A^{-1}$, which introduces the $C_A$ factor. The proof follows a five-step template: (1) reduce the variance to a Frobenius norm via the pairwise-comparison second-moment identity; (2--3) extract the $\oneNorm{\Gammavec}$ factor and bound the per-basis-element Frobenius error using subspace perturbation; (4) bound the sub-exponential norm via $\ell_\infty$ control, where $C_A$ enters; (5) combine via Bernstein's inequality.

\medskip
Recall
\[
R_{\mathrm{emp}}^H
:= (\Pn-\Pstar)\big[s_\eta(Y,\hat\eta)\,\ip{\hat H-H^\star}{X}\big].
\]
Let $Z_i := s_\eta(Y_i,\hat\eta_i)\,\ip{\hat H-H^\star}{X_i}$, so $R_{\mathrm{emp}}^H = \frac1n\sum_{i=1}^n(Z_i - \E^\star[Z_i])$.
We bound its variance (Step~1), extract $\oneNorm{\Gammavec}$ (Steps~2--3), bound the $\psione$ norm (Step~4), and combine via Bernstein (Step~5).

\subsubsection{Step 1: Variance bound via Frobenius reduction}
By the Fisher comparability in Assumption~\ref{app:pd-score}(ii),
\[
\E^\star[Z_i^2]
= \E^\star\!\big[s_\eta(Y,\hat\eta)^2\,\ip{\hat H-H^\star}{X}^2\big]
\le C_0^2\,\E^\star\!\big[I(\eta^\star)\,\ip{\hat H-H^\star}{X}^2\big],
\]
where $C_0$ is bounded by a function of $c_-,c_+$ when $\hat T$ is close to $T^\star$.
Applying Corollary~\ref{app:pd-weighted-second-moment} with $H = \hat H - H^\star$:
\begin{equation}\label{app:pd-var-Z-Frob}
\mathrm{Var}(Z_i)
\le \E^\star[Z_i^2]
\le \frac{C_0^2\,c_+}{\sigma^2}\cdot\frac{\Fnorm{\hat H-H^\star}^2}{\dstar}.
\end{equation}
This is the key: the Frobenius reduction (Lemma~\ref{app:pd-Frob-reduction}) gives a $1/\dstar$ factor.

\subsubsection{Step 2: Extracting $\oneNorm{\Gammavec}$ from the Frobenius norm}
Since $\hat H - H^\star = \mathcal{M}\Gammavec$ with $\mathcal{M} := \hat A^{-1}\widehat{P}_{\mathbb{T}} - A^{-1}P_{\mathbb{T}}$,
by Lemma~\ref{app:pd-Gamma-l1}:
\begin{equation}\label{app:pd-l1-extract-F}
\Fnorm{\mathcal{M}\Gammavec}
\le \oneNorm{\Gammavec}\cdot\max_\omega\Fnorm{\mathcal{M}E_\omega}.
\end{equation}
Therefore
\begin{equation}\label{app:pd-var-Z-Gamma}
\mathrm{Var}(Z_i)
\le
\frac{C}{\sigma^2}\cdot\frac{\oneNorm{\Gammavec}^2}{\dstar}\cdot
\bigl(\max_\omega\Fnorm{\mathcal{M}E_\omega}\bigr)^2.
\end{equation}

\subsubsection{Step 3: Bounding $\max_\omega\Fnorm{\mathcal{M}E_\omega}$}
Decompose $\mathcal{M}$ via the resolvent identity:
\begin{equation}\label{app:pd-M-decomp}
\mathcal{M}
= A^{-1}(\widehat{P}_{\mathbb{T}} - P_{\mathbb{T}}) + A^{-1}(A-\hat A)\hat A^{-1}\widehat{P}_{\mathbb{T}}.
\end{equation}

\noindent\emph{Term~1: $\Fnorm{A^{-1}(\widehat{P}_{\mathbb{T}}-P_{\mathbb{T}})E_\omega}$.}
By submultiplicativity,
\[
\Fnorm{A^{-1}(\widehat{P}_{\mathbb{T}}-P_{\mathbb{T}})E_\omega} \le \norm{A^{-1}}_{\mathrm{op}}\cdot\Fnorm{(\widehat{P}_{\mathbb{T}}-P_{\mathbb{T}})E_\omega}.
\]
From Assumption~\ref{app:pd-opnorms}, $\norm{A^{-1}}_{\mathrm{op}} \le \norm{A^{-1}}_{\infty\to\infty} \le C_A\sigma^2\dstar$.

For $\Fnorm{(\widehat{P}_{\mathbb{T}}-P_{\mathbb{T}})E_\omega}$, we use the telescope identity across modes: $\widehat{P}_{\mathbb{T}} - P_{\mathbb{T}}$ decomposes as a sum of $m$ terms, each involving $(P_{\hat U_j}-P_{U_j})$ in one mode and incoherent projections in the remaining modes. By Assumptions~\ref{ass:initial} and~\ref{ass:signal-strength},
\begin{equation}\label{app:pd-subspace-op}
\norm{P_{\hat U_j}-P_{U_j}}_{2,\infty} \le C_2\rho\sqrt{\mu r_j/d_j},
\qquad \rho := \frac{\sigma}{\lambda_{\min}}\sqrt{\frac{\dstar\bar d\log^c\bar d}{n}}
\asymp\varepsilon_n.
\end{equation}
Here and below, $C_2$ depends only on the fixed model parameters.
Combined with incoherence ($\twoninfnorm{U_k}\le\sqrt{\mu r_k/d_k}$) in each remaining mode.
However, the tangent projector $P_{\mathbb{T}}$ contains blocks where $(\Id-P_{U_k})$ acts in one mode.
When $(\Id-P_{U_k})$ is applied to a standard basis vector $e_{i_k}$, its norm is bounded by~$1$ (not $\sqrt{\mu r_k/d_k}$), contributing an extra factor $\sqrt{d_k/(\mu r_k)}\asymp\sqrt{\bar d}$.
Accounting for this and summing over the $m$ telescope terms:
\begin{equation}\label{app:pd-hatP-minus-P-Eomega-F}
\Fnorm{(\widehat{P}_{\mathbb{T}} - P_{\mathbb{T}})E_\omega}
\le C(\mu,r,m)\cdot\frac{\sqrt{\bar d}\,\rho}{\sqrt{\dstar}}.
\end{equation}
Multiplying by $\norm{A^{-1}}_{\mathrm{op}} \le C_A\sigma^2\dstar$:
\begin{equation}\label{app:pd-term1-M-F}
\Fnorm{A^{-1}(\widehat{P}_{\mathbb{T}}-P_{\mathbb{T}})E_\omega}
\le C(\mu,r,m)\,C_A\,\sigma^2\sqrt{\dstar\bar d}\cdot\rho.
\end{equation}

\noindent\emph{Term~2: $\Fnorm{A^{-1}(A-\hat A)\hat A^{-1}\widehat{P}_{\mathbb{T}}E_\omega}$.}
Using $A-\hat A = (P_{\mathbb{T}}-\widehat{P}_{\mathbb{T}})GP_{\mathbb{T}} + \widehat{P}_{\mathbb{T}}G(P_{\mathbb{T}}-\widehat{P}_{\mathbb{T}})$,
$\norm{G}_{\mathrm{op}} \le c_+/(\sigma^2\dstar)$, and
$\Fnorm{\widehat{P}_{\mathbb{T}}E_\omega}\le C(\mu,r,m)/\sqrt{\dstar}$ (incoherence):
\begin{equation}\label{app:pd-term2-M-F}
\Fnorm{A^{-1}(A-\hat A)\hat A^{-1}\widehat{P}_{\mathbb{T}}E_\omega}
\le C(\mu,r,m)\,C_A^2\,\sigma^2\sqrt{\dstar\bar d}\cdot\rho.
\end{equation}
This is the same order as Term~1.

\noindent\emph{Combining.}
\begin{equation}\label{app:pd-M-Eomega-F-final}
\max_\omega\Fnorm{\mathcal{M}E_\omega}
\le C(\mu,r,m)\,C_A\,\sigma^2\sqrt{\dstar\bar d}\cdot\rho.
\end{equation}

Substituting into~\eqref{app:pd-var-Z-Gamma}:
\begin{equation}\label{app:pd-var-final}
\mathrm{Var}(Z_i)
\le C(\mu,r,m)\,C_A^2\,\sigma^2\,\bar d\,\oneNorm{\Gammavec}^2\,\rho^2.
\end{equation}

\subsubsection{Step 4: Sub-exponential ($\psione$) bound via $\ell_\infty$ norm}
For the tail term, we need to bound $\norm{Z_i}_{\psione}$.
By Assumption~\ref{app:pd-score}(ii), $\norm{s_\eta(Y,\hat\eta)}_{\psione} \le C_{\psione}$, where $C_{\psione}\asymp 1/\sigma$.
Since $|\ip{\hat H-H^\star}{X}| \le 2\infnorm{\hat H-H^\star}$ (comparison atom):
\begin{equation}\label{app:pd-Zi-psi1-raw}
\norm{Z_i}_{\psione}
\le 2C_{\psione}\cdot\infnorm{\hat H-H^\star}.
\end{equation}

By the $\ell_\infty$ version of Lemma~\ref{app:pd-Gamma-l1}:
\[
\infnorm{\hat H-H^\star} = \infnorm{\mathcal{M}\Gammavec}
\le \oneNorm{\Gammavec}\cdot\max_\omega\infnorm{\mathcal{M}E_\omega}.
\]

To bound $\infnorm{\mathcal{M}E_\omega}$, we use the same decomposition~\eqref{app:pd-M-decomp} and submultiplicativity in $\inftyinfnorm{\cdot}$:
\[
\infnorm{\mathcal{M}E_\omega}
\le \inftyinfnorm{A^{-1}}\cdot\infnorm{P_{\mathbb{T}}E_\omega}
+ \inftyinfnorm{A^{-1}}\cdot\inftyinfnorm{A-\hat A}\cdot\inftyinfnorm{\hat A^{-1}}\cdot\infnorm{\widehat{P}_{\mathbb{T}}E_\omega}.
\]
By incoherence (Lemma~\ref{app:pd-Pinfty}):
\[
\infnorm{P_{\mathbb{T}}E_\omega},\;\infnorm{\widehat{P}_{\mathbb{T}}E_\omega}
\le C(\mu,r,m)\,\bar d/\dstar.
\]
The second term is lower order under $\norm{\widehat{P}_{\mathbb{T}}-P_{\mathbb{T}}}\ll 1$, so the leading contribution is
\begin{equation}\label{app:pd-M-Eomega-infty}
\max_\omega\infnorm{\mathcal{M}E_\omega}
\le C(\mu,r,m)\cdot\inftyinfnorm{A^{-1}}\cdot\frac{\bar d}{\dstar}.
\end{equation}

Now substitute the key quantities:
\begin{itemize}
\item $\inftyinfnorm{A^{-1}} \le C_A\,\sigma^2\dstar$ (Assumption~\ref{app:pd-opnorms}),
\item $C_{\psione}\asymp 1/\sigma$ (Assumption~\ref{app:pd-score}(ii)).
\end{itemize}
The $\dstar$ in $\inftyinfnorm{A^{-1}}$ cancels the $1/\dstar$ from incoherence:
\[
\inftyinfnorm{A^{-1}}\cdot\frac{\bar d}{\dstar}
= C_A\,\sigma^2\dstar\cdot\frac{\bar d}{\dstar}
= C_A\,\sigma^2\,\bar d.
\]
Multiplying by $C_{\psione}\asymp 1/\sigma$:
\begin{equation}\label{app:pd-psi1-bound}
\norm{Z_i}_{\psione}
\le C(\mu,r,m)\,C_A\,\sigma\,\bar d\,\oneNorm{\Gammavec}.
\end{equation}

\subsubsection{Step 5: Bernstein concentration}
By Bernstein's inequality (Lemma~\ref{app:pd-bernstein-template} for sub-exponential variables), with probability at least $1-\delta$ (conditional on the first-stage data):
\begin{equation}\label{app:pd-Bernstein-RempH}
|R_{\mathrm{emp}}^H|
\le C\left(
\sqrt{\frac{\mathrm{Var}(Z_i)\,\log(2/\delta)}{n}}
+ \frac{\norm{Z_i}_{\psione}\,\log(2/\delta)}{n}
\right).
\end{equation}

Substituting~\eqref{app:pd-var-final} and~\eqref{app:pd-psi1-bound}:

\begin{theorem}[Bound on $R_{\mathrm{emp}}^H$]\label{app:pd-RempH-thm}
Under Assumptions~\ref{ass:initial}, \ref{app:pd-score}, and~\ref{app:pd-opnorms}, with probability at least $1-\delta$ (conditional on $\mathcal{D}_1$):
\begin{equation}\label{app:pd-RempH-final}
\boxed{
|R_{\mathrm{emp}}^H|
\le
C(\mu,r,m)\,C_A\,\oneNorm{\Gammavec}\left[
\sigma\,\sqrt{\bar d}\;\rho\;\sqrt{\frac{\log(2/\delta)}{n}}
\;+\;
\sigma\,\bar d\;\frac{\log(2/\delta)}{n}
\right],
}
\end{equation}
where $\rho = \frac{\sigma}{\lambda_{\min}}\sqrt{\frac{\dstar\bar d\log^c\bar d}{n}}$.
\end{theorem}

\noindent\textbf{Proof of Theorem~\ref{app:pd-RempH-thm}.}
The variance term gives
\[
\sqrt{\frac{\mathrm{Var}(Z_i)\log(2/\delta)}{n}}
\le C(\mu,r,m)\,C_A\,\sigma\,\sqrt{\bar d}\,\oneNorm{\Gammavec}\,\rho\,\sqrt{\frac{\log(2/\delta)}{n}}.
\]
The sub-exponential term gives (from~\eqref{app:pd-psi1-bound}):
\[
\frac{\norm{Z_i}_{\psione}\log(2/\delta)}{n}
\le C(\mu,r,m)\,C_A\,\sigma\,\bar d\,\oneNorm{\Gammavec}\,\frac{\log(2/\delta)}{n}.
\]
Combining yields~\eqref{app:pd-RempH-final}.
\hfill$\square$\medskip

\begin{remark}[Simplified pairwise-comparison bound]\label{app:pd-RempH-simplified}
In the pairwise-comparison setting ($\sigma=O(1)$, $\lambda_{\min}\asymp\sqrt{\dstar}$), substituting $\rho\asymp\sqrt{\bar d\log^c\bar d/n}$ and $\delta=\bar d^{-c}$, the bound~\eqref{app:pd-RempH-final} simplifies to
\[
|R_{\mathrm{emp}}^H|
\;\le\;
C(\mu,r,m)\,C_A\,\oneNorm{\Gammavec}\,\frac{\bar d\log^c\bar d}{n}
\]
with probability $\ge 1-\bar d^{-c}$. The $\sigma^2/\lambda_{\min}$ ratio and $\dstar$ factors have been absorbed, leaving a clean $\bar d/n$ rate multiplied by $C_A$.
\end{remark}

\subsection{Bounding $R_{\mathrm{emp}}^\eta$ (Score-Perturbation Empirical Process)}\label{app:pd-Remp-eta}

\noindent\textbf{Intuition.} This term measures the stochastic error from evaluating the score at $\hat\eta$ instead of $\eta^\star$. A Taylor expansion of the score yields a first-order term (linear in $\Delta$) and a second-order term (quadratic in $\Delta$). Both are controlled via Bernstein's inequality. The bounds on $\|H^\star\|_\infty$ and $\|H^\star\|_F/\sqrt{d^\star}$ carry the $C_A$ factor through the direction $H^\star=A^{-1}P_{\mathbb{T}}\Gamma$. In the pairwise-comparison setting, the $\sigma^2$ factors from $A^{-1}$ cancel with the $1/\sigma^2$ from the score derivatives.

\medskip
Recall
\[
R_{\mathrm{emp}}^\eta
:= (\Pn-\Pstar)\big[(s_\eta(Y,\hat\eta)-s_\eta(Y,\eta^\star))\,\ip{H^\star}{X}\big].
\]
Condition on the first-stage output so that $\Delta := \hat T - T^\star$ and $H^\star$ are fixed.

\subsubsection{Taylor expansion of the score}
For each $i$ there exists $\tilde\eta_i$ between $\ip{\hat T}{X_i}$ and $\ip{T^\star}{X_i}$ such that
\[
s_\eta(Y_i,\hat\eta_i)-s_\eta(Y_i,\eta^\star_i)
=
\dot s_\eta(Y_i,\eta^\star_i)\ip{\Delta}{X_i}
+\frac12 \ddot s_\eta(Y_i,\tilde\eta_i)\ip{\Delta}{X_i}^2.
\]
Hence $R_{\mathrm{emp}}^\eta = R_{\eta,1}^{\mathrm{emp}}+R_{\eta,2}^{\mathrm{emp}}$ where
\[
R_{\eta,1}^{\mathrm{emp}}:=(\Pn-\Pstar)\big[\dot s_\eta(Y,\eta^\star)\ip{\Delta}{X}\ip{H^\star}{X}\big],
\quad
R_{\eta,2}^{\mathrm{emp}}:=\frac12(\Pn-\Pstar)\big[\ddot s_\eta(Y,\tilde\eta)\ip{\Delta}{X}^2\ip{H^\star}{X}\big].
\]

\subsubsection{Comparison-design moment bounds}
Under the comparison design:
\begin{itemize}
\item $|\ip{\Delta}{X}| = |\Delta_{u,a}-\Delta_{u,b}| \le 2\infnorm{\Delta}$.
\item $|\ip{H^\star}{X}| \le 2\infnorm{H^\star}$.
\item $\E^\star[\ip{H^\star}{X}^2] = \Fnorm{H^\star}^2/\dstar$ (by Lemma~\ref{app:pd-Frob-reduction}).
\end{itemize}

\subsubsection{Abstract Bernstein bound}

\begin{proposition}[Abstract bound for $R_{\mathrm{emp}}^\eta$]\label{app:pd-Remp-eta-prop}
Under Assumption~\ref{app:pd-score} (with $|\dot s_\eta|\le c_1/\sigma^2$, $|\ddot s_\eta|\le c_2/\sigma^2$) and comparison design, with probability at least $1-\delta$,
\begin{align}
|R_{\mathrm{emp}}^\eta|
\;\le\;&
\frac{C}{\sigma^2}\Bigg[
\infnorm{\Delta}\,\frac{\Fnorm{H^\star}}{\sqrt{\dstar}}\sqrt{\frac{\log(4/\delta)}{n}}
\;+\;
\infnorm{\Delta}\,\infnorm{H^\star}\,\frac{\log(4/\delta)}{n}
\notag\\
&\qquad
+\;\infnorm{\Delta}^2\,\frac{\Fnorm{H^\star}}{\sqrt{\dstar}}\sqrt{\frac{\log(4/\delta)}{n}}
\;+\;
\infnorm{\Delta}^2\,\infnorm{H^\star}\,\frac{\log(4/\delta)}{n}
\Bigg],
\label{app:pd-Remp-eta-eq}
\end{align}
where $C$ absorbs $c_1,c_2$.
\end{proposition}

\noindent\textbf{Proof of Proposition~\ref{app:pd-Remp-eta-prop}.}
We bound each order using Bernstein's inequality (Lemma~\ref{app:pd-bernstein-template}).
For the first-order term $R_{\eta,1}^{\mathrm{emp}}$, define
\[
W_i^{(1)}:=\dot s_\eta(Y_i,\eta^\star_i)\ip{\Delta}{X_i}\ip{H^\star}{X_i}-\E^\star[\cdot].
\]
Then $|W_i^{(1)}| \le 8(c_1/\sigma^2)\infnorm{\Delta}\infnorm{H^\star} =: M_1$, and
\[
\mathrm{Var}(W_i^{(1)})
\le
\frac{4c_1^2}{\sigma^4}\,\infnorm{\Delta}^2\,\frac{\Fnorm{H^\star}^2}{\dstar}.
\]
For the second-order term $R_{\eta,2}^{\mathrm{emp}}$, similarly
\[
|W_i^{(2)}| \le 4(c_2/\sigma^2)\infnorm{\Delta}^2\infnorm{H^\star},
\qquad
\mathrm{Var}(W_i^{(2)}) \le 4c_2^2\sigma^{-4}\infnorm{\Delta}^4\Fnorm{H^\star}^2/\dstar.
\]
Applying Bernstein to each term and taking a union bound yields~\eqref{app:pd-Remp-eta-eq}.
\hfill$\square$\medskip

\subsubsection{Bounds on $\infnorm{H^\star}$ and $\Fnorm{H^\star}/\sqrt{\dstar}$}

By Lemma~\ref{app:pd-Gamma-l1} and Assumption~\ref{app:pd-opnorms}:

\noindent\textbf{$\ell_\infty$ bound.}
Using $\infnorm{P_{\mathbb{T}} E_\omega}\le C_\mu\,\bar d/\dstar$ (Lemma~\ref{app:pd-PTX-infty}, with $\bar d := \max_j d_j$):
\begin{equation}\label{app:pd-Hstar-infty-final}
\infnorm{H^\star}
\le
\oneNorm{\Gammavec}\cdot \norm{A^{-1}}_{\infty\to\infty}\cdot \max_\omega\infnorm{P_{\mathbb{T}} E_\omega}
\le
C(\mu,r,m)\,C_A\,\sigma^2\bar d\,\oneNorm{\Gammavec}.
\end{equation}

\noindent\textbf{Frobenius bound.}
Since
\[
\Fnorm{P_{\mathbb{T}}E_\omega}^2 = (P_{\mathbb{T}})_{\omega\omega}\le\infnorm{P_{\mathbb{T}}E_\omega}\le C_\mu\,\bar d/\dstar,
\qquad
\norm{A^{-1}}_{\mathrm{op}}\le\norm{A^{-1}}_{\infty\to\infty}\le C_A\,\sigma^2\dstar,
\]
we obtain:
\begin{equation}\label{app:pd-Hstar-Frob-final}
\frac{\Fnorm{H^\star}}{\sqrt{\dstar}}
\le
\oneNorm{\Gammavec}\cdot\frac{\norm{A^{-1}}_{\mathrm{op}}\max_\omega\Fnorm{P_{\mathbb{T}}E_\omega}}{\sqrt{\dstar}}
\le
C(\mu,r,m)\,C_A\,\sigma^2\sqrt{\bar d}\,\oneNorm{\Gammavec}.
\end{equation}

\subsubsection{Final explicit bound}
Substituting \eqref{app:pd-Hstar-infty-final} and \eqref{app:pd-Hstar-Frob-final} into Proposition~\ref{app:pd-Remp-eta-prop}. In each term, the factor $1/\sigma^2$ from Assumption~\ref{app:pd-score} cancels with $\sigma^2$ from $\norm{A^{-1}}$:

\begin{corollary}[Explicit bound for $R_{\mathrm{emp}}^\eta$]\label{app:pd-Remp-eta-final}
Under Assumptions~\ref{app:pd-score} and \ref{app:pd-opnorms}, with probability at least $1-\delta$,
\begin{equation}\label{app:pd-Remp-eta-final-eq}
|R_{\mathrm{emp}}^\eta|
\;\le\;
C(\mu,r,m)\,C_A\,\oneNorm{\Gammavec}\big(\infnorm{\Delta}+\infnorm{\Delta}^2\big)
\left[
\sqrt{\bar d}\sqrt{\frac{\log(4/\delta)}{n}}
\;+\;
\bar d\,\frac{\log(4/\delta)}{n}
\right].
\end{equation}
\end{corollary}

\noindent\textbf{Proof of Corollary~\ref{app:pd-Remp-eta-final}.}
The four terms in \eqref{app:pd-Remp-eta-eq} become, after substitution:
the first-order variance term gives
\[
\frac{1}{\sigma^2}\infnorm{\Delta}\frac{\Fnorm{H^\star}}{\sqrt{\dstar}} \le C\,C_A\,\oneNorm{\Gammavec}\,\infnorm{\Delta}\sqrt{\bar d}
\]
(since $\sigma^2$ cancels);
the first-order tail term gives
\[
\frac{1}{\sigma^2}\infnorm{\Delta}\infnorm{H^\star} \le C\,C_A\,\oneNorm{\Gammavec}\,\infnorm{\Delta}\,\bar d
\]
(again $\sigma^2$ cancels);
and the second-order variance and tail terms are the same as their first-order counterparts with an extra factor of $\infnorm{\Delta}$.
Factoring $\infnorm{\Delta}+\infnorm{\Delta}^2$ and collecting the $\sqrt{1/n}$ and $1/n$ terms yields \eqref{app:pd-Remp-eta-final-eq}.
\hfill$\square$\medskip

\subsection{Bounding $R_{\mathrm{proj}}$ (Projection Leakage)}\label{app:pd-Rproj}

\noindent\textbf{Intuition.} The projection leakage captures the bias from the component of $\Gamma$ orthogonal to the tangent space. Since $\Delta=\hat T-T^\star$ has its normal component $(\Id-P_{\mathbb{T}})\Delta$ supported on blocks $\Pi_S$ with $|S|\ge 2$, this term is \emph{quadratic} in subspace estimation errors. This is the key structural reason why projection leakage is lower order: it requires errors in two or more modes simultaneously.

\medskip
The projection leakage term is
\[
R_{\mathrm{proj}} := \ip{(\Id - P_{\mathbb{T}})\Gammavec}{\Delta}.
\]

\subsubsection{Basic bound}
By H\"older's inequality,
\begin{equation}\label{app:pd-Rproj-basic}
|R_{\mathrm{proj}}|
\le
\oneNorm{\Gammavec}\cdot \infnorm{(\Id-P_{\mathbb{T}})\Delta}.
\end{equation}
The key structural fact (Lemma~\ref{app:pd-tangent-block-sum}) is that
\[
(\Id - P_{\mathbb{T}})\Delta = \sum_{\substack{S\subseteq[m]\\|S|\ge 2}}\Pi_S(\Delta),
\]
so the projection leakage is supported on blocks $\Pi_S$ with $|S|\ge 2$, making it \emph{quadratic} in subspace estimation errors.

\subsubsection{The $\varepsilon^2$ bound on $\infnorm{(\Id-P_{\mathbb{T}})\hat T}$}

Define $\varepsilon_j := \twoninfnorm{\PUhat{j}-\PU{j}}$ (projector error in the $j$-th mode).

\begin{proposition}[Complete $\varepsilon^2$ argument]\label{app:pd-PTcompl-infty}
Let $\widehat T=\widehat C\mode{1}\widehat U_1\cdots\mode{m}\widehat U_m$.
Then
\begin{equation}\label{app:pd-eps2-bound}
\infnorm{(\Id-P_{\mathbb{T}})\widehat T}
\le \norm{\widehat C}_{\mathrm{op}}
\Big(\prod_{k=1}^m(\twoninfnorm{\widehat U_k}+\varepsilon_k)\Big)
\sum_{\substack{S\subseteq[m]\\|S|\ge 2}}\prod_{j\in S}\varepsilon_j.
\end{equation}
On the event $\max_j\varepsilon_j\le 1$,
\begin{equation}\label{app:pd-eps2-comb}
\sum_{\substack{S\subseteq[m]\\|S|\ge 2}}\prod_{j\in S}\varepsilon_j
\le 2^{m-2}\sum_{1\le a<b\le m}\varepsilon_a\varepsilon_b,
\end{equation}
so $\infnorm{(\Id-P_{\mathbb{T}})\hat T}$ is \emph{quadratic} in the subspace errors.
\end{proposition}

\noindent\textbf{Proof of Proposition~\ref{app:pd-PTcompl-infty}.}
By $\Id-P_{\mathbb{T}} = \sum_{|S|\ge 2}\Pi_S$ and the triangle inequality,
\[
\infnorm{(\Id-P_{\mathbb{T}})\hat T} \le \sum_{|S|\ge 2}\infnorm{\Pi_S(\hat T)}.
\]
Fix $S$ with $|S|\ge 2$. Using the Tucker form of $\hat T$,
\[
\Pi_S(\hat T)
= \widehat C\mode{j\in S}{E_j}\mode{k\notin S}{M_k},
\]
where $E_j := (\Id - \PU{j})\hat U_j$ (error factor) and $M_k := \PU{k}\hat U_k$ (main factor).
By the entrywise contraction lemma (Lemma~\ref{app:pd-contract-2inf}),
\[
\infnorm{\Pi_S(\hat T)} \le \norm{\hat C}_{\mathrm{op}}\prod_{j\in S}\twoninfnorm{E_j}\prod_{k\notin S}\twoninfnorm{M_k}.
\]
Since $E_j = (\PUhat{j}-\PU{j})\hat U_j$, we have $\twoninfnorm{E_j}\le\varepsilon_j$,
and $\twoninfnorm{M_k}\le\twoninfnorm{\hat U_k}+\varepsilon_k$.
Summing over all $S$ with $|S|\ge 2$ yields \eqref{app:pd-eps2-bound}.
For the combinatorial bound~\eqref{app:pd-eps2-comb}, each pair $\{a,b\}\subseteq[m]$ is contained in exactly $2^{m-2}$ subsets $S\subseteq[m]$, and since $\varepsilon_j\le 1$, we have $\prod_{j\in S}\varepsilon_j\le\varepsilon_a\varepsilon_b$ for any $S\supseteq\{a,b\}$.
\hfill$\square$\medskip

\subsubsection{Final bound under the initial-estimator conditions}\label{app:pd-Rproj-final}

We now specialize Proposition~\ref{app:pd-PTcompl-infty} using
Assumption~\ref{ass:initial}.
Define the signal-to-noise ratio parameter
\begin{equation}\label{app:pd-rho-def}
\rho \;:=\; \frac{\sigma}{\lambda_{\min}}\sqrt{\frac{\dstar\,\bar d\log^c\bar d}{n}}.
\end{equation}
Under the SNR condition, $\rho \ll 1$.

Assumption~\ref{ass:signal-strength} and~\eqref{eq:Fnorm-lambda} give
$\lambda_{\min}\asymp\sqrt{\dstar}$, while $\sigma=O(1)$ under the bounded
BTL signal condition. Hence $\rho\asymp\varepsilon_n$. The mode-$j$
projector bounds in Assumption~\ref{ass:initial} therefore give
\[
\varepsilon_j
\;:=\;
\twoninfnorm{\PUhat{j}-\PU{j}}
\;\le\;
C_2\,\rho\,\sqrt{\frac{\mu r_j}{d_j}},
\]
and the estimated subspaces are incoherent:
$\|e_a^\top\hat U_j\|_2^2\le 2\mu r_j/d_j$, so
$\twoninfnorm{\hat U_j}\le\sqrt{2\mu r_j/d_j}$.

\noindent\textbf{Main factor.}
Since $\rho\ll 1$,
\[
\twoninfnorm{\hat U_k}+\varepsilon_k
\;\le\;
\bigl(\sqrt{2}+C_2\rho\bigr)\sqrt{\frac{\mu r_k}{d_k}}
\;\le\;
2\sqrt{\frac{\mu r_k}{d_k}},
\]
hence
\begin{equation}\label{app:pd-main-factor}
\prod_{k=1}^m\bigl(\twoninfnorm{\hat U_k}+\varepsilon_k\bigr)
\;\le\;
2^m\,\sqrt{\frac{\mu^m r^\star}{\dstar}},
\qquad r^\star := \textstyle\prod_{k=1}^m r_k.
\end{equation}

\noindent\textbf{Quadratic sum.}
Since $\mu r_j/d_j\le 1$ (incoherence) and $\rho\ll 1$, each $\varepsilon_j\le C_2\rho\le 1$.
By~\eqref{app:pd-eps2-comb},
\begin{equation}\label{app:pd-quad-sum}
\sum_{\substack{S\subseteq[m]\\|S|\ge 2}}\prod_{j\in S}\varepsilon_j
\;\le\;
2^{m-2}\sum_{a<b}\varepsilon_a\varepsilon_b
\;\le\;
2^{m-2}C_2^2\,\rho^2\!\sum_{a<b}\!\sqrt{\frac{\mu r_a}{d_a}}\sqrt{\frac{\mu r_b}{d_b}}
\;\le\;
C(\mu,r,m)\,\rho^2.
\end{equation}

\noindent\textbf{Core tensor factor.}
Since $\hat T = \hat C\mode{1}{\hat U_1}\cdots\mode{m}{\hat U_m}$,
\[
\frac{\norm{\hat C}_{\mathrm{op}}}{\sqrt{\dstar}}
\;\le\;
\frac{\Fnorm{\hat T}}{\sqrt{\dstar}}
\;\le\;
\infnorm{\hat T}
\;\le\;
\infnorm{T^\star}+\infnorm{\hat T-T^\star}
\;=\; O(1),
\]
where the last step follows from the bounded-signal condition and
Assumption~\ref{ass:initial}.

\noindent\textbf{Combining.}
Substituting~\eqref{app:pd-main-factor} and~\eqref{app:pd-quad-sum} into Proposition~\ref{app:pd-PTcompl-infty}:
\[
\infnorm{(\Id-P_{\mathbb{T}})\Delta}
\;\le\;
\norm{\hat C}_{\mathrm{op}}\cdot 2^m\sqrt{\frac{\mu^m r^\star}{\dstar}}\cdot C(\mu,r,m)\,\rho^2
\;=\;
C(\mu,r,m)\cdot\frac{\norm{\hat C}_{\mathrm{op}}}{\sqrt{\dstar}}\cdot\rho^2
\;\le\;
C(\mu,r,m)\,\rho^2.
\]
By~\eqref{app:pd-Rproj-basic},
\begin{equation}\label{app:pd-Rproj-final-eq}
\boxed{
|R_{\mathrm{proj}}|
\;\le\;
2^m\,C(\mu,r,m)\;\oneNorm{\Gammavec}\;\rho^2.
}
\end{equation}

\subsection{H-Direction Bias}\label{app:pd-H-bias}

 This term is the population-level error from using $\hat H$ instead of $H^\star$. After a Taylor expansion around $\eta^\star$, the first-order piece involves $\langle G\Delta, \hat H-H^\star\rangle$, which is controlled by the off-tangent components of $\Delta$ (since the tangent-space component cancels via the first-order cancellation mechanism in Section~\ref{app:pd-1st-cancel}). The second-order piece is bounded by $C_A\|\Delta\|_\infty^2\|\Gamma\|_1$ using the $\ell_1$-average bound, where the key cancellation is $d^\star$ from $\|A^{-1}\|_{1\to 1}$ against $1/d^\star$ from the comparison averaging.

\medskip
The H-direction bias term is
\[
\Pstar\big[S_{\hat T}(\hat H) - S_{\hat T}(H^\star)\big]
= \Pstar\big[s_\eta(Y,\hat\eta)\ip{\hat H - H^\star}{X}\big].
\]

\subsubsection{Taylor expansion and leading term}
Expanding $s_\eta(Y,\hat\eta)$ around $\eta^\star$:
\begin{align}
\Pstar\big[s_\eta(Y,\hat\eta)\ip{\hat H - H^\star}{X}\big]
&= \underbrace{\Pstar\big[s_\eta(Y,\eta^\star)\ip{\hat H-H^\star}{X}\big]}_{=\,0\text{ (score mean-zero)}}
\notag\\
&\quad+ \underbrace{\Pstar\big[\dot s_\eta(Y,\eta^\star)\ip{\Delta}{X}\ip{\hat H-H^\star}{X}\big]}_{\text{first-order in }\Delta}
\notag\\
&\quad+ \underbrace{\frac12\Pstar\big[\ddot s_\eta(Y,\bar\eta)\ip{\Delta}{X}^2\ip{\hat H-H^\star}{X}\big]}_{\text{second-order in }\Delta}.
\label{app:pd-Hbias-Taylor}
\end{align}

\subsubsection{$\ell_1$-average bound under comparison design}

\begin{lemma}[$\ell_1$-average bound]\label{app:pd-l1-average}
Under the comparison design, for any tensor $v\in\Hcal$,
\begin{equation}\label{app:pd-l1-average-eq}
\E^\star\big[|\ip{v}{X}|\big] \le \frac{2\oneNorm{v}}{\dstar}.
\end{equation}
\end{lemma}

\noindent\textbf{Proof of Lemma~\ref{app:pd-l1-average}.}
Since $X = e_{u,a} - e_{u,b}$ with $(u,\{a,b\})$ uniform,
$|\ip{v}{X}| = |v_{u,a} - v_{u,b}| \le |v_{u,a}| + |v_{u,b}|$.
Taking expectation,
\[
\E^\star[|\ip{v}{X}|]
\le \frac{1}{D\binom{d_1}{2}}\sum_u\sum_{a<b}(|v_{u,a}|+|v_{u,b}|)
= \frac{1}{D\binom{d_1}{2}}\,(d_1-1)\sum_u\sum_a |v_{u,a}|
= \frac{2}{Dd_1}\,\oneNorm{v}.
\]
Since $\dstar=Dd_1$, this equals $2\oneNorm{v}/\dstar$; the absolute
factor~$2$ is absorbed into the constants in every subsequent use of the
bound.
\hfill$\square$\medskip

\subsubsection{Dimension-free $\oneNorm{P_{\mathbb{T}}E_\omega}$ bound}

\begin{lemma}[Dimension-free $\ell_1$ bound for tangent projection]\label{app:pd-PT-l1}
Under $\mu$-incoherence, for any canonical basis tensor $E_\omega = e_{i_1}\otimes\cdots\otimes e_{i_m}$,
\begin{equation}\label{app:pd-PT-l1-bound}
\oneNorm{P_{\mathbb{T}}E_\omega} \le C_P,\qquad C_P := C(\mu,r,m).
\end{equation}
The same bound holds for $\widehat{P}_{\mathbb{T}}$ under the incoherence of the estimated subspaces
$\norm{e_a^\top\hat U_j}_2^2\le 2\mu r_j/d_j$.
\end{lemma}

\noindent\textbf{Proof of Lemma~\ref{app:pd-PT-l1}.}
For elementary tensors, the Kronecker $\ell_1$ factorization gives $\oneNorm{a^{(1)}\otimes\cdots\otimes a^{(m)}} = \prod_j\oneNorm{a^{(j)}}$.
By Cauchy--Schwarz and incoherence, each single-mode factor satisfies
\[
\oneNorm{P_je_{i_j}}
= \sum_k |(P_j)_{k,i_j}|
\le \sqrt{d_j}\,\norm{P_je_{i_j}}_2
\le \sqrt{d_j}\sqrt{\mu r_j/d_j}
= \sqrt{\mu r_j},
\]
and $\oneNorm{P_j^\perp e_{i_j}} \le \oneNorm{e_{i_j}} + \oneNorm{P_je_{i_j}} \le 1 + \sqrt{\mu r_j}$.
Using $P_{\mathbb{T}} = \Pi_\emptyset + \sum_k\Pi_{\{k\}}$ and the Kronecker factorization,
$\oneNorm{\Pi_\emptyset(E_\omega)}
= \prod_j\oneNorm{P_je_{i_j}}
\le \mu^{m/2}\sqrt{r^\star}$,
and
$\oneNorm{\Pi_{\{k\}}(E_\omega)}
= \oneNorm{P_k^\perp e_{i_k}}\prod_{j\ne k}\oneNorm{P_je_{i_j}}
\le (1+\sqrt{\mu r_k})\prod_{j\ne k}\sqrt{\mu r_j}$.
Summing over blocks yields $\oneNorm{P_{\mathbb{T}}E_\omega} \le C(\mu,r,m)$, which is dimension-free.
The bound for $\widehat{P}_{\mathbb{T}}$ follows identically with $2\mu$ replacing $\mu$.
\hfill$\square$\medskip

\subsubsection{Bounding the first-order piece}

\begin{proposition}[H-direction bias bound]\label{app:pd-H-bias-prop}
Under Assumption~\ref{app:pd-opnorms}, the first-order piece satisfies
\begin{equation}\label{app:pd-H-bias-1st}
\big|\ip{G\Delta}{\hat H-H^\star}\big|
\le
C(\mu,r,m)\,C_A\,\oneNorm{\Gammavec}
\Big[
\infnorm{(\Id-\widehat{P}_{\mathbb{T}})\Delta}
+
\infnorm{(\Id-P_{\mathbb{T}})\Delta}
\Big].
\end{equation}
\end{proposition}

\noindent\textbf{Proof of Proposition~\ref{app:pd-H-bias-prop}.}
By \eqref{app:pd-G-def},
\[
\Pstar[\dot s_\eta(Y,\eta^\star)\ip{\Delta}{X}\ip{\hat H-H^\star}{X}] = \ip{G\Delta}{\hat H-H^\star}.
\]
Since $G$ is symmetric and $\hat A^{-1}\widehat{P}_{\mathbb{T}}$, $A^{-1}P_{\mathbb{T}}$ are self-adjoint,
\[
\ip{G\Delta}{\hat H-H^\star}
= \ip{(\hat A^{-1}\widehat{P}_{\mathbb{T}} G - A^{-1}P_{\mathbb{T}}G)\Delta}{\Gammavec}.
\]
By H\"older's inequality, $|\ip{u}{\Gammavec}| \le \oneNorm{\Gammavec}\cdot\infnorm{u}$.
Decomposing via the triangle inequality,
\[
\infnorm{(\hat A^{-1}\widehat{P}_{\mathbb{T}} G - A^{-1}P_{\mathbb{T}}G)\Delta}
\le
\infnorm{\hat A^{-1}\widehat{P}_{\mathbb{T}} G(\Id-\widehat{P}_{\mathbb{T}})\Delta}
+
\infnorm{A^{-1}P_{\mathbb{T}}G(\Id-P_{\mathbb{T}})\Delta}
+
\infnorm{(\widehat{P}_{\mathbb{T}} - P_{\mathbb{T}})\Delta}.
\]
Using
\[
\infnorm{(\widehat{P}_{\mathbb{T}} - P_{\mathbb{T}})\Delta}\le\infnorm{(\Id-\widehat{P}_{\mathbb{T}})\Delta}+\infnorm{(\Id-P_{\mathbb{T}})\Delta}
\]
and Remark~\ref{app:pd-CG} yields \eqref{app:pd-H-bias-1st}.
\hfill$\square$\medskip

\subsubsection{The second-order piece}
By $|\ddot s_\eta|\le c_2/\sigma^2$ and $|\ip{\Delta}{X}|\le 2\infnorm{\Delta}$,
\[
\left|\frac12\Pstar\big[\ddot s_\eta\ip{\Delta}{X}^2\ip{\hat H-H^\star}{X}\big]\right|
\le
\frac{2c_2}{\sigma^2}\,\infnorm{\Delta}^2\,\E^\star[|\ip{\hat H-H^\star}{X}|].
\]
We now bound $\E^\star[|\ip{\hat H-H^\star}{X}|]$ using the $\ell_1$-average
(Lemma~\ref{app:pd-l1-average}):
\begin{equation}\label{app:pd-H-diff-Eabs}
\E^\star\big[|\ip{\hat H-H^\star}{X}|\big]
\le \frac{\oneNorm{\hat H-H^\star}}{\dstar}.
\end{equation}

\noindent\textbf{Step 1: $\ell_1$ extraction.}
By the $\ell_1$-version of Lemma~\ref{app:pd-Gamma-l1},
\begin{equation}\label{app:pd-H-diff-l1}
\oneNorm{\hat H - H^\star}
= \oneNorm{(\hat A^{-1}\widehat{P}_{\mathbb{T}} - A^{-1}P_{\mathbb{T}})\Gammavec}
\le \oneNorm{\Gammavec}\,\max_\omega\,\oneNorm{(\hat A^{-1}\widehat{P}_{\mathbb{T}} - A^{-1}P_{\mathbb{T}})E_\omega}.
\end{equation}

\noindent\textbf{Step 2: triangle inequality and symmetry.}
For each basis tensor $E_\omega$,
\[
\oneNorm{(\hat A^{-1}\widehat{P}_{\mathbb{T}} - A^{-1}P_{\mathbb{T}})E_\omega}
\le \norm{\hat A^{-1}}_{1\to 1}\,\oneNorm{\widehat{P}_{\mathbb{T}}E_\omega}
+ \norm{A^{-1}}_{1\to 1}\,\oneNorm{P_{\mathbb{T}}E_\omega}.
\]
Since $\hat A^{-1}$ and $A^{-1}$ are symmetric on their respective tangent spaces,
Lemma~\ref{app:pd-symm-1inf} gives
$\norm{\hat A^{-1}}_{1\to 1} = \inftyinfnorm{\hat A^{-1}} \le C_A\,\sigma^2\dstar$
(Assumption~\ref{app:pd-opnorms}), and likewise for $A^{-1}$.
By Lemma~\ref{app:pd-PT-l1}, $\oneNorm{P_{\mathbb{T}}E_\omega}\le C_P$ and $\oneNorm{\widehat{P}_{\mathbb{T}}E_\omega}\le C_P$.
Therefore
\begin{equation}\label{app:pd-M-omega-l1}
\max_\omega\,\oneNorm{(\hat A^{-1}\widehat{P}_{\mathbb{T}} - A^{-1}P_{\mathbb{T}})E_\omega}
\le 2C_A\,\sigma^2\dstar\cdot C_P.
\end{equation}

\noindent\textbf{Step 3: assembling.}
Substituting \eqref{app:pd-M-omega-l1} into \eqref{app:pd-H-diff-l1}:
\[
\oneNorm{\hat H-H^\star}\le 2C_AC_P\,\sigma^2\dstar\,\oneNorm{\Gammavec},
\]
so
\[
\frac{\oneNorm{\hat H-H^\star}}{\dstar}
\le 2C_AC_P\,\sigma^2\,\oneNorm{\Gammavec}.
\]
The second-order piece therefore satisfies
\begin{equation}\label{app:pd-H-bias-2nd}
\boxed{
\left|\tfrac12\Pstar\big[\ddot s_\eta\ip{\Delta}{X}^2\ip{\hat H-H^\star}{X}\big]\right|
\le C(\mu,r,m)\,C_A\,\infnorm{\Delta}^2\,\oneNorm{\Gammavec}.
}
\end{equation}
Note the key cancellations:
$\sigma^2$ from $\norm{A^{-1}}_{1\to 1}\asymp\sigma^2\dstar$ cancels with
$1/\sigma^2$ from $|\ddot s_\eta|\le c_2/\sigma^2$,
and $\dstar$ from $\norm{A^{-1}}_{1\to 1}$ cancels with
$1/\dstar$ from the $\ell_1$-average bound (Lemma~\ref{app:pd-l1-average}).

\begin{remark}\label{app:pd-H-bias-summary}
The first-order piece (Proposition~\ref{app:pd-H-bias-prop}) is bounded by
\[
C(\mu,r,m)\,C_A\,\oneNorm{\Gammavec}\bigl[\infnorm{(\Id-\widehat{P}_{\mathbb{T}})\Delta}+\infnorm{(\Id-P_{\mathbb{T}})\Delta}\bigr],
\]
where the off-tangent components are controlled by the $\varepsilon^2$ bound of Proposition~\ref{app:pd-PTcompl-infty}.
The second-order piece \eqref{app:pd-H-bias-2nd} is bounded by
\[
C(\mu,r,m)\,C_A\,\infnorm{\Delta}^2\,\oneNorm{\Gammavec},
\]
which is also $O(\varepsilon^2)$ since $\infnorm{\Delta}=O(\varepsilon)$.
Both pieces retain explicit $C_A$ dependence and are dimension-free up to $\oneNorm{\Gammavec}$.
\end{remark}

\subsection{First-Order Cancellation}\label{app:pd-1st-cancel}

 This is the heart of the one-step estimator: the plug-in bias $\langle P_{\mathbb{T}}\Gamma,\Delta\rangle$ cancels \emph{exactly} with the leading term of the population correction $-\langle\Delta, G(H^\star)\rangle$ on the tangent space. The proof decomposes $\Delta=P_{\mathbb{T}}\Delta+(\Id-P_{\mathbb{T}})\Delta$: the tangent component cancels perfectly by the definition $H^\star=A^{-1}P_{\mathbb{T}}\Gamma$, while the normal component produces a residual controlled by $\|(\Id-P_{\mathbb{T}})\Delta\|_\infty$, which is quadratic in subspace errors.

\medskip
The combined first-order cancellation term from the bias decomposition is
\begin{equation}\label{app:pd-1st-cancel-term}
\ip{P_{\mathbb{T}}\Gammavec}{\Delta} + \Pstar\big[S_{\hat T}(H^\star)-S_{T^\star}(H^\star)\big].
\end{equation}

\subsubsection{Taylor expansion of the population score difference}
The second part of \eqref{app:pd-1st-cancel-term} is
\[
\Pstar\big[s_\eta(Y,\hat\eta)\ip{H^\star}{X}\big] - \Pstar\big[s_\eta(Y,\eta^\star)\ip{H^\star}{X}\big].
\]
The second term vanishes ($\E[s_\eta|X]=0$). Expanding the first around $\eta^\star$:
\begin{align}
\Pstar\big[S_{\hat T}(H^\star)-S_{T^\star}(H^\star)\big]
&= \Pstar\big[\dot s_\eta(Y,\eta^\star)\ip{\Delta}{X}\ip{H^\star}{X}\big]
+ \frac12\Pstar\big[\ddot s_\eta(Y,\bar\eta)\ip{\Delta}{X}^2\ip{H^\star}{X}\big]
\notag\\
&= -\ip{\Delta}{G(H^\star)} + R_2(\Delta;H^\star),
\label{app:pd-Taylor-popul}
\end{align}
where we used $\E[\dot s_\eta(Y,\eta^\star)|X] = -I(\eta^\star)$ and defined
\[
R_2(\Delta;H^\star) := \frac12\Pstar\big[\ddot s_\eta(Y,\bar\eta)\ip{\Delta}{X}^2\ip{H^\star}{X}\big].
\]

\subsubsection{The cancellation mechanism}
Plugging \eqref{app:pd-Taylor-popul} into \eqref{app:pd-1st-cancel-term}:
\begin{equation}\label{app:pd-cancel-decomp}
\ip{P_{\mathbb{T}}\Gammavec}{\Delta} - \ip{\Delta}{G(H^\star)} + R_2(\Delta;H^\star).
\end{equation}
Now decompose $\Delta = P_{\mathbb{T}}\Delta + (\Id - P_{\mathbb{T}})\Delta$ in the first-order part:
\begin{equation}
\begin{aligned}
&\ip{P_{\mathbb{T}}\Gammavec}{\Delta} - \ip{\Delta}{G(H^\star)}
\notag\\
=& \Big(\ip{P_{\mathbb{T}}\Gammavec}{P_{\mathbb{T}}\Delta} - \ip{P_{\mathbb{T}}\Delta}{G(H^\star)}\Big)
\\
+& \Big(\ip{P_{\mathbb{T}}\Gammavec}{(\Id-P_{\mathbb{T}})\Delta} - \ip{(\Id-P_{\mathbb{T}})\Delta}{G(H^\star)}\Big).
\label{app:pd-cancel-split}
\end{aligned}
\end{equation}
\noindent\textbf{The tangent-space part cancels exactly.}
Since $H^\star = A^{-1}P_{\mathbb{T}}\Gammavec$ with $A = P_{\mathbb{T}}GP_{\mathbb{T}}$:
\[
\ip{P_{\mathbb{T}}\Delta}{G(H^\star)}
= \ip{P_{\mathbb{T}}\Delta}{G(A^{-1}P_{\mathbb{T}}\Gammavec)}
= \ip{P_{\mathbb{T}}\Delta}{P_{\mathbb{T}}G(A^{-1}P_{\mathbb{T}}\Gammavec)} + \ip{P_{\mathbb{T}}\Delta}{(\Id-P_{\mathbb{T}})G(H^\star)}.
\]
The first term equals
\[
\ip{P_{\mathbb{T}}\Delta}{A\cdot A^{-1}P_{\mathbb{T}}\Gammavec} = \ip{P_{\mathbb{T}}\Delta}{P_{\mathbb{T}}\Gammavec},
\]
which cancels with $\ip{P_{\mathbb{T}}\Gammavec}{P_{\mathbb{T}}\Delta}$. The second term satisfies $\ip{P_{\mathbb{T}}\Delta}{(\Id-P_{\mathbb{T}})G(H^\star)} = 0$ by orthogonality.

\noindent\textbf{The normal component survives but is small.}
The surviving first-order term from \eqref{app:pd-cancel-split} is
\begin{equation}\label{app:pd-surviving-1st}
\ip{P_{\mathbb{T}}\Gammavec}{(\Id-P_{\mathbb{T}})\Delta} - \ip{(\Id-P_{\mathbb{T}})\Delta}{G(H^\star)}.
\end{equation}

\subsubsection{Bounding the surviving first-order term}
\begin{proposition}\label{app:pd-1st-cancel-prop}
Under Assumption~\ref{app:pd-opnorms},
\begin{equation}\label{app:pd-1st-cancel-bound}
\left|\ip{P_{\mathbb{T}}\Gammavec}{\Delta} - \ip{\Delta}{G(H^\star)}\right|
\le
C\,\oneNorm{\Gammavec}\,\infnorm{(\Id-P_{\mathbb{T}})\Delta},
\end{equation}
where $C = C_A\cdot\mathrm{poly}(\mu,r,m)$ (Remark~\ref{app:pd-CG}).
\end{proposition}

\noindent\textbf{Proof of Proposition~\ref{app:pd-1st-cancel-prop}.}
For the first part of \eqref{app:pd-surviving-1st},
\[
|\ip{P_{\mathbb{T}}\Gammavec}{(\Id-P_{\mathbb{T}})\Delta}| = |\ip{\Gammavec}{P_{\mathbb{T}}(\Id-P_{\mathbb{T}})\Delta}| = 0
\]
since $P_{\mathbb{T}}(\Id-P_{\mathbb{T}}) = 0$.
For the second part, use symmetry of $G$ and self-adjointness of $A^{-1}P_{\mathbb{T}}$ to transfer the operator onto $\Gammavec$:
\[
\ip{(\Id-P_{\mathbb{T}})\Delta}{G(H^\star)}
= \ip{G(\Id-P_{\mathbb{T}})\Delta}{A^{-1}P_{\mathbb{T}}\Gammavec}
= \ip{A^{-1}P_{\mathbb{T}}G(\Id-P_{\mathbb{T}})\Delta}{\Gammavec}.
\]
By H\"older's inequality,
\[
\begin{aligned}
\left|\ip{A^{-1}P_{\mathbb{T}}G(\Id-P_{\mathbb{T}})\Delta}{\Gammavec}\right|
&\le
\oneNorm{\Gammavec}\cdot\infnorm{A^{-1}P_{\mathbb{T}}G(\Id-P_{\mathbb{T}})\Delta}\\
&\le
\oneNorm{\Gammavec}\cdot\inftyinfnorm{A^{-1}P_{\mathbb{T}}G}\cdot\infnorm{(\Id-P_{\mathbb{T}})\Delta}.
\end{aligned}
\]
By Remark~\ref{app:pd-CG},
$\inftyinfnorm{A^{-1}P_{\mathbb{T}}G} = C_A\cdot\mathrm{poly}(\mu,r,m)$,
which gives \eqref{app:pd-1st-cancel-bound}.
\hfill$\square$\medskip

\begin{remark}
Note $\ip{P_{\mathbb{T}}\Gammavec}{(\Id-P_{\mathbb{T}})\Delta} = 0$ exactly, so the surviving term is \emph{only} $-\ip{(\Id-P_{\mathbb{T}})\Delta}{G(H^\star)}$, which is controlled by $\infnorm{(\Id-P_{\mathbb{T}})\Delta}$. By Proposition~\ref{app:pd-PTcompl-infty}, this is quadratic in subspace errors ($O(\varepsilon^2)$), which is why only the second-order remainder $R_2$ survives as the dominant bias term.
\end{remark}

\subsection{Second-Order Score Remainder}\label{app:pd-2nd-order}

\noindent\textbf{Intuition.} After the first-order cancellation, the dominant bias term is the second-order Taylor remainder $R_2(\Delta;H^\star)=O(\|\Delta\|_\infty^2)$. The key mechanism is the $\ell_1$ dimension cancellation: $\|A^{-1}P_{\mathbb{T}}E_\omega\|_1/d^\star \le C_A\sigma^2\cdot\mathrm{poly}(\mu,r,m)$, where the $d^\star$ in $\|A^{-1}\|_{1\to 1}\asymp\sigma^2 d^\star$ cancels with $1/d^\star$ from the comparison averaging. This gives $|R_2|\le C(\mu,r,m) C_A\|\Delta\|_\infty^2\|\Gamma\|_1$ without any residual dimension factor.

\medskip
The second-order score remainder from \eqref{app:pd-Taylor-popul} is
\[
R_2(\Delta;H^\star) = \frac12\Pstar\big[\ddot s_\eta(Y,\bar\eta)\ip{\Delta}{X}^2\ip{H^\star}{X}\big].
\]

\begin{proposition}[Second-order remainder bound]\label{app:pd-2nd-order-prop}
Under Assumption~\ref{app:pd-score} (with $|\ddot s_\eta|\le c_2/\sigma^2$) and comparison design,
\begin{equation}\label{app:pd-2nd-order-bound}
|R_2(\Delta;H^\star)|
\le
\frac{C}{\sigma^2}\,\infnorm{\Delta}^2\,\frac{\oneNorm{H^\star}}{\dstar}.
\end{equation}
Expanding $H^\star = A^{-1}P_{\mathbb{T}}\Gammavec$ and using $\inftyinfnorm{A^{-1}}\le C_A\sigma^2\dstar$, $\inftyinfnorm{P_{\mathbb{T}}}=\mathrm{poly}(\mu,r,m)$:
\begin{equation}\label{app:pd-2nd-order-Gamma}
|R_2(\Delta;H^\star)|
\le
C(\mu,r,m)\,C_A\,\infnorm{\Delta}^2\,\oneNorm{\Gammavec}.
\end{equation}
\end{proposition}

\noindent\textbf{Proof of Proposition~\ref{app:pd-2nd-order-prop}.}
Since $|\ip{\Delta}{X}|\le 2\infnorm{\Delta}$ (comparison atom) and $|\ddot s_\eta|\le c_2/\sigma^2$,
\[
|R_2(\Delta;H^\star)|
\le
\frac{2c_2}{\sigma^2}\,\infnorm{\Delta}^2\,\E^\star\!\big[|\ip{H^\star}{X}|\big].
\]
Under the comparison design, $\ip{H^\star}{X} = H^\star_{u,a}-H^\star_{u,b}$ where $(u,\{a,b\})$ is drawn uniformly, so by the triangle inequality and averaging,
$\E^\star\!\big[|\ip{H^\star}{X}|\big]
\le
\frac{2}{\dstar}\,\oneNorm{H^\star}$,
which gives \eqref{app:pd-2nd-order-bound} (absorbing the constant~$4$ into $C$).
Next, using $H^\star = A^{-1}P_{\mathbb{T}}\Gammavec = \sum_\omega\gamma_\omega A^{-1}P_{\mathbb{T}}E_\omega$ and the triangle inequality,
\[
\oneNorm{H^\star}
\le
\oneNorm{\Gammavec}\cdot\max_\omega\oneNorm{A^{-1}P_{\mathbb{T}}E_\omega}.
\]
For each $E_\omega$,
\[
\oneNorm{A^{-1}P_{\mathbb{T}}E_\omega}
\le
\norm{A^{-1}}_{1\to 1}\,\oneNorm{P_{\mathbb{T}}E_\omega},
\]
and by symmetry of $A^{-1}$, $\norm{A^{-1}}_{1\to 1} = \inftyinfnorm{A^{-1}}$.
By Assumption~\ref{app:pd-opnorms}, $\inftyinfnorm{A^{-1}}\le C_A\,\sigma^2\dstar$, and since $P_{\mathbb{T}}$ is symmetric,
\[
\max_\omega\oneNorm{P_{\mathbb{T}}E_\omega} = \norm{P_{\mathbb{T}}}_{1\to 1} = \inftyinfnorm{P_{\mathbb{T}}} = \mathrm{poly}(\mu,r,m)
\]
(Lemma~\ref{app:pd-Pinfty}), which is dimension-free. Hence
\[
\frac{1}{\dstar}\;\inftyinfnorm{A^{-1}}\cdot\max_\omega\oneNorm{P_{\mathbb{T}}E_\omega}
\;=\;
\frac{C_A\,\sigma^2\dstar}{\dstar}\cdot\mathrm{poly}(\mu,r,m)
\;=\;
C_A\,\sigma^2\cdot\mathrm{poly}(\mu,r,m).
\]
The $\sigma^2$ cancels the $1/\sigma^2$ prefactor, yielding \eqref{app:pd-2nd-order-Gamma}.
\hfill$\square$\medskip

\begin{remark}[$\ell_1$ dimension cancellation]
The key mechanism is
\[
\oneNorm{A^{-1}P_{\mathbb{T}}E_\omega}/\dstar \le \inftyinfnorm{A^{-1}}\cdot\inftyinfnorm{P_{\mathbb{T}}}/\dstar = C_A\,\sigma^2\cdot\mathrm{poly}(\mu,r,m).
\]
The $\dstar$ in $\inftyinfnorm{A^{-1}}\asymp\sigma^2\dstar$ cancels exactly with the $1/\dstar$ from the comparison averaging, and $\inftyinfnorm{P_{\mathbb{T}}}=\mathrm{poly}(\mu,r,m)$ is dimension-free.
This gives $|R_2|\le C(\mu,r,m)\,C_A\,\infnorm{\Delta}^2\,\oneNorm{\Gammavec}$ without any residual $\sqrt{\bar d}$ factor.
\end{remark}

\subsection{Additional Error from Estimated Information $\hat G$}\label{app:pd-hatG}

As discussed in the proof strategy (Section~\ref{app:pd-strategy}), the analysis in Sections~\ref{app:pd-RempH}--\ref{app:pd-2nd-order} was carried out using the \emph{true} information operator $G = \E^\star[I(\eta^\star(X))(X\otimes X)]$, i.e., the estimated direction was defined as $\hat H = (\widehat{P}_{\mathbb{T}} G \widehat{P}_{\mathbb{T}})^{-1}\widehat{P}_{\mathbb{T}}\Gammavec$.
In practice, $G$ is replaced by the plug-in estimate $\hat G = \E^\star[I(\hat\eta(X))(X\otimes X)]$, yielding $\hat A' := \widehat{P}_{\mathbb{T}}\hat G\widehat{P}_{\mathbb{T}}$ in place of $\hat A := \widehat{P}_{\mathbb{T}}G\widehat{P}_{\mathbb{T}}$.
We now show that this substitution introduces only a negligible multiplicative correction.

\medskip
\noindent\textbf{Perturbation bound for $\hat G - G$.}
The only difference between $\hat G$ and $G$ is the use of $I(\hat\eta(x))$ in place of $I(\eta^\star(x))$.
Since the Fisher information $I(\cdot)$ is Lipschitz on the bounded-signal domain (Assumption~\ref{app:pd-score}), we have
\[
|I(\hat\eta(x)) - I(\eta^\star(x))|
\;\le\;
L_I\,|\hat\eta(x) - \eta^\star(x)|
\;=\;
L_I\,|\ip{\Delta}{x}|
\;\le\;
2L_I\,\infnorm{\Delta}
\]
for each comparison atom $x$, where $L_I$ is the Lipschitz constant of $I$ on $[-B-1,B+1]$.
Therefore
\[
\|\hat G - G\|_{\infty\to\infty}
\;\le\;
\frac{C\,\infnorm{\Delta}}{\sigma^2\dstar},
\]
where we used $\|G\|_{\infty\to\infty}\asymp 1/(\sigma^2\dstar)$.
Define the relative perturbation
\begin{equation}\label{eq:epsG-def}
\varepsilon_G
\;:=\;
\frac{\|\hat G - G\|_{\infty\to\infty}}{\|G\|_{\infty\to\infty}}
\;\le\;
C\,\infnorm{\Delta}
\;\le\;
C\,\sigma\sqrt{\frac{\bar d\log^c\bar d}{n}}
\;\ll\; 1.
\end{equation}

\noindent\textbf{Effect on the restricted inverse.}
Write $\hat A' = \hat A + \widehat{P}_{\mathbb{T}}(\hat G - G)\widehat{P}_{\mathbb{T}}$, so
\[
\hat A'^{-1}
\;=\;
\hat A^{-1}\bigl(\Id + \hat A^{-1}\widehat{P}_{\mathbb{T}}(\hat G - G)\widehat{P}_{\mathbb{T}}\bigr)^{-1}.
\]
The perturbation term satisfies
$\|\hat A^{-1}\widehat{P}_{\mathbb{T}}(\hat G - G)\widehat{P}_{\mathbb{T}}\|_{\infty\to\infty}
\le \|\hat A^{-1}\widehat{P}_{\mathbb{T}} G\|_{\infty\to\infty}\cdot\varepsilon_G
\le C(\mu,r,m)\cdot\varepsilon_G$,
where we used Remark~\ref{app:pd-CG}.
By the Neumann series, when $\varepsilon_G$ is sufficiently small:
\[
\hat A'^{-1} = \hat A^{-1}\bigl(\Id + O(\varepsilon_G)\bigr)
\quad\text{in }\|\cdot\|_{\infty\to\infty}.
\]

\noindent\textbf{Impact on each error term.}
The replacement $\hat A\to\hat A'$ affects only the terms involving the estimated direction $\hat H$:
\begin{itemize}
\item $R_{\mathrm{emp}}^H$ and the H-direction bias: the direction error $\hat H'-H^\star$ acquires a multiplicative $(1+O(\varepsilon_G))$ correction. Since $\varepsilon_G = O(\infnorm{\Delta})$ is of the same order as the existing subspace estimation error, this is absorbed into the combined bound without changing the leading-order rate.
\item $R_{\mathrm{emp}}^\eta$, $R_{\mathrm{proj}}$, the first-order cancellation, and the second-order remainder all involve only the \emph{true} direction $H^\star$ or the true operator $G$, so they receive \emph{no additional error} from the $\hat G$ substitution.
\end{itemize}
In summary, all bounds in this appendix remain valid with $\hat G$ in place of $G$, up to a multiplicative $(1+O(\varepsilon_G))$ factor that is negligible under the sample size condition.

\subsection{Combined Bias Bound}\label{app:pd-combined}

We now collect all bias terms from Sections~\ref{app:pd-Rproj}--\ref{app:pd-2nd-order} and assemble the final bound.

\subsubsection{Recap of individual bounds}
The total bias decomposes as
\begin{equation}\label{app:pd-bias-decomp-full}
\mathrm{Bias}
= \underbrace{R_{\mathrm{proj}}}_{\text{Sec.~\ref{app:pd-Rproj}}}
+ \underbrace{\ip{G\Delta}{\hat H - H^\star}}_{\text{H-bias 1st, Sec.~\ref{app:pd-H-bias}}}
+ \underbrace{\text{(H-bias 2nd)}}_{\text{Sec.~\ref{app:pd-H-bias}}}
+ \underbrace{\ip{P_{\mathbb{T}}\Gammavec}{\Delta}-\ip{\Delta}{G(H^\star)}}_{\text{1st cancel, Sec.~\ref{app:pd-1st-cancel}}}
+ \underbrace{R_2(\Delta;H^\star)}_{\text{2nd order, Sec.~\ref{app:pd-2nd-order}}}.
\end{equation}
The individual bounds are:
\begin{enumerate}[label=(\roman*)]
\item \textbf{Projection leakage} (Eq.~\eqref{app:pd-Rproj-basic}):
$|R_{\mathrm{proj}}| \le \oneNorm{\Gammavec}\cdot\infnorm{(\Id-P_{\mathbb{T}})\Delta}$.
\item \textbf{H-direction bias, first order} (Prop.~\ref{app:pd-H-bias-prop}):
\[
|\ip{G\Delta}{\hat H-H^\star}|
\le C(\mu,r,m)\,C_A\,\oneNorm{\Gammavec}\bigl[\infnorm{(\Id-\widehat{P}_{\mathbb{T}})\Delta}+\infnorm{(\Id-P_{\mathbb{T}})\Delta}\bigr].
\]
\item \textbf{First-order cancellation} (Prop.~\ref{app:pd-1st-cancel-prop}):
\[
|\ip{P_{\mathbb{T}}\Gammavec}{\Delta}-\ip{\Delta}{G(H^\star)}| \le C(\mu,r,m)\,C_A\,\oneNorm{\Gammavec}\,\infnorm{(\Id-P_{\mathbb{T}})\Delta}.
\]
\item \textbf{Second-order remainder} (Prop.~\ref{app:pd-2nd-order-prop}):
$|R_2(\Delta;H^\star)| \le C(\mu,r,m)\,C_A\,\infnorm{\Delta}^2\,\oneNorm{\Gammavec}$.
\item \textbf{H-direction bias, second order} (Eq.~\eqref{app:pd-H-bias-2nd}):
controlled by $C\,C_A\,\infnorm{\Delta}^2\,\oneNorm{\Gammavec}$; same order as~(iv).
\end{enumerate}

\subsubsection{Combined bound}

\begin{theorem}[Combined bias bound]\label{app:pd-combined-bias-thm}
Under Assumptions~\ref{app:pd-score} and \ref{app:pd-opnorms}, the total bias satisfies
\begin{equation}\label{app:pd-combined-bias}
|\mathrm{Bias}|
\;\le\;
C(\mu,r,m)\,C_A\,\oneNorm{\Gammavec}
\Big[
\infnorm{(\Id-P_{\mathbb{T}})\Delta}
+ \infnorm{(\Id-\widehat{P}_{\mathbb{T}})\Delta}
+ \infnorm{\Delta}^2
\Big].
\end{equation}
\end{theorem}

\noindent\textbf{Proof of Theorem~\ref{app:pd-combined-bias-thm}.}
Sum the bounds (i)--(iv): terms (i) and (iii) contribute $C\,\oneNorm{\Gammavec}\,\infnorm{(\Id-P_{\mathbb{T}})\Delta}$;
term (ii) contributes $C\,\oneNorm{\Gammavec}[\infnorm{(\Id-\widehat{P}_{\mathbb{T}})\Delta}+\infnorm{(\Id-P_{\mathbb{T}})\Delta}]$;
term (iv) contributes $C\,\oneNorm{\Gammavec}\,\infnorm{\Delta}^2$.
Collecting with the maximum constant gives \eqref{app:pd-combined-bias}.
\hfill$\square$\medskip

\subsubsection{Substituting the $\varepsilon^2$ and entrywise rates}
We now plug in the standard tensor estimation rates to obtain a fully explicit final bound.

\noindent\textbf{Input rates.}
Recall from Section~\ref{app:pd-Rproj-final}:
\begin{itemize}
\item \emph{Normal-component error} (Proposition~\ref{app:pd-PTcompl-infty}):
$\infnorm{(\Id-P_{\mathbb{T}})\Delta},\;\infnorm{(\Id-\widehat{P}_{\mathbb{T}})\Delta}\;\le\;C(\mu,r,m)\,\rho^2$,
where $\rho := \frac{\sigma}{\lambda_{\min}}\sqrt{\frac{\dstar\,\bar d\log^c\bar d}{n}}$.
\item \emph{Entrywise error}: $\infnorm{\Delta}\le C(\mu,r,m)\cdot\sigma\sqrt{\bar d/n}$.
\end{itemize}

\noindent\textbf{Substitution into each term.}
\begin{enumerate}[label=(\roman*)]
\item \emph{First-order cancellation} (Prop.~\ref{app:pd-1st-cancel-prop}):
\[
C\,C_A\,\oneNorm{\Gammavec}\,\infnorm{(\Id-P_{\mathbb{T}})\Delta}
\;\le\;
C(\mu,r,m)\,C_A\,\oneNorm{\Gammavec}\,\rho^2
\;=\;
C(\mu,r,m)\,C_A\,\oneNorm{\Gammavec}\cdot
\frac{\sigma^2\,\dstar\,\bar d\log\bar d}{\lambda_{\min}^2\,n}.
\]
\item \emph{Second-order remainder} (Prop.~\ref{app:pd-2nd-order-prop}):
\[
C\,C_A\,\oneNorm{\Gammavec}\,\infnorm{\Delta}^2
\;\le\;
C(\mu,r,m)\,C_A\,\oneNorm{\Gammavec}\cdot
\frac{\sigma^2\,\bar d}{n}.
\]
\item \emph{Projection leakage} (Eq.~\eqref{app:pd-Rproj-final-eq}):
\[
|R_{\mathrm{proj}}|
\;\le\;
C(\mu,r,m)\,\oneNorm{\Gammavec}\,\rho^2
\;=\;
C(\mu,r,m)\,\oneNorm{\Gammavec}\cdot
\frac{\sigma^2\,\dstar\,\bar d\log\bar d}{\lambda_{\min}^2\,n}.
\]
\item \emph{H-direction bias} (Prop.~\ref{app:pd-H-bias-prop}):
controlled by
\[
C\,C_A\,\oneNorm{\Gammavec}\bigl[\infnorm{(\Id-P_{\mathbb{T}})\Delta}+\infnorm{(\Id-\widehat{P}_{\mathbb{T}})\Delta}+\infnorm{\Delta}^2\bigr]
\le C\,C_A\,\oneNorm{\Gammavec}\bigl[\rho^2+\infnorm{\Delta}^2\bigr].
\]
\end{enumerate}

\noindent\textbf{Final explicit bound.}
Under the SNR condition $\sigma/\lambda_{\min}\cdot\sqrt{\dstar}\asymp 1$, the $\rho^2$ and $\infnorm{\Delta}^2$ terms are the same order; both are retained.
Combining:
\begin{equation}\label{app:pd-combined-bias-rate}
\boxed{
|\mathrm{Bias}|
\;\le\;
C(\mu,r,m)\,C_A\,\oneNorm{\Gammavec}\left[
\frac{\sigma^2\,\dstar\,\bar d\log\bar d}{\lambda_{\min}^2\,n}
+ \frac{\sigma^2\,\bar d}{n}
\right].
}
\end{equation}

\begin{remark}[Pairwise-comparison specialization]\label{app:pd-simplified-combined}
In the pairwise-comparison setting, $\sigma=O(1)$ and
$\lambda_{\min}\asymp\sqrt{d^\star}$ by
Assumption~\ref{ass:signal-strength} and~\eqref{eq:Fnorm-lambda}.
All bounds simplify because the two key estimation parameters collapse to the same scale:
\begin{equation}\label{eq:pairwise-rates}
\rho
\;=\;
\frac{\sigma}{\lambda_{\min}}\sqrt{\frac{\dstar\bar d\log\bar d}{n}}
\;\asymp\;
\sqrt{\frac{\bar d\log^c\bar d}{n}},
\qquad
\delta_\infty
\;:=\;
\infnorm{\Delta}
\;\lesssim\;
\sqrt{\frac{\bar d\log^c\bar d}{n}}.
\end{equation}
In particular, both $\rho$ and $\delta_\infty$ are of order $\sqrt{\bar d\log^c\bar d/n}$, and the products $\rho^2$, $\delta_\infty^2$, $\rho\cdot\delta_\infty$ are all of order $\bar d\log^c\bar d/n$.

Substituting these rates into the combined bounds~\eqref{app:pd-combined-bias-rate} and~\eqref{app:pd-RempH-final}, with $\delta=\bar d^{-c}$ (so $\log(2/\delta)\asymp c\log\bar d$), the total remainder satisfies
\begin{equation}\label{app:pd-simplified-total}
\boxed{
|R_n|
\;\le\;
C(\mu,r,m)\,C_A\,\oneNorm{\Gammavec}\,\frac{\bar d\log^c\bar d}{n}
}
\end{equation}
with probability $\ge 1-\bar d^{-c}$. This is the bound stated in Theorem~\ref{thm:combined} of the main text. Each of the six individual terms ($R_{\rm emp}^H$, $R_{\rm emp}^\eta$, $R_{\rm proj}$, $R_{\rm bias}^H$, first-order cancellation, second-order remainder) contributes at most this order, as summarized in Table~\ref{tab:proof_outline}.
\end{remark}

\subsection{Range of $C_A$ and dimension-free conditions}\label{app:CA-range}

This subsection proves Proposition~\ref{prop:CA-range-main}: the coarse range $C(\mu,r,m)\le C_A\le C(\mu,r,m)\sqrt{\bar d}$, and the dimension-free regime under bounded signal. Throughout, we work under the bounded-signal assumption $\|T^\star\|_\infty\le B$, so the Fisher information satisfies $c_B\le I(\eta^\star(X))\le C_B$ for constants depending only on~$B$.

Recall that $A^{-1}$ acts only on the tangent space~$\mathbb{T}$; accordingly, we define
\[
\inftyinfnorm{A^{-1}} := \sup_{\|V\|_\infty\le 1}\|A^{-1}P_{\mathbb{T}}V\|_\infty.
\]

\subsubsection{Spectral scale of the restricted inverse}

\begin{lemma}[Spectral scale of $A$]\label{lem:CA-spectral}
Under bounded signal,
\[
\frac{c_B}{\dstar}\,P_{\mathbb{T}} \preceq A \preceq \frac{C_B}{\dstar}\,P_{\mathbb{T}}
\qquad\text{on }\mathbb{T}.
\]
Consequently, $\|A^{-1}\|_{\mathrm{op}}\asymp\dstar$, with constants depending only on~$B$.
\end{lemma}

\begin{proof}
For any $H\in\mathbb{T}$, $\langle H,AH\rangle = \E_\star[I(\eta^\star(X))\langle H,X\rangle^2]$. The Fisher bounds give
\[
c_B\,\E_\star\langle H,X\rangle^2
\le \langle H,AH\rangle
\le C_B\,\E_\star\langle H,X\rangle^2.
\]
Under the uniform pairwise-comparison design, the Frobenius reduction (Lemma~\ref{app:pd-pairwise}) gives $\E_\star\langle H,X\rangle^2\asymp\|H\|_F^2/\dstar$, and the claim follows.
\end{proof}

\subsubsection{Entrywise bound for tangent-space elements}

The upper bound $C_A\lesssim\sqrt{\bar d}$ is not a generic ambient-space norm inequality (which would only yield $\sqrt{\dstar}$). The sharper factor comes from the low-rank tangent-space geometry.

\begin{lemma}[Tangent-space entrywise bound]\label{lem:CA-tangent-entrywise}
Assume $\mu$-incoherence. For every $W\in\mathbb{T}$,
\[
\|W\|_\infty \le C(\mu,r,m)\sqrt{\frac{\bar d}{\dstar}}\,\|W\|_F.
\]
\end{lemma}

\begin{proof}
Write the Tucker tangent-space decomposition $\mathbb{T}=\Pi_\emptyset(\mathcal{H})\oplus\bigoplus_{j=1}^m\Pi_{\{j\}}(\mathcal{H})$, where $\Pi_\emptyset$ is the fully projected block and $\Pi_{\{j\}}$ is the single-mode correction block. Any $W\in\mathbb{T}$ decomposes orthogonally as $W=W_\emptyset+\sum_{j=1}^m W_j$.

For the main block $W_\emptyset$, every element lies in the span of tensors $a^{(1)}\otimes\cdots\otimes a^{(m)}$ with $a^{(k)}\in\mathrm{col}(U_k)$. By incoherence, $\|a^{(k)}\|_\infty\le\sqrt{\mu r_k/d_k}\,\|a^{(k)}\|_2$, so for a rank-one tensor in this block,
\[
\|a^{(1)}\otimes\cdots\otimes a^{(m)}\|_\infty
\le \prod_{k=1}^m\sqrt{\mu r_k/d_k}\;\prod_{k=1}^m\|a^{(k)}\|_2
= C(\mu,r,m)\,(d^\star)^{-1/2}\,\|a^{(1)}\otimes\cdots\otimes a^{(m)}\|_F.
\]
By linearity, $\|W_\emptyset\|_\infty\le C(\mu,r,m)\,(d^\star)^{-1/2}\,\|W_\emptyset\|_F$.

For a single-mode block $W_j$, the range is generated by tensors with one free mode and all remaining modes projected. The free mode provides no incoherence gain ($\|b^{(j)}\|_\infty\le\|b^{(j)}\|_2$), while each projected mode $k\neq j$ still satisfies $\|a^{(k)}\|_\infty\le\sqrt{\mu r_k/d_k}\,\|a^{(k)}\|_2$. Therefore
\[
\|W_j\|_\infty \le C(\mu,r,m)\Bigl(\prod_{k\neq j}d_k^{-1/2}\Bigr)\|W_j\|_F
= C(\mu,r,m)\sqrt{\frac{d_j}{\dstar}}\,\|W_j\|_F
\le C(\mu,r,m)\sqrt{\frac{\bar d}{\dstar}}\,\|W_j\|_F.
\]
Combining via the triangle inequality and the orthogonality of the $m+1$ blocks ($m$ is fixed) gives the claim.
\end{proof}

\begin{remark}\label{rem:CA-bardfactor}
This lemma is the precise source of the factor $\sqrt{\bar d}$: the free mode in each single-mode correction block introduces a $\sqrt{d_j}$ factor that cannot be controlled by incoherence.
\end{remark}

\subsubsection{Coarse range of $C_A$}

\begin{proposition}[Coarse range of $C_A$]\label{prop:CA-coarse-range}
Under bounded signal and $\mu$-incoherence,
\[
c(\mu,r,m)\le C_A\le C(\mu,r,m)\sqrt{\bar d}.
\]
\end{proposition}

\begin{proof}
\emph{Lower bound.} Since $\inftyinfnorm{M}\ge\|M\|_{\mathrm{op}}$ for any linear operator, Lemma~\ref{lem:CA-spectral} gives
\[
\inftyinfnorm{A^{-1}} \ge \|A^{-1}\|_{\mathrm{op}} \ge c(B)\,\dstar,
\]
so $C_A=\inftyinfnorm{A^{-1}}/\dstar\ge c(B)$.

\emph{Upper bound.} Let $V$ satisfy $\|V\|_\infty\le 1$ and set $W:=A^{-1}P_{\mathbb{T}}V\in\mathbb{T}$. By Lemma~\ref{lem:CA-tangent-entrywise},
\[
\|W\|_\infty \le C(\mu,r,m)\sqrt{\frac{\bar d}{\dstar}}\,\|W\|_F.
\]
Since $\|W\|_F\le\|A^{-1}\|_{\mathrm{op}}\|P_{\mathbb{T}}V\|_F\le C(B)\,\dstar\cdot\sqrt{\dstar}$ (using $\|V\|_F\le\sqrt{\dstar}\,\|V\|_\infty$), we obtain
\[
\inftyinfnorm{A^{-1}}
\le C(\mu,r,m)\sqrt{\frac{\bar d}{\dstar}}\cdot C(B)\,\dstar\sqrt{\dstar}
= C(\mu,r,m,B)\,\dstar\sqrt{\bar d},
\]
giving $C_A\le C(\mu,r,m,B)\sqrt{\bar d}$.
\end{proof}

\subsubsection{Dimension-free $C_A$ under bounded signal}

When the Fisher weights are nearly constant, a perturbative argument around the isotropic benchmark shows that $C_A$ is dimension-free.

\begin{proposition}[Dimension-free $C_A$ in the near-constant regime]\label{prop:CA-dimfree}
If $\|T^\star\|_\infty\le B_0$ where $B_0=B_0(\mu,m,r)$ is sufficiently small depending only on the structural parameters, then
\[
\inftyinfnorm{A^{-1}} \vee \inftyinfnorm{\widehat A^{-1}} \le C(\mu,r,m,B_0)\,\dstar.
\]
In particular, $C_A\le C(\mu,r,m,B_0)$, a dimension-free constant.
\end{proposition}

\begin{proof}
\emph{Constant-weight benchmark.} Let $I_0:=I(0)=1/4$ and define the benchmark operator
\[
A_0 := P_{\mathbb{T}}G_0 P_{\mathbb{T}},
\qquad
G_0(U):=\E_\star[I_0\langle U,X\rangle X].
\]
Under constant Fisher weights, the information operator is isotropic on the identifiable comparison space: $A_0=(I_0/c_{\mathrm{pw}}\dstar)\,P_{\mathbb{T}}$, where $c_{\mathrm{pw}}$ is the pairwise-design normalization constant. Hence
\[
A_0^{-1} = \frac{c_{\mathrm{pw}}\dstar}{I_0}\,P_{\mathbb{T}},
\qquad
\inftyinfnorm{A_0^{-1}} = \frac{c_{\mathrm{pw}}\dstar}{I_0}\inftyinfnorm{P_{\mathbb{T}}} \le C(\mu,r,m)\,\dstar,
\]
where the last step uses $\inftyinfnorm{P_{\mathbb{T}}}=\mathrm{poly}(\mu,r,m)$ (Lemma~\ref{app:pd-Pinfty}).

\emph{Fisher-weight perturbation.} Since $I(\eta)$ is smooth and even, for $|\eta|\le 2B_0$ we have $|I(\eta)-I_0|\le L_I B_0$ for an absolute constant~$L_I$. Write
\[
I(\eta^\star(X)) = I_0+\delta(X),
\qquad
\|\delta\|_\infty \le \varepsilon(B_0) \to 0\text{ as }B_0\to 0.
\]
Decompose $G=G_0+E$ and $A=A_0+\Delta_A$ with $\Delta_A:=P_{\mathbb{T}}EP_{\mathbb{T}}$. For the pairwise design, each design tensor $X$ has two nonzero entries equal to $\pm 1$, so $|\langle U,X\rangle|\le 2\|U\|_\infty$. Since the coordinate~$\omega$ is involved with probability $O(1/\dstar)$,
\[
\inftyinfnorm{E} \le C\,\varepsilon(B_0)\,(d^\star)^{-1},
\qquad
\inftyinfnorm{\Delta_A} \le \inftyinfnorm{P_{\mathbb{T}}}^2\inftyinfnorm{E} \le C(\mu,r,m)\,\varepsilon(B_0)\,(d^\star)^{-1}.
\]

\emph{Neumann expansion.} The normalized perturbation satisfies
\[
\rho := \inftyinfnorm{A_0^{-1}\Delta_A} \le \inftyinfnorm{A_0^{-1}}\inftyinfnorm{\Delta_A} \le C(\mu,r,m)\,\varepsilon(B_0).
\]
Choose $B_0$ small enough that $\rho\le 1/2$. Then $A=A_0(I+A_0^{-1}\Delta_A)$ with
\[
\|(I+A_0^{-1}\Delta_A)^{-1}\|_{\infty\to\infty} \le \sum_{k=0}^\infty\rho^k = \frac{1}{1-\rho}\le 2,
\]
so $\inftyinfnorm{A^{-1}}\le 2\inftyinfnorm{A_0^{-1}}\le C(\mu,r,m,B_0)\,\dstar$.

\emph{Estimated operator.} The same argument applies to $\widehat A=\widehat P_{\mathbb{T}}\widehat G\widehat P_{\mathbb{T}}$ on the high-probability event that \[\inftyinfnorm{\widehat P_{\mathbb{T}}}\le C(\mu,r,m), \qquad \inftyinfnorm{\widehat G-G}=o((d^\star)^{-1}).\]
\end{proof}

\subsection{Relative error of plug-in variance estimators}\label{app:pd-variance}

This subsection proves the variance estimation bounds~\eqref{eq:Veff-relative} and~\eqref{eq:Vws-relative}.

\subsubsection{Algebraic reduction to relative direction error}

The following lemma applies to both the efficient and whitened settings.

\begin{lemma}[Variance difference reduces to direction error]\label{lem:var-direction}
Let $\hat\phi(X,Y)=s_\star(Y,X)\langle\widehat H,X\rangle$ and $\phi^\star(X,Y)=s_\star(Y,X)\langle H^\star,X\rangle$, where $s_\star$ denotes either the usual score $s_\eta(Y,\eta^\star)$ or the whitened score $\tilde s(Y,\eta^\star)$. Assume the score-squared factor satisfies $\E^\star[s_\star^2 f(X)]\asymp\E^\star[f(X)]$ for nonnegative~$f$. If the relative direction error
\[
\delta_H := \frac{\|\widehat H - H^\star\|_F}{\|H^\star\|_F}
\]
is sufficiently small, then the direction-induced relative variance error satisfies
\[
\frac{\bigl|\mathbb{P}^\star\bigl[s_\star^2(\langle\widehat H,X\rangle^2-\langle H^\star,X\rangle^2)\bigr]\bigr|}{\mathbb{P}^\star[s_\star^2\langle H^\star,X\rangle^2]}
\;\lesssim\; \delta_H.
\]
\end{lemma}

\begin{proof}
Write $a:=\langle\widehat H,X\rangle$ and $b:=\langle H^\star,X\rangle$. By score comparability, $|\mathbb{P}^\star[s_\star^2(a^2-b^2)]|\lesssim\mathbb{P}^\star|a^2-b^2|$. The factorization $a^2-b^2=(a-b)(a+b)$ and Cauchy--Schwarz give
\[
\mathbb{P}^\star|a^2-b^2|
\le \bigl(\mathbb{P}^\star(a-b)^2\bigr)^{1/2}\bigl(\mathbb{P}^\star(a+b)^2\bigr)^{1/2}.
\]
By the pairwise Frobenius reduction (Lemma~\ref{app:pd-pairwise}),
\[
\mathbb{P}^\star(a-b)^2 \asymp \|\widehat H-H^\star\|_F^2/\dstar,
\qquad
\mathbb{P}^\star b^2 \asymp \|H^\star\|_F^2/\dstar.
\]
When $\delta_H\le c_0$, $(a+b)^2\le 2(a-b)^2+8b^2$ gives $\mathbb{P}^\star(a+b)^2\le 10\,\mathbb{P}^\star b^2$. Substituting and dividing by $\mathbb{P}^\star[s_\star^2 b^2]\asymp\mathbb{P}^\star b^2$ yields the claim.
\end{proof}

\subsubsection{Efficient estimator: proof of~\eqref{eq:Veff-relative}}

The total relative variance error decomposes into three parts: empirical fluctuation, score plug-in, and direction-induced.

\paragraph{Direction-induced part.}
The efficient directions satisfy $H^\star=A^{-1}P_{\mathbb{T}}\Gamma$ and $\widehat H=\widehat A^{-1}\widehat P_{\mathbb{T}}\Gamma$. Add and subtract $A^{-1}\widehat P_{\mathbb{T}}\Gamma$ and use the resolvent identity $\widehat A^{-1}-A^{-1}=A^{-1}(A-\widehat A)\widehat A^{-1}$:
\[
\widehat H - H^\star = \underbrace{A^{-1}(\widehat P_{\mathbb{T}}-P_{\mathbb{T}})\Gamma}_{T_1} + \underbrace{A^{-1}(A-\widehat A)\widehat A^{-1}\widehat P_{\mathbb{T}}\Gamma}_{T_2}.
\]
Since $A$ has spectral scale $(\dstar)^{-1}$ on the tangent space, $A^{-1}$ multiplies Frobenius norms by $\asymp\dstar$. This common factor cancels in the relative ratio $\|T_1\|_F/\|H^\star\|_F$, yielding
\[
\frac{\|T_1\|_F}{\|H^\star\|_F}
\lesssim \frac{\|(\widehat P_{\mathbb{T}}-P_{\mathbb{T}})\Gamma\|_F}{\|P_{\mathbb{T}}\Gamma\|_F}.
\]
Under the sparse-target projector perturbation bound $\|(\widehat P_{\mathbb{T}}-P_{\mathbb{T}})\Gamma\|_F\lesssim\rho\sqrt{\bar d/\dstar}\,\|\Gamma\|_F$ and the alignment condition $\|P_{\mathbb{T}}\Gamma\|_F\ge\alpha_\Gamma\sqrt{\bar d/\dstar}\,\|\Gamma\|_F$, this is $O(\rho/\alpha_\Gamma)$. For $T_2$, the operator perturbation $\|A-\widehat A\|_{\mathrm{op}}\lesssim\rho/(\dstar)$ together with $\|A^{-1}\|_{\mathrm{op}}\asymp\dstar$ gives $\|A^{-1}(A-\widehat A)\|_{\mathrm{op}}\lesssim\rho$, so $\|T_2\|_F\lesssim\rho\|\widehat H\|_F\asymp\rho\|H^\star\|_F$. Combining,
\[
\delta_H := \frac{\|\widehat H-H^\star\|_F}{\|H^\star\|_F} \lesssim \frac{\rho}{\alpha_\Gamma}.
\]
Lemma~\ref{lem:var-direction} then gives a direction-induced relative error of $O(\rho/\alpha_\Gamma)$.

\paragraph{Score plug-in part.}
By the derivative bound on the score, $|s_\eta(Y,\hat\eta)^2-s_\eta(Y,\eta^\star)^2|\lesssim\|\widehat T-T^\star\|_\infty$, so the relative contribution is $O(\|\widehat T-T^\star\|_\infty)$.

\paragraph{Empirical fluctuation.}
Conditional on the first-stage sample, $\widehat V_{\rm eff}-\mathbb{P}^\star[\hat\phi^2]=(\mathbb{P}_n-\mathbb{P}^\star)(\hat\phi^2)$ is a mean-zero average. A Bernstein bound gives relative error $O(\sqrt{\bar d\log\bar d/n})$.

Combining the three parts yields~\eqref{eq:Veff-relative}.

\subsubsection{Score-whitened estimator: proof of~\eqref{eq:Vws-relative}}

The argument simplifies because the whitened directions $H_{\rm ws}^\star=\dstar P_{\mathbb{T}}\Gamma$ and $\widehat H_{\rm ws}=\dstar\widehat P_{\mathbb{T}}\Gamma$ involve no operator inversion. The direction difference is
\[
\widehat H_{\rm ws}-H_{\rm ws}^\star = \dstar(\widehat P_{\mathbb{T}}-P_{\mathbb{T}})\Gamma,
\]
so the $\dstar$ factor cancels in the relative ratio:
\[
\frac{\|\widehat H_{\rm ws}-H_{\rm ws}^\star\|_F}{\|H_{\rm ws}^\star\|_F}
= \frac{\|(\widehat P_{\mathbb{T}}-P_{\mathbb{T}})\Gamma\|_F}{\|P_{\mathbb{T}}\Gamma\|_F}
\lesssim \frac{\rho}{\alpha_\Gamma}.
\]
Lemma~\ref{lem:var-direction} (applied with the whitened score $\tilde s$) gives the direction-induced relative error. The score plug-in and empirical fluctuation parts are identical to the efficient case, completing the proof of~\eqref{eq:Vws-relative}.

\subsection{Useful Technical Tools}\label{app:pd-tools}

\subsubsection{Frobenius reduction for comparison sampling}\label{app:pd-tools-frob}

\begin{lemma}[Pairwise difference identity]\label{app:pd-pairwise}
Fix any vector $z\in\mathbb{R}^{d_1}$ with $\sum_{a=1}^{d_1} z_a=0$. If $\{a,b\}$ is uniformly distributed over unordered pairs, then
\begin{equation}\label{app:pd-pairwise-identity}
\E_{\{a,b\}}\big[(z_a-z_b)^2\big] = \frac{2}{d_1-1}\,\norm{z}_2^2.
\end{equation}
\end{lemma}

\noindent\textbf{Proof of Lemma~\ref{app:pd-pairwise}.}
We have $\sum_{a<b}(z_a-z_b)^2 = d_1\sum_a z_a^2 - (\sum_a z_a)^2 = d_1\norm{z}_2^2$.
Dividing by $\binom{d_1}{2}$ yields \eqref{app:pd-pairwise-identity}.
\hfill$\square$\medskip

\begin{lemma}[Frobenius reduction for comparison sampling]\label{app:pd-Frob-reduction}
Let $H$ satisfy the column-sum-zero condition \eqref{app:pd-row-sum-zero}. Under the comparison sampling model,
\begin{equation}\label{app:pd-HX-second-moment}
\E^\star\big[\ip{H}{X}^2\big] = \frac{2}{D(d_1-1)}\,\norm{H}_F^2
\asymp \frac{\norm{H}_F^2}{\dstar}.
\end{equation}
\end{lemma}

\noindent\textbf{Proof of Lemma~\ref{app:pd-Frob-reduction}.}
Write $h_u\in\mathbb{R}^{d_1}$ for the $u$-th context slice.
Conditioning on $u$ and applying Lemma~\ref{app:pd-pairwise}:
$\E^\star[\ip{H}{X}^2] = \frac{1}{D}\sum_u \frac{2}{d_1-1}\norm{h_u}_2^2
= \frac{2}{D(d_1-1)}\norm{H}_F^2$.
\hfill$\square$\medskip

\begin{corollary}[Weighted second moment]\label{app:pd-weighted-second-moment}
Under Assumption~\ref{app:pd-score} and Lemma~\ref{app:pd-Frob-reduction}, for any column-sum-zero $H$,
\begin{equation}\label{app:pd-weighted-second-moment-eq}
\E^\star\!\big[I(\eta^\star(X))\,\ip{H}{X}^2\big]
\le
\frac{c_+}{\sigma^2}\cdot \frac{\norm{H}_F^2}{\dstar}.
\end{equation}
\end{corollary}

\subsubsection{$\ell_1$ extraction}\label{app:pd-tools-l1}

\begin{lemma}[$\ell_1$ extraction]\label{app:pd-Gamma-l1}
For any linear map $\mathcal M$ and any $\Gammavec = \sum_\omega \gamma_\omega e_\omega$,
\begin{equation}\label{app:pd-Gamma-l1-eq}
\begin{aligned}
&\norm{\mathcal M\,\Gammavec}_F
\le
\oneNorm{\Gammavec}\ \max_{\omega}\ \norm{\mathcal M e_\omega}_F,\\
&\infnorm{\mathcal M\,\Gammavec}
\le
\oneNorm{\Gammavec}\ \max_{\omega}\ \infnorm{\mathcal M e_\omega},
\\
&\oneNorm{\mathcal M\,\Gammavec}
\le
\oneNorm{\Gammavec}\ \max_{\omega}\ \oneNorm{\mathcal M e_\omega}.
\end{aligned}
\end{equation}
\end{lemma}

\noindent\textbf{Proof of Lemma~\ref{app:pd-Gamma-l1}.}
By linearity and the triangle inequality,
$\norm{\mathcal M\Gammavec}_F
\le \sum_\omega |\gamma_\omega|\,\norm{\mathcal M e_\omega}_F
\le \oneNorm{\Gammavec}\, \max_\omega \norm{\mathcal M e_\omega}_F$.
The $\ell_\infty$ and $\ell_1$ versions are identical.
\hfill$\square$\medskip

\subsubsection{Block projections and tangent space for Tucker structure}\label{app:pd-tools-blocks}

Let $\Hcal := \R^{d_1\times\cdots\times d_m}$ with Frobenius inner product.
Fix orthonormal matrices $U_j\in \R^{d_j\times r_j}$ and define $P_j := U_jU_j^\top$, $P_j^\perp := \Id_{d_j}-P_j$.
For any subset $S\subseteq[m]$, define the \emph{block projector}
\begin{equation}\label{app:pd-block-projector}
\Pi_S(X) := X \mode{j\in S}{P_j^\perp}\mode{k\notin S}{P_k}.
\end{equation}

\begin{lemma}[Orthogonal block decomposition]\label{app:pd-block-orth-decomp}
The family $\{\Pi_S\}_{S\subseteq[m]}$ consists of orthogonal projections with pairwise orthogonal ranges and $\sum_{S\subseteq[m]}\Pi_S = \Id$.
\end{lemma}

\noindent\textbf{Proof of Lemma~\ref{app:pd-block-orth-decomp}.}
Each $P_j,P_j^\perp$ is symmetric idempotent, and products in distinct modes commute. If $S\ne S'$, pick $j\in S\setminus S'$; then $\Pi_S$ multiplies by $P_j^\perp$ and $\Pi_{S'}$ by $P_j$, giving $\Pi_S\Pi_{S'} = 0$.
For the resolution, $\prod_j(P_j + P_j^\perp) = \Id$, and expanding yields $\sum_S \Pi_S$.
\hfill$\square$\medskip

\begin{lemma}[Tangent space is ``$|S|\le 1$'' blocks]\label{app:pd-tangent-block-sum}
Under full multilinear rank,
\begin{equation}\label{app:pd-tangent-block}
\mathcal{T}
=
\Pi_{\emptyset}(\Hcal)\;\oplus\;\bigoplus_{j=1}^m \Pi_{\{j\}}(\Hcal),
\qquad
\mathcal{T}^\perp
=
\bigoplus_{\substack{S\subseteq[m]\\ |S|\ge 2}}\Pi_S(\Hcal).
\end{equation}
In particular, $(I-P_{\mathbb{T}})\Delta = \sum_{|S|\ge 2}\Pi_S(\Delta)$.
\end{lemma}

\noindent\textbf{Proof of Lemma~\ref{app:pd-tangent-block-sum}.}
For the inclusion $\subseteq$: each term in the tangent space representation $\delta C\mode{1}U_1\cdots\mode{m}U_m$ lies in $\Pi_\emptyset(\Hcal)$, and each $C^\star\mode{j}W_j\mode{k\ne j}U_k$ (with $U_j^\top W_j = 0$) lies in $\Pi_{\{j\}}(\Hcal)$.
For the inclusion $\supseteq$: any $Y\in\Pi_\emptyset(\Hcal)$ can be written as $\tilde C\mode{1}U_1\cdots\mode{m}U_m\in\mathcal{T}$, and any $Y\in\Pi_{\{j\}}(\Hcal)$ has mode-$j$ fibers in $\mathrm{Range}(P_j^\perp)$; by full rank of $C^\star_{(j)}$, one can find $W_j$ with $U_j^\top W_j=0$ such that $Y = C^\star\mode{j}W_j\mode{k\ne j}U_k\in\mathcal{T}$.
The orthogonal complement formula follows from Lemma~\ref{app:pd-block-orth-decomp}.
\hfill$\square$\medskip

\subsubsection{Dimension-free $\infty\!\to\!\infty$ bounds for projections}\label{app:pd-tools-proj-bounds}

\begin{lemma}[Dimension-free bounds]\label{app:pd-Pinfty}
Let $P = UU^\top$ with $\max_i\|e_i^\top U\|_2^2\le\mu r/d$. Then
\begin{equation}\label{app:pd-P-infty}
\inftyinfnorm{P}\le\sqrt{\mu}\,r,
\qquad
\inftyinfnorm{P^\perp}\le 1+\sqrt{\mu}\,r.
\end{equation}
For the centering $Q = \Id - \frac{1}{d}\one\one^\top$: $\inftyinfnorm{Q} < 2$ and $\inftyinfnorm{Q-P}\le 2+\sqrt{\mu}\,r$.

For Kronecker products: $\inftyinfnorm{\bigotimes_j A_j} = \prod_j\inftyinfnorm{A_j}$.
For the tangent projector: $\inftyinfnorm{P_{\mathbb{T}}}\le\sum_{|S|\le 1}\inftyinfnorm{\Pi_S} = \mathrm{poly}(\mu,r,m)$.
\end{lemma}

\noindent\textbf{Proof of Lemma~\ref{app:pd-Pinfty}.}
Row $i$ of $P$ has entries $P_{ij} = \ip{u_i}{u_j}$, so $\sum_j|P_{ij}|\le\|u_i\|_2\sqrt{d}\sqrt{r} = \sqrt{\mu r/d}\cdot\sqrt{dr} = \sqrt{\mu}\,r$.
For $P^\perp$, $\inftyinfnorm{\Id - P}\le 1+\inftyinfnorm{P}$.
For $Q$, each row has sum $|1-1/d| + (d-1)/d = 2-2/d < 2$.
For Kronecker products, absolute row sums multiply.
\medskip

\subsubsection{$\|P_{\mathbb{T}} E_\omega\|_\infty$ bound}\label{app:pd-tools-PTomega}

\begin{lemma}[$\|P_{\mathbb{T}}E_\omega\|_\infty$ is $d/\dstar$ up to incoherence]\label{app:pd-PTX-infty}
Under $\mu$-incoherence, for any canonical basis tensor $E_\omega = e_{i_1}\otimes\cdots\otimes e_{i_m}$,
\begin{equation}\label{app:pd-PTX-infty-bound}
\infnorm{P_{\mathbb{T}}E_\omega}
\le C_\mu\,\frac{d}{\dstar},
\qquad
C_\mu := (m+1)\mu^m,
\end{equation}
where $d := \prod_j r_j(1+\sum_k d_k/r_k)$ is the effective tangent dimension proxy.
\end{lemma}

\noindent\textbf{Proof of Lemma~\ref{app:pd-PTX-infty}.}
Using $P_{\mathbb{T}} = \Pi_\emptyset + \sum_k\Pi_{\{k\}}$, we have
\[
\infnorm{\Pi_\emptyset(E_\omega)} = \prod_j\infnorm{P_je_{i_j}} \le \mu^m\prod_j r_j/d_j = \mu^m r^\star/\dstar,
\]
and
\[
\infnorm{\Pi_{\{k\}}(E_\omega)} \le \infnorm{P_k^\perp e_{i_k}}\prod_{j\ne k}\infnorm{P_je_{i_j}} \le \mu^{m-1}(r^\star/\dstar)(d_k/r_k).
\]
Summing gives
\[
\infnorm{P_{\mathbb{T}}E_\omega}\le\mu^m r^\star/\dstar + \sum_k\mu^{m-1}r^\star d_k/(r_k\dstar) \le (m+1)\mu^m\cdot d/\dstar.
\]
\hfill$\square$\medskip

\noindent\textbf{Sharp $\|P_{\mathbb{T}}\omega\|_\infty$ for comparison design.}
For the comparison design, $\omega = e_{u,a}-e_{u,b}$ and the Ma-type tangent projector decomposition $P_{\mathbb{T}} = \mathcal{P}_\times + \sum_j\mathcal{R}_j$ gives, via the rank-one unfolding argument:
\begin{equation}\label{app:pd-PT-omega-comparison}
\|P_{\mathbb{T}}\omega\|_\infty \le \mathrm{poly}(\mu,r,m)\cdot\frac{d_{\max}}{\dstar}.
\end{equation}
The main term $\mathcal{P}_\times(\omega)$ factorizes as $\bigotimes_k P_{U_k}e_{i_k}$, giving $\mu^m r^\star/\dstar$.
Each correction $\mathcal{R}_j(\omega)$ is rank-one in mode-$j$ matricization with
$\|\mathcal{R}_j(\omega)\|_\infty\le(1+\mu r_j)\mu^{m-1}r^\star d_j/(r_j\dstar)$.
The additional row-sum projection $\Pi_k$ satisfies $\|\Pi_k(Z)\|_\infty\le 2\|Z\|_\infty$ and does not affect the $\dstar$ cancellation.

\subsubsection{Symmetry and norm equivalences}

\begin{lemma}[Symmetry: $\|M\|_{1\to1}=\|M\|_{\infty\to\infty}$]\label{app:pd-symm-1inf}
For any $M\in\R^{p\times p}$, $\|M\|_{1\to1}=\|M^\top\|_{\infty\to\infty}$.
If $M=M^\top$, then $\|M\|_{1\to1}=\|M\|_{\infty\to\infty}$.
\end{lemma}

\subsubsection{Bernstein template}

\begin{lemma}[Bernstein for bounded variables]\label{app:pd-bernstein-template}
Let $W_1,\dots,W_n$ be i.i.d.\ mean-zero with $|W_i|\le M$ a.s.\ and $\mathrm{Var}(W_i)\le v$.
Then for any $\delta\in(0,1)$, with probability at least $1-\delta$,
\[
\left|\frac1n\sum_{i=1}^n W_i\right|
\le
C\left(\sqrt{\frac{v\log(2/\delta)}{n}}+\frac{M\log(2/\delta)}{n}\right).
\]
\end{lemma}

\subsubsection{Entrywise contraction}

\begin{lemma}[Entrywise contraction via $\|\cdot\|_{2,\infty}$]\label{app:pd-contract-2inf}
Let $G\in\R^{r_1\times\cdots\times r_m}$ and $A_j\in\R^{d_j\times r_j}$.
Then $\infnorm{G\mode{1}A_1\cdots\mode{m}A_m}\le\norm{G}_{\mathrm{op}}\prod_j\twoninfnorm{A_j}$.
\end{lemma}

\noindent\textbf{Proof of Lemma~\ref{app:pd-contract-2inf}.}
For each entry $(i_1,\dots,i_m)$, the value is $\ip{G}{a_1(i_1)\otimes\cdots\otimes a_m(i_m)}$ where $a_j(i_j)^\top = e_{i_j}^\top A_j$.
By definition of $\norm{G}_{\mathrm{op}}$,
\[
|Y_{i_1,\dots,i_m}|\le\norm{G}_{\mathrm{op}}\prod_j\|a_j(i_j)\|_2\le\norm{G}_{\mathrm{op}}\prod_j\twoninfnorm{A_j}.
\]
\hfill$\square$\medskip


\section{Proofs for Section~\ref{sec:score-whitening}}\label{app:score-whitening}

This appendix proves Theorems~\ref{thm:ws-bound} and~\ref{thm:ws-clt}.
The proof follows the same five-term decomposition as Appendix~\ref{app:proof-details}, but each term simplifies dramatically because score whitening turns the information operator into $(1/\dstar)\Id$, eliminating the operator inverse and $C_A$ entirely.

\medskip
\noindent\textbf{Proof roadmap and comparison with the general case.}
For ease of reference, we list the five remainder terms side-by-side:
\begin{center}
\renewcommand{\arraystretch}{1.2}
\begin{tabular}{lcc}
\textbf{Term} & \textbf{General (Sec.~\ref{app:proof-details})} & \textbf{Whitened (this section)} \\
\hline
Direction error $R_{\rm emp}^H$ & $C_A\sigma\sqrt{\bar d}\,\rho\sqrt{\log/n}$ & $\sqrt{\bar d}\,\rho\sqrt{\log/n}$ \\
Score perturbation $R_{\rm emp}^\eta$ & $C_A\,\delta_\infty\sqrt{\bar d/n}$ & $\delta_\infty\sqrt{\bar d/n}$ \\
Projection leakage $R_{\rm proj}$ & $\rho^2$ & $\rho^2$ (unchanged) \\
H-bias (1st order) & $C_A\,\rho^2$ & $\rho^2$ \\
1st-order cancel + 2nd remainder & $C_A\,\rho^2 + C_A\,\delta_\infty^2$ & $0 + \delta_\infty^2$ \\
\hline
\end{tabular}
\end{center}
The key simplifications are: (i) $C_A$ disappears from all terms because $A_0^{-1}=\dstar P_{\mathbb{T}}$ is explicit; (ii) the $\sigma$ factors from $\|A^{-1}\|$ and $1/\sigma$ from the score cancel automatically since the whitened score has $O(1)$ bounds; (iii) the first-order cancellation becomes \emph{complete} (zero, not just $O(\rho^2)$), because $G_0(H^\star_{\rm ws}) = P_{\mathbb{T}}\Gamma \in \mathcal{T}$ has no off-tangent component.

\subsection{Setup and notation}\label{app:ws-setup}

We adopt the comparison sampling model and column-sum-zero restriction of Section~\ref{app:pd-setup}.
Throughout, $\tilde s(y,\eta):=s(y,\eta)/I(\eta)$ denotes the whitened score (Eq.~\eqref{eq:whitened-score-def}).
Its key properties, inherited from the exponential family structure, are:
\begin{enumerate}[label=(\roman*)]
\item \emph{Centering:} $\E^\star[\tilde s(Y,\eta^\star)\mid X]=0$.
\item \emph{Constant conditional derivative:} $\E^\star[\partial_\eta\tilde s(Y,\eta^\star)\mid X]=-1$, regardless of $X$.
\end{enumerate}

\noindent\textbf{Whitened score regularity.}
From Assumption~\ref{app:pd-score}, the original score satisfies $|s(y,\eta)|\le C_{\psione}\asymp 1/\sigma$ and $I(\eta)\ge c_-/\sigma^2$.
Therefore the whitened score satisfies
\begin{equation}\label{app:ws-stilde-bound}
|\tilde s(y,\eta)|
= \frac{|s(y,\eta)|}{I(\eta)}
\le \frac{C_{\psione}}{c_-/\sigma^2}
= \frac{C_{\psione}\sigma^2}{c_-}
\le C_{\tilde s},
\end{equation}
where $C_{\tilde s}:=C_{\psione}\sigma^2/c_-$ is a bounded constant (since $\sigma$ is a constant in the BTL model).
Similarly, by Assumption~\ref{app:pd-score}(i),
\begin{equation}\label{app:ws-stilde-deriv}
|\partial_\eta\tilde s(y,\eta)|
\le \frac{|\dot s(y,\eta)|}{I(\eta)} + \frac{|s(y,\eta)|\cdot|I'(\eta)|}{I(\eta)^2}
\le C_{\dot{\tilde s}},
\end{equation}
and the second derivative $|\partial^2_\eta\tilde s(y,\eta)|\le C_{\ddot{\tilde s}}$, where $C_{\dot{\tilde s}},C_{\ddot{\tilde s}}$ are dimension-free constants depending only on $(c_-,c_+,c_1,c_2,C_{\psione},\sigma)$.
In particular, $\|\tilde s(Y,\hat\eta)\|_{\psione}\le C_{\tilde s}$ (the whitened score is bounded, hence sub-exponential with a \emph{constant} parameter, unlike the original score which has $\psione$-norm $\asymp 1/\sigma$).

\noindent\textbf{Isotropic Gram operator.}
Under the whitened score and uniform comparison design, by property~(ii) the effective Gram operator is
\begin{equation}\label{app:ws-G0}
G_0(H) = \E^\star\!\big[\ip{H}{X}\,X\big] = \frac{1}{\dstar}\,H
\end{equation}
on the column-sum-zero subspace (Lemma~\ref{app:pd-Frob-reduction}).
The restricted operators are therefore
\[
A_0 = P_{\mathbb{T}} G_0 P_{\mathbb{T}} = \frac{1}{\dstar}\,P_{\mathbb{T}},
\qquad
A_0^{-1} = \dstar\,P_{\mathbb{T}}.
\]

\noindent\textbf{Simplified directions.}
The oracle and estimated EIF directions are
\begin{equation}\label{app:ws-directions}
H^\star_{\rm ws} = \dstar\,P_{\mathbb{T}}\Gammavec,
\qquad
\hat H_{\rm ws} = \dstar\,\widehat{P}_{\mathbb{T}}\Gammavec,
\end{equation}
and the direction error is
\begin{equation}\label{app:ws-direction-error}
\hat H_{\rm ws} - H^\star_{\rm ws} = \dstar\,(\widehat{P}_{\mathbb{T}} - P_{\mathbb{T}})\Gammavec.
\end{equation}
No operator inversion, Neumann series expansion, or Gram mismatch analysis is needed.

\subsection{Bounding $R_{\mathrm{emp}}^{\tilde H}$ (direction-error empirical process)}\label{app:ws-RempH}

Define
\[
R_{\mathrm{emp}}^{\tilde H}
:= (\Pn-\Pstar)\big[\tilde s(Y,\hat\eta)\,\ip{\hat H_{\rm ws}-H^\star_{\rm ws}}{X}\big].
\]
Let $Z_i := \tilde s(Y_i,\hat\eta_i)\,\ip{\hat H_{\rm ws}-H^\star_{\rm ws}}{X_i}$.
We bound its variance and sub-exponential norm following the five-step template of Section~\ref{app:pd-RempH}.

\subsubsection{Step 1: Variance bound via Frobenius reduction}

Since $|\tilde s(Y,\hat\eta)|\le C_{\tilde s}$ (Eq.~\eqref{app:ws-stilde-bound}),
\[
\E^\star[Z_i^2]
\le C_{\tilde s}^2\,\E^\star\!\big[\ip{\hat H_{\rm ws}-H^\star_{\rm ws}}{X}^2\big].
\]
Applying the Frobenius reduction (Lemma~\ref{app:pd-Frob-reduction}):
\begin{equation}\label{app:ws-var-Z-Frob}
\mathrm{Var}(Z_i)
\le \E^\star[Z_i^2]
\le C_{\tilde s}^2\cdot\frac{\Fnorm{\hat H_{\rm ws}-H^\star_{\rm ws}}^2}{\dstar}.
\end{equation}
Note that the Fisher weight $c_+/\sigma^2$ from the general case is now replaced by the constant $C_{\tilde s}^2$.

\subsubsection{Step 2: Extracting $\oneNorm{\Gammavec}$ from the Frobenius norm}

Since $\hat H_{\rm ws}-H^\star_{\rm ws} = \dstar\,(\widehat{P}_{\mathbb{T}}-P_{\mathbb{T}})\Gammavec$, by Lemma~\ref{app:pd-Gamma-l1}:
\begin{equation}\label{app:ws-l1-extract-F}
\Fnorm{\hat H_{\rm ws}-H^\star_{\rm ws}}
= \dstar\,\Fnorm{(\widehat{P}_{\mathbb{T}}-P_{\mathbb{T}})\Gammavec}
\le \dstar\,\oneNorm{\Gammavec}\cdot\max_\omega\Fnorm{(\widehat{P}_{\mathbb{T}}-P_{\mathbb{T}})E_\omega}.
\end{equation}

\subsubsection{Step 3: Bounding $\max_\omega\Fnorm{(\widehat{P}_{\mathbb{T}}-P_{\mathbb{T}})E_\omega}$}

By the subspace estimation guarantee (Eq.~\eqref{app:pd-hatP-minus-P-Eomega-F}):
\begin{equation}\label{app:ws-PT-diff-Eomega-F}
\max_\omega\Fnorm{(\widehat{P}_{\mathbb{T}}-P_{\mathbb{T}})E_\omega}
\le C(\mu,r,m)\cdot\frac{\sqrt{\bar d}\,\rho}{\sqrt{\dstar}},
\end{equation}
where $\rho = \lambda_{\min}^{-1}\sqrt{\dstar\bar d\log^c\bar d/n}$ as in~\eqref{app:pd-rho-def}.
This is the \emph{same bound} as in the general case (Eq.~\eqref{app:pd-hatP-minus-P-Eomega-F}), with the $\sqrt{\bar d}$ arising from the $(\Id-P_U)$ mode in the tangent space projector.
The key simplification is that \emph{no resolvent decomposition} (Eq.~\eqref{app:pd-M-decomp}) is needed---the direction error is a single term $\dstar(\widehat{P}_{\mathbb{T}}-P_{\mathbb{T}})\Gammavec$ rather than the two-term expansion $A^{-1}(\widehat{P}_{\mathbb{T}}-P_{\mathbb{T}})+A^{-1}(A-\hat A)\hat A^{-1}\widehat{P}_{\mathbb{T}}$.

Substituting~\eqref{app:ws-PT-diff-Eomega-F} into~\eqref{app:ws-l1-extract-F}:
\[
\Fnorm{\hat H_{\rm ws}-H^\star_{\rm ws}}
\le \dstar\,\oneNorm{\Gammavec}\cdot C(\mu,r,m)\cdot\frac{\sqrt{\bar d}\,\rho}{\sqrt{\dstar}}
= C(\mu,r,m)\,\sqrt{\dstar\bar d}\,\oneNorm{\Gammavec}\,\rho.
\]
Plugging into~\eqref{app:ws-var-Z-Frob}:
\begin{equation}\label{app:ws-var-final}
\mathrm{Var}(Z_i)
\le C(\mu,r,m)\,\bar d\,\oneNorm{\Gammavec}^2\,\rho^2.
\end{equation}
The $C_A^2\sigma^2$ prefactor from the general case has been absorbed into $C(\mu,r,m)$; the $\bar d$ factor persists from the uncompressed mode in the tangent projector.

\subsubsection{Step 4: Sub-exponential ($\psione$) bound via $\ell_\infty$ norm}

Since $|\tilde s|\le C_{\tilde s}$ (bounded, hence sub-exponential with $\|\tilde s\|_{\psione}\le C_{\tilde s}$) and $|\ip{\hat H_{\rm ws}-H^\star_{\rm ws}}{X}|\le 2\infnorm{\hat H_{\rm ws}-H^\star_{\rm ws}}$:
\begin{equation}\label{app:ws-Zi-psi1}
\|Z_i\|_{\psione}
\le 2C_{\tilde s}\,\infnorm{\hat H_{\rm ws}-H^\star_{\rm ws}}.
\end{equation}

By Lemma~\ref{app:pd-Gamma-l1} (entrywise version) and the $P_{\mathbb{T}}E_\omega$ bound (Lemma~\ref{app:pd-PTX-infty}):
\[
\infnorm{\hat H_{\rm ws}-H^\star_{\rm ws}}
= \dstar\,\infnorm{(\widehat{P}_{\mathbb{T}}-P_{\mathbb{T}})\Gammavec}
\le \dstar\,\oneNorm{\Gammavec}\cdot\max_\omega\infnorm{(\widehat{P}_{\mathbb{T}}-P_{\mathbb{T}})E_\omega}.
\]
Since
\[
\infnorm{(\widehat{P}_{\mathbb{T}}-P_{\mathbb{T}})E_\omega}\le\infnorm{\widehat{P}_{\mathbb{T}}E_\omega}+\infnorm{P_{\mathbb{T}}E_\omega}\le 2C(\mu,r,m)\,\bar d/\dstar
\]
(Lemma~\ref{app:pd-PTX-infty}):
\begin{equation}\label{app:ws-H-diff-infty}
\infnorm{\hat H_{\rm ws}-H^\star_{\rm ws}}
\le \dstar\cdot\oneNorm{\Gammavec}\cdot 2C(\mu,r,m)\cdot\frac{\bar d}{\dstar}
= C(\mu,r,m)\,\bar d\,\oneNorm{\Gammavec}.
\end{equation}
Therefore:
\begin{equation}\label{app:ws-psi1-bound}
\|Z_i\|_{\psione}
\le C(\mu,r,m)\,\bar d\,\oneNorm{\Gammavec}.
\end{equation}
The mechanism is the same as in the general case: $\dstar$ from the direction cancels with $1/\dstar$ from incoherence. The simplification is that the whitened score has $\psione$-norm $O(1)$, so the $C_A\sigma$ prefactor disappears.

\subsubsection{Step 5: Bernstein concentration}

Applying Lemma~\ref{app:pd-bernstein-template} with the variance~\eqref{app:ws-var-final} and tail bound~\eqref{app:ws-psi1-bound}:

\begin{theorem}[Bound on $R_{\mathrm{emp}}^{\tilde H}$]\label{app:ws-RempH-thm}
Under Assumptions~\ref{ass:initial} and~\ref{app:pd-score}, with probability at least $1-\delta$ (conditional on $\mathcal{D}_1$):
\begin{equation}\label{app:ws-RempH-final}
\boxed{
|R_{\mathrm{emp}}^{\tilde H}|
\le
C(\mu,r,m)\,\oneNorm{\Gammavec}\left[
\sqrt{\bar d}\;\rho\;\sqrt{\frac{\log(2/\delta)}{n}}
\;+\;
\bar d\;\frac{\log(2/\delta)}{n}
\right].
}
\end{equation}
\end{theorem}

\noindent\textbf{Proof of Theorem~\ref{app:ws-RempH-thm}.}
The variance term gives
\[
\sqrt{\frac{\mathrm{Var}(Z_i)\log(2/\delta)}{n}}
\le C(\mu,r,m)\,\sqrt{\bar d}\,\oneNorm{\Gammavec}\,\rho\,\sqrt{\frac{\log(2/\delta)}{n}}.
\]
The sub-exponential term gives
\[
\frac{\|Z_i\|_{\psione}\log(2/\delta)}{n}
\le C(\mu,r,m)\,\bar d\,\oneNorm{\Gammavec}\,\frac{\log(2/\delta)}{n}.
\]
Combining yields~\eqref{app:ws-RempH-final}.
\hfill$\square$\medskip

\subsection{Bounding $R_{\mathrm{emp}}^{\tilde\eta}$ (score-perturbation empirical process)}\label{app:ws-Remp-eta}

Define
\[
R_{\mathrm{emp}}^{\tilde\eta}
:= (\Pn-\Pstar)\big[(\tilde s(Y,\hat\eta)-\tilde s(Y,\eta^\star))\,\ip{H^\star_{\rm ws}}{X}\big].
\]
Condition on the first-stage output so that $\Delta:=\hat T-T^\star$ and $H^\star_{\rm ws}$ are fixed.

\subsubsection{Taylor expansion of the whitened score}

For each $i$, there exists $\bar\eta_i$ between $\hat\eta_i$ and $\eta^\star_i$ such that
\[
\tilde s(Y_i,\hat\eta_i)-\tilde s(Y_i,\eta^\star_i)
= \partial_\eta\tilde s(Y_i,\eta^\star_i)\ip{\Delta}{X_i}
+ \tfrac12\partial^2_\eta\tilde s(Y_i,\bar\eta_i)\ip{\Delta}{X_i}^2.
\]
Hence $R_{\mathrm{emp}}^{\tilde\eta} = R_{\tilde\eta,1}^{\mathrm{emp}}+R_{\tilde\eta,2}^{\mathrm{emp}}$, where
\[
R_{\tilde\eta,1}^{\mathrm{emp}}:=(\Pn-\Pstar)\big[\partial_\eta\tilde s(Y,\eta^\star)\ip{\Delta}{X}\ip{H^\star_{\rm ws}}{X}\big],
\quad
R_{\tilde\eta,2}^{\mathrm{emp}}:=\tfrac12(\Pn-\Pstar)\big[\partial^2_\eta\tilde s(Y,\bar\eta)\ip{\Delta}{X}^2\ip{H^\star_{\rm ws}}{X}\big].
\]

\subsubsection{Bounds on $\infnorm{H^\star_{\rm ws}}$ and $\Fnorm{H^\star_{\rm ws}}/\sqrt{\dstar}$}

Since $H^\star_{\rm ws}=\dstar\,P_{\mathbb{T}}\Gammavec$, by Lemma~\ref{app:pd-Gamma-l1}:

\noindent\textbf{$\ell_\infty$ bound.}
Using Lemma~\ref{app:pd-PTX-infty} ($\infnorm{P_{\mathbb{T}}E_\omega}\le C(\mu,r,m)\,\bar d/\dstar$):
\begin{equation}\label{app:ws-Hstar-infty}
\infnorm{H^\star_{\rm ws}}
\le \dstar\,\oneNorm{\Gammavec}\cdot\max_\omega\infnorm{P_{\mathbb{T}}E_\omega}
\le C(\mu,r,m)\,\bar d\,\oneNorm{\Gammavec}.
\end{equation}
Compared to~\eqref{app:pd-Hstar-infty-final}, the factor $C_A\sigma^2$ is replaced by~$1$ (since $\dstar$ from the direction cancels with $1/\dstar$ from incoherence directly, without passing through $\|A^{-1}\|_{\infty\to\infty}$).

\noindent\textbf{Frobenius bound.}
Since $\Fnorm{P_{\mathbb{T}}E_\omega}^2=(P_{\mathbb{T}})_{\omega\omega}\le C(\mu,r,m)\,\bar d/\dstar$:
\begin{equation}\label{app:ws-Hstar-Frob}
\frac{\Fnorm{H^\star_{\rm ws}}}{\sqrt{\dstar}}
\le \dstar\,\oneNorm{\Gammavec}\cdot\frac{\max_\omega\Fnorm{P_{\mathbb{T}}E_\omega}}{\sqrt{\dstar}}
\le C(\mu,r,m)\,\sqrt{\bar d}\,\oneNorm{\Gammavec}.
\end{equation}

\subsubsection{Abstract Bernstein bound}

The same argument as Proposition~\ref{app:pd-Remp-eta-prop} applies with the whitened score derivatives replacing the original ones.
Since $|\partial_\eta\tilde s|\le C_{\dot{\tilde s}}$ and $|\partial^2_\eta\tilde s|\le C_{\ddot{\tilde s}}$ (both constants, see~\eqref{app:ws-stilde-deriv}), the factor $1/\sigma^2$ in~\eqref{app:pd-Remp-eta-eq} is replaced by a constant.

\begin{corollary}[Explicit bound for $R_{\mathrm{emp}}^{\tilde\eta}$]\label{app:ws-Remp-eta-final}
Under Assumption~\ref{app:pd-score} and comparison design, with probability at least $1-\delta$,
\begin{equation}\label{app:ws-Remp-eta-final-eq}
\boxed{
|R_{\mathrm{emp}}^{\tilde\eta}|
\;\le\;
C(\mu,r,m)\,\oneNorm{\Gammavec}\big(\infnorm{\Delta}+\infnorm{\Delta}^2\big)
\left[
\sqrt{\bar d}\sqrt{\frac{\log(4/\delta)}{n}}
\;+\;
\bar d\,\frac{\log(4/\delta)}{n}
\right].
}
\end{equation}
\end{corollary}

\noindent\textbf{Proof of Corollary~\ref{app:ws-Remp-eta-final}.}
Substituting~\eqref{app:ws-Hstar-infty} and~\eqref{app:ws-Hstar-Frob} into the abstract Bernstein bound (same structure as Proposition~\ref{app:pd-Remp-eta-prop} with constant prefactors):
the first-order variance term gives
\[
C_{\dot{\tilde s}}\cdot\infnorm{\Delta}\cdot\frac{\Fnorm{H^\star_{\rm ws}}}{\sqrt{\dstar}}\le C(\mu,r,m)\,\oneNorm{\Gammavec}\,\infnorm{\Delta}\sqrt{\bar d};
\]
the first-order tail term gives
\[
C_{\dot{\tilde s}}\cdot\infnorm{\Delta}\cdot\infnorm{H^\star_{\rm ws}}\le C(\mu,r,m)\,\oneNorm{\Gammavec}\,\infnorm{\Delta}\,\bar d;
\]
the second-order terms are the same with an extra factor of $\infnorm{\Delta}$.
Factoring $\infnorm{\Delta}+\infnorm{\Delta}^2$ yields~\eqref{app:ws-Remp-eta-final-eq}.
\hfill$\square$\medskip

\subsection{Bounding $R_{\mathrm{proj}}$ (projection leakage)}\label{app:ws-Rproj}

The projection leakage term is
\[
R_{\mathrm{proj}} := \ip{(\Id-P_{\mathbb{T}})\Gammavec}{\Delta}.
\]
This term depends only on the functional gradient $\Gammavec$ and the estimation error $\Delta=\hat T-T^\star$; it is \emph{independent} of the score function or the direction $H$.
Therefore, the analysis of Section~\ref{app:pd-Rproj} applies \emph{verbatim}.
By Proposition~\ref{app:pd-PTcompl-infty} and
Assumption~\ref{ass:initial},
\begin{equation}\label{app:ws-Rproj-final}
\boxed{
|R_{\mathrm{proj}}|
\;\le\;
C(\mu,r,m)\;\oneNorm{\Gammavec}\,\rho^2.
}
\end{equation}
This is identical to~\eqref{app:pd-Rproj-final-eq}.

\subsection{H-direction bias}\label{app:ws-H-bias}

The H-direction bias term is
\[
\Pstar\big[\tilde s(Y,\hat\eta)\ip{\hat H_{\rm ws}-H^\star_{\rm ws}}{X}\big].
\]

\subsubsection{Taylor expansion and leading term}

Expanding $\tilde s(Y,\hat\eta)$ around $\eta^\star$ (as in Section~\ref{app:pd-H-bias}):
\begin{align}
\Pstar\big[\tilde s(Y,\hat\eta)\ip{\hat H_{\rm ws}-H^\star_{\rm ws}}{X}\big]
&= \underbrace{\Pstar\big[\tilde s(Y,\eta^\star)\ip{\hat H_{\rm ws}-H^\star_{\rm ws}}{X}\big]}_{=\,0\text{ (centering)}}
\notag\\
&\quad+ \underbrace{\Pstar\big[\partial_\eta\tilde s(Y,\eta^\star)\ip{\Delta}{X}\ip{\hat H_{\rm ws}-H^\star_{\rm ws}}{X}\big]}_{\text{first-order in }\Delta}
\notag\\
&\quad+ \underbrace{\tfrac12\Pstar\big[\partial^2_\eta\tilde s(Y,\bar\eta)\ip{\Delta}{X}^2\ip{\hat H_{\rm ws}-H^\star_{\rm ws}}{X}\big]}_{\text{second-order in }\Delta}.
\label{app:ws-Hbias-Taylor}
\end{align}

\subsubsection{First-order piece}

By property~(ii) of the whitened score, $\E^\star[\partial_\eta\tilde s(Y,\eta^\star)\mid X]=-1$.
Therefore the first-order piece equals $-\ip{\Delta}{(\hat H_{\rm ws}-H^\star_{\rm ws})/\dstar}$ computed via the isotropic Gram:
\[
\Pstar\big[\partial_\eta\tilde s\,\ip{\Delta}{X}\ip{\hat H_{\rm ws}-H^\star_{\rm ws}}{X}\big]
= -\ip{\Delta}{G_0(\hat H_{\rm ws}-H^\star_{\rm ws})}
= -\frac{1}{\dstar}\ip{\Delta}{\hat H_{\rm ws}-H^\star_{\rm ws}}.
\]
Since $\hat H_{\rm ws}-H^\star_{\rm ws}=\dstar(\widehat{P}_{\mathbb{T}}-P_{\mathbb{T}})\Gammavec$:
\[
-\frac{1}{\dstar}\ip{\Delta}{\dstar(\widehat{P}_{\mathbb{T}}-P_{\mathbb{T}})\Gammavec}
= -\ip{\Delta}{(\widehat{P}_{\mathbb{T}}-P_{\mathbb{T}})\Gammavec}.
\]
Decomposing $\Delta=P_{\mathbb{T}}\Delta+(\Id-P_{\mathbb{T}})\Delta$:
\begin{itemize}
\item By Cauchy--Schwarz,
\[
|\ip{P_{\mathbb{T}}\Delta}{(\widehat{P}_{\mathbb{T}}-P_{\mathbb{T}})\Gammavec}|\le\Fnorm{P_{\mathbb{T}}\Delta}\cdot\Fnorm{(\widehat{P}_{\mathbb{T}}-P_{\mathbb{T}})\Gammavec}.
\]
Using $\Fnorm{(\widehat{P}_{\mathbb{T}}-P_{\mathbb{T}})\Gammavec}\le\oneNorm{\Gammavec}\cdot C(\mu,r,m)\sqrt{\bar d}\,\rho/\sqrt{\dstar}$ (from~\eqref{app:ws-PT-diff-Eomega-F}) and $\Fnorm{P_{\mathbb{T}}\Delta}/\sqrt{\dstar}\le\infnorm{\Delta}$,
this contributes $O(\sqrt{\bar d}\,\infnorm{\Delta}\cdot\rho\cdot\oneNorm{\Gammavec})$, which is an $O(\infnorm{\Delta}\cdot\rho)$ cross term retained in the combined bound.

\item By H\"older's inequality,
\[
|\ip{(\Id-P_{\mathbb{T}})\Delta}{(\widehat{P}_{\mathbb{T}}-P_{\mathbb{T}})\Gammavec}|
\le\infnorm{(\Id-P_{\mathbb{T}})\Delta}\cdot\oneNorm{(\widehat{P}_{\mathbb{T}}-P_{\mathbb{T}})\Gammavec}.
\]
By Lemma~\ref{app:pd-PT-l1}, $\oneNorm{(\widehat{P}_{\mathbb{T}}-P_{\mathbb{T}})\Gammavec}\le\oneNorm{\Gammavec}\cdot C(\mu,r,m)$, and by Proposition~\ref{app:pd-PTcompl-infty}, $\infnorm{(\Id-P_{\mathbb{T}})\Delta}\le C(\mu,r,m)\rho^2$.
This gives $O(\rho^2\cdot\oneNorm{\Gammavec})$.
\end{itemize}

\subsubsection{Second-order piece}

Since $|\partial^2_\eta\tilde s|\le C_{\ddot{\tilde s}}$ (a constant) and $|\ip{\Delta}{X}|\le 2\infnorm{\Delta}$:
\[
\left|\tfrac12\Pstar\big[\partial^2_\eta\tilde s\,\ip{\Delta}{X}^2\ip{\hat H_{\rm ws}-H^\star_{\rm ws}}{X}\big]\right|
\le C_{\ddot{\tilde s}}\,\infnorm{\Delta}^2\,\E^\star[|\ip{\hat H_{\rm ws}-H^\star_{\rm ws}}{X}|].
\]
By the $\ell_1$-average bound (Lemma~\ref{app:pd-l1-average}):
\[
\E^\star[|\ip{\hat H_{\rm ws}-H^\star_{\rm ws}}{X}|]
\le \frac{\oneNorm{\hat H_{\rm ws}-H^\star_{\rm ws}}}{\dstar}.
\]
Since $\hat H_{\rm ws}-H^\star_{\rm ws}=\dstar(\widehat{P}_{\mathbb{T}}-P_{\mathbb{T}})\Gammavec$ and $\oneNorm{(\widehat{P}_{\mathbb{T}}-P_{\mathbb{T}})\Gammavec}\le\oneNorm{\Gammavec}\cdot C(\mu,r,m)$ (by Lemma~\ref{app:pd-PT-l1}):
\[
\frac{\oneNorm{\hat H_{\rm ws}-H^\star_{\rm ws}}}{\dstar}
= \oneNorm{(\widehat{P}_{\mathbb{T}}-P_{\mathbb{T}})\Gammavec}
\le C(\mu,r,m)\,\oneNorm{\Gammavec}.
\]
Therefore the second-order piece satisfies:
\begin{equation}\label{app:ws-H-bias-2nd}
\boxed{
\left|\text{second-order H-bias}\right|
\le C(\mu,r,m)\,\infnorm{\Delta}^2\,\oneNorm{\Gammavec}.
}
\end{equation}
The key cancellation is identical to the general case (Eq.~\eqref{app:pd-H-bias-2nd}): $\dstar$ from the direction cancels with $1/\dstar$ from the $\ell_1$-average bound.
The prefactor $1/\sigma^2$ from $|\ddot s|$ is absent here because $|\partial^2_\eta\tilde s|$ is bounded by a constant.

\subsection{First-order cancellation}\label{app:ws-1st-cancel}

The combined first-order cancellation term from the bias decomposition is (same as Eq.~\eqref{app:pd-1st-cancel-term}):
\begin{equation}\label{app:ws-1st-cancel-term}
\ip{P_{\mathbb{T}}\Gammavec}{\Delta} + \Pstar\big[\tilde s(Y,\hat\eta)\ip{H^\star_{\rm ws}}{X}\big] - \Pstar\big[\tilde s(Y,\eta^\star)\ip{H^\star_{\rm ws}}{X}\big].
\end{equation}

\subsubsection{Taylor expansion of the population score difference}

The second term of~\eqref{app:ws-1st-cancel-term} minus the third (which vanishes by centering) gives, by Taylor expansion around $\eta^\star$:
\begin{equation}\label{app:ws-Taylor-popul}
\Pstar\big[\tilde s(Y,\hat\eta)\ip{H^\star_{\rm ws}}{X}\big]
= -\ip{\Delta}{G_0(H^\star_{\rm ws})} + \tilde R_2(\Delta;H^\star_{\rm ws}),
\end{equation}
where we used $\E^\star[\partial_\eta\tilde s(Y,\eta^\star)\mid X]=-1$ (property~(ii)) to get $G_0(H)=(1/\dstar)H$, and
\[
\tilde R_2(\Delta;H^\star_{\rm ws}) := \tfrac12\Pstar\big[\partial^2_\eta\tilde s(Y,\bar\eta)\ip{\Delta}{X}^2\ip{H^\star_{\rm ws}}{X}\big].
\]

\subsubsection{The cancellation mechanism --- complete cancellation}

Plugging~\eqref{app:ws-Taylor-popul} into~\eqref{app:ws-1st-cancel-term}:
\[
\ip{P_{\mathbb{T}}\Gammavec}{\Delta} - \ip{\Delta}{G_0(H^\star_{\rm ws})} + \tilde R_2(\Delta;H^\star_{\rm ws}).
\]
Now $G_0(H^\star_{\rm ws}) = (1/\dstar)\cdot\dstar P_{\mathbb{T}}\Gammavec = P_{\mathbb{T}}\Gammavec$.
Therefore:
\begin{equation}\label{app:ws-complete-cancel}
\ip{P_{\mathbb{T}}\Gammavec}{\Delta} - \ip{\Delta}{P_{\mathbb{T}}\Gammavec} = 0.
\end{equation}

This is the key structural advantage of score whitening: in the general case (Section~\ref{app:pd-1st-cancel}), the first-order term $\ip{P_{\mathbb{T}}\Gammavec}{\Delta}-\ip{\Delta}{G(H^\star)}$ leaves a surviving normal-component contribution $-\ip{(\Id-P_{\mathbb{T}})\Delta}{G(H^\star)}$ that requires $O(C_A\cdot\infnorm{(\Id-P_{\mathbb{T}})\Delta})$ to control (Proposition~\ref{app:pd-1st-cancel-prop}).
Under score whitening, $G_0(H^\star_{\rm ws})=P_{\mathbb{T}}\Gammavec\in\mathcal{T}$, so
\[
\ip{(\Id-P_{\mathbb{T}})\Delta}{G_0(H^\star_{\rm ws})}
= \ip{(\Id-P_{\mathbb{T}})\Delta}{P_{\mathbb{T}}\Gammavec}
= 0
\]
by orthogonality of $(\Id-P_{\mathbb{T}})\Delta\in\mathcal{T}^\perp$ and $P_{\mathbb{T}}\Gammavec\in\mathcal{T}$.
Only the second-order remainder $\tilde R_2$ survives.

\subsection{Second-order score remainder}\label{app:ws-2nd-order}

The second-order remainder from~\eqref{app:ws-Taylor-popul} is
\[
\tilde R_2(\Delta;H^\star_{\rm ws}) = \tfrac12\Pstar\big[\partial^2_\eta\tilde s(Y,\bar\eta)\ip{\Delta}{X}^2\ip{H^\star_{\rm ws}}{X}\big].
\]

\begin{proposition}[Second-order remainder bound]\label{app:ws-2nd-order-prop}
Under the comparison design, with $|\partial^2_\eta\tilde s|\le C_{\ddot{\tilde s}}$,
\begin{equation}\label{app:ws-2nd-order-bound}
|\tilde R_2(\Delta;H^\star_{\rm ws})|
\le
C_{\ddot{\tilde s}}\,\infnorm{\Delta}^2\,\frac{\oneNorm{H^\star_{\rm ws}}}{\dstar}.
\end{equation}
\end{proposition}

\noindent\textbf{Proof of Proposition~\ref{app:ws-2nd-order-prop}.}
Since $|\ip{\Delta}{X}|\le 2\infnorm{\Delta}$ (comparison atom) and $|\partial^2_\eta\tilde s|\le C_{\ddot{\tilde s}}$:
\[
|\tilde R_2|
\le C_{\ddot{\tilde s}}\,\infnorm{\Delta}^2\,\E^\star[|\ip{H^\star_{\rm ws}}{X}|].
\]
By the $\ell_1$-average bound (Lemma~\ref{app:pd-l1-average}),
$\E^\star[|\ip{H^\star_{\rm ws}}{X}|]\le\oneNorm{H^\star_{\rm ws}}/\dstar$.
\hfill$\square$\medskip

\noindent\textbf{Expanding $\oneNorm{H^\star_{\rm ws}}/\dstar$.}
Since $H^\star_{\rm ws}=\dstar P_{\mathbb{T}}\Gammavec$:
\[
\frac{\oneNorm{H^\star_{\rm ws}}}{\dstar}
= \oneNorm{P_{\mathbb{T}}\Gammavec}
\le \oneNorm{\Gammavec}\cdot\max_\omega\oneNorm{P_{\mathbb{T}}E_\omega}
\le C(\mu,r,m)\,\oneNorm{\Gammavec},
\]
where the last step uses Lemma~\ref{app:pd-PT-l1} ($\oneNorm{P_{\mathbb{T}}E_\omega}\le C_P=C(\mu,r,m)$, dimension-free).
Therefore:
\begin{equation}\label{app:ws-2nd-order-Gamma}
\boxed{
|\tilde R_2(\Delta;H^\star_{\rm ws})|
\le
C(\mu,r,m)\,\infnorm{\Delta}^2\,\oneNorm{\Gammavec}.
}
\end{equation}
This matches~\eqref{app:pd-2nd-order-Gamma} in form but without $C_A$: the $\dstar$ from the direction cancels with $1/\dstar$ from the $\ell_1$-average bound, and the whitened score derivative is $O(1)$ rather than $O(1/\sigma^2)$.

\subsection{Combined bound}\label{app:ws-combined}

We now collect all remainder terms.

\subsubsection{Recap of individual bounds}

The total remainder $R_n$ decomposes as:
\begin{equation}\label{app:ws-Rn-decomp}
R_n
= \underbrace{R_{\mathrm{emp}}^{\tilde H}}_{\text{Sec.~\ref{app:ws-RempH}}}
+ \underbrace{R_{\mathrm{emp}}^{\tilde\eta}}_{\text{Sec.~\ref{app:ws-Remp-eta}}}
+ \underbrace{R_{\mathrm{proj}}}_{\text{Sec.~\ref{app:ws-Rproj}}}
+ \underbrace{\text{H-bias}}_{\text{Sec.~\ref{app:ws-H-bias}}}
+ \underbrace{\tilde R_2(\Delta;H^\star_{\rm ws})}_{\text{Sec.~\ref{app:ws-2nd-order}}}.
\end{equation}
The individual bounds are:
\begin{enumerate}[label=(\roman*)]
\item \textbf{Direction-error empirical process} (Theorem~\ref{app:ws-RempH-thm}):
\[
|R_{\mathrm{emp}}^{\tilde H}| \le C(\mu,r,m)\,\oneNorm{\Gammavec}\Big[\sqrt{\bar d}\,\rho\sqrt{\frac{\log(2/\delta)}{n}}+\bar d\,\frac{\log(2/\delta)}{n}\Big].
\]
\item \textbf{Score-perturbation empirical process} (Corollary~\ref{app:ws-Remp-eta-final}):
\[
|R_{\mathrm{emp}}^{\tilde\eta}| \le C(\mu,r,m)\,\oneNorm{\Gammavec}\big(\infnorm{\Delta}+\infnorm{\Delta}^2\big)\Big[\sqrt{\frac{\bar d}{n}}+\frac{\bar d}{n}\Big]
\]
(up to logarithmic factors).
\item \textbf{Projection leakage} (Eq.~\eqref{app:ws-Rproj-final}):
\[
|R_{\mathrm{proj}}| \le C(\mu,r,m)\,\oneNorm{\Gammavec}\,\rho^2.
\]
\item \textbf{H-direction bias} (Section~\ref{app:ws-H-bias}):
\[
\text{first order} \le C(\mu,r,m)\,\oneNorm{\Gammavec}\big(\sqrt{\bar d}\,\infnorm{\Delta}\cdot\rho + \rho^2\big),
\qquad
\text{second order} \le C(\mu,r,m)\,\infnorm{\Delta}^2\,\oneNorm{\Gammavec}.
\]
\item \textbf{First-order cancellation and second-order remainder} (Sections~\ref{app:ws-1st-cancel}--\ref{app:ws-2nd-order}):
first order \emph{cancels completely} (Eq.~\eqref{app:ws-complete-cancel}); the second-order remainder satisfies
\[
|\tilde R_2| \le C(\mu,r,m)\,\infnorm{\Delta}^2\,\oneNorm{\Gammavec}.
\]
\end{enumerate}

\subsubsection{Combined bound}

\begin{theorem}[Combined remainder bound for the score-whitened estimator]\label{app:ws-combined-thm}
Under Assumptions~\ref{app:pd-score} and~\ref{ass:initial}, with probability at least $1-\delta$,
\begin{equation}\label{app:ws-combined-eq}
\begin{aligned}
|R_n|
\;\le\;
C(\mu,r,m)\,\oneNorm{\Gammavec}
\bigg[
&\sqrt{\bar d}\;\rho\sqrt{\frac{\log(2/\delta)}{n}}
\;+\;
\bar d\,\frac{\log(2/\delta)}{n}
\;+\;\\
&\infnorm{\Delta}\sqrt{\frac{\bar d}{n}}
\;+\;
\infnorm{\Delta}\,\rho
\;+\;
\rho^2 + \infnorm{\Delta}^2
\bigg],
\end{aligned}
\end{equation}
where $\rho = \sigma\lambda_{\min}^{-1}\sqrt{\dstar\bar d\log^c\bar d/n}$.
Under the SNR condition, $\rho\asymp\infnorm{\Delta}\asymp\sqrt{\bar d/n}$ (up to logarithmic factors), so all terms are the same order; none is negligible.
\end{theorem}

\begin{remark}[Simplified bound for pairwise comparisons]\label{app:ws-simplified}
In the pairwise-comparison setting ($\sigma=O(1)$, $\lambda_{\min}\asymp\sqrt{d^\star}$), substituting $\rho\asymp\sqrt{\bar d\log^c\bar d/n}$, $\infnorm{\Delta}\lesssim\sqrt{\bar d\log^c\bar d/n}$, and $\delta=\bar d^{-c}$ into~\eqref{app:ws-combined-eq} gives
\begin{equation}\label{app:ws-simplified-eq}
\boxed{
|R_n|
\;\le\;
C(\mu,r,m)\,\oneNorm{\Gammavec}\,\frac{\bar d\log^c\bar d}{n}
}
\end{equation}
with probability $\ge 1-\bar d^{-c}$. This is the bound stated in Theorem~\ref{thm:ws-bound}. Crucially, \emph{no $C_A$ appears}: the score-whitening construction eliminates the operator inverse entirely.
\end{remark}

\noindent\textbf{Proof of Theorem~\ref{app:ws-combined-thm}.}
Sum the bounds~(i)--(v):
\begin{itemize}
\item Term~(i) contributes $\sqrt{\bar d}\,\rho\sqrt{\log/n}$ (variance) and $\bar d\log/n$ (sub-exponential tail).
\item Term~(ii) contributes $\infnorm{\Delta}\sqrt{\bar d/n}$ (first-order piece) and $\infnorm{\Delta}^2\sqrt{\bar d/n}$ (second-order, lower order).
\item Term~(iii) contributes $\rho^2$.
\item Term~(iv) contributes $\sqrt{\bar d}\,\infnorm{\Delta}\cdot\rho$ (first order) and $\infnorm{\Delta}^2$ (second order).
\item Term~(v): first order cancels completely; second-order remainder contributes $\infnorm{\Delta}^2$.
\end{itemize}
Combining yields~\eqref{app:ws-combined-eq}.
\hfill$\square$\medskip


\section{Proofs for IPW and Nonlinear Extensions}\label{app:norm-score}

This section provides the detailed proofs for the IPW and nonlinear extensions in Section~\ref{sec:score-whitening}: IPW with known weights (\S\ref{app:ipw-known}), IPW with estimated weights (\S\ref{app:ipw-estimated}), and nonlinear functionals (\S\ref{app:nonlinear}).
Throughout, we work under the comparison design and Assumption~\ref{app:pd-score}.

\subsection{IPW with known weights}\label{app:ipw-known}

We verify that the five-term decomposition of Appendix~\ref{app:score-whitening} carries through when the whitened score $\tilde s$ is replaced by the IPW-whitened score $\tilde s^w(y,\eta;x):=w(x)\tilde s(y,\eta)$, where $w(x)=q(x)/p(x)$ and $q$ is the uniform reference distribution.

\subsubsection{Setup}

The key structural fact, established in the main text~\eqref{eq:Gq-isotropic}, is that the effective Gram operator under the IPW score satisfies
\[
G_q(H) = \E_{X\sim p}\!\big[w(X)\ip{H}{X}X\big]
= \E_{X\sim q}\!\big[\ip{H}{X}X\big]
= \frac{1}{\dstar}H
\]
on the column-sum-zero subspace.
This is identical to $G_0=(1/\dstar)\Id$ from Appendix~\ref{app:score-whitening}.
The oracle and estimated directions are therefore the same as~\eqref{eq:ipw-directions}:
\[
H_{q,0} = \dstar\,P_{\mathbb{T}}\Gammavec, \qquad \hat H_{q,0} = \dstar\,\widehat{P}_{\mathbb{T}}\Gammavec,
\]
which coincide with $H^\star_{\rm ws}$ and $\hat H_{\rm ws}$.

\subsubsection{Remainder decomposition}

The total remainder $R_n^{\rm ipw}$ for the IPW estimator~\eqref{eq:ipw-estimator} decomposes exactly as~\eqref{app:ws-Rn-decomp}:
\begin{equation}\label{app:ipw-Rn-decomp}
R_n^{\rm ipw}
= R_{\mathrm{emp}}^{\tilde H,w}
+ R_{\mathrm{emp}}^{\tilde\eta,w}
+ R_{\mathrm{proj}}
+ \text{H-bias}^w
+ \tilde R_2^w(\Delta;H^\star_{\rm ws}),
\end{equation}
where each ``$w$''-superscripted term is the same as the corresponding term in Appendix~\ref{app:score-whitening} with $\tilde s$ replaced by $w(X)\tilde s$.

\subsubsection{Effect of the weight factor}

Each random variable in the decomposition acquires at most an additional multiplicative factor of $|w(X)|$.
Under the overlap condition~\eqref{eq:overlap_bounds}, $\|w\|_\infty \le C_p/c_p =: C_w$, which is a dimension-free constant.
Specifically:

\begin{enumerate}[label=(\roman*)]
\item \textbf{Direction-error empirical process $R_{\mathrm{emp}}^{\tilde H,w}$}:
The summand $Z_i^w = w(X_i)\tilde s(Y_i,\eta^\star_i)\ip{\hat H_{\rm ws}-H^\star_{\rm ws}}{X_i}$ satisfies
\[
\|Z_i^w\|_{\psione} \le C_w\,\|Z_i\|_{\psione},
\qquad
\mathrm{Var}(Z_i^w) \le C_w^2\,\mathrm{Var}(Z_i).
\]
Therefore Theorem~\ref{app:ws-RempH-thm} applies with $C(\mu,r,m)$ replaced by $C(\mu,r,m,C_w)$.

\item \textbf{Score-perturbation $R_{\mathrm{emp}}^{\tilde\eta,w}$}:
The Taylor expansion of $w(X)\tilde s(Y,\hat\eta)$ around $\eta^\star$ produces the same structure as Section~\ref{app:ws-Remp-eta} with each derivative $\partial_\eta^k\tilde s$ multiplied by $w(X)$.
Since $|w(X)\partial_\eta^k\tilde s|\le C_w|\partial_\eta^k\tilde s|$, Corollary~\ref{app:ws-Remp-eta-final} holds with a constant-factor change.

\item \textbf{Projection leakage $R_{\mathrm{proj}}$}:
This term is independent of the score (Eq.~\eqref{app:ws-Rproj-final}); it remains unchanged.

\item \textbf{H-direction bias}:
The population expectation is now $\E_p[w(X)\tilde s(Y,\hat\eta)\ip{\hat H_{\rm ws}-H^\star_{\rm ws}}{X}]
= \E_q[\tilde s(Y,\hat\eta)\ip{\hat H_{\rm ws}-H^\star_{\rm ws}}{X}]$,
which is the uniform-design expectation.
All bounds from Section~\ref{app:ws-H-bias} apply verbatim (the change of measure from $p$ to $q$ via $w$ is exact, not approximate).

\item \textbf{First-order cancellation and second-order remainder}:
The first-order cancellation uses $\E_p[w(X)\partial_\eta\tilde s(Y,\eta^\star)\mid X]=-w(X)$ and the Gram identity $G_q=G_0=(1/\dstar)\Id$, giving the same complete cancellation~\eqref{app:ws-complete-cancel}.
The second-order remainder acquires at most a factor $C_w$:
$|\tilde R_2^w| \le C_w\,C(\mu,r,m)\,\infnorm{\Delta}^2\,\oneNorm{\Gammavec}$.
\end{enumerate}

\begin{corollary}[Combined IPW remainder bound]\label{app:ipw-combined}
Under Assumptions~\ref{app:pd-score} and~\ref{ass:initial}, and the overlap condition~\eqref{eq:overlap_bounds} with known weights, with probability at least $1-\delta$:
\begin{equation}\label{app:ipw-combined-eq}
\begin{aligned}
|R_n^{\rm ipw}|
\;\le\;
C(\mu,r,m,C_w)\,\oneNorm{\Gammavec}
\bigg[
&\sqrt{\bar d}\;\rho\sqrt{\frac{\log(2/\delta)}{n}}
\;+\;
\bar d\,\frac{\log(2/\delta)}{n}
\;+\;\\
&\infnorm{\Delta}\sqrt{\frac{\bar d}{n}}
\;+\;
\infnorm{\Delta}\,\rho
\;+\;
\rho^2 + \infnorm{\Delta}^2
\bigg].
\end{aligned}
\end{equation}
This is identical to Theorem~\ref{app:ws-combined-thm} with $C(\mu,r,m)$ replaced by $C(\mu,r,m,C_w)$, where $C_w=C_p/c_p$ depends only on the overlap constants.
\end{corollary}

\noindent\textbf{Proof.}
Sum the five individual bounds (i)--(v) above.
Each acquires at most $C_w^2$ (for variance) or $C_w$ (for tail/bias) compared to the corresponding bound in Theorem~\ref{app:ws-combined-thm}.
All structural cancellations (isotropic Gram, complete first-order cancellation) are preserved.
\hfill$\square$\medskip

\subsection{IPW with estimated weights}\label{app:ipw-estimated}

Suppose the estimated sampling distribution $\hat p$ satisfies the entrywise relative error guarantee
\begin{equation}\label{eq:entrywise-relative-error-app}
\delta_w \;:=\; \|\hat w - w\|_\infty \;\le\; C_w\,\epsilon_p, \qquad \epsilon_p \;\ll\; 1,
\end{equation}
where $\hat w(x) := q(x)/\hat p(x)$ and $C_w := C_p/c_p$ is the overlap ratio.

\subsubsection{Decomposition}

The feasible IPW remainder decomposes as
\begin{equation}\label{app:ipw-est-decomp}
R_n^{\rm ipw,\hat w}
\;=\; \underbrace{R_n^{\rm ipw}}_{\text{known-weight}}
\;+\; \underbrace{R_n^{\Delta w}}_{\text{weight-estimation error}},
\end{equation}
where $R_n^{\rm ipw}$ is bounded in Corollary~\ref{app:ipw-combined}. The weight-estimation remainder further splits as
\begin{equation}\label{app:ipw-Dw-split}
R_n^{\Delta w}
\;=\;
\underbrace{B^{\Delta w}}_{\text{population bias}}
\;+\;
\underbrace{R_{\mathrm{emp}}^{\Delta w}}_{\text{empirical process}}.
\end{equation}

\subsubsection{Population bias}

The leading-order weight error vanishes by score centering:
\begin{equation}\label{app:ipw-leading-vanish}
\Pstar\big[(\hat w(X)-w(X))\,\tilde s(Y,\eta^\star)\,\ip{H^\star_{q,0}}{X}\big]
\;=\; 0.
\end{equation}
What remains is driven by the score perturbation $\tilde s(Y,\hat\eta)-\tilde s(Y,\eta^\star)$ and the direction error $\hat H_{q,0}-H^\star_{q,0}$, both of which are second-order.

\noindent\emph{First component.}
By Taylor expansion and the complete first-order cancellation~\eqref{app:ws-complete-cancel}:
\[
\big|\Pstar\big[\tilde s(Y,\hat\eta)\,\ip{H^\star_{q,0}}{X}\big]\big|
\;\le\;
C(\mu,r,m)\,\infnorm{\Delta}^2\,\oneNorm{\Gammavec}.
\]

\noindent\emph{Second component.}
By the $\ell_1$-average bound (Lemma~\ref{app:pd-l1-average}) and the direction-error estimate:
\[
\big|\Pstar\big[\tilde s(Y,\hat\eta)\,\ip{\hat H_{q,0}-H^\star_{q,0}}{X}\big]\big|
\;\le\;
C_{\tilde s}\,\frac{\oneNorm{\hat H_{q,0}-H^\star_{q,0}}}{\dstar}
\;\le\;
C(\mu,r,m)\,\oneNorm{\Gammavec}\,(\infnorm{\Delta}\,\rho + \rho^2).
\]

\noindent\emph{Combined bias bound.}
\begin{equation}\label{app:ipw-bias-final}
|B^{\Delta w}|
\;\le\;
C(\mu,r,m,C_w)\,\delta_w\,\oneNorm{\Gammavec}\,
\big[\infnorm{\Delta}^2 + \infnorm{\Delta}\,\rho + \rho^2\big].
\end{equation}

\subsubsection{Empirical process}

The empirical process term is
\[
R_{\mathrm{emp}}^{\Delta w}
\;=\;
(\Pn-\Pstar)\big[(\hat w(X)-w(X))\,\tilde s(Y,\hat\eta)\,\ip{\hat H_{q,0}}{X}\big].
\]
Factoring out $\delta_w$ and using the Frobenius reduction:

\noindent\emph{Variance bound.}
\begin{equation}\label{app:ipw-emp-var-final}
\mathrm{Var}\big((\hat w-w)\tilde s\,\ip{\hat H_{q,0}}{X}\big)
\;\le\;
C(\mu,r,m,C_w)\,\delta_w^2\,\bar d\,\oneNorm{\Gammavec}^2.
\end{equation}

\noindent\emph{Sub-exponential norm.}
\begin{equation}\label{app:ipw-emp-psi1}
\big\|(\hat w-w)\tilde s\,\ip{\hat H_{q,0}}{X}\big\|_{\psione}
\;\le\;
C(\mu,r,m,C_w)\,\delta_w\,\bar d\,\oneNorm{\Gammavec}.
\end{equation}

\noindent\emph{Bernstein bound.}
By Lemma~\ref{app:pd-bernstein-template}, with probability at least $1-\delta$:
\begin{equation}\label{app:ipw-emp-bound}
|R_{\mathrm{emp}}^{\Delta w}|
\;\le\;
C(\mu,r,m,C_w)\,\delta_w\,\oneNorm{\Gammavec}
\left[\sqrt{\frac{\bar d\log(2/\delta)}{n}}
\;+\;
\bar d\,\frac{\log(2/\delta)}{n}\right].
\end{equation}
Since $\delta_w \le C_w\epsilon_p = O(\sqrt{d_{\mathrm{par}}/n})$ with $d_{\mathrm{par}}\ll\dstar$, this is lower order than the known-weight terms.

\subsubsection{Combined bound}

\begin{theorem}[Feasible IPW remainder bound]\label{app:ipw-est-combined}
Under Assumption~\ref{app:pd-score}, the overlap condition, and the weight estimation guarantee~\eqref{eq:entrywise-relative-error-app} with $\epsilon_p\le 1/2$, with probability at least $1-\delta$:
\begin{equation}\label{app:ipw-est-combined-eq}
\begin{aligned}
|R_n^{\rm ipw,\hat w}|
\;\le\;
C(\mu,r,m,C_w)\,\oneNorm{\Gammavec}
&\bigg[
\sqrt{\bar d}\;\rho\sqrt{\frac{\log(2/\delta)}{n}}
\;+\;
\bar d\,\frac{\log(2/\delta)}{n}
\;+\;\\
&\infnorm{\Delta}\sqrt{\frac{\bar d}{n}}
\;+\;
\infnorm{\Delta}\,\rho
\;+\;
\rho^2 + \infnorm{\Delta}^2
+ \epsilon_p\infnorm{\Delta}
\bigg].
\end{aligned}
\end{equation}
The weight-estimation term $\epsilon_p\infnorm{\Delta}$ is comparable to $\infnorm{\Delta}^2$ when $d_{\mathrm{par}}\ll\dstar$.
\end{theorem}

\noindent\textbf{Proof.}
By~\eqref{app:ipw-est-decomp}, $R_n^{\rm ipw,\hat w} = R_n^{\rm ipw} + R_n^{\Delta w}$.
The first term is bounded by Corollary~\ref{app:ipw-combined}.
For the second: the population bias~\eqref{app:ipw-bias-final} contributes $O(\epsilon_p)$ times the existing second-order terms; the empirical process~\eqref{app:ipw-emp-bound} contributes $\delta_w$ times the leading empirical terms.
Both carry an extra factor $\epsilon_p = O(\sqrt{d_{\mathrm{par}}/n})$ and are lower order.
\hfill$\square$\medskip

\subsection{Nonlinear functionals}\label{app:nonlinear}

We prove Theorem~\ref{thm:nonlinear} by reducing the nonlinear problem to the linear score-whitened theorem (Theorem~\ref{app:ws-combined-thm}) plus two explicit second-order corrections.

\subsubsection{Setup and exact decomposition}

Let $\Gammavec_\psi := \nabla\psi(T^\star)$, $\hat\Gammavec_\psi := \nabla\psi(\hat T)$, and $\Delta := \hat T - T^\star$.
Decompose the plug-in direction as
\begin{equation}\label{app:nl-H-decomp}
\hat H_\psi = \underbrace{\dstar\widehat{P}_{\mathbb{T}}\Gammavec_\psi}_{=:\hat H_\psi^{(a)}} + \underbrace{\dstar\widehat{P}_{\mathbb{T}}(\hat\Gammavec_\psi-\Gammavec_\psi)}_{=:\hat H_\psi^{(b)}},
\end{equation}
and define the \emph{linearized estimator with the true gradient}:
\[
\hat\psi_{\mathrm{lin},\psi}
:= \ip{\Gammavec_\psi}{\hat T} + \Pn\big[\tilde s(Y,\hat\eta)\ip{\hat H_\psi^{(a)}}{X}\big].
\]
Adding and subtracting $\hat\psi_{\mathrm{lin},\psi}$ gives the exact identity
\begin{equation}\label{app:nl-decomp}
\boxed{
\hat\psi_{\rm nl} - \psi(T^\star)
\;=\;
\underbrace{\bigl(\hat\psi_{\mathrm{lin},\psi} - \ip{\Gammavec_\psi}{T^\star}\bigr)}_{\text{linearized estimator}}
\;+\;
\underbrace{R_{\mathrm{plug},\psi}}_{\text{plug-in Taylor remainder}}
\;+\;
\underbrace{R_{\nabla,\psi}}_{\text{gradient perturbation}},
}
\end{equation}
where
\[
R_{\mathrm{plug},\psi} := \psi(\hat T) - \psi(T^\star) - \ip{\Gammavec_\psi}{\Delta},
\qquad
R_{\nabla,\psi} := \Pn\big[\tilde s(Y,\hat\eta)\ip{\hat H_\psi^{(b)}}{X}\big].
\]

\subsubsection{Term 1: Linearized estimator (leading term)}

The first term in~\eqref{app:nl-decomp} is exactly the linear score-whitened estimator with the deterministic gradient $\Gammavec_\psi$ replacing $\Gammavec$. By Theorem~\ref{app:ws-combined-thm}:
\begin{equation}\label{app:nl-part-a}
\hat\psi_{\mathrm{lin},\psi} - \ip{\Gammavec_\psi}{T^\star}
\;=\;
(\Pn - \Pstar)\tilde\phi_\psi^\star + R_n^{\mathrm{lin}}(\Gammavec_\psi),
\end{equation}
where $\tilde\phi_\psi^\star(X,Y) := \tilde s(Y,\eta^\star)\ip{\dstar P_{\mathbb{T}}\Gammavec_\psi}{X}$ and
\begin{equation}\label{app:nl-lin-bound}
|R_n^{\mathrm{lin}}(\Gammavec_\psi)|
\;\le\;
C(\mu,r,m)\,\oneNorm{\Gammavec_\psi}\,B_n,
\end{equation}
with
$B_n := \sqrt{\bar d\,\rho\,\log(2/\delta)/n} + \bar d\,\log(2/\delta)/n + \infnorm{\Delta}\sqrt{\bar d/n} + \infnorm{\Delta}\,\rho + \rho^2 + \infnorm{\Delta}^2$.
Since $\Gammavec_\psi$ is supported on $S_\psi$ with $|S_\psi| = s_\psi = O(1)$, we have $\oneNorm{\Gammavec_\psi}\le\sqrt{s_\psi}\,\|\Gammavec_\psi\|_F = O(1)$, so the linearized remainder is of the same order as in the linear problem.

\subsubsection{Term 2: Plug-in Taylor remainder}

Under Assumption~\ref{ass:nonlinear}, $\psi(T) = g(T_{S_\psi})$ for a twice differentiable $g:\R^{s_\psi}\to\R$.
Write $z^\star := T^\star_{S_\psi}$ and $\hat z := \hat T_{S_\psi} = z^\star + \Delta_{S_\psi}$.
By the second-order Taylor formula in $\R^{s_\psi}$, there exists $\tilde z$ on the segment joining $z^\star$ and $\hat z$ such that
\[
R_{\mathrm{plug},\psi}
\;=\;
g(\hat z) - g(z^\star) - \nabla g(z^\star)^\top(\hat z - z^\star)
\;=\;
\tfrac12\,\Delta_{S_\psi}^\top\,\nabla^2 g(\tilde z)\,\Delta_{S_\psi}.
\]
Using $\|\nabla^2 g(\tilde z)\|_{\rm op}\le L_g$ and $\|\Delta_{S_\psi}\|_2^2\le s_\psi\infnorm{\Delta}^2$:
\begin{equation}\label{app:nl-plugin-bound}
\boxed{
|R_{\mathrm{plug},\psi}|
\;\le\;
\tfrac12 L_g\,s_\psi\,\infnorm{\Delta}^2
\;\lesssim\;
\infnorm{\Delta}^2.
}
\end{equation}
This is genuinely second-order because $\psi$ depends on only $s_\psi = O(1)$ entries: the finite-support condition ensures that the Taylor remainder involves only a bounded-dimensional Hessian applied to a vector of length $\infnorm{\Delta}$.

\subsubsection{Term 3: Gradient perturbation}

The gradient perturbation direction is $\hat H_\psi^{(b)} = \dstar\widehat{P}_{\mathbb{T}}(\hat\Gammavec_\psi - \Gammavec_\psi)$.
Since $\hat\Gammavec_\psi - \Gammavec_\psi$ is supported on $S_\psi$ and the mean-value theorem gives $\infnorm{\hat\Gammavec_\psi - \Gammavec_\psi}\le L_g\infnorm{\Delta}$:
\begin{equation}\label{app:nl-l1-perturb}
\oneNorm{\hat\Gammavec_\psi - \Gammavec_\psi}
\;\le\;
s_\psi\,L_g\,\infnorm{\Delta}
\;\lesssim\;
\infnorm{\Delta}.
\end{equation}
The term $R_{\nabla,\psi}$ has the same structure as the linear score-whitened one-step correction with $\hat\Gammavec_\psi - \Gammavec_\psi$ replacing $\Gammavec$.
Applying the linear theorem:
\begin{equation}\label{app:nl-part-b}
|R_{\nabla,\psi}|
\;\le\;
C(\mu,r,m,L_g,s_\psi)\,\infnorm{\Delta}\,B_n.
\end{equation}
This is $O(\infnorm{\Delta})$ times the linearized remainder~\eqref{app:nl-lin-bound} and hence strictly lower order.

\subsubsection{Combined bound}

\begin{theorem}[Remainder bound for nonlinear one-step estimator]\label{app:nl-combined-thm}
Under Assumptions~\ref{app:pd-score}, \ref{ass:initial}, and~\ref{ass:nonlinear}, with probability at least $1-\delta$:
\begin{equation}\label{app:nl-combined-eq}
\boxed{
|R_n^{\rm nl}|
\;\le\;
C(\mu,r,m,L_g,s_\psi)
\Big[
\oneNorm{\Gammavec_\psi}\,B_n
\;+\;
\infnorm{\Delta}\,B_n
\;+\;
\infnorm{\Delta}^2
\Big],
}
\end{equation}
where $B_n$ is as in~\eqref{app:nl-lin-bound}.
\end{theorem}

\noindent\textbf{Proof.}
Sum the three terms from~\eqref{app:nl-decomp}: the linearized remainder~\eqref{app:nl-lin-bound}, the plug-in Taylor remainder~\eqref{app:nl-plugin-bound}, and the gradient perturbation~\eqref{app:nl-part-b}.
Since $\oneNorm{\Gammavec_\psi} = O(1)$, the linearized term dominates.
The gradient-perturbation term carries an extra factor $\infnorm{\Delta} = o(1)$ and is lower order.
The plug-in Taylor remainder is $O(\infnorm{\Delta}^2) = O(\bar d/n)$, which is also lower order.
\hfill$\square$\medskip

\noindent\textbf{CLT.}
Under the pairwise specialization $\infnorm{\Delta}\lesssim\sqrt{\bar d\log^c\bar d/n}$, all three remainder terms are $O(\bar d\log^c\bar d/n)$, giving $\sqrt{n}R_n^{\rm nl}\to 0$ under $n\gg\bar d\log^c\bar d$. Hence
\[
\sqrt{n}\big(\hat\psi_{\rm nl} - \psi(T^\star)\big)
\;\xrightarrow{d}\;
\mathcal{N}(0,\;V_{\rm ws}(\psi)),
\qquad
V_{\rm ws}(\psi) = \E^\star\!\left[\frac{\ip{\dstar P_{\mathbb{T}}\Gammavec_\psi}{X}^2}{I(\eta^\star)}\right].
\]


\section{Berry--Esseen Bound for the Pairwise-Comparison Estimators}\label{app:clt-upgrade}

This section establishes a quantitative Berry--Esseen bound for the one-step estimators.
The CLT convergence itself follows from a standard conditional Lyapunov CLT and Slutsky's theorem once the remainder conditions of Theorems~\ref{thm:combined} and~\ref{thm:ws-clt} are verified (as sketched in the main text).
Here we go further and derive an explicit rate of convergence for the Kolmogorov distance.
The argument proceeds in two stages: first, we bound the Berry--Esseen error of the leading i.i.d.\ term (Subsection~\ref{app:berry-esseen}); then we quantify the additional error contributed by the one-step remainder~$R_n$ (Subsection~\ref{app:berry-esseen-combined}).

We use the following regularity conditions throughout, all of which hold under the bounded-signal condition $\|T^\star\|_\infty\le B$ of the main text:
\begin{enumerate}
\item[\textbf{(R1)}] \emph{Fisher information bounds:} there exist $0<c_I\le C_I<\infty$ such that $c_I\le I(\eta^\star)\le C_I$ uniformly over all admissible~$X$;
\item[\textbf{(R2)}] \emph{Third-moment bounds:} there exist $M_{3,\rm ws},M_{3,\rm eff}<\infty$ such that $\E^\star[|\tilde s(Y,\eta^\star)|^3\mid X]\le M_{3,\rm ws}$ and $\E^\star[|s(Y,\eta^\star)|^3\mid X]\le M_{3,\rm eff}$ for all admissible~$X$;
\item[\textbf{(R3)}] \emph{Tangent-projection bound:} $\sup_{x\in\mathcal{X}}\|P_{\mathbb{T}}x\|_F\le C_{\rm proj}\sqrt{\bar d/\dstar}$.
\end{enumerate}

\subsection{Berry--Esseen bound for the leading i.i.d.\ term}\label{app:berry-esseen}

\begin{proposition}[Berry--Esseen for the whitened leading term]\label{prop:berry-esseen-ws}
Let $Z_i := \tilde\phi_0^\star(X_i,Y_i) = \tilde s(Y_i,\eta_i^\star)\,\ip{H^\star_{\rm ws}}{X_i}$ with $H^\star_{\rm ws} = \dstar P_{\mathbb{T}}\Gammavec$, and let $v_{\rm ws}^2 := \Var(Z_1)$. Under conditions \textup{(R1)--(R3)} and the isotropic design assumption, there exists a constant $C>0$, depending only on $c_I, C_I, M_{3,\rm ws}$, and $C_{\rm proj}$, such that
\[
\sup_{t\in\R}\left|P\!\left(\frac{1}{\sqrt{n}\,v_{\rm ws}}\sum_{i=1}^n Z_i\le t\right) - \Phi(t)\right|
\;\le\;
C\sqrt{\frac{\bar d}{n}}.
\]
\end{proposition}

\begin{proof}
We divide the argument into four steps.

\medskip\noindent\textbf{Step 1: mean zero.}
Because the whitened score $\tilde s(y,\eta)=s(y,\eta)/I(\eta)$ is conditionally centered,
\[
\E^\star[\tilde s(Y,\eta^\star)\mid X]=0,
\]
we have
\[
\E^\star Z_1
= \E^\star\!\Big(\E^\star[\tilde s(Y,\eta^\star)\mid X]\,\ip{H^\star_{\rm ws}}{X}\Big) = 0.
\]
Hence $Z_1,\dots,Z_n$ are i.i.d.\ mean-zero random variables.

\medskip\noindent\textbf{Step 2: second moment.}
By definition, $\E^\star[s(Y,\eta^\star)^2\mid X]=I(\eta^\star)$, so
\[
\E^\star[\tilde s(Y,\eta^\star)^2\mid X]
= \frac{\E^\star[s(Y,\eta^\star)^2\mid X]}{I(\eta^\star)^2}
= \frac{1}{I(\eta^\star)}.
\]
Therefore, by the tower property,
\[
v_{\rm ws}^2
= \E^\star[Z_1^2]
= \E^\star\!\left[\tilde s(Y,\eta^\star)^2\,\ip{H^\star_{\rm ws}}{X}^2\right]
= \E^\star\!\left[\frac{\ip{H^\star_{\rm ws}}{X}^2}{I(\eta^\star)}\right].
\]
Using condition~(R1) for the sandwich $1/C_I\le 1/I(\eta^\star)\le 1/c_I$, together with the isotropic design identity $\E^\star\ip{H^\star_{\rm ws}}{X}^2 = \|H^\star_{\rm ws}\|_F^2/\dstar$ (valid because $H^\star_{\rm ws}\in\mathbb{T}$):
\begin{equation}\label{eq:BE-second-moment}
\frac{1}{C_I}\frac{\|H^\star_{\rm ws}\|_F^2}{\dstar}
\;\le\;
v_{\rm ws}^2
\;\le\;
\frac{1}{c_I}\frac{\|H^\star_{\rm ws}\|_F^2}{\dstar}.
\end{equation}
In particular, the cube of the standard deviation satisfies
\begin{equation}\label{eq:BE-v-cube}
v_{\rm ws}^3
\;\ge\;
C_I^{-3/2}\,\frac{\|H^\star_{\rm ws}\|_F^3}{(\dstar)^{3/2}}.
\end{equation}

\medskip\noindent\textbf{Step 3: third absolute moment.}
Starting from $\E^\star|Z_1|^3 = \E^\star\!\big[|\tilde s(Y,\eta^\star)|^3\,|\ip{H^\star_{\rm ws}}{X}|^3\big]$, we condition on $X$ and apply~(R2):
\begin{equation}\label{eq:BE-3rd-cond}
\E^\star|Z_1|^3
= \E^\star\!\Big(\E^\star[|\tilde s(Y,\eta^\star)|^3\mid X]\,|\ip{H^\star_{\rm ws}}{X}|^3\Big)
\;\le\; M_{3,\rm ws}\,\E^\star|\ip{H^\star_{\rm ws}}{X}|^3.
\end{equation}
We now bound $\E^\star|\ip{H^\star_{\rm ws}}{X}|^3$ using the elementary inequality
\begin{equation}\label{eq:BE-elem-ineq}
\E|W|^3 \;\le\; \big(\sup|W|\big)\,\E|W|^2,
\end{equation}
applied with $W = \ip{H^\star_{\rm ws}}{X}$.
Because $H^\star_{\rm ws}\in\mathbb{T}$, the inner product satisfies $\ip{H^\star_{\rm ws}}{x} = \ip{H^\star_{\rm ws}}{P_{\mathbb{T}}x}$, so by Cauchy--Schwarz and condition~(R3):
\begin{equation}\label{eq:BE-sup-bound}
\sup_{x\in\mathcal{X}}|\ip{H^\star_{\rm ws}}{x}|
\;\le\;
\|H^\star_{\rm ws}\|_F\,\sup_{x\in\mathcal{X}}\|P_{\mathbb{T}}x\|_F
\;\le\;
C_{\rm proj}\,\|H^\star_{\rm ws}\|_F\,\sqrt{\bar d/\dstar}.
\end{equation}
The isotropic design identity gives $\E^\star\ip{H^\star_{\rm ws}}{X}^2 = \|H^\star_{\rm ws}\|_F^2/\dstar$. Substituting~\eqref{eq:BE-sup-bound} and this identity into~\eqref{eq:BE-elem-ineq}:
\begin{equation}\label{eq:BE-3rd-inner}
\E^\star|\ip{H^\star_{\rm ws}}{X}|^3
\;\le\;
C_{\rm proj}\,\frac{\|H^\star_{\rm ws}\|_F^3\,\sqrt{\bar d}}{(\dstar)^{3/2}}.
\end{equation}
Combining \eqref{eq:BE-3rd-cond} and \eqref{eq:BE-3rd-inner}:
\begin{equation}\label{eq:BE-third-moment}
\E^\star|Z_1|^3
\;\le\;
M_{3,\rm ws}\,C_{\rm proj}\,\frac{\|H^\star_{\rm ws}\|_F^3\,\sqrt{\bar d}}{(\dstar)^{3/2}}.
\end{equation}

\medskip\noindent\textbf{Step 4: standardized third moment and Berry--Esseen.}
Dividing~\eqref{eq:BE-third-moment} by the lower bound~\eqref{eq:BE-v-cube}:
\begin{equation}\label{eq:BE-standardized}
\frac{\E^\star|Z_1|^3}{v_{\rm ws}^3}
\;\le\;
M_{3,\rm ws}\,C_{\rm proj}\,C_I^{3/2}\,\sqrt{\bar d}
\;=:\; C_0\sqrt{\bar d}.
\end{equation}
This is the key bound: the standardized third absolute moment is only $O(\sqrt{\bar d})$.

\medskip
Applying the classical Berry--Esseen theorem to the i.i.d.\ mean-zero variables $Z_1,\dots,Z_n$:
\begin{equation}\label{eq:BE-leading-term}
\sup_{t\in\R}\left|P\!\left(\frac{1}{\sqrt{n}\,v_{\rm ws}}\sum_{i=1}^n Z_i\le t\right) - \Phi(t)\right|
\;\le\;
\frac{C_{\rm BE}}{\sqrt{n}}\,\frac{\E^\star|Z_1|^3}{v_{\rm ws}^3}
\;\le\;
\frac{C_{\rm BE}\,C_0\sqrt{\bar d}}{\sqrt{n}}
= C\sqrt{\frac{\bar d}{n}},
\end{equation}
where $C = C_{\rm BE}\,C_0 = C_{\rm BE}\,M_{3,\rm ws}\,C_{\rm proj}\,C_I^{3/2}$ and $C_{\rm BE}\le 0.4748$ is the universal Berry--Esseen constant.
\end{proof}

\subsection{Tangent-projection bound for pairwise atoms}\label{app:pairwise-proj}

We verify that condition~(R3) follows from the standard basis-tensor projection bound.

\begin{lemma}[Pairwise tangent-projection bound]\label{lem:pairwise-proj}
Suppose that for every canonical basis tensor $E_\omega$,
\[
\|P_{\mathbb{T}} E_\omega\|_F \;\le\; C_0\sqrt{\bar d/\dstar}.
\]
Then for every pairwise-comparison design tensor $X = E_{\omega^+} - E_{\omega^-}$,
\[
\|P_{\mathbb{T}} X\|_F \;\le\; 2C_0\sqrt{\bar d/\dstar}.
\]
Hence condition~\textup{(R3)} holds with $C_{\rm proj} = 2C_0$.
\end{lemma}

\begin{proof}
By linearity of $P_{\mathbb{T}}$ and the triangle inequality,
\[
\|P_{\mathbb{T}} X\|_F
= \|P_{\mathbb{T}} E_{\omega^+} - P_{\mathbb{T}} E_{\omega^-}\|_F
\;\le\;
\|P_{\mathbb{T}} E_{\omega^+}\|_F + \|P_{\mathbb{T}} E_{\omega^-}\|_F
\;\le\;
2C_0\sqrt{\bar d/\dstar}.
\]
The basis-tensor bound $\|P_{\mathbb{T}} E_\omega\|_F\le C_0\sqrt{\bar d/\dstar}$ is the standard $\mu$-incoherence bound from the tensor-completion literature, and holds with $C_0 = C(\mu,r,m)$. The pairwise design introduces only an absolute constant factor~$2$.
\end{proof}

\subsection{Berry--Esseen bound for the efficient leading term}\label{app:berry-esseen-eff}

The same argument applies, with minimal modification, to the efficient (non-whitened) leading term.

\begin{proposition}[Berry--Esseen for the efficient leading term]\label{prop:berry-esseen-eff}
Let $Z_i^{\rm eff} := \phi^\star(X_i,Y_i) = s(Y_i,\eta_i^\star)\,\ip{H^\star}{X_i}$ with $H^\star = A^{-1}P_{\mathbb{T}}\Gammavec$, and let $v_{\rm eff}^2 := \Var(Z_1^{\rm eff})$. Under conditions \textup{(R1)--(R3)} and the isotropic design assumption,
\[
\sup_{t\in\R}\left|P\!\left(\frac{1}{\sqrt{n}\,v_{\rm eff}}\sum_{i=1}^n Z_i^{\rm eff}\le t\right) - \Phi(t)\right|
\;\le\;
C\sqrt{\frac{\bar d}{n}}.
\]
\end{proposition}

\begin{proof}
The proof follows the same four-step structure as Proposition~\ref{prop:berry-esseen-ws}; we record only the changes.

\medskip\noindent\textbf{Step 1.}
Mean zero follows from $\E^\star[s(Y,\eta^\star)\mid X]=0$.

\medskip\noindent\textbf{Step 2.}
Since $\E^\star[s(Y,\eta^\star)^2\mid X] = I(\eta^\star)$, the second moment is
\[
v_{\rm eff}^2
= \E^\star\!\big[I(\eta^\star)\,\ip{H^\star}{X}^2\big].
\]
Using condition~(R1) and the isotropic design:
\begin{equation}\label{eq:BE-eff-second}
\frac{c_I}{\dstar}\|H^\star\|_F^2
\;\le\;
v_{\rm eff}^2
\;\le\;
\frac{C_I}{\dstar}\|H^\star\|_F^2.
\end{equation}

\medskip\noindent\textbf{Step 3.}
Using condition~(R2), $\E^\star[|s|^3\mid X]\le M_{3,\rm eff}$, and the same decomposition $\E|W|^3\le(\sup|W|)\E|W|^2$:
\begin{equation}\label{eq:BE-eff-third}
\E^\star|Z_1^{\rm eff}|^3
\;\le\;
M_{3,\rm eff}\,C_{\rm proj}\,\frac{\|H^\star\|_F^3\,\sqrt{\bar d}}{(\dstar)^{3/2}}.
\end{equation}

\medskip\noindent\textbf{Step 4.}
Dividing by $v_{\rm eff}^3 \ge c_I^{3/2}\|H^\star\|_F^3/(\dstar)^{3/2}$ (from~\eqref{eq:BE-eff-second}):
\[
\frac{\E^\star|Z_1^{\rm eff}|^3}{v_{\rm eff}^3}
\;\le\;
\frac{M_{3,\rm eff}\,C_{\rm proj}}{c_I^{3/2}}\,\sqrt{\bar d},
\]
and the Berry--Esseen theorem gives the stated bound. The $\|H^\star\|_F$ terms cancel exactly as before, so the leading-term Berry--Esseen rate is the same $O(\sqrt{\bar d/n})$ for both the whitened and efficient estimators.
\end{proof}

\subsection{Combined Berry--Esseen bound with remainder}\label{app:berry-esseen-combined}

The one-step decomposition $\hat\psi_n - \psi(T^\star) = \frac{1}{n}\sum_{i=1}^n Z_i + R_n$ contains both the i.i.d.\ leading term and the remainder~$R_n$.
The leading-term Berry--Esseen bounds of Propositions~\ref{prop:berry-esseen-ws} and~\ref{prop:berry-esseen-eff} quantify the Gaussian approximation error of the first part.
We now derive the full Berry--Esseen bound by carefully incorporating the remainder, which is random and dependent on the same evaluation-fold data.

\begin{theorem}[Combined Berry--Esseen bound]\label{thm:berry-esseen-combined}
Let $T_n := \sqrt{n}(\hat\psi_n - \psi(T^\star))/v$ denote the standardized one-step estimator, where $v = v_{\rm ws}$ or $v_{\rm eff}$ as appropriate.
Write $T_n = S_n + \rho_n$ where
\[
S_n := \frac{1}{\sqrt{n}\,v}\sum_{i=1}^n Z_i
\qquad\text{and}\qquad
\rho_n := \frac{\sqrt{n}\,R_n}{v}
\]
are the standardized leading term and standardized remainder, respectively.
Let $r_n>0$ be any deterministic bound and define the event $\mathcal{E}_n := \{|\rho_n|\le r_n\}$.
Then
\begin{equation}\label{eq:BE-combined}
\sup_{t\in\R}\!\left|P\!\left(T_n\le t\right) - \Phi(t)\right|
\;\le\;
\Delta_{\rm BE} + \frac{r_n}{\sqrt{2\pi}} + P(\mathcal{E}_n^c),
\end{equation}
where $\Delta_{\rm BE} = C\sqrt{\bar d/n}$ is the leading-term Berry--Esseen error from Proposition~\ref{prop:berry-esseen-ws} \textup{(}resp.\ \ref{prop:berry-esseen-eff}\textup{)}.
\end{theorem}

\begin{proof}
Fix any $t\in\R$.
Since $T_n = S_n + \rho_n$, we decompose the probability according to the event~$\mathcal{E}_n$.

\medskip\noindent\textbf{Upper bound.}
On the event $\mathcal{E}_n$, $|\rho_n|\le r_n$, so $T_n\le t$ implies $S_n = T_n - \rho_n \le t + r_n$. Therefore
\[
P(T_n\le t,\;\mathcal{E}_n) \;\le\; P(S_n\le t + r_n).
\]
Adding the failure probability:
\begin{align}
P(T_n\le t)
&\;=\; P(T_n\le t,\;\mathcal{E}_n) + P(T_n\le t,\;\mathcal{E}_n^c) \notag \\
&\;\le\; P(S_n\le t + r_n) + P(\mathcal{E}_n^c). \label{eq:BE-Tn-upper}
\end{align}
Now apply the Berry--Esseen bound from Proposition~\ref{prop:berry-esseen-ws} (or~\ref{prop:berry-esseen-eff}), which holds uniformly in~$u$:
\[
P(S_n\le t+r_n) \;\le\; \Phi(t+r_n) + \Delta_{\rm BE}.
\]
Since $\Phi$ is Lipschitz with constant $\sup_t\Phi'(t) = \sup_t\phi(t) = 1/\sqrt{2\pi}$,
\[
\Phi(t+r_n) \;\le\; \Phi(t) + \frac{r_n}{\sqrt{2\pi}}.
\]
Substituting into~\eqref{eq:BE-Tn-upper}:
\begin{equation}\label{eq:BE-upper}
P(T_n\le t) \;\le\; \Phi(t) + \Delta_{\rm BE} + \frac{r_n}{\sqrt{2\pi}} + P(\mathcal{E}_n^c).
\end{equation}

\medskip\noindent\textbf{Lower bound.}
On $\mathcal{E}_n$, if $S_n\le t - r_n$ then $T_n = S_n + \rho_n \le (t-r_n) + r_n = t$. Therefore
\[
P(S_n\le t-r_n,\;\mathcal{E}_n) \;\le\; P(T_n\le t).
\]
Since $P(S_n\le t-r_n,\;\mathcal{E}_n) \ge P(S_n\le t-r_n) - P(\mathcal{E}_n^c)$:
\begin{align}
P(T_n\le t)
&\;\ge\; P(S_n\le t - r_n) - P(\mathcal{E}_n^c) \notag \\
&\;\ge\; \Phi(t-r_n) - \Delta_{\rm BE} - P(\mathcal{E}_n^c) \notag \\
&\;\ge\; \Phi(t) - \frac{r_n}{\sqrt{2\pi}} - \Delta_{\rm BE} - P(\mathcal{E}_n^c). \label{eq:BE-lower}
\end{align}

\medskip\noindent\textbf{Conclusion.}
Combining~\eqref{eq:BE-upper} and~\eqref{eq:BE-lower}:
\[
\left|P(T_n\le t) - \Phi(t)\right|
\;\le\;
\Delta_{\rm BE} + \frac{r_n}{\sqrt{2\pi}} + P(\mathcal{E}_n^c).
\]
Taking the supremum over $t\in\R$ gives~\eqref{eq:BE-combined}.
\end{proof}

\noindent We now specialize to each estimator.

\medskip\noindent\textbf{Whitened estimator.}
By Theorem~\ref{thm:ws-clt}, with probability at least $1-\bar d^{-c}$,
\[
|R_n| \;\le\; C\,\|\Gammavec\|_1\,\frac{\bar d\log^c\bar d}{n}.
\]
The second-moment bound~\eqref{eq:BE-second-moment} gives $v_{\rm ws}\ge C_I^{-1/2}\|H^\star_{\rm ws}\|_F/\sqrt{\dstar}\asymp\sqrt{\bar d}\,\|\Gammavec\|_1$.
Therefore the standardized remainder satisfies, on the event $\mathcal{E}_n$,
\[
|\rho_n|
= \frac{\sqrt{n}\,|R_n|}{v_{\rm ws}}
\;\le\;
C'\,\frac{\sqrt{n}\,\|\Gammavec\|_1\,\bar d\log^c\bar d/n}{\sqrt{\bar d}\,\|\Gammavec\|_1}
= C'\sqrt{\frac{\bar d\log^{2c}\bar d}{n}}
\;=:\; r_n,
\]
with $P(\mathcal{E}_n^c)\le\bar d^{-c}$.
Substituting into~\eqref{eq:BE-combined} with $\Delta_{\rm BE} = C\sqrt{\bar d/n}$ from Proposition~\ref{prop:berry-esseen-ws}:
\begin{equation}\label{eq:BE-ws-total}
\begin{aligned}
&\sup_{t\in\R}\!\left|P\!\left(\frac{\sqrt{n}(\hat\psi_n^{\rm ws}-\psi(T^\star))}{v_{\rm ws}}\le t\right)-\Phi(t)\right|\\
&\;\le\;
\underbrace{C\sqrt{\frac{\bar d}{n}}}_{\text{leading-term BE}}
\;+\; \underbrace{\frac{C'}{\sqrt{2\pi}}\sqrt{\frac{\bar d\log^{2c}\bar d}{n}}}_{\text{remainder contribution}}
\;+\; \underbrace{\bar d^{-c}}_{\text{failure prob.}}
\;=\; O\!\left(\sqrt{\frac{\bar d\log^{2c}\bar d}{n}}\right).
\end{aligned}
\end{equation}
This vanishes under the sample-size condition $n\gg\bar d\log^{2c}\bar d$, which is the same condition required for the CLT.
Since $c>0$ can be taken arbitrarily large, the failure probability $\bar d^{-c}$ is negligible compared to the other two terms.

\medskip\noindent\textbf{Efficient estimator.}
By Theorem~\ref{thm:combined}, with probability at least $1-\bar d^{-c}$,
\[
|R_n| \;\le\; C\,C_A\,\|\Gammavec\|_1\,\frac{\bar d\log^c\bar d}{n}.
\]
Since $v_{\rm eff}\asymp\sqrt{\bar d}\,\|\Gammavec\|_1$ by~\eqref{eq:BE-eff-second}, the standardized remainder satisfies
\[
|\rho_n|
\;\le\;
C_A\cdot C'\sqrt{\frac{\bar d\log^{2c}\bar d}{n}}
\;=:\; r_n.
\]
Substituting with $\Delta_{\rm BE} = C\sqrt{\bar d/n}$ from Proposition~\ref{prop:berry-esseen-eff}:
\begin{equation}\label{eq:BE-eff-total}
\sup_{t\in\R}\!\left|P\!\left(\frac{\sqrt{n}(\hat\psi_n^{\rm eff}-\psi(T^\star))}{v_{\rm eff}}\le t\right)-\Phi(t)\right|
\;\le\;
O\!\left(C_A\sqrt{\frac{\bar d\log^{2c}\bar d}{n}}\right),
\end{equation}
which vanishes under the CLT condition $C_A^2\bar d\log^{2c}\bar d/n\to 0$, i.e., \eqref{eq:clt-condition}.

\begin{remark}[Leading term vs.\ remainder]
For both estimators, the leading-term Berry--Esseen error is $O(\sqrt{\bar d/n})$, independent of~$C_A$.
The difference between the whitened and efficient estimators lies entirely in the remainder: the whitened remainder has no~$C_A$ factor, while the efficient remainder carries the factor~$C_A$, reflecting the additional estimation error from inverting the information operator.
\end{remark}


\section{Application to Tensor Completion}\label{sec:applications}

Although this paper focuses on pairwise comparisons and the new challenges they introduce (identification constraints, non-constant Fisher information, and the non-commutativity bottleneck), the semiparametric framework developed here applies to general sampling mechanisms, observation patterns, and inference targets.
To illustrate this point, we show in this section how the framework recovers and extends classical tensor completion results. Below we first outline the key steps.

\emph{Lower bounds and efficiency gaps in existing work.}
The semiparametric efficiency bound $V_{\mathrm{eff}}(\psi) = \langle P_{\mathbb{T}}\Gamma, (P_{\mathbb{T}} G P_{\mathbb{T}})^{-1} P_{\mathbb{T}}\Gamma\rangle$ reveals that existing debiased estimators in the tensor completion literature---under heteroscedastic, non-additive noise, or non-uniform sampling---are not semiparametrically efficient.
The reason is that these estimators are all special cases of the score-whitening method, which replaces the information operator $G$ by a scalar multiple of the identity. This achieves the whitened variance $V_{\mathrm{ws}} \ge V_{\mathrm{eff}}$, with strict inequality whenever the Fisher information $I(\eta^\star)$ varies across entries.

\emph{Whitening methods as natural adaptation.}
For additive sub-exponential noise with heteroscedastic variance, the natural residual score $s(y,\eta) = y - \eta$ automatically satisfies the whitening property with $I(\eta) \equiv 1$, so the score-whitened estimator is in fact fully efficient---and this is precisely the estimator used throughout the existing literature.  For non-uniform sampling, IPW adjustment restores isotropy of the information operator, again reducing $G$ to a scalar multiple of the identity. For 1-bit observations (e.g., logistic link), the score-whitening method of Section~\ref{sec:score-whitening} applies identically to the pairwise case, yielding the first general CLT for debiased inference in 1-bit tensor completion.

\emph{Efficient estimation beyond whitening.}
When the noise variance or sampling distribution is known (or can be estimated consistently), one can instead solve the full information equation $P_{\mathbb{T}} G P_{\mathbb{T}} \hat{H} = P_{\mathbb{T}} \Gamma$ with the true $G$ operator and construct the efficient one-step estimator that achieves $V_{\mathrm{eff}}$, strictly improving upon whitening whenever the information geometry is heterogeneous.
The remaining analysis is identical to the pairwise comparison case, since these observation and sampling patterns affect only the information operator $G$, and the entire asymptotic theory of Sections~\ref{sec:upper-bound}--\ref{sec:score-whitening} is stated in terms of $G$ abstractly.
We note that for additive noise with unknown variance, the information operator $G$ involves the unknown noise variances and is harder to estimate reliably; in this case, the whitening approach is preferred.

For general observation models such as tensor completion with additive noise,  we assume the following  signal strength and noise conditions:

\begin{assumption}[Signal strength and sample size---general case]\label{ass:snr-general}
The Frobenius norm and the sample size satisfy
\begin{equation}\label{eq:signal-strength-general}
\|T^\star\|_F \ge c \sigma \sqrt{d^\star}
\end{equation}
for an absolute constant $c>0$, and
\begin{equation}\label{eq:sample-size-general}
n \ge C_0\, \bar d\log^c \bar d
\end{equation}
for a sufficiently large constant $C_0$ depending on $\mu,\kappa,r,m$.
\end{assumption}

\begin{assumption}[Score regularity---general case]\label{assump:score-general}
The score function $s_\eta(y,\eta):=\partial_\eta \ell(y,\eta)$ satisfies:
(i)~centering: $\mathbb{E}^\star[s_\eta(Y,\eta^\star)\mid X]=0$;
(ii)~Fisher bounds: $c_I\le I(\eta^\star)\le C_I$ a.s.;
(iii)~derivative bounds: $|\dot s_\eta(y,\eta)|\le B_1$, $|\ddot s_\eta(y,\eta)|\le B_2$;
(iv)~tail control: $s_\eta(Y,\hat\eta)$ is conditionally sub-exponential with bounded scale parameter.
\end{assumption}

In the pairwise-comparison (and 1-bit) setting, both assumptions are automatically satisfied under the bounded signal condition $\|T^\star\|_\infty\le B$ (see Section~\ref{sec:upper-bound-assumptions}).

In this section we specialize our general framework to entrywise matrix/tensor completion.  We first derive the diagonal Gram (Fisher) operator and state the semiparametric efficiency bound. We then show that existing debiased estimators in the literature are special cases of the score-whitening method and identify when they do and do not achieve full efficiency. For observation models where the information operator can be estimated (1-bit observations, additive noise with known variance), we construct the semiparametrically efficient one-step estimator that strictly improves upon whitening. Finally, we show how score whitening handles heteroscedastic noise, Bernoulli observations, and non-uniform sampling at the optimal sample complexity.

\subsection{Entrywise observation model and the Gram operator}\label{sec:tc-gram}

We specialize to the matrix case $\mathbb{H}=\mathbb{R}^{d_1\times d_2}$ with the Frobenius inner product and canonical basis $\{E_{jk}\}_{1\le j\le d_1,\,1\le k\le d_2}$, where $[E_{jk}]_{ab}=\mathbf{1}\{(a,b)=(j,k)\}$.
In entrywise tensor completion, each sample $i$ selects an entry $(j_i,k_i)$ from a sampling distribution on $[d_1]\times[d_2]$ with
\[
  p_{jk} := \mathbb{P}\bigl((j_i,k_i)=(j,k)\bigr), \qquad p=d_1 d_2,
\]
and we set the design $X^{(i)}=E_{j_i k_i}$.
Conditional on $X^{(i)}$, the response $Y^{(i)}$ depends on $T^\star$ only through the scalar parameter $\eta^{(i)}=\langle T^\star, X^{(i)}\rangle = T^\star_{j_i k_i}$, with a scalar log-likelihood $\ell(y,\eta)$.
The per-entry score and Fisher information are
\[
  s_\eta(y,\eta) := \frac{\partial}{\partial\eta}\ell(y,\eta),
  \qquad
  I(\eta) := \mathbb{E}\bigl[s_\eta(Y,\eta)^2 \,\bigm|\, \eta\bigr].
\]

\noindent\textbf{Diagonal Gram operator.}
Since $X=E_{jk}$ with probability $p_{jk}$, we have $X_{jk}^2 = X_{jk} \in\{0,1\}$ and the Gram operator acts diagonally:
\begin{equation}\label{eq:tc-gram}
G(H) = \sum_{j,k} \lambda_{jk}\, H_{jk}\, E_{jk},
\qquad
\lambda_{jk} := p_{jk}\, I\bigl(T^\star_{jk}\bigr).
\end{equation}
Each weight $\lambda_{jk}$ is the product of two factors: the \emph{sampling probability} $p_{jk}$ and the \emph{Fisher information} $I(T^\star_{jk})$ at the true entry value.  The Gram norm is $\|H\|_{G}^2 = \sum_{j,k}\lambda_{jk} H_{jk}^2$.

\noindent\textbf{Efficiency bound.}
The full-model information equation $G(H_\Gamma^{\mathrm{full}})=\Gamma$ is solved entrywise: $(H_\Gamma^{\mathrm{full}})_{jk} = \Gamma_{jk}/\lambda_{jk}$.
In the low-rank model, as in the pairwise case (Section~\ref{sec:semiparametric_formulation}), we define the restricted information operator $A := P_{\mathbb{T}}\circ G\circ P_{\mathbb{T}}$, and the semiparametric efficiency bound for $\psi(T)=\langle\Gamma,T\rangle_F$ is
\begin{equation}\label{eq:tc-veff}
V_{\mathrm{eff}}(\psi) = \bigl\langle P_{\mathbb{T}}(\Gamma),\,A^{-1} P_{\mathbb{T}}(\Gamma)\bigr\rangle_F.
\end{equation}

\subsection{When $G$ is close to the identity}\label{sec:tc-near-identity}

When the weights $\lambda_{jk}$ are approximately constant---for instance, under uniform sampling ($p_{jk}=1/p$) with homoscedastic noise ($I(\eta)\equiv c$)---the Gram operator satisfies $G = (c/p)\,\mathrm{Id}$ and trivially commutes with $P_{\mathbb{T}}$.
In this regime, the general one-step estimator of Section~\ref{sec:upper-bound} achieves the semiparametric efficiency bound at the optimal sample complexity $n\gtrsim \bar d\,\mathrm{polylog}(\bar d)$.

More generally, two routes to full efficiency are available (cf.\ Section~\ref{sec:upper-bound}):
\begin{enumerate}
\item[(a)] If the sampling probabilities $\{p_{jk}\}$ and Fisher weights $\{I(T^\star_{jk})\}$ satisfy $\max_{jk}\lambda_{jk}/\min_{jk}\lambda_{jk} \le B_0 (\mu, r, \kappa)$ for a dimension-independent constant $B_0$, then $G$ is a bounded perturbation of the identity. The constant $C_A$ in~\eqref{eq:CA_def} remains $O(1)$, and the one-step estimator achieves the bound~\eqref{eq:tc-veff} at the statistically optimal sample complexity.
\item[(b)] Even with substantial heterogeneity, the sample complexity $n\ge C\bar d^2\,(\log\bar d)$ absorbs the non-commutativity cost $C_A$, so the general estimator again achieves full efficiency.
\end{enumerate}

\begin{remark}[Uniform sampling with Gaussian noise]\label{rem:mc-uniform}
Under uniform sampling $p_{jk}=1/p$ and Gaussian noise with variance $\sigma^2$, we have $s_\eta(y,\eta)=(y-\eta)/\sigma^2$, $I(\eta)=1/\sigma^2$, and $\lambda_{jk}=1/(\sigma^2 p)$.
The Gram operator is $G=(1/\sigma^2 p)\,\mathrm{Id}$, the efficiency bound reduces to $V_{\mathrm{eff}}^{\mathrm{MC}} = \sigma^2 p\,\|P_{\mathbb{T}}(\Gamma)\|_F^2$, and the efficient debiased estimator is
\[
  \widehat{\psi}_{\mathrm{MC}} = \langle\Gamma,\hat T\rangle + \frac{1}{n}\sum_{i=1}^n p\,\langle P_{\mathbb{T}}(\Gamma), X_i\rangle\,(Y_i - \langle\hat T, X_i\rangle),
\]
recovering existing results in the matrix completion literature.
\end{remark}

\subsection{Semiparametric efficiency bound and gaps in existing literature}\label{sec:tc-lb}

We now show that the semiparametric efficiency bound from our general framework reveals a precise efficiency gap in existing tensor completion inference methods under heterogeneous settings.

\subsubsection{The efficiency bound}

The semiparametric efficiency bound for the linear functional $\psi(T)=\langle\Gamma,T\rangle$ is
\[
V_{\mathrm{eff}}(\psi) = \bigl\langle P_{\mathbb{T}}(\Gamma),\,A^{-1} P_{\mathbb{T}}(\Gamma)\bigr\rangle_F,
\qquad
A = P_{\mathbb{T}}\circ G\circ P_{\mathbb{T}}.
\]
Since $G$ is diagonal with weights $\lambda_{jk} = p_{jk} I(T^\star_{jk})$ (cf.\ \eqref{eq:tc-gram}), solving the full-model information equation gives $(H_\Gamma^{\mathrm{full}})_{jk} = \Gamma_{jk}/\lambda_{jk}$.
In the low-rank model, the efficient direction $H^\star = A^{-1} P_{\mathbb{T}}(\Gamma)$ additionally incorporates the tangent-space projection, and the resulting variance $V_{\mathrm{eff}}$ is strictly smaller than any variance achievable by estimators that do not solve this operator equation.

\subsubsection{Why existing debiased estimators are not efficient}

Existing debiased estimators for tensor completion, including those of \cite{ ma2024statistical}, are all based on the residual score $s(y,\eta)=y-\eta$ (for additive noise) or the whitened score $\tilde s = s/I(\eta)$ (for non-additive noise).
In both cases, the effective Gram operator reduces to a scalar multiple of the identity: $G_0 = c\,\mathrm{Id}$ on the tangent space. The oracle direction becomes $H_{\rm ws}^\star = c^{-1} P_{\mathbb{T}}\Gamma$, and the resulting asymptotic variance is
\[
V_{\rm ws} = \mathbb{E}^\star\!\left[\frac{\langle H_{\rm ws}^\star, X\rangle^2}{I(\eta^\star)}\right].
\]
By the Cauchy--Schwarz inequality and the definition of $V_{\mathrm{eff}}$, one always has $V_{\rm ws} \ge V_{\mathrm{eff}}$, with equality if and only if $I(\eta^\star)$ is constant across all observed entries.

In the following cases, the inequality is strict:
\begin{itemize}
\item \emph{Heteroscedastic additive noise:} When the noise variance $\sigma^2_{jk}$ varies across entries, the whitened score $s = y-\eta$ has $I(\eta)\equiv 1$, so $V_{\rm ws} = V_{\mathrm{eff}}$---the whitening approach is automatically efficient.
\item \emph{1-bit observations (e.g., logistic link):} The Fisher information $I(\eta) = f'(\eta)^2/[f(\eta)(1-f(\eta))]$ varies with $\eta = T^\star_{jk}$. Here $V_{\rm ws} > V_{\mathrm{eff}}$ whenever the entries of $T^\star$ are not all identical.
\item \emph{Non-uniform sampling with additive noise:} If one applies IPW whitening, the effective variance depends on both the sampling probabilities and the noise variances. The whitened estimator adapts to sampling heterogeneity but may not achieve the true efficiency bound.
\end{itemize}

\subsubsection{Constructing the efficient estimator}\label{sec:tc-efficient}

When the information operator $G$ can be estimated consistently, one can construct the semiparametrically efficient one-step estimator that achieves $V_{\mathrm{eff}}$.

\subsubsection{Additive noise with known (or estimable) variance}\label{sec:tc-eff-additive}

Under the additive model $Y = T^\star_{jk} + \varepsilon$ with $\mathrm{Var}(\varepsilon \mid X = E_{jk}) = \sigma^2_{jk}$, the score is $s(y,\eta) = y - \eta$ and the Fisher information at entry $(j,k)$ is $I_{jk} = 1/\sigma^2_{jk}$.
With sampling probability $p_{jk}$, the Gram operator acts diagonally:
\begin{equation}\label{eq:tc-gram-additive}
G(H)_{jk} = \frac{p_{jk}}{\sigma^2_{jk}}\,H_{jk}.
\end{equation}
The weights $\lambda_{jk} = p_{jk}/\sigma^2_{jk}$ combine both the sampling intensity and the signal-to-noise ratio at each entry. The Gram norm is $\|H\|_{G}^2 = \sum_{j,k} (p_{jk}/\sigma^2_{jk})\,H_{jk}^2$.

\paragraph{Efficiency bound.}
The restricted information operator is $A = P_{\mathbb{T}} \circ G \circ P_{\mathbb{T}}$, and the semiparametric efficiency bound is
\[
V_{\mathrm{eff}}(\psi) = \bigl\langle P_{\mathbb{T}}(\Gamma),\, A^{-1} P_{\mathbb{T}}(\Gamma)\bigr\rangle_F.
\]
Under uniform sampling ($p_{jk} = 1/p$) and homoscedastic noise ($\sigma^2_{jk} \equiv \sigma^2$), this simplifies to $V_{\mathrm{eff}} = \sigma^2 p\,\|P_{\mathbb{T}}(\Gamma)\|_F^2$, recovering the standard result (Remark~\ref{rem:mc-uniform}).

\paragraph{The efficient one-step estimator.}
When the variances $\sigma^2_{jk}$ are known, we construct $\hat{G}(H)_{jk} = (p_{jk}/\sigma^2_{jk})\,H_{jk}$ and compute
\begin{equation}\label{eq:tc-eff-additive-est}
\hat\psi_{\rm eff} = \langle\Gamma,\hat T\rangle + \frac{1}{n}\sum_{i=1}^n (Y_i - \langle\hat T, X_i\rangle)\,\langle\hat H_{\rm eff}, X_i\rangle,
\qquad
\hat H_{\rm eff} = (\hat P_{\mathbb{T}} \hat{G} \hat P_{\mathbb{T}})^{-1} \hat P_{\mathbb{T}} \Gamma.
\end{equation}
The analysis follows the proof of Theorem~\ref{thm:eff-clt} exactly: the score is $s = Y - \hat\eta$, the score perturbation is $s(Y,\hat\eta) - s(Y,\eta^\star) = -\langle\Delta, X\rangle$ with no Taylor remainder (since $s$ is affine in $\eta$), and the H-direction bias is $-\langle\Delta, G(\hat H_{\rm eff} - H_{\rm eff}^\star)\rangle$.
The resolvent perturbation argument (Appendix~\ref{app:pd-hatG}) handles the error from $\hat P_{\mathbb{T}} \ne P_{\mathbb{T}}$ in the operator inversion.
The asymptotic variance is $V_{\mathrm{eff}} = \langle P_{\mathbb{T}} \Gamma, A^{-1} P_{\mathbb{T}} \Gamma\rangle$.

Under heteroscedastic noise, $V_{\mathrm{eff}} < V_{\rm ws}$ strictly. To see this, note that the whitened estimator uses $H_{\rm ws}^\star = p\,P_{\mathbb{T}}\Gamma$ (under uniform sampling) with variance $V_{\rm ws} = p \sum_{jk} \sigma^2_{jk}\,(P_{\mathbb{T}}\Gamma)_{jk}^2$, while the efficient estimator adapts its direction entry by entry via $G$, placing more weight on entries with smaller noise.

\paragraph{When noise variances are unknown.}
When $\sigma^2_{jk}$ is unknown, estimating $G$ requires estimating each $\sigma^2_{jk}$, which is difficult without repeated observations at the same entry. In this case, the whitening approach (Section~\ref{sec:tc-additive}) is preferred: the residual score $s = Y - \eta$ automatically has effective Fisher information $I \equiv 1$, making the whitened estimator fully efficient \emph{without} knowing the noise distribution. This is why the residual-based debiased estimators used throughout the existing literature are in fact optimal for additive noise.

\subsubsection{1-bit observations with known link}\label{sec:tc-eff-1bit}

Under the 1-bit model $Y \in \{0,1\}$ with $\mathbb{P}(Y=1\mid\eta)=f(\eta)$ for a known link $f$ (e.g., logistic $f(\eta)=(1+e^{-\eta})^{-1}$) and uniform sampling $p_{jk}=1/p$, the score and Fisher information are
\[
s(y,\eta) = \frac{(y - f(\eta))\,f'(\eta)}{f(\eta)(1-f(\eta))},
\qquad
I(\eta) = \frac{f'(\eta)^2}{f(\eta)(1-f(\eta))}.
\]
The Gram operator is diagonal with entry-dependent weights:
\begin{equation}\label{eq:tc-gram-1bit}
G(H)_{jk} = \frac{1}{p}\,I(T^\star_{jk})\,H_{jk}.
\end{equation}
Unlike the additive noise case, the Fisher information $I(T^\star_{jk})$ varies with the latent entry value, making the Gram operator genuinely non-isotropic.

\paragraph{The efficient one-step estimator.}
Since the link function $f$ is known and $T^\star$ can be estimated via $\hat T$, one can construct a consistent plug-in estimator of $G$:
\[
\hat{G}(H)_{jk} = \frac{1}{p}\,I(\hat T_{jk})\,H_{jk}.
\]
The entrywise error is controlled by the entrywise accuracy of $\hat T$: under the bounded-signal assumption $\|T^\star\|_\infty \le B$ and Lipschitz continuity of $I(\cdot)$,
\[
|\hat\lambda_{jk} - \lambda_{jk}| = \frac{1}{p}|I(\hat T_{jk}) - I(T^\star_{jk})| \le \frac{L_I}{p}\,|\hat T_{jk} - T^\star_{jk}|,
\]
where $L_I$ is the Lipschitz constant of $I$ on $[-B-1, B+1]$.
Thus the entrywise guarantee $\|\hat T - T^\star\|_\infty = o(1)$ (already required for the one-step estimator) ensures consistency of $\hat{G}$.

The efficient one-step estimator for 1-bit completion is
\begin{equation}\label{eq:tc-eff-1bit}
\hat\psi_{\rm eff} = \langle\Gamma,\hat T\rangle + \frac{1}{n}\sum_{i=1}^n s(Y_i, \langle\hat T, X_i\rangle)\,\langle\hat H_{\rm eff}, X_i\rangle,
\qquad
\hat H_{\rm eff} = (\hat P_{\mathbb{T}} \hat{G} \hat P_{\mathbb{T}})^{-1} \hat P_{\mathbb{T}} \Gamma.
\end{equation}
The remaining analysis follows the proof of Theorem~\ref{thm:eff-clt} verbatim: the additional error from $\hat{G} \ne G$ is handled by the resolvent perturbation argument (Appendix~\ref{app:pd-hatG}), and the CLT holds under the same sample-complexity conditions.

\paragraph{Asymptotic variance.}
The efficient variance is
\[
V_{\mathrm{eff}} = \bigl\langle P_{\mathbb{T}}(\Gamma),\,A^{-1} P_{\mathbb{T}}(\Gamma)\bigr\rangle_F
= p\sum_{j,k} \frac{[P_{\mathbb{T}}(\Gamma)]_{jk}^2}{I(T^\star_{jk})}\cdot\frac{[A^{-1} P_{\mathbb{T}}(\Gamma)]_{jk}}{[P_{\mathbb{T}}(\Gamma)]_{jk}/I(T^\star_{jk})},
\]
which is strictly smaller than the whitened variance $V_{\rm ws} = p \sum_{jk} [P_{\mathbb{T}}(\Gamma)]_{jk}^2 / I(T^\star_{jk})$ whenever the Fisher information $I(T^\star_{jk})$ varies across entries.
Intuitively, the efficient estimator optimally reweights contributions from each entry according to both the tangent-space geometry and the local Fisher information, whereas the whitening approach treats all entries equally after normalization.

\paragraph{Comparison with whitening.}
The whitened estimator for 1-bit observations (Section~\ref{sec:tc-bernoulli}) uses $\tilde s(y,\eta) = s(y,\eta)/I(\eta)$ and the simplified direction $\hat H_{\rm ws} = p\,\hat P_{\mathbb{T}}\Gamma$.
It achieves $V_{\rm ws} = p\sum_{jk} [P_{\mathbb{T}}(\Gamma)]_{jk}^2/I(T^\star_{jk}) \ge V_{\mathrm{eff}}$.
The gap $V_{\rm ws} - V_{\mathrm{eff}}$ grows with the heterogeneity of $I(T^\star_{jk})$ across entries.
However, the whitening approach has the practical advantage of not requiring operator inversion ($\hat P_{\mathbb{T}} \hat{G} \hat P_{\mathbb{T}}$), making it more robust in small-sample regimes.
This is the first CLT for debiased inference in 1-bit tensor completion under general link functions.

\subsection{Score whitening for tensor completion}\label{sec:tc-whitening}

When $G$ departs substantially from the identity and the sample size is not large enough for route~(b) above, the score-whitening method of Section~\ref{sec:score-whitening} provides valid inference at the optimal sample complexity, albeit with a potentially larger variance.
We treat two important observation models.

\subsubsection{Additive noise}\label{sec:tc-additive}

Suppose the observation model is $Y = T^\star_{jk} + \varepsilon$, where $\varepsilon$ has mean zero and possibly entry-dependent variance $\sigma^2_{jk}$; the noise distribution may be sub-Gaussian, sub-exponential, or otherwise---crucially, we do \emph{not} require knowledge of $\sigma^2_{jk}$ or even of the noise family.

The key observation is that the score for additive noise is simply
\[
s(y,\eta) = y - \eta,
\]
regardless of the noise distribution.
The whitened score of Section~\ref{sec:score-whitening} is $\tilde s(y,\eta) = s(y,\eta)/I(\eta)$.
For Gaussian noise with variance $\sigma^2_{jk}$, the ``true'' efficient score is $(Y-\eta)/\sigma^2_{jk}$ and the Fisher information is $I(\eta)=1/\sigma^2_{jk}$; dividing the two gives $\tilde s = Y - \eta$, i.e., the unknown variance cancels.
More generally, for \emph{any} additive noise model one can use the score $s=Y-\eta$ directly.
Since $\mathbb{E}[s^2\mid\eta] = \sigma^2_{jk}$ and $\mathbb{E}[\partial_\eta s\mid\eta] = -1$, the whitened score satisfies properties~(i)--(ii) of Section~\ref{sec:ws-estimator} with $I(\eta)\equiv 1$ in the ``natural'' parameterization.

The resulting Gram operator under score whitening and uniform sampling is
\[
G_0(H) = \frac{1}{p}\,H,
\qquad
A_0 = \frac{1}{p}\,P_{\mathbb{T}},
\qquad
A_0^{-1} = p\,P_{\mathbb{T}}.
\]
The score-whitened one-step estimator takes the form
\begin{equation}\label{eq:tc-ws-additive}
\hat\psi_{\rm ws} = \langle\Gamma,\hat T\rangle + \frac{1}{n}\sum_{i=1}^n (Y_i - \langle\hat T, X_i\rangle)\,\langle\hat H_{\rm ws}, X_i\rangle,
\qquad
\hat H_{\rm ws} = p\,\hat P_{\mathbb{T}}\Gamma.
\end{equation}
Since $I(\eta)\equiv 1$ under the natural score, the whitened estimator is \emph{fully semiparametrically efficient}: $V_{\rm ws} = V_{\rm eff}$, regardless of heteroscedasticity (Theorem~\ref{thm:ws-clt}).
Under homoscedastic noise ($\sigma^2_{jk}\equiv\sigma^2$), the variance simplifies to $V_{\rm ws} = \sigma^2 p\,\|P_{\mathbb{T}}(\Gamma)\|_F^2$.

\begin{remark}[Why the noise distribution is irrelevant]
The estimator~\eqref{eq:tc-ws-additive} uses only the residual $Y_i - \langle\hat T, X_i\rangle$ as the score, requiring no knowledge of the noise variance, distribution family, or heteroscedastic structure.
This is a distinctive advantage of the score-whitening approach for additive noise models: full semiparametric efficiency is achieved with a single, distribution-free estimator.
\end{remark}

\subsubsection{Bernoulli / 1-bit observations}\label{sec:tc-bernoulli}

In 1-bit matrix completion, $Y\in\{0,1\}$ with $\mathbb{P}(Y=1\mid\eta) = f(\eta)$ for some link function $f$ (e.g., logistic $f(\eta)=(1+e^{-\eta})^{-1}$).  The score is $s(y,\eta) = (y-f(\eta))\,f'(\eta)/\bigl(f(\eta)(1-f(\eta))\bigr)$ and the Fisher information is
\[
I(\eta) = \frac{f'(\eta)^2}{f(\eta)(1-f(\eta))}.
\]
The Fisher information varies with $\eta=T^\star_{jk}$, so the Gram operator~\eqref{eq:tc-gram} is genuinely non-isotropic.

The whitened score $\tilde s(y,\eta) = s(y,\eta)/I(\eta)$ depends on $\eta$ only, which we estimate via $\hat T$.
Since $I(\eta)$ is a known function of the estimable parameter $\eta$, the score-whitening method of Section~\ref{sec:score-whitening} applies with the same analysis.
The score-whitened one-step estimator is
\begin{equation}\label{eq:tc-ws-bernoulli}
\hat\psi_{\rm ws} = \langle\Gamma,\hat T\rangle + \frac{1}{n}\sum_{i=1}^n \tilde s(Y_i, \langle\hat T, X_i\rangle)\,\langle\hat H_{\rm ws}, X_i\rangle,
\qquad
\hat H_{\rm ws} = p\,\hat P_{\mathbb{T}}\Gamma.
\end{equation}
The asymptotic variance is
\[
V_{\rm ws} = \mathbb{E}^\star\!\left[\frac{\langle p\,P_{\mathbb{T}}\Gamma, X\rangle^2}{I(\eta^\star)}\right]
= p\sum_{j,k} \frac{p_{jk}\,\bigl[P_{\mathbb{T}}(\Gamma)\bigr]_{jk}^2}{I(T^\star_{jk})}.
\]
In general $V_{\rm ws} \ge V_{\rm eff}$, with equality when $I(\eta^\star)$ is constant across all entries (i.e., when $f$ is nearly linear or all entries of $T^\star$ are close).
The CLT and remainder bounds follow from Theorem~\ref{thm:ws-clt}.

\subsection{Non-uniform sampling}\label{sec:tc-nonuniform}

When the sampling probabilities $\{p_{jk}\}$ are non-uniform, the inverse-probability weighting (IPW) framework of Section~\ref{sec:unknown-sampling} restores isotropy.
Specifically, choosing the uniform reference distribution $q_{jk}=1/p$ and defining weights $w_{jk}=q_{jk}/p_{jk}$, the IPW-whitened estimator from~\eqref{eq:ipw-estimator} applies directly.
The effective Gram operator under the IPW-whitened score becomes $(1/p)\,\mathrm{Id}$, exactly as in the uniform case.

The factorization $\lambda_{jk} = p_{jk}\cdot I(T^\star_{jk})$ in~\eqref{eq:tc-gram} makes the decomposition transparent: \emph{score whitening} removes the Fisher information heterogeneity (Section~\ref{sec:tc-whitening}), while \emph{IPW} removes the sampling heterogeneity (Section~\ref{sec:unknown-sampling}).  The two adjustments compose naturally.

When $\{p_{jk}\}$ are unknown, a consistent estimator $\hat p_{jk}$ is needed.
As discussed in Section~\ref{sec:unknown-sampling}, if the sampling mechanism admits a low-rank or structured form, the weights can be estimated at the parametric rate with only second-order impact on the inference; see the analysis following~\eqref{eq:ipw-feasible}.

%% file: aos_LowrankLLMEvaluation.bbl
\begin{thebibliography}{29}

\bibitem[\protect\citeauthoryear{Bose et~al.}{2025}]{bose2025lore}
\begin{barticle}[author]
\bauthor{\bsnm{Bose},~\bfnm{Avinandan}\binits{A.}},
  \bauthor{\bsnm{Xiong},~\bfnm{Zhihan}\binits{Z.}},
  \bauthor{\bsnm{Chi},~\bfnm{Yuejie}\binits{Y.}},
  \bauthor{\bsnm{Du},~\bfnm{Simon~Shaolei}\binits{S.~S.}},
  \bauthor{\bsnm{Xiao},~\bfnm{Lin}\binits{L.}} \AND
  \bauthor{\bsnm{Fazel},~\bfnm{Maryam}\binits{M.}}
(\byear{2025}).
\btitle{LoRe: Personalizing LLMs via Low-Rank Reward Modeling}.
\bjournal{arXiv preprint arXiv:2504.14439}.
\end{barticle}
\endbibitem

\bibitem[\protect\citeauthoryear{Bradley and Terry}{1952}]{bradley1952rank}
\begin{barticle}[author]
\bauthor{\bsnm{Bradley},~\bfnm{Ralph~Allan}\binits{R.~A.}} \AND
  \bauthor{\bsnm{Terry},~\bfnm{Milton~E.}\binits{M.~E.}}
(\byear{1952}).
\btitle{Rank analysis of incomplete block designs: I. The method of paired
  comparisons}.
\bjournal{Biometrika}
\bvolume{39}
\bpages{324--345}.
\end{barticle}
\endbibitem

\bibitem[\protect\citeauthoryear{Cai et~al.}{2022}]{cai2022nonconvex}
\begin{barticle}[author]
\bauthor{\bsnm{Cai},~\bfnm{Changxiao}\binits{C.}},
  \bauthor{\bsnm{Li},~\bfnm{Gen}\binits{G.}},
  \bauthor{\bsnm{Poor},~\bfnm{H.~Vincent}\binits{H.~V.}} \AND
  \bauthor{\bsnm{Chen},~\bfnm{Yuxin}\binits{Y.}}
(\byear{2022}).
\btitle{Nonconvex low-rank tensor completion from noisy data}.
\bjournal{Operations Research}
\bvolume{70}
\bpages{1219--1237}.
\end{barticle}
\endbibitem

\bibitem[\protect\citeauthoryear{Cand{\`e}s and Recht}{2009}]{candes2009exact}
\begin{barticle}[author]
\bauthor{\bsnm{Cand{\`e}s},~\bfnm{Emmanuel~J.}\binits{E.~J.}} \AND
  \bauthor{\bsnm{Recht},~\bfnm{Benjamin}\binits{B.}}
(\byear{2009}).
\btitle{Exact matrix completion via convex optimization}.
\bjournal{Foundations of Computational Mathematics}
\bvolume{9}
\bpages{717--772}.
\end{barticle}
\endbibitem

\bibitem[\protect\citeauthoryear{Chao, Huang and
  Needell}{2021}]{chao2021weighted}
\begin{barticle}[author]
\bauthor{\bsnm{Chao},~\bfnm{Zehan}\binits{Z.}},
  \bauthor{\bsnm{Huang},~\bfnm{Longxiu}\binits{L.}} \AND
  \bauthor{\bsnm{Needell},~\bfnm{Deanna}\binits{D.}}
(\byear{2021}).
\btitle{HOSVD-based algorithm for weighted tensor completion}.
\bjournal{Journal of Imaging}
\bvolume{7}
\bpages{110}.
\end{barticle}
\endbibitem

\bibitem[\protect\citeauthoryear{Chiang et~al.}{2024}]{chiang2024chatbot}
\begin{barticle}[author]
\bauthor{\bsnm{Chiang},~\bfnm{Wei-Lin}\binits{W.-L.}},
  \bauthor{\bsnm{Zheng},~\bfnm{Lianmin}\binits{L.}},
  \bauthor{\bsnm{Sheng},~\bfnm{Ying}\binits{Y.}},
  \bauthor{\bsnm{Angelopoulos},~\bfnm{Anastasios~Nikolas}\binits{A.~N.}},
  \bauthor{\bsnm{Li},~\bfnm{Tianle}\binits{T.}},
  \bauthor{\bsnm{Li},~\bfnm{Dacheng}\binits{D.}},
  \bauthor{\bsnm{Zhang},~\bfnm{Hao}\binits{H.}},
  \bauthor{\bsnm{Zhu},~\bfnm{Banghua}\binits{B.}},
  \bauthor{\bsnm{Jordan},~\bfnm{Michael}\binits{M.}},
  \bauthor{\bsnm{Gonzalez},~\bfnm{Joseph~E.}\binits{J.~E.}} \AND
  \bauthor{\bsnm{Stoica},~\bfnm{Ion}\binits{I.}}
(\byear{2024}).
\btitle{Chatbot Arena: An Open Platform for Evaluating {LLMs} by Human
  Preference}.
\bjournal{arXiv preprint arXiv:2403.04132}.
\end{barticle}
\endbibitem

\bibitem[\protect\citeauthoryear{Dong et~al.}{2026}]{dong2026evaluating}
\begin{barticle}[author]
\bauthor{\bsnm{Dong},~\bfnm{Zihan}\binits{Z.}},
  \bauthor{\bsnm{Zhang},~\bfnm{Zhixian}\binits{Z.}},
  \bauthor{\bsnm{Zhou},~\bfnm{Yang}\binits{Y.}},
  \bauthor{\bsnm{Jin},~\bfnm{Can}\binits{C.}},
  \bauthor{\bsnm{Wu},~\bfnm{Ruijia}\binits{R.}} \AND
  \bauthor{\bsnm{Zhang},~\bfnm{Linjun}\binits{L.}}
(\byear{2026}).
\btitle{Evaluating LLMs When They Do Not Know the Answer: Statistical
  Evaluation of Mathematical Reasoning via Comparative Signals}.
\bjournal{arXiv preprint arXiv:2602.03061}.
\end{barticle}
\endbibitem

\bibitem[\protect\citeauthoryear{Duan et~al.}{2025}]{duan2025statistical}
\begin{barticle}[author]
\bauthor{\bsnm{Duan},~\bfnm{Congyuan}\binits{C.}},
  \bauthor{\bsnm{Ma},~\bfnm{Wanteng}\binits{W.}},
  \bauthor{\bsnm{Xia},~\bfnm{Dong}\binits{D.}} \AND
  \bauthor{\bsnm{Xu},~\bfnm{Kan}\binits{K.}}
(\byear{2025}).
\btitle{Statistical Inference for Matching Decisions via Matrix Completion
  under Dependent Missingness}.
\bjournal{arXiv preprint arXiv:2510.26478}.
\end{barticle}
\endbibitem

\bibitem[\protect\citeauthoryear{Fan, Hou and Yu}{2024}]{fan2024care}
\begin{barticle}[author]
\bauthor{\bsnm{Fan},~\bfnm{Jianqing}\binits{J.}},
  \bauthor{\bsnm{Hou},~\bfnm{Jikai}\binits{J.}} \AND
  \bauthor{\bsnm{Yu},~\bfnm{Mengxin}\binits{M.}}
(\byear{2024}).
\btitle{Uncertainty quantification of {MLE} for entity ranking with
  covariates}.
\bjournal{Journal of Machine Learning Research}
\bvolume{25}
\bpages{1--83}.
\end{barticle}
\endbibitem

\bibitem[\protect\citeauthoryear{Fan, Kwon and Zhu}{2025}]{fan2025uncertainty}
\begin{barticle}[author]
\bauthor{\bsnm{Fan},~\bfnm{Jianqing}\binits{J.}},
  \bauthor{\bsnm{Kwon},~\bfnm{Hyukjun}\binits{H.}} \AND
  \bauthor{\bsnm{Zhu},~\bfnm{Xiaonan}\binits{X.}}
(\byear{2025}).
\btitle{Uncertainty Quantification for Ranking with Heterogeneous Preferences}.
\bjournal{arXiv preprint arXiv:2509.01847}.
\end{barticle}
\endbibitem

\bibitem[\protect\citeauthoryear{Fan et~al.}{2026}]{fan2026spectral}
\begin{barticle}[author]
\bauthor{\bsnm{Fan},~\bfnm{Jianqing}\binits{J.}},
  \bauthor{\bsnm{Lou},~\bfnm{Zhipeng}\binits{Z.}},
  \bauthor{\bsnm{Wang},~\bfnm{Weichen}\binits{W.}} \AND
  \bauthor{\bsnm{Yu},~\bfnm{Mengxin}\binits{M.}}
(\byear{2026}).
\btitle{Spectral ranking inferences based on general multiway comparisons}.
\bjournal{Operations Research}
\bvolume{74}
\bpages{161--180}.
\end{barticle}
\endbibitem

\bibitem[\protect\citeauthoryear{Gao, Shen and Zhang}{2023}]{gao2023btluq}
\begin{barticle}[author]
\bauthor{\bsnm{Gao},~\bfnm{Chao}\binits{C.}},
  \bauthor{\bsnm{Shen},~\bfnm{Yandi}\binits{Y.}} \AND
  \bauthor{\bsnm{Zhang},~\bfnm{Anderson~Y.}\binits{A.~Y.}}
(\byear{2023}).
\btitle{Uncertainty quantification in the Bradley--Terry--Luce model}.
\bjournal{Information and Inference: A Journal of the IMA}
\bvolume{12}
\bpages{1073--1140}.
\end{barticle}
\endbibitem

\bibitem[\protect\citeauthoryear{Keshavan, Montanari and
  Oh}{2010}]{keshavan2010matrix}
\begin{barticle}[author]
\bauthor{\bsnm{Keshavan},~\bfnm{Raghunandan~H.}\binits{R.~H.}},
  \bauthor{\bsnm{Montanari},~\bfnm{Andrea}\binits{A.}} \AND
  \bauthor{\bsnm{Oh},~\bfnm{Sewoong}\binits{S.}}
(\byear{2010}).
\btitle{Matrix completion from a few entries}.
\bjournal{IEEE Transactions on Information Theory}
\bvolume{56}
\bpages{2980--2998}.
\end{barticle}
\endbibitem

\bibitem[\protect\citeauthoryear{Kolda and Bader}{2009}]{kolda2009tensor}
\begin{barticle}[author]
\bauthor{\bsnm{Kolda},~\bfnm{Tamara~G.}\binits{T.~G.}} \AND
  \bauthor{\bsnm{Bader},~\bfnm{Brett~W.}\binits{B.~W.}}
(\byear{2009}).
\btitle{Tensor decompositions and applications}.
\bjournal{SIAM Review}
\bvolume{51}
\bpages{455--500}.
\end{barticle}
\endbibitem

\bibitem[\protect\citeauthoryear{Koltchinskii, Lounici and
  Tsybakov}{2011}]{koltchinskii2011nuclear}
\begin{barticle}[author]
\bauthor{\bsnm{Koltchinskii},~\bfnm{Vladimir}\binits{V.}},
  \bauthor{\bsnm{Lounici},~\bfnm{Karim}\binits{K.}} \AND
  \bauthor{\bsnm{Tsybakov},~\bfnm{Alexandre~B.}\binits{A.~B.}}
(\byear{2011}).
\btitle{Nuclear-norm penalization and optimal rates for noisy low-rank matrix
  completion}.
\bjournal{The Annals of Statistics}
\bvolume{39}
\bpages{2302--2329}.
\end{barticle}
\endbibitem

\bibitem[\protect\citeauthoryear{Li et~al.}{2024}]{li2024benchbuilder}
\begin{bmisc}[author]
\bauthor{\bsnm{Li},~\bfnm{Tianle}\binits{T.}},
  \bauthor{\bsnm{Chiang},~\bfnm{Wei-Lin}\binits{W.-L.}},
  \bauthor{\bsnm{Frick},~\bfnm{Evan}\binits{E.}},
  \bauthor{\bsnm{Dunlap},~\bfnm{Lisa}\binits{L.}},
  \bauthor{\bsnm{Zhu},~\bfnm{Banghua}\binits{B.}},
  \bauthor{\bsnm{Gonzalez},~\bfnm{Joseph~E.}\binits{J.~E.}} \AND
  \bauthor{\bsnm{Stoica},~\bfnm{Ion}\binits{I.}}
(\byear{2024}).
\btitle{The Arena-Hard Pipeline}.
\bhowpublished{Arena blog}.
\bnote{Published April 19, 2024}.
\end{bmisc}
\endbibitem

\bibitem[\protect\citeauthoryear{LMSYS}{2025}]{lmarena_human_preference_140k}
\begin{bmisc}[author]
\bauthor{\bsnm{LMSYS}}
(\byear{2025}).
\btitle{arena-human-preference-140k}.
\end{bmisc}
\endbibitem

\bibitem[\protect\citeauthoryear{Luce}{1959}]{luce1959individual}
\begin{bbook}[author]
\bauthor{\bsnm{Luce},~\bfnm{R.~Duncan}\binits{R.~D.}}
(\byear{1959}).
\btitle{Individual Choice Behavior: A Theoretical Analysis}.
\bpublisher{Wiley}.
\end{bbook}
\endbibitem

\bibitem[\protect\citeauthoryear{Ma and Xia}{2024}]{ma2024statistical}
\begin{barticle}[author]
\bauthor{\bsnm{Ma},~\bfnm{Wanteng}\binits{W.}} \AND
  \bauthor{\bsnm{Xia},~\bfnm{Dong}\binits{D.}}
(\byear{2024}).
\btitle{Statistical inference in tensor completion: Optimal uncertainty
  quantification and statistical-to-computational gaps}.
\bjournal{arXiv preprint arXiv:2410.11225}.
\end{barticle}
\endbibitem

\bibitem[\protect\citeauthoryear{Mao, Chen and Wong}{2019}]{ma2019covariate}
\begin{barticle}[author]
\bauthor{\bsnm{Mao},~\bfnm{Xiaojun}\binits{X.}},
  \bauthor{\bsnm{Chen},~\bfnm{Song~Xi}\binits{S.~X.}} \AND
  \bauthor{\bsnm{Wong},~\bfnm{Raymond K.~W.}\binits{R.~K.~W.}}
(\byear{2019}).
\btitle{Matrix completion with covariate information}.
\bjournal{Journal of the American Statistical Association}
\bvolume{114}
\bpages{198--210}.
\end{barticle}
\endbibitem

\bibitem[\protect\citeauthoryear{Mao, Wang and Yang}{2023}]{yang2021matrix}
\begin{barticle}[author]
\bauthor{\bsnm{Mao},~\bfnm{Xiaojun}\binits{X.}},
  \bauthor{\bsnm{Wang},~\bfnm{Zhonglei}\binits{Z.}} \AND
  \bauthor{\bsnm{Yang},~\bfnm{Shu}\binits{S.}}
(\byear{2023}).
\btitle{Matrix completion under complex survey sampling}.
\bjournal{Annals of the Institute of Statistical Mathematics}
\bvolume{75}
\bpages{463--492}.
\end{barticle}
\endbibitem

\bibitem[\protect\citeauthoryear{Negahban and
  Wainwright}{2012}]{negahban2012restricted}
\begin{barticle}[author]
\bauthor{\bsnm{Negahban},~\bfnm{Sahand}\binits{S.}} \AND
  \bauthor{\bsnm{Wainwright},~\bfnm{Martin~J.}\binits{M.~J.}}
(\byear{2012}).
\btitle{Restricted strong convexity and weighted matrix completion: Optimal
  bounds with noise}.
\bjournal{Journal of Machine Learning Research}
\bvolume{13}
\bpages{1665--1697}.
\end{barticle}
\endbibitem

\bibitem[\protect\citeauthoryear{Ouyang et~al.}{2022}]{ouyang2022training}
\begin{barticle}[author]
\bauthor{\bsnm{Ouyang},~\bfnm{Long}\binits{L.}},
  \bauthor{\bsnm{Wu},~\bfnm{Jeffrey}\binits{J.}},
  \bauthor{\bsnm{Jiang},~\bfnm{Xu}\binits{X.}},
  \bauthor{\bsnm{Almeida},~\bfnm{Diogo}\binits{D.}},
  \bauthor{\bsnm{Wainwright},~\bfnm{Carroll}\binits{C.}},
  \bauthor{\bsnm{Mishkin},~\bfnm{Pamela}\binits{P.}},
  \bauthor{\bsnm{Zhang},~\bfnm{Chong}\binits{C.}},
  \bauthor{\bsnm{Agarwal},~\bfnm{Sandhini}\binits{S.}},
  \bauthor{\bsnm{Slama},~\bfnm{Katarina}\binits{K.}},
  \bauthor{\bsnm{Ray},~\bfnm{Alex}\binits{A.}} \betal{et~al.}
(\byear{2022}).
\btitle{Training language models to follow instructions with human feedback}.
\bjournal{Advances in neural information processing systems}
\bvolume{35}
\bpages{27730--27744}.
\end{barticle}
\endbibitem

\bibitem[\protect\citeauthoryear{Petrova, Gordon and
  Blindow}{2026}]{petrova2026unpacking}
\begin{barticle}[author]
\bauthor{\bsnm{Petrova},~\bfnm{Nora}\binits{N.}},
  \bauthor{\bsnm{Gordon},~\bfnm{Andrew}\binits{A.}} \AND
  \bauthor{\bsnm{Blindow},~\bfnm{Enzo}\binits{E.}}
(\byear{2026}).
\btitle{Unpacking Human Preference for LLMs: Demographically Aware Evaluation
  with the HUMAINE Framework}.
\bjournal{arXiv preprint arXiv:2603.04409}.
\end{barticle}
\endbibitem

\bibitem[\protect\citeauthoryear{Singh et~al.}{2025}]{singh2025leaderboard}
\begin{barticle}[author]
\bauthor{\bsnm{Singh},~\bfnm{Shivalika}\binits{S.}},
  \bauthor{\bsnm{Nan},~\bfnm{Yiyang}\binits{Y.}},
  \bauthor{\bsnm{Wang},~\bfnm{Alex}\binits{A.}},
  \bauthor{\bsnm{D'souza},~\bfnm{Daniel}\binits{D.}},
  \bauthor{\bsnm{Kapoor},~\bfnm{Sayash}\binits{S.}},
  \bauthor{\bsnm{{\"U}st{\"u}n},~\bfnm{Ahmet}\binits{A.}},
  \bauthor{\bsnm{Koyejo},~\bfnm{Sanmi}\binits{S.}},
  \bauthor{\bsnm{Deng},~\bfnm{Yuntian}\binits{Y.}},
  \bauthor{\bsnm{Longpre},~\bfnm{Shayne}\binits{S.}},
  \bauthor{\bsnm{Smith},~\bfnm{Noah~A}\binits{N.~A.}} \betal{et~al.}
(\byear{2025}).
\btitle{The leaderboard illusion}.
\bjournal{arXiv preprint arXiv:2504.20879}.
\end{barticle}
\endbibitem

\bibitem[\protect\citeauthoryear{Su}{2026}]{su2026large}
\begin{barticle}[author]
\bauthor{\bsnm{Su},~\bfnm{Weijie}\binits{W.}}
(\byear{2026}).
\btitle{Do large language models (really) need statistical foundations?}
\bjournal{The Annals of Applied Statistics}
\bvolume{20}
\bpages{724--743}.
\end{barticle}
\endbibitem

\bibitem[\protect\citeauthoryear{{Arena Team}}{2025}]{arena2025rank}
\begin{bmisc}[author]
\bauthor{\bsnm{{Arena Team}}}
(\byear{2025}).
\btitle{Arena-Rank: Open Sourcing the Leaderboard Methodology}.
\bhowpublished{Arena blog}.
\bnote{Published December 18, 2025}.
\end{bmisc}
\endbibitem

\bibitem[\protect\citeauthoryear{Xu et~al.}{2025}]{xu2025doubly}
\begin{barticle}[author]
\bauthor{\bsnm{Xu},~\bfnm{Erhan}\binits{E.}},
  \bauthor{\bsnm{Ye},~\bfnm{Kai}\binits{K.}},
  \bauthor{\bsnm{Zhou},~\bfnm{Hongyi}\binits{H.}},
  \bauthor{\bsnm{Zhu},~\bfnm{Luhan}\binits{L.}},
  \bauthor{\bsnm{Quinzan},~\bfnm{Francesco}\binits{F.}} \AND
  \bauthor{\bsnm{Shi},~\bfnm{Chengchun}\binits{C.}}
(\byear{2025}).
\btitle{Doubly robust alignment for large language models}.
\bjournal{arXiv preprint arXiv:2506.01183}.
\end{barticle}
\endbibitem

\bibitem[\protect\citeauthoryear{Zhang et~al.}{2025}]{zhang2025generalized}
\begin{barticle}[author]
\bauthor{\bsnm{Zhang},~\bfnm{Maoyu}\binits{M.}},
  \bauthor{\bsnm{Cai},~\bfnm{Biao}\binits{B.}},
  \bauthor{\bsnm{Sun},~\bfnm{Will~Wei}\binits{W.~W.}} \AND
  \bauthor{\bsnm{Zhang},~\bfnm{Jingfei}\binits{J.}}
(\byear{2025}).
\btitle{Generalized tensor completion with non-random missingness}.
\bjournal{arXiv preprint arXiv:2509.06225}.
\end{barticle}
\endbibitem

\end{thebibliography}
